\documentclass[letterpaper,11pt]{article}

\usepackage[font=small,labelfont=bf]{caption}
\usepackage[utf8]{inputenc}
\usepackage{microtype}
\usepackage{xcolor}
\usepackage{amsthm}
\usepackage{amsmath}
\usepackage{amsthm}
\usepackage{amssymb}
\usepackage{xspace}
\usepackage{graphicx}
\usepackage{fullpage}
\usepackage[colorlinks=false]{hyperref}
\usepackage{url}
\usepackage[capitalise]{cleveref}
\usepackage[shortlabels]{enumitem}
\usepackage{subfig}
\usepackage[export]{adjustbox}
\usepackage{thm-restate}

\newtheorem{theorem}{Theorem}[section]
\newtheorem{lemma}[theorem]{Lemma}
\newtheorem{proposition}[theorem]{Proposition}
\newtheorem{definition}[theorem]{Definition}
\newtheorem{corollary}[theorem]{Corollary}
\newtheorem{question}[theorem]{Question}
\newtheorem{claim}[theorem]{Claim}
\newtheorem{observation}[theorem]{Observation}

\newenvironment{innerproof}
 {\proof}
 {\endproof}

\newcommand{\defparproblem}[4]{
	\vspace{3mm}
	\noindent\fbox{
		\begin{minipage}{0.96\linewidth}
			\begin{tabular*}{\linewidth}{@{\extracolsep{\fill}}lr} \textsc{#1} & {\bf{Parameter:}} #3 \\ \end{tabular*}
			{\bf{Input:}} #2 \\
			{\bf{Task:}} #4
		\end{minipage}
	}
	\vspace{2mm}
}

\newcommand{\tw}{\mathrm{\textbf{tw}}}
\newcommand{\twsf}{\ensuremath{\mathsf{tw}}\xspace}
\newcommand{\br}[1]{\left(#1\right)}

\newcommand{\Oh}{\mathcal{O}}

\newcommand{\sub}{\subseteq}
\newcommand{\sm}{\setminus}

\newcommand{\rr}{\ensuremath{\mathbb{R}}}
\newcommand{\nn}{\ensuremath{\mathbb{N}}}
\newcommand{\pp}{\ensuremath{\mathcal{P}}}
\newcommand{\qq}{\ensuremath{\mathcal{Q}}}
\newcommand{\ff}{\ensuremath{\mathcal{F}}}
\newcommand{\bb}{\ensuremath{\mathbf{b}}}
\newcommand{\vv}{\ensuremath{\mathbf{v}}}
\newcommand{\tcal}{\ensuremath{\mathcal{T}}}
\newcommand{\scal}{\ensuremath{\mathcal{S}}}
\newcommand{\func}{\ensuremath{\gamma}}

\newcommand{\dist}{\ensuremath{\mathsf{dist}}}
\newcommand{\rdist}{\ensuremath{\mathsf{rdist}}}
\newcommand{\inter}{\ensuremath{\mathsf{int}}}
\newcommand{\disc}{\ensuremath{\mathsf{Disc}}}
\newcommand{\ring}{\ensuremath{\mathsf{Ring}}}
\newcommand{\ringfull}{\ensuremath{\mathsf{Ring}(I_{in}, I_{out})}\xspace}

\newcommand{\nullnoncross}{\textsc{Non-crossing Multicommodity Flow}\xspace}

\ifdefined\DEBUG{}
\newcommand{\mic}[1]{{\color{blue}{#1}}}

\def\rem#1{{\marginpar{\raggedright\scriptsize #1}}}
\newcommand{\micr}[1]{\rem{\textcolor{blue}{\(\bullet \) #1}}}

\newcommand{\meir}[1]{\rem{\textcolor{purple}{\(\bullet \) #1}}}

\else
\newcommand{\mic}[1]{#1}

\newcommand{\micr}[1]{ }
\newcommand{\meir}[1]{ }
\fi

\title{Planar Disjoint Paths, Treewidth, and Kernels}

\author{Micha{\l} W{\l}odarczyk\footnote{Address: \texttt{michal.wloda@gmail.com}} \\ Ben Gurion University
\and Meirav Zehavi\footnote{Address: \texttt{meiravze@bgu.ac.il}} \\ Ben Gurion University}
\date{}

\begin{document}
\pagestyle{empty}
\maketitle

\thispagestyle{empty}

\begin{abstract} 
In the {\sc Planar Disjoint Paths} problem, one is given an undirected planar graph with a set of $k$ vertex pairs $(s_i,t_i)$ and the task is to find $k$ pairwise vertex-disjoint paths such that the $i$-th path connects $s_i$ to $t_i$.
We study the problem through the lens of kernelization, aiming at efficiently reducing the input size in terms of a parameter.
We show that  {\sc Planar Disjoint Paths} does not admit a polynomial kernel when parameterized by $k$ unless coNP $\sub$ NP/poly, resolving an open problem by [Bodlaender, Thomass{\'e}, Yeo, ESA'09].
Moreover, we rule out the existence of a polynomial Turing kernel unless \mic{the} WK-hierarchy collapses.
Our reduction carries over to the setting of edge-disjoint paths, where the kernelization status remained open even in general graphs.

On the positive side, we present a polynomial kernel for  {\sc Planar Disjoint Paths} parameterized by $k + \twsf$, where $\twsf$ denotes the treewidth of the input graph.
As a consequence of both our results,
we rule out the possibility of a polynomial-time (Turing) treewidth reduction to $\twsf = k^{\Oh(1)}$ under the same assumptions.
To the best of our knowledge, this is the first hardness result of this kind.
\mic{Finally}, combining our kernel with the known techniques [Adler, Kolliopoulos, Krause, Lokshtanov, Saurabh, Thilikos, JCTB'17; Schrijver, SICOMP'94] yields an 
alternative (and arguably simpler) proof that {\sc Planar Disjoint Paths} can be solved in time $2^{\Oh(k^2)}\cdot n^{\Oh(1)}$, matching the result of [Lokshtanov, Misra, Pilipczuk, Saurabh, Zehavi, STOC'20].
\end{abstract}

\newpage


\newpage
\pagestyle{plain}
\setcounter{page}{1}

\section{Introduction}

{\sc Disjoint Paths} is a fundamental routing problem: for several decades, it has been extensively studied in a wide variety of areas in computer science and graph theory. We focus on the area of algorithm design, specifically of parameterized algorithms. Phrased as a parameterized problem, given an $n$-vertex undirected graph $G$  and a set of $k$ pairwise disjoint vertex pairs, $\{s_i, t_i\}^k_{i=1}$, the objective is to decide whether there exist $k$ pairwise vertex-disjoint paths connecting $s_i$ to $t_i$ for each $i \in \{1,\ldots, k\}$. Here, the classic parameter choice is $k$. The problem was shown to be NP-hard by Karp (attributing it to Knuth) in 1975~\cite{karp1975computational}, in a follow-up paper to his classic list of 21 NP-complete problems~\cite{DBLP:conf/coco/Karp72}. Since then, the problem was shown to be NP-hard on various simple graph classes~\cite{heggernes2015finding,lynch1975equivalence,natarajan1996disjoint} 
\mic{including} the class of grid graphs~\cite{kramer1984complexity}, a highly restricted subclass of planar graphs. Notably, {\sc Disjoint Paths} is a cornerstone for the widely celebrated graph minors project of Robertson and Seymour, considered to be one of the greatest feats of modern mathematics (see Section~1.5). Moreover, {\sc Disjoint Paths} finds  applications in various practical fields such as VLSI layout and virtual circuit routing~\cite{frank1990packing,ogier1993distributed,schrijver2003combinatorial,srinivas2005finding}.

Due to its computational hardness, {\sc Disjoint Paths} was studied from the perspectives of parameterized complexity and approximation algorithms. In particular, {\sc Disjoint Paths} was shown to be in FPT (that is, solvable in time $f(k)\cdot n^{O(1)}$ for some computational function $f$ of $k$) in 1995 as part of the graph minors project~\cite{robertson1995graph}, being one of the first problems classified in FPT. Here, the polynomial is $n^3$. In 2012, the polynomial was improved to $n^2$~\cite{kawarabayashi2012disjoint}. Unfortunately, the dependency on $k$ in both algorithms is ``galactic''~\cite{lipton2013people,johnson1985np}, being a tower of exponents. Concerning approximation algorithms, the state-of-the-art is  grim as well: despite substantial efforts, the currently best-known approximation algorithm is still a simple greedy one that achieves a ratio of~$O(\sqrt{n})$~\cite{kolliopoulos2004approximating}.

We focus on {\sc Planar Disjoint Paths}, the most well-studied special case of {\sc Disjoint Paths}, where the input graph is restricted to be planar.  Understanding this special case is critical for algorithms for {\sc Disjoint Paths}, {\sc Minor Testing} and {\sc Topological Minor Testing} on general graphs (see Section~1.5). Moreover, it finds most of the general case's applications. Fortunately, this special case is  known to be more tractable than the general one. 
Already in the 90s, {\sc Disjoint Paths} on planar~\cite{reed1995rooted,reed1993finding} and bounded genus graphs~\cite{reed1995rooted,dvovrak2009coloring,kobayashi2009algorithms} were shown to admit algorithms with running times whose dependency on $n$ in linear.
Regarding the dependency on $k$, \mic{the state-of-the-art for {\sc Planar Disjoint Paths} is} 
$2^{O(k^2)}\cdot n^{O(1)}$~\cite{LokshtanovMPSZ20}, improving upon earlier works \cite{AdlerKKLST17, reed1995rooted}.  Very recently, the dependency on $n$ was improved to be linear without compromising this dependency on~$k$~\cite{cho2023parameterized}. It is also noteworthy  that when extended to directed graphs, {\sc Planar Disjoint Paths} is in FPT~\cite{DBLP:conf/focs/CyganMPP13} (and, for three decades, already known to be in XP~\cite{schrijver1994finding}), while {\sc Disjoint Paths} is NP-hard \mic{already} when $k=2$~\cite{fortune1980directed}.  {\sc Planar Disjoint Paths} has also been intensively studied from the perspective of approximation algorithms, with a burst of activity in recent years. Some of the highlights of this line of works include a polynomial-time approximation algorithm with a factor of $n^{9/19} \log^{O(1)} n$~\cite{chuzhoy2016improved}, and, under reasonable complexity-theoretic assumptions, the proof of hardness of polynomial-time approximation within a factor of $2^{o(\sqrt{\log n})}$~\cite{chuzhoy2017new}.

\paragraph{1.1. Our focus: Kernelization of planar disjoint paths.}
From the perspective of parameterized complexity, the (arguably) biggest open question that remains regarding {\sc Planar Disjoint Paths} is whether it admits a polynomial kernel. Kernelization is \mic{a} mathematical paradigm for \mic{the} analysis of preprocessing procedures~\cite{fomin2019kernelization}. Due to the profound impact of preprocessing, kernelization has been termed ``the lost continent of polynomial time''~\cite{fellows2006lost}. Formally, a parameterized problem $\Pi$ admits a {\em kernel} if there is a polynomial-time algorithm (called a kernelization algorithm) that, given an instance $(I,k)$ of $\Pi$, translates it into an equivalent instance $(I',k')$ of $\Pi$ of size $f(k)$ for some computable function $f$ depending only on $k$. (Equivalence means that $(I, k)$ is a yes-instance if and only if $(I',k')$ is a yes-instance.) A (decidable) problem admits a kernel if and only if it is in FPT~\cite{DBLP:journals/apal/CaiCDF97}. So, the central question in kernelization is: Which problems admit kernels of size $f(k)$ where $f$ is polynomial in $k$, called {\em polynomial kernels}. Originally in 2009, {\sc Disjoint Paths} was shown not to admit a polynomial kernel with respect to $k$ unless coNP $\subseteq$ NP/poly~\cite{DBLP:conf/esa/BodlaenderTY09,bodlaender2011kernel}, being one of the first problems for which such a result was proved. In the same paper, it was already asked whether {\sc Planar Disjoint Paths} \mic{has} a polynomial kernel. 
Still, up until this paper, it was not even known whether \mic{its extension 
to directed planar} graphs admits a polynomial kernel. 

Remarkably, the literature abounds with problems that do not admit polynomial kernels on general graphs unless coNP $\subseteq$ NP/poly, but admit polynomial kernels on planar graphs~\cite{fomin2019kernelization}. What is more, many of them are W[1]-hard\footnote{A W[1]-hard problem is unlikely to be in FPT~\cite{cygan2015parameterized}.} or even W[2]-hard on general graphs, while the sizes of their polynomial kernels on planar graphs are, in fact, linear; {\sc Dominating Set} is a prime example for this phenomenon. Today, we have very general techniques to design such kernels on planar graphs~\cite{DBLP:journals/siamcomp/FominLST20,fomin2019kernelization}
\mic{and} there exist only few\footnote{We are aware of one example: for {\sc Steiner Tree} on planar graphs parameterized by the number of terminals, the unlikely existence of a polynomial kernel is implied by the combination of the lower and upper bounds given in~\cite{MarxPP18}.} natural problems that are in FPT on planar graphs, but 
have non-trivial \mic{kernelization lower bounds.}
By non-trivial we mean that the proof for planar graphs is not essentially the same as for general~graphs.


In this paper, we  decipher the complexity of preprocessing procedures (kernels and treewidth reductions) for {\sc Planar Disjoint Paths}. 
\mic{Below,} we present our main theorems and their implications. \mic{Next, we discuss} the role of our work in the efforts of making the graph minors theory~efficient.

\paragraph{\bf 1.2. On the negative side: Our first main theorem.} First, we resolve the almost decade-and-a-half open question of whether {\sc Planar Disjoint Paths} admits a polynomial kernel with respect to $k$: unless the polynomial hierarchy collapses, the answer~is~negative. 

\begin{theorem}[{\bf Main Theorem I}]\label{thm:noPolyKer}
Unless coNP $\subseteq$ NP/poly, {\sc Planar Disjoint Paths} does not admit a polynomial kernel with respect to $k$. 
\end{theorem}

Our reduction also shows that {\sc Planar Disjoint Paths} is WK[1]-hard\footnote{The WK-hierarchy organizes parameterized problems with respect to polynomial parameter transformations. WK[1]-complete problems include {\sc Set Cover} (par. by the universe size),  {\sc Connected Vertex Cover} (par. by the solution size), {\sc {Min Ones 3-SAT}} (par. by the number of 1's in an assignment), or {\sc {Binary Nondeterministic Turing Machine Halting}} (par. by the number of steps).}, which means (see~\cite{hermelin2015completeness}) that it is unlikely 
even to admit a weaker form of a preprocessing procedure called a {\em polynomial Turing kernel}. Formally, a parameterized problem $\Pi$ admits a {\em Turing kernel} if 
there exists a polynomial-time algorithm for $\Pi$ \mic{using an oracle that solves instances of $\Pi$ of size at most $f(k)$ for some computable function $f$}.
Similarly to standard kernelization, polynomial Turing \mic{kernel} refers to the case where $f$ is polynomial in $k$. Note that a kernel is a~special case of a Turing kernel where the algorithm can perform exactly one call to the oracle. To date, we know of many problems that admit a polynomial \mic{Turing} kernel but are unlikely to admit a polynomial kernel. 
\mic{This is the case} for the other most famous path problem in parameterized complexity, called {\sc $k$-Path} (determine whether a given undirected graph contains a path on $k$ vertices): while {\sc $k$-Path} is unlikely to admit a polynomial kernel when restricted to planar graphs (which can be shown by a trivial OR-composition~\cite{fomin2019kernelization}), it does admit a polynomial Turing kernel 
\mic{on} planar graphs~\cite{DBLP:journals/jcss/Jansen17} or even \mic{on} topological-minor-free graphs~\cite{DBLP:journals/algorithmica/JansenPW19}. In light of this result, we find Theorem~\ref{thm:MK2} quite  surprising.

\begin{theorem}\label{thm:MK2}
{\sc Planar Disjoint Paths} is WK[1]-hard.
\end{theorem}

Additionally, our reduction carries over to {\sc Planar Edge-Disjoint Paths}, 
where the solution paths are required to be edge-disjoint rather than vertex-disjoint. 
Specifically, we show that it is also unlikely 
to admit a polynomial kernel (or even a polynomial Turing kernel). Remarkably, prior to our work, it was not even known whether the problem admits a polynomial kernel on general graphs, although it was already asked as an open question close to a decade ago~\cite{bodlaender2014graph,heggernes2015finding}. 

\begin{theorem}\label{thm:noPolyKerEdge}
Unless coNP $\subseteq$ NP/poly, {\sc Planar Edge-Disjoint Paths} does not admit a polynomial kernel with respect to $k$. Moreover, {\sc Planar Edge-Disjoint Paths} is WK[1]-hard.
\end{theorem}

The {\sc Edge-Disjoint Paths} problem in general, and the {\sc Planar Edge-Disjoint Paths} problem in particular, have been intensively studied in the literature (see, e.g., \cite{andrews2005hardness,chekuri2004edge,chekuri2006edge,frank1985edge,okamura1981multicommodity,kawarabayashi2012disjoint,nishizeki2001edge,heggernes2015finding}), although perhaps to a lesser extent than their vertex counterparts. We remark that the vertex and edge versions sometimes behave very differently---for example, while {\sc Disjoint Paths} is in FPT with respect to treewidth~\cite{scheffler1994practical}, {\sc Edge-Disjoint Paths} is NP-complete even on series-parallel
graphs~\cite{nishizeki2001edge} and thus on graphs of treewidth at most 2. Still, the work of Robertson and Seymour implies that {\sc Edge-Disjoint Paths} is solvable in time $f(k)\cdot n^3$; later,  the polynomial factor was reduced to $n^2$ by Kawarabayashi et al.~\cite{kawarabayashi2012disjoint}.

\paragraph{1.3. On the positive side: Our second main theorem.} We prove that {\sc Planar Disjoint Paths} admits a polynomial kernel with respect to $k+\mathsf{tw}$, where $\mathsf{tw}$ is the treewidth of the input graph.\footnote{A trivial AND-composition~\cite{fomin2019kernelization} implies that parameterization by $\mathsf{tw}$ alone is unlikely to yield a polynomial kernel.}
This  theorem is (arguably) the broadest and the most involved positive result known to date regarding the kernelization complexity of {\sc Disjoint Paths}; other  results in the literature concern highly restricted graph classes: split graphs~\cite{heggernes2015finding,yang2018kernelization} and well-partitioned chordal graphs~\cite{ahn2020well}.

\begin{theorem}[{\bf Main Theorem II}]\label{thm:polyKer}
{\sc Planar Disjoint Paths} admits a polynomial kernel with respect to $k+\mathsf{tw}$, where $\mathsf{tw}$ is the treewidth of the input graph.
\end{theorem}

The interest in the parameterization of {\sc Disjoint Paths} and {\sc Planar Disjoint Paths} by $\mathsf{tw}$ stems, mainly, from the fact that all known algorithms for these problems as well as for 
\mic{{\sc (Topological) Minor Testing}} rely on {\em treewidth reduction} (defined below), and, in particular, require the resolution of these problems when $\mathsf{tw}$ is small as part of their execution (see Section~1.5). In fact, some of the running times are stated as a function of $k+\mathsf{tw}$ rather than $k$ alone (e.g., the algorithm of~\cite{LokshtanovMPSZ20} is stated to run in time $\mathsf{tw}^{O(k)}\cdot n^{O(1)}$). 
Moreover, 
\mic{treewidth} is the most well-studied structural parameter is parameterized complexity~\cite{cygan2015parameterized,DowneyF13}. It is known that {\sc Disjoint Paths} parameterized by $\mathsf{tw}$ is solvable in time $2^{O(\mathsf{tw}\log \mathsf{tw})}\cdot n$~\cite{scheffler1994practical}, while, under the Exponential Time Hypothesis (ETH), {\sc Disjoint Paths} and {\sc Planar Disjoint Paths} cannot be solved in times $2^{o(\mathsf{tw}\log \mathsf{tw})}\cdot n^{O(1)}$~\cite{lokshtanov2018slightly} and $2^{o(\mathsf{tw})}\cdot n^{O(1)}$~\cite{baste2015role}, respectively.

A treewidth reduction for a parameterized graph problem $\Pi$ is a polynomial-time algorithm that, given an instance $(I,k)$ of $\Pi$, translates it into an equivalent instance of $\Pi$ where the treewidth of the new graph is bounded by $f(k)$ for some computational function $f$ of $k$. For {\sc Disjoint Paths}, unfortunately, the best-known function is a tower of exponents~\cite{robertson1995graph,kawarabayashi2012disjoint}. However, for {\sc Planar Disjoint Paths}, $f(k)=2^{O(k)}$~\cite{AdlerKKLST17}.
Thus, since \mic{the problem} 
is solvable in time $n^{O(k)}$~\cite{schrijver1994finding}, Theorem~\ref{thm:polyKer} yields a $2^{O(k^2)}\cdot n^{O(1)}$-time algorithm: 
Reduce the treewidth of the graph to $2^{O(k)}$ in polynomial time~\cite{cho2023parameterized},
then run our kernelization algorithm in polynomial time, \mic{obtaining} an equivalent instance 
\mic{with} $2^{O(k)}$ vertices, and lastly solve the new instance  in time $2^{O(k^2)}$ using the $n^{O(k)}$-time algorithm.
This provides an alternative (and much shorter) proof of the result of~\cite{LokshtanovMPSZ20}. Unlike~\cite{cho2023parameterized, LokshtanovMPSZ20}, we use the algorithm of~\cite{schrijver1994finding} as a black box, 
\mic{so any improvement upon it (i.e., an~$n^{o(k)}$-time algorithm) would immediately entail an improvement also in the FPT running time.}

\begin{theorem}
The algorithm  of Schrijver~\cite{schrijver1994finding}  can be used in a black-box manner to solve {\sc Planar Disjoint Paths} in time $2^{O(k^2)}\cdot n^{O(1)}$.
\end{theorem}

\paragraph{1.4. Implication for treewidth reductions.} A remarkable  corollary of the combination of Theorems~\ref{thm:noPolyKer} and~\ref{thm:polyKer} rules out the existence of a polynomial treewidth reduction: If there existed a polynomial treewidth reduction for {\sc Planar Disjoint Paths} with respect to $k$, then combined with Theorem~\ref{thm:polyKer}, this would have yielded a polynomial kernel 
with respect to $k$, contradicting Theorem~\ref{thm:noPolyKer}.  This result can be viewed as a significant strengthening of Theorem~\ref{thm:noPolyKer}: not only we cannot efficiently  preprocess the graph so that its size will be polynomial in $k$, but we even cannot preprocess it so that only its treewidth will be polynomial in $k$.

\begin{theorem}\label{thm:noTwRed}
Unless coNP $\subseteq$ NP/poly, {\sc Planar Disjoint Paths} does not admit a polynomial treewidth reduction with respect to $k$.
\end{theorem}

To the best of our knowledge, this is the first non-trivial result of this form,\footnote{Here, by non-trivial, we mean that the result does not follow simply because the problem does not admit any treewidth reduction (e.g., since it is in FPT with respect to $\mathsf{tw}$ but it is not in FPT with respect to the parameter under consideration).}\meir{Changed the footnote. Old footnote commented.}
although treewidth reduction is a common tool in parameterized complexity, particularly since it is tightly linked to the irrelevant vertex technique as well as to Bidimensionality theory~\cite{cygan2015parameterized}. We refer to \cite{DBLP:conf/soda/LokshtanovP0SZ18,marx2013finding,jansen2014near,DBLP:journals/algorithmica/MarxS12,golovach2013obtaining,DBLP:conf/stoc/GroheKMW11,DBLP:conf/stoc/FominLP0Z20,DBLP:journals/algorithmica/Marx10,DBLP:journals/jcss/Grohe04,DBLP:journals/jacm/DemaineFHT05,DBLP:journals/cj/DemaineH08,DBLP:journals/csr/DornFT08} for a few illustrative examples of treewidth reductions for problems other than {\sc Disjoint Paths} and {\sc Minor Testing}. Prior to our work, there was hope that {\sc Planar Disjoint Paths} would admit a polynomial treewidth reduction with respect to $k$---notably, coupled with the $\mathsf{tw}^{O(k)}$-time algorithm of~\cite{LokshtanovMPSZ20}, this would have yielded a $2^{O(k\log k)}\cdot n^{O(1)}$-time algorithm. 

A negative hint was given by Adler et al.~\cite{AdlerLB}, who constructed yes-instances of {\sc Planar Disjoint Paths} where the treewidth of the graph is $2^{\Omega(k)}$ and every vertex is {\em relevant}, that is, the removal of any vertex would turn the instance into a no-instance. This is indeed a negative hint since all known algorithms for 
\mic{{\sc (Planar) Disjoint Paths}} apply treewidth reduction by iteratively finding and removing irrelevant vertices until the treewidth of the graph becomes ``small enough''. However, the result of~\cite{AdlerLB} does not imply that  {\sc Planar Disjoint Paths} does not admit a polynomial treewidth reduction---indeed, it {\em provably cannot} even show that the removal of irrelevant edges rather than irrelevant vertices is futile. 
To see this, consider any yes-instance 
and some solution of it, and remove from the graph all \mic{the} edges that are not part of \mic{the} solution (which are irrelevant edges)---then we are left with a collection of paths, having treewidth $1$. Our result rules out not just the success of removal of irrelevant edges, but the success of any method implementing a polynomial treewidth reduction for {\sc Planar Disjoint Paths}. In fact, since we show that 
\mic{the problem} is WK[1]-hard, our result even rules out a ``Turing-version''  of a polynomial treewidth reduction (defined in the natural way), 
strengthening Theorem~\ref{thm:MK2}.

\begin{theorem}
Unless the WK-hierarchy collapses, {\sc Planar Disjoint Paths} does not admit a polynomial Turing treewidth reduction with respect to $k$.
\end{theorem}

\paragraph{1.5. Part of the development of an efficient graph minors theory.}
The concept of a {\em minor} has been extensively studied already in the early 20th century, and it is defined as follows: A~graph $H$ is a {\em minor} of a graph $G$ if $H$ can be obtained from $G$ by deleting vertices and edges, and contracting edges. Kuratowski’s famous theorem states that a graph is planar if and only if it does not contain the graphs $K_{3,3}$ and $K_5$ as topological minors~\cite{kuratowski1930probleme}, which holds also for minors~\cite{wagner1937eigenschaft}. Thus, the class of planar graphs is characterized by a set of two forbidden minors. The graph minors project of Robertson and Seymour is a series of 23 papers spanning more than two decades, dedicated to proving the generalization of Kuratowski’s theorem called Wagner’s conjecture~\cite{wagner1937eigenschaft}: The class of all graphs is well-quasi ordered by the minor relation, or, equivalently, any minor-closed family of graphs can be characterized by a finite set of forbidden minors. The graph minors project had tremendous impact on various areas of theoretical computer science and on graph theory, particularly due to numerous concepts, structural results, and algorithms that it yielded. Notably, it is considered to be the  origin of the field of parameterized complexity~\cite{downey2012birth} and the source for a large number of its most central notions and techniques~\cite{DBLP:conf/birthday/Lokshtanov0Z20}.

Unfortunately, the dependencies on the parameters of most algorithms based on the graph minors project are huge, being towers of exponents, and so they are called ``galactic algorithms''~\cite{lipton2013people,johnson1985np}. As written in~\cite{DBLP:conf/birthday/Lokshtanov0Z20}: ``keeping in mind that the primary objective of the paradigm of parameterized complexity is to cope with computational intractability, we are facing a blatant discrepancy.'' Thus, \mic{the} holy grail of parameterized complexity is to amend this discrepancy, making the graph minors theory efficient. While substantial efforts have been made in this direction (e.g., see~\cite{grohe2013simple,kawarabayashi2012disjoint,kawarabayashi2010shorter,chekuri2016polynomial,chuzhoy2014improved,DBLP:journals/jctb/ChuzhoyT21}), it will likely take a 
\mic{long time for the matter to be well understood.} 
In particular, two central concrete goals posed for this purpose are to solve {\sc Minor Testing} and {\sc Disjoint Paths} efficiently~\cite{DBLP:conf/birthday/Lokshtanov0Z20}.

All known algorithms for 
\mic{{\sc (Topological) Minor Testing}} and {\sc Disjoint Paths} use the following case distinction. First, if the treewidth of the graph is ``small'', then they directly solve the problem using dynamic programming (classically) or a different mean~\cite{cho2023parameterized, LokshtanovMPSZ20}. Else, if the graph contains a large clique as a minor, then they use rerouting arguments to find an irrelevant vertex within it. Lastly, we are left with the case where the treewidth is large and the graph does not have a large clique as a minor.
\mic{This reduces the problem to almost-embeddable graphs~\cite{GM17}, where}
a~so-called flat wall theorem \mic{ensures the existence of} a large almost planar piece within the graph, which is afterwards analyzed.
\mic{Hence understanding the planar case is paramount to understand the problem in general. 
Our} contributions can be viewed as a piece of the ongoing efforts of many researchers to \mic{establish} which parts of the graph minor theory can be made algorithmically efficient.

\paragraph{1.6. Organization.} First, in Section~\ref{sec:outline}, we outline the proofs of our results. Afterwards, in Section~\ref{sec:prelims}, we present the more basic preliminaries required for our work. Then, in Sections~\ref{sec:polyKer} and \ref{sec:hardness}, we provide the full details of the proofs of our positive and negative results, respectively.
Finally, in Section~\ref{sec:conclusion}, we conclude the paper with \mic{several} open questions.
{Whenever we want to emphasize the importance of a statement (e.g., when it is a building block of the main proof), we use the ``proposition'' environment, instead of ``lemma''.}

\section{Outline}\label{sec:outline}
In this section, we give an informal overview of our technical contributions,
beginning from the positive result, which requires fewer intermediate steps.

\subsection{Polynomial kernel for parameter $k+\mathsf{tw}$} 
Our kernelization algorithm is based on several steps that reduce the size of certain subgraphs of $G$ while treating their boundary vertices as temporary terminals.
Since we have no control over which pairs of these terminals might be connected by paths in a solution, it is convenient to work in a slightly more general setting. 
We say that two graphs $G_1, G_2$ sharing a set of vertices $X$ are {\em $X$-linkage-equivalent} if for every set of pairs $\tcal \sub X^2$, the instances $(G_1,\tcal)$ and $(G_2,\tcal)$ of {\sc Disjoint Paths} are equivalent.
In fact, we prove a theorem that is more general than \cref{thm:polyKer} as we do not need to know in advance which pairs of terminals should be connected.

\begin{restatable}{theorem}{polyKer}
\label{thm:outline:polyKer}
Let $G$ be a planar graph of treewidth $\twsf$ and $X \sub V(G)$ be of size $k$. 
Then we can construct, in polynomial time, a planar graph $G'$ with $X \sub V(G')$ such that $|V(G')| = \Oh(k^{12}\twsf^{12})$ and $G'$ is $X$-linkage-equivalent to $G$.
\end{restatable}

\paragraph{Single-face case.}
The problem becomes simpler when we are equipped with an embedding of $G$ with all the terminals from $X$ lying on a single face (we can assume that this is the outer face by flipping the embedding).
In fact, in this case {\sc Disjoint Paths} is solvable in polynomial time~\cite{GM6}, similarly as {\sc Steiner Tree}~\cite{EricksonMV87}.
To design a useful subroutine, we need to reduce the size of~$G$ to be polynomial in $|X|$ while maintaining $X$-linkage-equivalency.
To this end, we take advantage of the criterion by Robertson and Seymour~\cite{GM6}, stating that when $X$ lies on the outer face of $G$ and $\tcal \sub X^2$, then the instance $(G,\tcal)$ is solvable if and only if (1) $\tcal$ is cross-free\footnote{
When $X$ lies on the outer face, then $\tcal \sub X^2$ is called {\em cross-free} if there are no pairs $(a,b), (c,d) \in \tcal$ such that $a,c,b,d$ lie in this order on the outer face.
} with respect to the cyclic ordering of $X$ and (2) for every partition of $X$ into continuous segments $(X_1,X_2)$ the number of requested paths with one endpoint in $X_1$ and the other one in $X_2$ is not greater than the minimum vertex $(X_1,X_2)$-cut (see \Cref{fig:split-outline}).
So, to compress $G$ we need a {\em mimicking network} that preserves such minimum cuts.
There are known constructions of mimicking networks for planar~\cite{GoranciHP20, KrauthgamerR20} and general~\cite{KratschW14} graphs, but they are designed to preserve edge-cuts.
We give a self-contained  construction of
a mimicking network of size $\Oh(|X|^6)$ preserving the necessary vertex-cuts.

\begin{figure}[t]
    \centering
\includegraphics[scale=0.8]{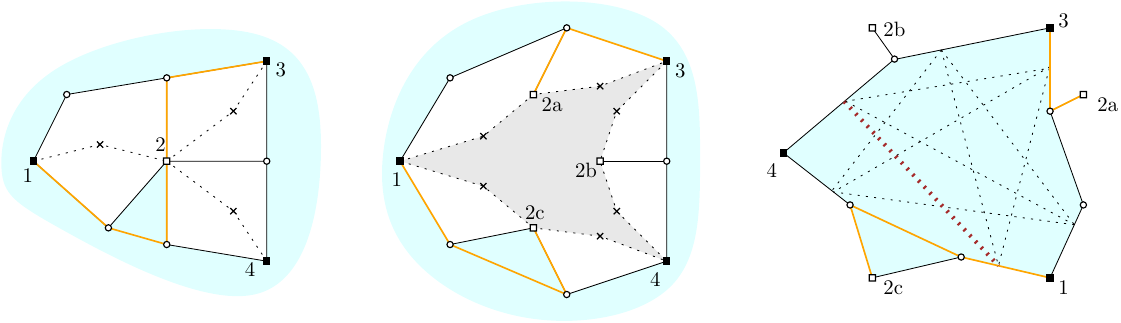}
\caption{A visualization of cutting  a graph open alongside a tree $T$ in the radial graph.
Left: The tree $T$ is sketched with dotted lines; we have $V(T) \cap V(G) = \{1,2,3,4\}$.
The vertices $1,3,4$ (black squares) belong to $X$.
A $(1,3)$-path $P$ is drawn with solid orange lines.
Middle: After opening the graph, the vertex 2 is split into three copies $2a, 2b, 2c$, while the path $P$ is split into a $(1, 2c)$-path and a $(2a, 3)$-path.
The white and black squares are the new terminals.
Right: After flipping the embedding, we can assume that the face where the cut happened is the outer face.
The dotted lines illustrate all the cuts that must be preserved in a mimicking network. \mic{As an example, the brown heavier line stands for the minimum cut between $\{4,2c\}$ and $\{1,2a,3,2b\}$}.
} 
\label{fig:split-outline}
\end{figure}

To reduce the general case to the single-face case, we follow the idea from the kernelization algorithm for \mic{{\sc Planar Steiner Tree}} 
\mic{parameterized by the solution size}~\cite{PilipczukPSvL18} (cf. \cite{BorradaileKM09}).
To~adapt it to our setting, we consider the {\em radial graph} of the plane graph $G$, 
obtained by inserting a vertex inside each face, connecting it to all vertices from $G$ lying on this face, and removing original edges from $E(G)$.
Let $T$ be a tree in the radial graph that spans all the terminals from $X$.
Imagine ``cutting the graph open'' alongside $T$ by widening the fissure marked by $T$ on the plane and duplicating the vertices of $G$ lying on this fissure (see \Cref{fig:split-outline}).
This operation creates a new face, incident to all the copies of terminals.
We could compress the obtained graph using the approach outlined above and afterwards stitch it back alongside $T$.
However, now we need to treat all the vertices lying on $T$ (not only the ones from $X$) as terminals in order to keep track of paths that might traverse $T$.
This means that, prior to opening the graph, we need to ensure that the tree $T$ is not too large, that is, the vertices from $X$ are close to each other in the radial graph.
In other words, we need to reduce the {\em radial diameter} of the graph, i.e., the maximum number of faces one must cross to reach a certain vertex from another one.
Such an approach has been applied in the reductions to a~single-face case for {\sc Vertex Multiway Cut}~\cite{JansenPvL19} and {\sc Vertex Planarization}~\cite{JansenW22}.

\paragraph{Radial diameter reduction.}
The radial diameter of a plane graph $G$ is proportional to the maximal number of concentric cycles in $G$ (i.e., these cycles are vertex-disjoint and each one is located in the \mic{interior of the next one}, resembling a well).
In particular, when the radial diameter is as large as $\Omega(k\cdot \mathsf{tw}^2)$, then one can find a sequence
$C_1, C_2, \dots, C_m$ of concentric cycles such that $m = \Omega(\mathsf{tw}^2)$ and each terminal from $X$ is located in either the interior of $C_1$ or the exterior of $C_m$.
We show that in this case, $G$ must contain an {\em irrelevant edge}, that is, an edge $e$ for which the graph $G\sm e$ is $X$-linkage-equivalent to $G$.
Our strategy is to iteratively remove irrelevant edges from $G$ until its radial diameter becomes bounded by $O(k\cdot \mathsf{tw}^2)$.
Afterwards, we will be able to find a Steiner tree $T$ of $X$ of size $O(k^2\cdot \mathsf{tw}^2)$ in the radial graph.
This will allow us to reduce the problem to the single-face case with $O(k^2\cdot \mathsf{tw}^2)$ terminals.

Consider a sequence of concentric cycles $C_1, C_2, \dots, C_t$.
It is known~\cite{GuT12} that the existence of $t$ vertex-disjoint paths between $V(C_1)$ and $V(C_t)$ yields a minor model of a $t\times t$-grid, thus implying that the treewidth of the graph is at least $t$.
Conversely, if we know that treewidth is less than $t$, Menger's theorem implies that there is a vertex $(V(C_1), V(C_t))$-separator of size less than $t$.
We can always find such a separator located within the well, that is, between $C_1$ and $C_t$ (inclusively).
In our setting, this implies that we can find a $(V(C_1), V(C_m))$-separator of size at most \twsf located inside the cycle $C_{\twsf+1}$, and similarly, one located outside $C_{m-\twsf-1}$.
Therefore, the search of an irrelevant edge can be reduced to the two-face case: we are given a plane graph $G$ with a set of terminals $V_{out}$ located on the outer face, another set of terminals $V_{in}$ located on some internal face, such that $|V_{in}|, |V_{out}| \le \twsf$, and a sequence of  $\Omega(\mathsf{tw}^2)$ concentric $C_1, C_2, \dots, C_m$ cycles around $V_{in}$ (see \Cref{fig:irrelevant-edge-outline}).
Now the task is to find an edge $e \in E(G)$ such that $G \sm e$ and $G$ are $(V_{in} \cup V_{out})$-linkage-equivalent.

We begin with computing the minimum $(V_{in}, V_{out})$-separator $S_{in}$ that is {\em closest} to $V_{in}$.
By ``closest to $V_{in}$'' we mean that for any other minimum separator $S$ the set of vertices reachable from $V_{in}$ in $G-S$ is a superset of those reachable from $V_{in}$ in $G-S_{in}$.
It is well-known~\cite[Thm. 8.4]{cygan2015parameterized} that such a separator exists.
Similarly, we compute  the minimum $(V_{in}, V_{out})$-separator $S_{out}$ closest to $V_{out}$.
Every inclusion-minimal vertex separator $S$ in a plane graph $G$ can be represented by a~{\em noose} in the plane that intersects the image of $G$ exactly at the vertices of $S$.
Since $|S_{in}|, |S_{out}| \le \twsf$, 
the corresponding nooses cannot cross more than $\twsf$ cycles from $C_1, C_2, \dots, C_m$.
These nooses form a partition of $G$ into three parts, one of which must contain  $\Omega(\mathsf{tw}^2)$  cycles from $C_1, C_2, \dots, C_m$.
Depending on which part is ``deep'',
we apply different strategies for detecting an irrelevant edge. 

\begin{figure}[t]
    \centering
\includegraphics[scale=0.8]{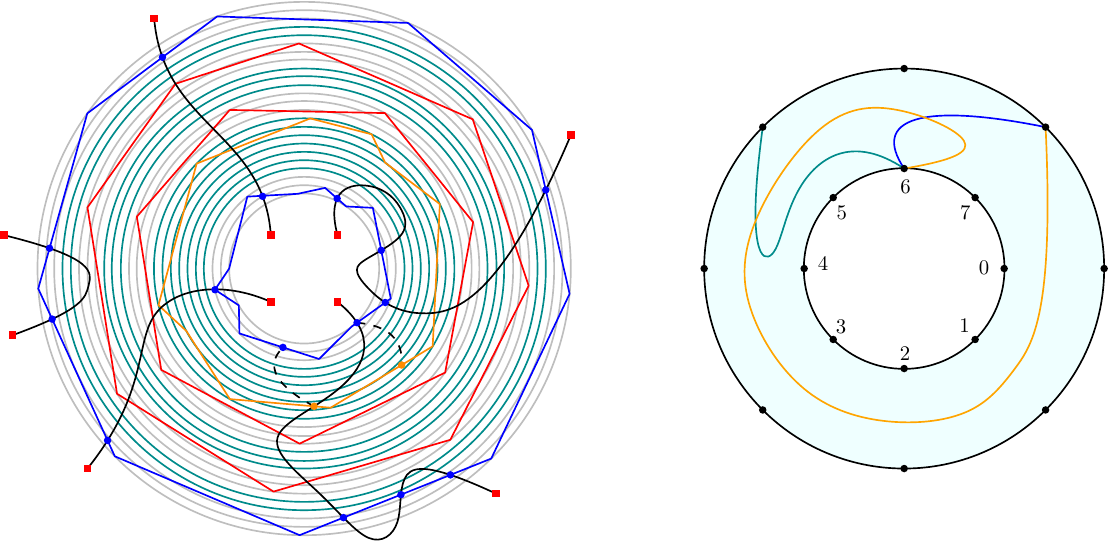}
\caption{Left: An overview of the two-face case.
The well comprises a sequence of concentric cycles $C_1, \dots, C_m$ such that each terminal from $X$ (the red squares) is either inside $C_1$ or outside $C_m$.
The nooses representing separators $V_{in}, V_{out}$ are drawn blue, while the nooses of $S_{in}, S_{out}$ are red; each of them can intersect at most $\twsf$ cycles.
The cycles contained entirely in each of the three parts of the graph are highlighted in green.
A family of vertex-disjoint paths with endpoints at $X$ is sketched with black lines.
We can assume that the $(V_{in}, V_{in})$-subpaths and the $(V_{out}, V_{out})$-subpaths, which do not traverse the well from inside to outside, intersect only few cycles.
Since the remaining  $(V_{in}, V_{out})$-subpaths must cross the separator $S_{in}$, which is the minimum $(V_{in}, V_{out})$-separator closest to $V_{in}$,
there is \mic{a space for an augmenting path} between $V_{in}$ and $S$ (the orange noose).
This fact, crucial for the analysis of Case I, is illustrated with two dashed lines.
Right: An illustration of the notion of the winding number from Case II.
The blue, green, and orange paths have winding numbers 1,~-1,~-7, respectively.
} 
\label{fig:irrelevant-edge-outline}
\end{figure}

\paragraph{Case I: Deep well in the interior/exterior.}
First consider the case that there are $\Omega(\mathsf{tw}^2)$ concentric cycles between $V_{in}$ and $S_{in}$.
(The case with a deep subgraph between $S_{out}$ and $V_{out}$ is analogous.)
Let $\pp$  be a family of vertex-disjoint paths (a linkage) with endpoints in $V_{in} \cup V_{out}$ and
$\pp_{long}$ denote the subfamily of paths in $\pp$ that connect $V_{in}$ to $V_{out}$.
By a standard argument, 
the paths from $\pp \sm \pp_{long}$ can be assumed to cross only few cycles from $C_1, C_2, \dots, C_m$, namely at most $\max(|V_{in}|, |V_{out}|) \le \twsf$, so our main focus is on the paths from $\pp_{long}$.
Each of these paths must traverse the separator $S_{in}$, let $p$ denote its size, so $|\pp_{long}| \le p$.
But since $S_{in}$ is the minimum $(V_{in}, V_{out})$-separator closest to $V_{in}$, for any  $(V_{in}, V_{out})$-separator $S$, located inclusively between $V_{in}$ and $S_{in}$, there exists at least $p+1$ vertex-disjoint $(V_{in}, S)$-paths, i.e., $\mu(V_{in}, S) > p$.
This allows us to focus on the following variant of the two-face case: $S$ replaces the set of terminals $V_{out}$ lying on the outer face, $\mu(V_{in}, S) > p$, there are still $\Omega(\mathsf{tw}^2)$ concentric cycles between $V_{in}$ and $S$, and we want to detect an edge $e$ that is not relevant for any family $\pp$ of at most $p$ 
vertex-disjoint $(V_{in}, S)$-paths.
Here, by ``not relevant for $\pp$'' we mean that there exists a linkage $\pp'$ in $G \sm e$ connecting the same pairs of vertices as $\pp$.
Let $V' \sub V_{in}$ (resp. $S' \sub S$) denote the endpoints of $\pp$ at $V_{in}$ (resp. at $S$).
We prove a criterion that under the given assumptions, such a linkage exists if and only if (1) the cyclic ordering of $V'$ matches the cyclic ordering of $S'$ and (2) the cut-condition $\mu(V', S') \ge |\pp|$ holds.
To prove this, we take advantage of the slack  $\mu(V_{in}, S) > p$ to 
show that any family of at most $p$ disjoint $(V', S')$-paths can be ``shifted'' clockwise using  $\Omega(\mathsf{tw})$ concentric cycles.
As we need at most \twsf shifts to transform any $(V', S')$-linkage into $\pp$, having  $\Omega(\mathsf{tw}^2)$ concentric cycles suffices.
With this criterion at hand, we show that there always exists an edge $e$ whose removal does not affect the cut-condition for any pair $(V',S')$, implying that $e$ is irrelevant.

\paragraph{Case II: Deep well in the middle.}
Now consider the case that there are $\Omega(\mathsf{tw}^2)$ concentric cycles between $S_{in}$ and $S_{out}$.
Recall that $p = |S_{in}| = |S_{out}| = \mu(S_{in}, S_{out})$.
Let $\pp$ be a
family of vertex-disjoint $(S_{in} \cup S_{out})$-paths.
If $\pp$ contains less than $p$ paths that connect $S_{in}$ to $S_{out}$, then the analysis is the same as in the previous case.
Therefore, the problem boils down to a very restricted case: every path in $\pp$ connects $S_{in}$ to $S_{out}$ and every vertex in $(S_{in} \cup S_{out})$ is an endpoint of a path from $\pp$.
We call such $\pp$ a {\em cylindrical linkage}.
It is now convenient to fix a concrete plane embedding of the graph: assume that the vertices from $S_{in}$ lie on the circle $\{(x,y) \in \rr^2 \mid x^2 + y^2 = 1\}$ and the vertices from $S_{out}$ lie on 
$\{(x,y) \in \rr^2 \mid x^2 + y^2 = 4\}$.
Furthermore, assume that the $j$-th element of $S_{in}$, $0 \le j < p$,
 has polar coordinates $(1,\frac{-2\pi}{p} j)$ and the $j$-th element of $S_{out}$ has polar coordinates $(2,\frac{-2\pi}{p} j)$.
 For an $(S_{in},S_{out})$-path $P$, we define its {\em winding number} $\theta(P) \in \mathbb{Z}$ as $\frac{p}{2\pi}$ times the total angle traversed by the curve corresponding to $P$, measured clockwise (see \Cref{fig:irrelevant-edge-outline}).
 It is easy to see that all paths in a cylindrical linkage $\pp$ share the same winding number, so we can also define a winding number $\theta(\pp)$ of $\pp$.
 We say that $\theta \in \mathbb{Z}$ is {\em feasible} if there exists a cylindrical linkage $\pp$ with $\theta(\pp) = \theta$.
 Note that if $\theta(\pp_1) \equiv \theta(\pp_2) \mod p$ then the linkages $\pp_1, \pp_2$ connect the same pairs of terminals.
 Therefore there are at most $p$ different connection-patterns  that we need to preserve.

The structure of cylindrical linkages has been studied by Robertson and Seymour~\cite{GM6} who showed that when $\theta_1 < \theta_2 < \theta_3$ and $\theta_1, \theta_3$ are feasible, then so is $\theta_2$.
This means that it suffices to preserve just the minimal and maximal values $\theta_{min}, \theta_{max}$ that are feasible.
They can be efficiently computed because the problem becomes polynomial-time solvable in this special case~\cite{GM6}.
Moreover, the observation above allows us to assume that the values  $\theta_{min}, \theta_{max}$ differ by at most $p-1$.
We prove that there exist cylindrical linkages $\pp_1, \pp_2$ with $\theta(\pp_1) = \theta_{min}$, $\theta(\pp_2) = \theta_{max}$, such that the intersection of any $P_1 \in \pp_1$ and $P_2 \in \pp_2$ has at most one connected component.
Combined with the presence of many concentric cycles between $S_{in}$ and $S_{out}$, this implies that there exists an edge $e$ used by neither $\pp_1$ nor $\pp_2$.
Consequently, removing $e$ from the graph preserves the set of feasible values of $\theta$ (modulo $p$) and hence
yields an $(S_{in} \cup S_{out})$-linkage-equivalent instance.
This concludes the description of radial diameter reduction.

\paragraph{Comparison to the previous approach.}
The known
$2^{O(k^2)}\cdot n^{O(1)}$-time algorithms for {\sc Planar Disjoint Paths}~\cite{cho2023parameterized, LokshtanovMPSZ20} are based on bounding the number of relevant homotopy classes of a~solution by $2^{O(k^2)}$.
Afterwards, for each fixed homotopy class,
the problem \mic{can be} 
solved in polynomial time~\cite{schrijver1994finding}.
An important case in the analysis of  the homotopy classes resembles the two-face case, described above.
To bound the winding numbers of paths, these proofs rely on a highly technical argument originating from the FPT algorithm for the directed variant of the problem~\cite{DBLP:conf/focs/CyganMPP13}.
By dividing the analysis into two cases, one with non-maximal linkages and one with highly structured linkages, we  avoid these technicalities and obtain a stronger result (i.e., a kernel) by simpler means.

\subsection{Kernelization hardness for parameter $k$}

To establish the kernelization hardness, we present a polynomial-time reduction from {\sc Set Cover} with parameter (the universe size) $k$ to {\sc Planar Disjoint Paths} with parameter (the number of terminal pairs) $k' = k^{\Oh(1)}$.
Under this parameterization, {\sc Set Cover} is known not to admit a polynomial kernel unless  coNP $\subseteq$ NP/poly~\cite{DomLS09} and to be WK[1]-complete~\cite{hermelin2015completeness}.
Hence, such a reduction entails Theorems~\ref{thm:noPolyKer} and~\ref{thm:MK2}.

Before we present the reduction, we discuss the intuition that guided its construction (and, in particular, of the so-called vector-containment gadget described later). Recall that {\sc Planar Disjoint Paths} is solvable in polynomial time once we fix the {\em homotopy class} of a sought solution~\cite{schrijver1994finding}.
This suggests that a polynomial reduction from an NP-hard problem should map different NP-witnesses (each encoding a solution candidate that can be verified in polynomial time) into different homotopy classes of a solution. Notice that the solution candidates to {\sc Set Cover} are tuples of sets (whose number can be huge), while the number of different homotopy classes is bounded by $n^{\Oh(k')}$~\cite{schrijver1994finding}.\footnote{The number of different homotopy classes is also known to be bounded by $2^{\Oh(k^2)}$~\cite{LokshtanovMPSZ20}, if we restrict ourselves to \mic{only} ``relevant'' ones. The precise definition of ``relevant'' in this context is immaterial for our work.}
Thus,
{\em we must map the solution candidates to {\sc Set Cover} into homotopies in an economical fashion.}
The main crux of the reduction is an intricate mechanism that allows to encode a set family of size as large as $2^{\Omega(k)}$ using a homotopy class of just $k^{\Oh(1)}$ paths.

\paragraph{Non-crossing multicommodity flow.}
We present our reduction in the language of non-crossing edge-disjoint walks.
Focusing on edge-disjoint walks allows us to utilize the convenient link between max-flows and shortest paths in the dual graph.
What is more, this setting generalizes finding both vertex-disjoint or edge-disjoint paths in planar graphs~\cite{BercziK17}\footnote{
We remark that there is a flaw in \cite[Proposition 12]{BercziK17} because
replacing a vertex with merely a cycle is not sufficient.
We give a correct argument using a cylindrical wall instead of a single cycle.
}, which will make Theorems~\ref{thm:noPolyKer},~\ref{thm:MK2},~\ref{thm:noPolyKerEdge} simple corollaries from the main reduction.

For a multigraph $G$ with a fixed plane embedding, 
two pairs of edges $(e_1,f_1)$ and $(e_2,f_2)$, with all edges incident to a vertex $v \in V(G)$, {\em cross} if $e_1,e_2,f_1,f_2$ appear in this order in the cyclic ordering of edges around $v$.
Next, two edge-disjoint walks $W_1, W_2$ in $G$ are {\em non-crossing} if there are no
pairs of consecutive edges $(e_1,f_1)$ in $W_1$ and $(e_2,f_2)$ in $W_2$ that cross (see \Cref{fig:non-crossing-setcover}).
In the {\sc Non-crossing Multicommodity Flow} problem (cf.~\cite{BercziK17,LokshtanovMPSZ20}), we are given a plane multigraph $G$ and a family $\tcal$ of $k$ tuples $(s_i,t_i,d_i) \in V(G) \times V(G) \times \nn$, called {\em requests}.
A solution is a family $\pp$ of pairwise edge-disjoint non-crossing walks containing $d_i$ walks connecting $s_i$ to $t_i$, for $i \in [k]$, called a {\em non-crossing $\tcal$-flow}.\footnote{
In this paper, we consider only integral flows.
}
We add one more technical requirement for a solution, which is irrelevant in this informal outline (see \cref{def:reduction:flow}).
Since we do not impose any bounds on the {\em demands} $d_i$,  the total size of the family $\pp$ may be exponential in the parameter $k$.
We address this issue later.

\begin{figure}[t]%
    \centering
    \subfloat{
    {\includegraphics[width=5.5cm,valign=c]{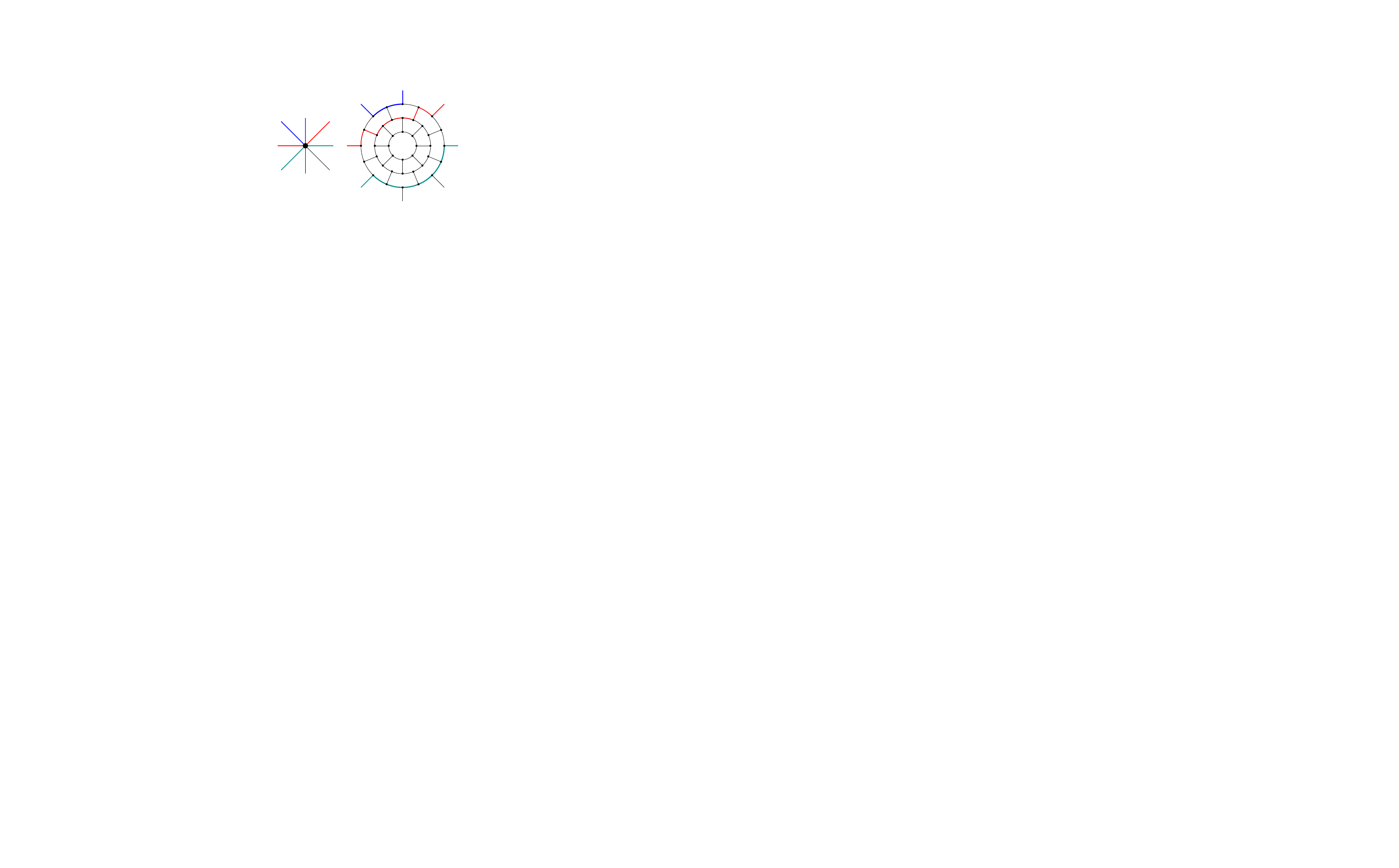}}
    }%
    \qquad\qquad
    \subfloat 
    {{\includegraphics[width=7.5cm,valign=c]{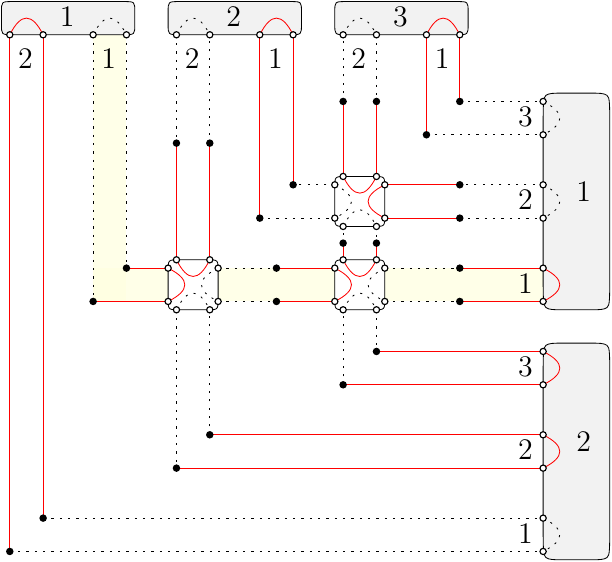}}}%
\caption{
Left: Three non-crossing walks traversing a vertex.
In the reduction from \nullnoncross to {\sc Planar (Edge-)Disjoint Paths} we replace each vertex with a cylindrical wall.
This transformation makes the graph \mic{simple and} subcubic, so the three notions of (a) edge-disjoint non-crossing walks, (b) edge-disjoint paths, and (c) vertex-disjoint paths become equivalent.
Right: A system of gadgets for $k = 3$, $\ell = 2$, in a reduction from {\sc Set Cover}.
The existential gadgets are on the top and the subset gadgets are on the right.
The terminal pairs in each gadget are numbered clockwise. 
The three squares in the middle are the junction gadgets.
The highlighted stripe shows a way of communication between the first existential gadget and the first subset gadget.
The red flow encodes a solution $S_1 = \{1\}$, $S_2 = \{2,3\}$.}%
    \label{fig:non-crossing-setcover}%
\end{figure}

\paragraph{The main gadgets.}

We employ three types of gadgets sharing a common interface: each gadget is a plane multigraph $G$ with a set of requests $\tcal \sub  V(G) \times V(G) \times \nn$ and distinguished vertices $s_1, t_1, \dots, s_m, t_m$ lying on the outer face in this clockwise order.
For a subset $F \sub [m]$, let $\tcal_F = \{(s_i,t_i,1) \mid i \in F\}$.
A gadget $(G,\tcal)$ encodes some downward-closed family of sets $\ff$ as follows: the instance $(G, \tcal \cup \tcal_F)$ should be solvable if and only if $F \in \ff$.
In other words, routing all $(s_i,t_i)$-walks for $i\in F$ through the gadget should be possible only when $F$ satisfies property $\ff$.

The first type of a gadget is an $\ell$-$\mathsf{Existential\, Gadget}$ with $\ell$ terminal pairs $(s_i,t_i)$.
For this gadget, $\ff$ is defined as a family of all proper subsets of $[\ell]$.
That is, $(G, \tcal \cup \tcal_F)$ is solvable if and only if $|F| < \ell$.
Suppose we are given an instance $(k,\scal,\ell)$ of \textsc{Set Cover}, i.e., $\scal$ is a family of subsets of $[k]$ and we ask whether there are $\ell$ sets in $\scal$ that cover $[k]$.
We will make a single copy of an $\ell$-$\mathsf{Existential\, Gadget}$ for each $i \in [k]$.
For $j \in [\ell]$ the intended meaning of  $j \not\in F$ in the $i$-th gadget is that the element $i$ should be covered by the $j$-th set in a solution.
The $\ell$-$\mathsf{Existential\, Gadget}$ ensures that for at least one $j \in [\ell]$ this condition will hold.

Next, we introduce an $(r,k,\scal)$-$\mathsf{Subset\, Gadget}$.
By padding the family $\scal$ with empty sets, we can assume that $|\scal| = 2^r$ for some integer $r \le k$.
The $(r,k,\scal)$-$\mathsf{Subset\, Gadget}$ has $k$ terminal pairs $(s_i,t_i)$ and we require $F \in \ff$ if and only if there exists $S \in \scal$ with $F \sub S$.
In other words, the set of additional terminals should encode a subset of some set from $\scal$.

Imagine the following construction: we make $k$ copies of an $\ell$-$\mathsf{Existential\, Gadget}$, $\ell$ copies of an $(r,k,\scal)$-$\mathsf{Subset\, Gadget}$ and, for each $i \in [k]$, $j \in [\ell]$, we add terminals $u_{i,j}$, $v_{i,j}$, connected to the $j$-th pair of terminals in the $i$-th existential gadget  and the $i$-th pair of terminals in the $j$-th subset gadget.
For each created pair $u_{i,j}$, $v_{i,j}$, we demand a single unit of flow between $u_{i,j}$ and $v_{i,j}$.
By the property of an $\ell$-$\mathsf{Existential\, Gadget}$, for each $i \in [k]$ there should be at least one $j \in [\ell]$ for which the $(u_{i,j}, v_{i,j})$-walk goes through the $j$-th subset gadget.
On the other hand, for each $j \in [\ell]$ the set of such indices $i$ 
forms a subset of some set from $\scal$. 
Therefore satisfying all the requests is possible exactly when there are $\ell$ sets in $\scal$ whose union is $[k]$, as intended.

The problem with this construction is that already for $\ell = k = 3$ such a graph contains $K_{3,3}$ as a minor, so it cannot be planar.
To circumvent this, we need yet another gadget to allow the links between each $i$-th existential gadget and each $j$-th subset gadget to cross.
To this end, we utilize a $\mathsf{Junction\, Gadget}$ $(G, \tcal)$ with 4 terminal pairs $(s_i,t_i)$.
We demand that $(G, \tcal \cup \tcal_F)$ should be solvable if and only if $\{1,3\} \not\sub F$ and $\{2,4\} \not\sub F$.
That is, when we allow a walk on the left then we cannot route a walk on the right, and when we allow a walk at the top then we cannot route a walk at the bottom, and vice versa.
These two exclusion mechanisms are independent from each other, thus allowing two bits of information to ``travel'' in a crossing fashion (see \Cref{fig:non-crossing-setcover}).

The existential and junction gadgets have been employed in the original NP-hardness proof of {\sc Planar Disjoint Paths}~\cite{kramer1984complexity} and we can easily adapt them for our purposes.
The main challenge though is to construct the $(r,k,\scal)$-$\mathsf{Subset\, Gadget}$.
In order to design a meaningful reduction we can produce only $(r+k)^{\Oh(1)}$ requests while we need to encode as many as $2^r$ sets from $\scal$.

\paragraph{Subset gadget: The first attempt.}

We begin with a simplified construction, first presenting the pattern propagation mechanism alone.
This also reflects how the full construction is presented in Sections \ref{sec:subset-simple} and \ref{sec:subset-full}.
Additionally, here we do not delve into formulas regarding the numbers of parallel edges and the non-crucial demands $d_i$, aiming at the simplest presentation of the main ideas.   
By a slight abuse of notation, we treat $\scal$ as a function from $\{0,1\}^r$ to subsets of $[k]$.
We will build the gadget from $k$ blocks, each of which could ``choose'' a {pattern} encoded by a vector $\bb \in \{0,1\}^r$.
When $i \in F$ and the $i$-th block chooses a vector $\bb^i$, this should imply $i \in \scal(\bb^i)$.
The pattern propagation will ensure that each block chooses exactly the same vector $\bb$.
In turn, this will imply that $i \in F \Rightarrow i \in \scal(\bb)$, matching the gadget specification.

Let $r$-ladder be the $(r+1) \times 2$-grid, with internal faces numbered bottom-up as $f_1, \dots, f_r$.
We construct the $i$-th block using two $r$-ladders, an upper one $L_i^+$ and a lower one $L_i^-$, and connect the consecutive blocks as depicted in \Cref{fig:ladder1-outline}.
We also connect the first and the last block, thus creating a ring-like structure with the lower ladders in its interior.
In each ladder $L$ we attach two vertices $L[u_0]$, $L[u_1]$ to, respectively, the bottom and the top of $L$, and add a request  $(L[u_0]$, $L[u_1],1)$ to $\tcal$.
We will enforce that the $(L[u_0]$, $L[u_1])$-walk $W$ must be entirely contained within $L$, and so it can be associated with a vector $\bb^L \in \{0,1\}^r$ encoding which faces $f_1, \dots, f_r$ are to the left  of $W$ (then the corresponding bit in $\bb^L$ is set to 0) and which are to the right (then the corresponding bit is 1).
We shall call $\bb^L$ the {\em pattern} in $L$.

Next, for each $i \in [k]$ and $j \in [r]$ we attach a vertex inside the face $f_j$ of the ladder $L_i^+$ (we refer to this vertex as $L_i^+[x_j]$) and a vertex inside the face $f_j$ of the ladder $L_i^-$ (denoted $L_i^-[x_j]$).
We create a request $(L_i^+[x_j], L_i^-[x_j], 2^{j-1})$.
Because the $(L_i^+[x_j], L_i^-[x_j])$-walks cannot cross the $(L_i^+[u_0]$, $L_i^+[u_1])$-walk, they need to use the passage on the left when the $j$-th bit of the pattern is 0, or the passage on the right when this bit is 1.
This already implies that the patterns in the ladders $L_i^-$, $L_i^+$ must be the same, and we will refer to this common pattern as $\bb^i$.

\begin{figure}[t]
    \centering
\includegraphics[scale=0.4]{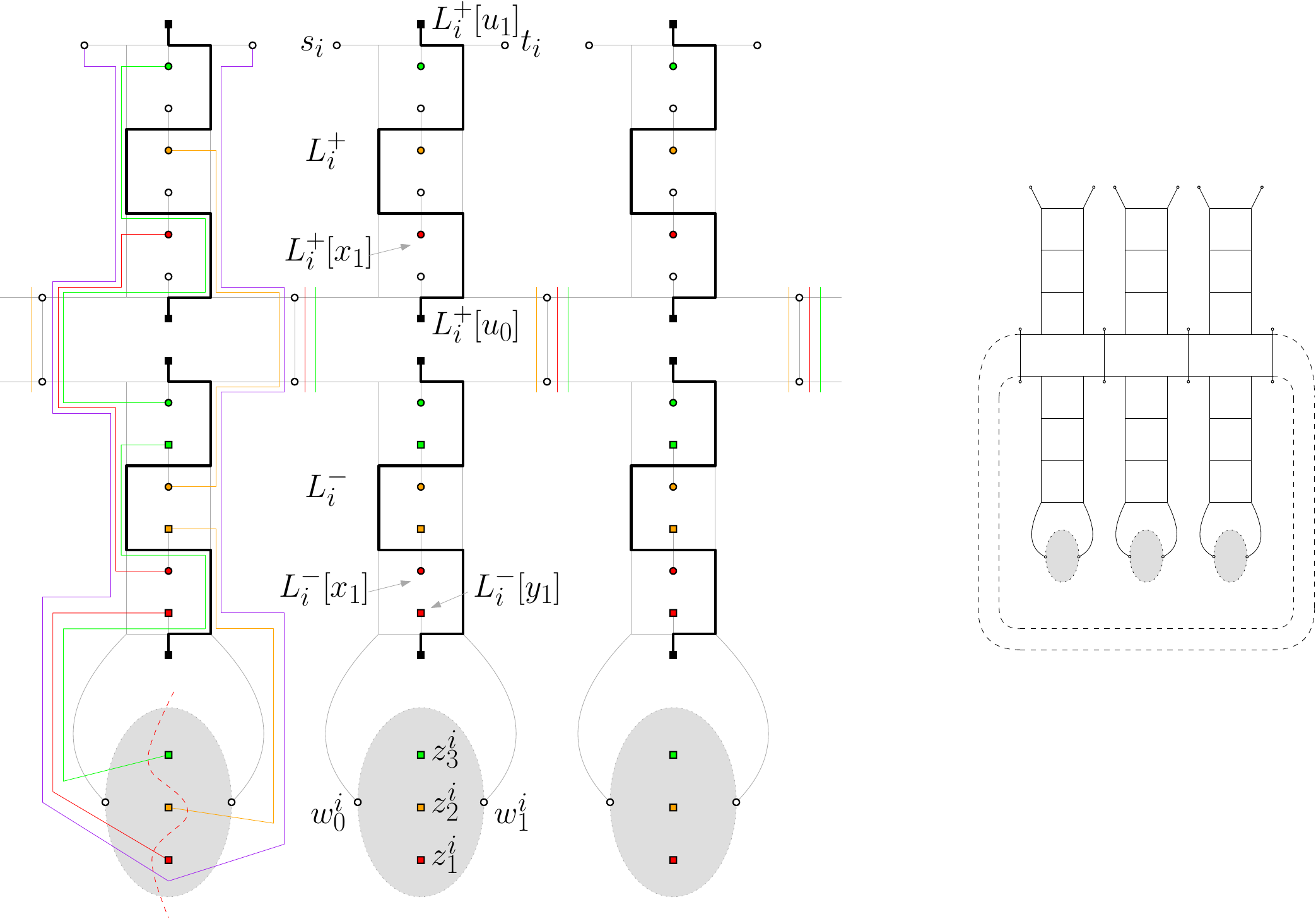}
\caption{
A simplified construction of a $(3,k,\scal)$-$\mathsf{Subset\, Gadget}$ and a sketch of a solution (the colorful lines).
The vertices that need to be connected by walks share common colors and shapes.
The $i$-th block is built from two $3$-ladders $L_i^{+}, L_i^-$, and a $(3,Z^\scal_i)$-$\mathsf{Vector\, Containment\, Gadget}$.
The blocks are combined into a ring-like structure (on the right).
The area separating the upper and lower ladder is referred to as the {\em middle belt}.
Due to  pattern propagation, the choice of which colors are being routed through the left or right passage    must the same in all blocks.
The common pattern $\bb = (010)$ chosen by the solution is sketched with solid lines.
Accommodating an $(s_i,t_i)$-walk (purple) through the $i$-th vector-containment gadget is possible if and only if $\bb \in Z^\scal_i$. 
} 
\label{fig:ladder1-outline}
\end{figure}

We set the capacity of each passage going through the middle belt to be $2^{r}-1$, i.e., we place this many parallel edges in each of $k$ passages.
Note that the $x$-requests from the $i$-th block send $\sum_{j=1}^r 2^{j-1} \cdot 1_{[\bb^i_j = 0]}$ units of flow through the passage to the left of the $i$-th block and $\sum_{j=1}^r 2^{j-1} \cdot 1_{[\bb^i_j = 1]}$ units of flow through the right passage.
If all the patterns are the same, then the total amount of flow going through each passage is $\sum_{j=1}^r 2^{j-1} = 2^{r} - 1$.
Because we work on a ring structure, when some patterns differ then there is a passage through which one would need to push at least $2^{r}$ units of flow.
Consequently, all the vectors $\bb^i$ must coincide and so pattern propagation works~as~intended.

So far we have established a mechanism that makes a solution choose a single vector  $\bb \in \{0,1\}^r$ that appears as a pattern in each block of the $(r,k,\scal)$-$\mathsf{Subset\, Gadget}$.
The next step is to enforce that $i \in F \Rightarrow i \in \scal(\bb)$.
It is convenient to define $Z^\scal_i$ as the set of vectors $\bb$ for which $i \in \scal(\bb)$.
Then the set $Z^\scal_i$ can be encoded in the graph, and we need to check whether the pattern $\bb$ chosen by a solution belongs to $Z^\scal_i$.
To this end, we will need another kind of a gadget, with a slightly different interface.
For $Z \sub \{0,1\}^r$, an $(r,Z)$-$\mathsf{Vector\, Containment\, Gadget}$ is a plane multigraph $G$ with distinguished vertices $z_1, \dots, z_r$ and $w_0, w_1$, with the last two lying on the outer face.
For $\bb \in \{0,1\}^r$ we define $\tcal_{\bb,d}$ as the family of following requests:
\begin{enumerate}[nolistsep]
        \item $(w_0, z_j, 1)$ for each $j \in [r]$ with $\bb_j = 0$,
        \item $(w_1, z_j, 1)$ for each $j \in [r]$ with $\bb_j = 1$,
        \item the request $(w_0, w_1, d)$. 
\end{enumerate}
(The vertex $z_j$ needs to be connected by a walk to either $w_0$ or $w_1$ depending on the $j$-th bit in $\bb$.)
We require that the instance $(G,\tcal_{\bb,0})$ is satisfiable for any $\bb \in \{0,1\}^r$ but the instance $(G,\tcal_{\bb,1})$ is satisfiable if and only if $\bb \in Z$. 
In other words, the choice of the vector $\bb$ governs whether we can insert a $(w_0, w_1)$-walk on top of the $r$ walks with endpoints at $z_1, \dots, z_r$.

Assuming that such a gadget exists, we could finish the construction of the subset gadget as follows.
For each $i \in [k]$ we insert an $(r,Z^\scal_i)$-$\mathsf{Vector\, Containment\, Gadget}$ below the $i$-th block and connect it to the lower corners of the ladder $L_i^-$.
We refer to its distinguished vertices with $i$ in the superscript, e.g., $w^i_0, z^i_j$.
Next, for each $i \in [k]$, $j \in [r]$, we create another vertex inside the face $f_j$ of the ladder $L_i^-$ (denoted $L_i^-[y_j]$) and add a request $(L_i^-[y_j], z^i_j, 1)$. 
By the same argument as above, these walks must go through $w^i_0$ when $\bb_j = 0$ or through $w^i_1$ when $\bb_j = 1$.
Therefore, they contain subwalks satisfying the requests from $\tcal_{\bb,0}$.
Next, we create vertices $s_i, t_i$ attached to the upper corners of the ladder $L_i^+$; note that they end up in the exterior of the ring structure, i.e., on the outer face of the subset gadget.
The only possible way from $s_i$ to $t_i$ leads through the $i$-th vector-containment gadget.
Hence when $i \in F$ and we need to satisfy the request $(s_i,t_i,1)$, there must exist a non-crossing $\tcal_{\bb,1}$-flow in the $(r,Z^\scal_i)$-$\mathsf{Vector\, Containment\, Gadget}$.
By its definition, this implies $\bb \in Z^\scal_i$ and so $i \in \scal(\bb)$, as intended.
The last issue is that the passages through the middle belt are already saturated by the $x$-requests.
This is not a big problem though as we can multiply the demands in the $x$-requests by a constant and similarly multiply the number of the parallel edges in the passages,
creating a little slack sufficient for routing the $(s_i,t_i)$-walks.

\paragraph{Vector-containment gadget.}

Unfortunately, the construction above does not work because we do not know how to construct an $(r,Z)$-$\mathsf{Vector\, Containment\, Gadget}$.
Instead, we now present a~gadget with a slightly more complicated specification.
First, we explain how to construct a gadget that behaves almost like an $(r,Z)$-$\mathsf{Vector\, Containment\, Gadget}$ and then augment the construction of the subset gadget with additional elements necessary to plug in the proper vector-containment~gadget.

It is convenient to analyze the gadget from the dual perspective.
For each $\bb \in \{0,1\}^r$ we can draw a curve (the red dashed curve in \Cref{fig:ladder1-outline}) that has to intersect all the walks in a $\tcal_{\bb,d}$-flow.
When there is a~path $P$ in the dual graph whose homotopy aligns with this curve (that is, $P$ traverses $z_j$ from the side of $w_0$ when $\bb_j =0$ and  from the side of $w_1$ when $\bb_j =1$) then the length of $P$ imposes an upper bound on the maximal number of walks in a $\tcal_{\bb,d}$-flow, what in turn entails an upper bound on $d$.
Given the set $Z$, we would like to construct a plane graph, being a prototype of a dual graph of the gadget, with two vertices $s,t$ on the outer face and $r$ distinguished internal faces $f_1,\dots,f_r$ with the following property:
the length of a shortest $(s,t)$-path with homotopy class encoding the vector $\bb$ (with respect to the faces $f_1,\dots,f_r$) depends on whether $\bb \in Z$.

\begin{figure}
    \centering
\includegraphics[scale=1.3]{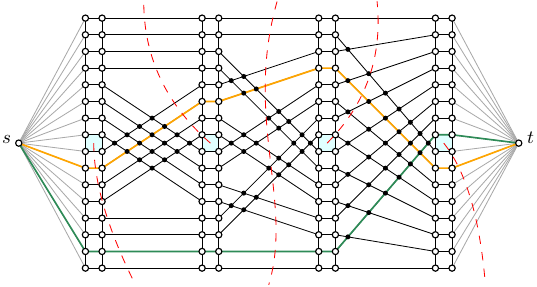}
\caption{The graph $H_4$ with $2^4$ vertices in each vertical path. 
Four distinguished faces $f_1,f_2,f_3,f_4$ are highlighted in light blue.
For $\bb = (1001)$ the path $P_\bb$ is drawn in orange and
for $\vv = (1110)$ the path $P_\vv$ is drawn in green.
These vectors encode whether a path passes above or below each of the faces $f_1,f_2,f_3,f_4$.
The red dashed curves illustrate the directions of walks in a $\tcal_{\bb,d}$-flow in the dual of $H_4$: note that each of them intersects $P_\bb$.
} 
\label{fig:homotopy-outline}
\end{figure}

We construct a plane graph $H_r$ (see \Cref{fig:homotopy-outline}) having a {\em unique} shortest $(s,t)$-path $P_\bb$ for each homotopy class given by  $\bb \in \{0,1\}^r$.
Moreover, the paths $P_\bb$ are pairwise edge-disjoint.
The graph $H_r$ has size $2^{\Oh(r)}$, which is polynomial in the size of family $\scal$, and its treewidth is $2^{\Omega(r)}$ due to large grid subgraphs; this is inevitable in the light of \cref{thm:polyKer}.
The faces $f_1,\dots,f_r$ are the dual counterparts of the vertices $z_1,\dots,z_r$, while the areas above and below $H_r$ are the placeholders for $w_0,w_1$.
In order to achieve our goal, we would like to modify $H_r$ to increase the length of the path $P_\bb$ exactly when $\bb \in Z$, thus allowing more slack for flows in the dual of $H_r$.
This can be obtained by subdividing the first edge (incident to $s$) on the path $P_\bb$ when  $\bb \in Z$; due to edge-disjointness of the paths $P_\bb$, this does not affect the remaining homotopy classes.
This modification raises the upper bound on the size of a $\tcal_{\bb,d}$-flow by one; by applying some other adjustments to $H_r$ we can make this upper bound tight for non-crossing flows.
As a consequence, we can accommodate one more $(w_0,w_1)$-walk in the dual exactly when the pattern $\bb$ belongs to $Z$, as intended.

The main technical hurdle comes from the fact that the length of the path $P_\bb$ in $H_r$ depends on~$\bb$: it is very short for $\bb$ being a 0-vector and very long for $\bb$ comprising alternating 0's and 1's.
So the bound on the size of a $\tcal_{\bb,d}$-walk depends not only on whether  $\bb \in Z$ but also on some function $\gamma(\bb)$, making it useless for the current construction of the subset gadget.
To circumvent this, we first prove that  the function $\gamma$ enjoys a very special form, which will play a crucial role later.

\[
\gamma(b_1b_2\dots b_r) = \sum_{1 \le p < q \le r} 1_{[b_p \ne b_q]} \cdot 2^{r-q+p-1}
\]

We will now work with a generalization of an  $(r,Z)$-$\mathsf{Vector\, Containment\, Gadget}$, namely an  $(r,\gamma,Z)$-$\mathsf{Vector\, Containment\, Gadget}$.
The difference is that a non-crossing $\tcal_{\bb,d}$-flow should be feasible exactly when $d \le \gamma(\bb) + 1_{[\bb \in Z]}$.
In fact, the proper definition (\ref{def:homo:elemenet-gadget}) requires more subtleties concerning factors depending on $r$ but we omit them here.

\begin{figure}[t]
    \centering
\includegraphics[scale=0.35]{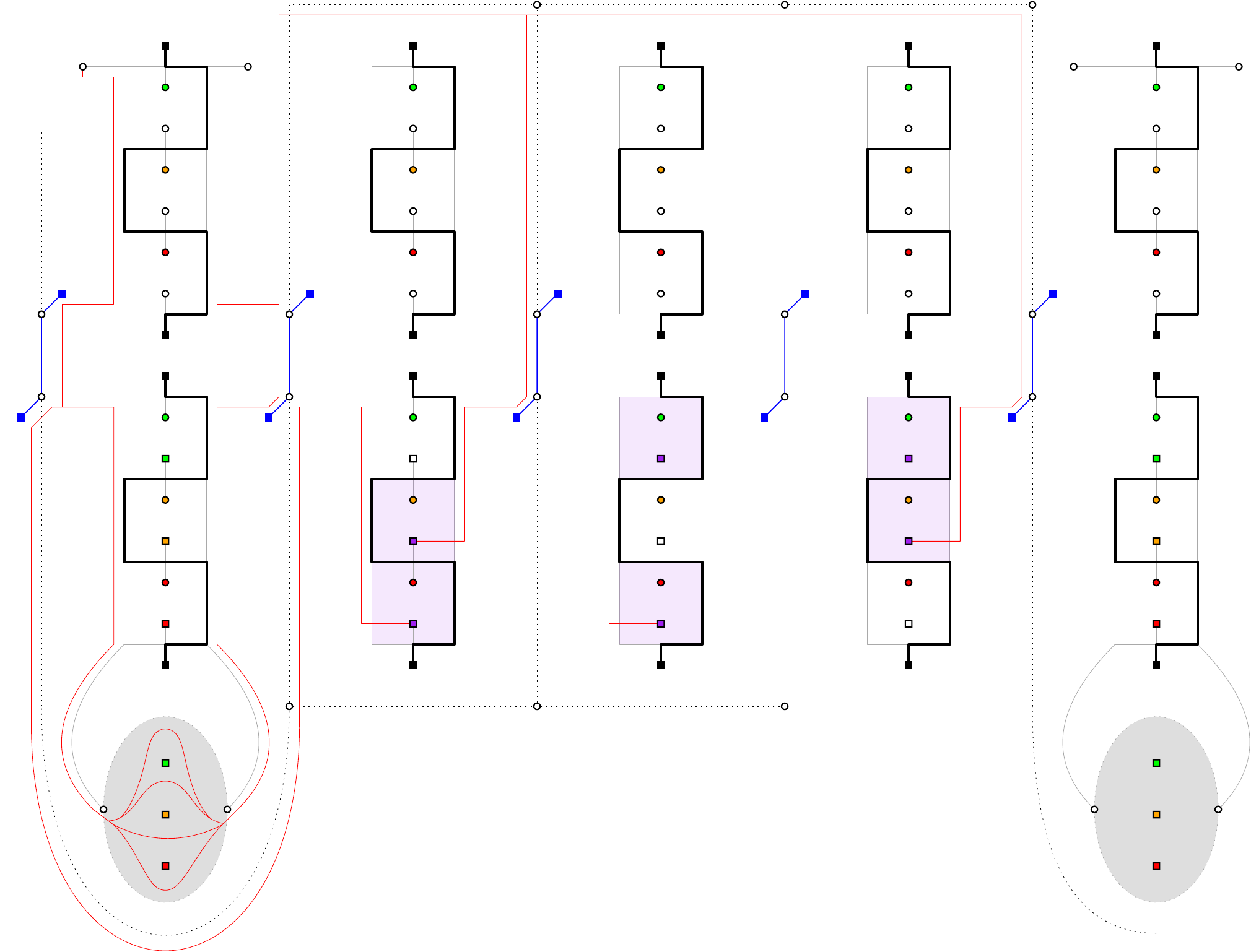}
\caption{
A refined construction of a $(3,k,\scal)$-$\mathsf{Subset\, Gadget}$, showing new blocks inserted between the $i$-th block and the $(i+1)$-th block from the previous construction.
The faces with the endpoints of the new requests are highlighted.
The ``pipes'' are drawn with dotted lines.
The blue paths show other new requests whose purpose is to block the connections between the upper and lower pipes.
The common pattern of the solution is depicted with solid lines while the red lines illustrate the workaround for the new requests traversing the $i$-th vector-containment gadget.
Observe that one pair of highlighted faces is not being separated by the pattern so the corresponding request can be satisfied locally. 
} 
\label{fig:ladder-full-outline}
\end{figure}

\paragraph{Subset gadget: The real deal.}

The current task is to adapt the construction of an $(r,k,\scal)$-$\mathsf{Subset\, Gadget}$ to be compatible with the more complicated $(r,\gamma,Z)$-$\mathsf{Vector\, Containment\, Gadget}$.
The condition $\bb \in Z$ becomes meaningful for the vector-containment gadget only when it needs to additionally accommodate $\gamma(\bb)$ units of $(w_0,w_1)$-flow.
We will extend the previous construction with ``dynamic flow generators'': new requests that could be satisfied either locally, within their ladder, or via walks passing through the vector-containment gadget.
We insert ${r \choose 2}$ new blocks (without vector-containment gadgets) between each pair of blocks in the ring structure.
In each new block, labeled now with a triple $(i,p,q)$, $i\in [k]$, $1 \le p < q \le r$, we request $2^{r-q+p-1}$ walks from a vertex within the face $f_p$ of the corresponding lower ladder $L$ to a vertex within the face $f_q$ of the same ladder (see \Cref{fig:ladder-full-outline}).
Due to pattern propagation, these faces will end up on the same side of the $(L[u_0], L[u_1])$-walk only when a solution chooses a pattern $\bb$ with $\bb_p = \bb_q$.
In this case the new requests can be satisfied by walks contained in $L$.
However, when $\bb_p \ne \bb_q$ the faces $f_p, f_q$ in the ladder $L$ become separated and the corresponding $2^{r-q+p-1}$ walks must be routed elsewhere.
Observe that the total amount of flow that cannot be satisfied locally, summing over all pairs $1 \le p < q \le r$, matches the formula for $\gamma(\bb)$.

We need a way to satisfy the requests, that cannot be satisfied locally, with walks that traverse the $i$-th vector-containment gadget.
To this end, we design an ``irrigation system'' of workarounds that are too narrow \mic{(in the number of parallel edges)} to be used by the other requests. 
It comprises two groups of ``pipes'', one inside the ring structure, gathering the walks starting on the left-sides of the ladders and leading towards $w^i_0$, and the other one outside  the ring structure, gathering the walks starting on the right-sides of the ladders  and leading towards $w^i_1$.
By designing these systems carefully, we can ensure that (a) the walks going through the pipes can traverse the vector-containment gadget without crossing each other or the remaining walks, (b) the only way from the upper pipe system to the lower pipe system leads through the vector-containment gadget, and (c) the terminals $s_i,t_i$ stay on the outer face of the subset gadget, as required by its specification.
This concludes the description of the $(r,k,\scal)$-$\mathsf{Subset\, Gadget}$ $(G,\tcal)$ with the total number of requests $|\tcal| = \Oh(k\cdot r^3)$.

\paragraph{Implementing weights.}

The last issue is to deal with the fact in the {\sc Non-crossing Multicommodity Flow} problem we allow requests of the form $(s_i, t_i, d_i)$ where the demand $d_i$ may be exponentially large in the parameter.
As our final goal is to reduce the problem to {\sc Planar Disjoint Paths}, we would like to set all the demands to $d_i = 1$ without increasing the number of requests too much.
To this end, we take advantage of the construction by Adler and Krause~\cite{AdlerLB} who presented an instance of {\sc Planar Disjoint Paths} with parameter $\ell$ and a roughly $2^\ell \times 2^\ell$-grid subgraph in which every vertex must be used in the unique solution  (hence no vertex is irrelevant despite large treewidth).
What is important for us, the unique solution must traverse the grid $2^\ell-1$ times from left to right
(see \Cref{fig:topological-dpp} on page \pageref{fig:topological-dpp}).
We utilize this property to implement a request of the form $(s_i, t_i, 2^\ell-1)$ using only $\ell$ unitary requests.
We replace the endpoint $s_i$ with a gadget mimicking  the left-side of the grid and $t_i$ with a gadget for the right-side of the grid.
By the arguments from~\cite{AdlerLB} we obtain that a solution must contain $2^\ell-1$ subwalks between these two gadgets, which can be translated back into $2^\ell-1$ walks between $s_i$ and $t_i$.
Finally, we can partition each request  $(s_i, t_i, d_i)$ into $\Oh(\log d_i)$ requests of the form as above, thus increasing the total size of $\tcal$ by a factor of $\Oh(k^2)$ when the demands are bounded by $2^{\Oh(k)}$.
This concludes the outline of the reduction.

\section{Preliminaries}\label{sec:prelims}

The set $\{1,\ldots,p\}$ is denoted by $[p]$.
A~graph $G$ has vertex set $V(G)$ and edge set $E(G)$ of distinct pairs of vertices.
\mic{We also consider multigraphs that may have parallel edges but no loops, i.e., $E(G)$ becomes a multiset of pairs of distinct vertices.}
For a vertex $v \in V(G)$ we denote by $E_G(v)$ the set of edges incident to $v$.
For $A,B \subseteq V(G)$ we define $E_G(A,B) = \{uv \mid u \in A, v \in B, uv \in E(G)\}$.
 The {open} neighborhood of $v \in V(G)$ is $N_G(v) := \{u \mid uv \in E(G)\}$. 
 For a vertex set~$S \subseteq V(G)$ the open neighborhood of~$S$, denoted~$N_G(S)$, is defined as~$\bigcup _{v \in S} N_G(v) \setminus S$.
For $S \subseteq V(G)$, the graph induced by $S$ is denoted by $G[S]$. 
We use shorthand $G-S$ for the graph $G[V(G) \setminus S]$. For $v \in V(G)$, we write $G-v$ instead of $G-\{v\}$.
{For~$A \subseteq E(G)$ we denote by~$G \sm A$ the graph with vertex set~$V(G)$ and edge set~$E(G) \sm A$. For~$e \in E(G)$ we write~$G \setminus e$ instead of~$G \setminus \{e\}$.}

\paragraph{Paths, linkages, and separators.}

For $X,Y \sub V(G)$ an $(X,Y)$-walk in $G$ is an alternating sequence of vertices and edges from $G$ so that the first element is a vertex from $X$, the last one element is a vertex from $Y$, and each edge is incident to the proceeding and  succeeding vertex.
An~$(X,Y)$-path is an $(X,Y)$-walk without vertex repetitions; we often identity a path with a subgraph of $G$.
\mic{When $X = \{x\}$, $Y = \{y\}$, we simply refer to $(x,y)$-walks and $(x,y)$-paths.
The length of a path equals the number of its edges.
For $x,y \in V(G)$ we define $\dist_G(x,y)$ as the length of the shortest $(x,y)$-path.
}

A {\em linkage} in $G$ is a family of vertex-disjoint paths in $G$.
For $X,Y \sub V(G)$ we say that $\pp$ is an {\em $(X,Y)$-linkage} when $\pp$ comprises $(X,Y)$-paths. 
We use shorthand {\em $X$-linkage} for an $(X,X)$-linkage.
For $\tcal \sub V(G) \times V(G)$ we say that $\pp$ is a {\em $\tcal$-linkage} when $|\pp| = |\tcal|$ and $\pp$ contains an
$(s_i, t_i)$-path for each $(s_i,t_i) \in \tcal$.
We say that $\tcal$ is {\em realizable} in $G$ if there exists a $\tcal$-linkage in $G$.
Two linkages $\pp_1, \pp_2$ are {\em aligned} if there is a bijection $f \colon \pp_1 \to \pp_2$ such that $P \in \pp_1$ has the same endpoints as $f(P)$.
For $\tcal \sub V(G) \times V(G)$ we denote by $V_\tcal$ the set of all vertices occurring in $\tcal$.

For two sets $X,Y \subseteq V(G)$, a set $S \subseteq V(G)$ is an $(X,Y)$-separator if no connected component of $G-S$ contains a vertex from both {$X \setminus S$} and {$Y \setminus S$}. Note that such a separator may intersect $X \cup Y$. 
{By Menger's theorem, the minimum cardinality of such a separator is equal to the maximum cardinality of an $(X,Y)$-linkage.}
We denote this quantity by $\mu_G(X,Y)$.
An $(X,Y)$-separator $S$ is {\em inclusion-minimal} if no proper subset of $S$ is an $(X,Y)$-separator.

\begin{theorem}[{\cite[Thm. 8.4, 8.5]{cygan2015parameterized}}]
\label{thm:prelim:closets-cut}
    Let $G$ be a graph and $X, Y \sub V(G)$ be disjoint sets of vertices.
    There exists a minimum-size
    $(X,Y)$-separator $S$ such that for any other minimum-size $(X,Y)$-separator $S'$ the set of vertices reachable from $X$ in $G-S$ is a subset of the set of vertices reachable from $X$ in $G-S'$.
    Furthermore, $S$ can be constructed in polynomial time.
\end{theorem}

We say that $S$ is the minimum-size
$(X,Y)$-separator closest to $X$. 
\mic{Next, we shall need the following fact, which follows from the analysis of the residual graph in the Ford-Fulkerson algorithm.}

\begin{lemma}[Augmenting path, \cite{ford1956maximal,cormen2022introduction}~(Implicit)]
\label{lem:prelim:augmenting}
Let $X,Y \sub V(G)$ and $X' \subset X,\, Y' \subset Y$ be such that $|X'| = |Y'| = \mu_G(X',Y') < \mu_G(X,Y)$.
Then there exist $x \in X \sm X'$, $y \in Y \sm Y'$, and an $(X' \cup \{x\} ,Y' \cup \{y\})$-linkage of size $|\mu_G(X',Y')|+1$. 
\end{lemma}

\paragraph{Contractions and minors.}
The operation of contracting an edge $uv \in E(G)$ introduces a~new vertex adjacent to all of {$N_G(\{u,v\})$}, after which $u$ and $v$ are deleted. 
\mic{When working with multigraphs, we accumulate multiplicities of edges with a common endpoint.} The result of contracting $uv \in E(G)$ is denoted $G / uv$. For $A \subseteq V(G)$ such that $G[A]$ is connected, we say we contract $A$ if we simultaneously contract all edges in $G[A]$ and introduce a single new vertex.

We say that $H$ is a {minor} of $G$, if we can turn $G$ into $H$ by a (possibly empty) series of edge contractions, edge deletions, and vertex deletions.
The result of such a process is given by a minor-model, i.e., a~mapping $\Pi \colon V(H) \to 2^{V(G)}$, such that
the branch sets $(\Pi(h))_{h\in V(H)}$ are pairwise disjoint, induce connected subgraphs of $G$, 
and $h_1h_2 \in E(H)$ implies that $E_G(\Pi(h_1), \Pi(h_2)) \ne \emptyset$.

\paragraph{Treewidth.}
A tree decomposition of graph $G$ is a pair $(T, \chi)$ where~$T$ is a rooted tree, and~$\chi \colon V(T) \to 2^{V(G)}$ is a function, such that:
\begin{enumerate}
    \item For each~$v \in V(G)$ the nodes~$\{t \mid v \in \chi(t)\}$ form a {non-empty} connected subtree of~$T$. 
    \item For each edge~$uv \in E(G)$ there is a node~$t \in V(G)$ with $\{u,v\} \subseteq \chi(t)$.
\end{enumerate}
The \emph{width} of a tree decomposition is defined as~$\max_{t \in V(T)} |\chi(t)| - 1$. The \emph{treewidth} of a graph~$G$, denoted $\tw(G)$, is the minimum width of a tree decomposition of~$G$.


\paragraph{Planar graphs and multigraphs.} 
\mic{We provide only the necessary background here and refer to the textbook~\cite{Mohar2001GraphsOS} for more details.}
A plane embedding of a multigraph $G$ is given by a mapping from $V(G)$ to $\mathbb{R}^2$ and a mapping that associates with each edge $uv \in E(G)$ a simple curve on the plane connecting the images of $u$ and $v$, such that the curves given by two distinct edges can intersect only in a common endpoint. 
A multigraph is called planar if it admits a plane embedding.
{A plane (multi)graph is a (multi)graph with a fixed planar embedding, in which we identify the set of vertices with the set of their images on the plane.} 
For a vertex $v$ in a plane multigraph $G$ we
denote by $\pi_G(v)$ the clockwise ordering of the set $E_G(v)$.
The family of such orderings is called a {\em rotation system}.
\mic{For a topological disc $D \sub \rr^2$ such that $G[V(G) \cap D]$ is connected, the outcome of contracting $V(G) \cap D$ is the unique (with respect to the rotation system) plane multigraph obtained by contracting $G \cap D$ into a point.}

A {\em face} in a plane embedding of a multigraph $G$ is a {maximal connected subset of the plane minus the image of $G$}.
We say that a vertex or an edge lies on a face $f$ if its images  belongs to the closure of $f$.
In every plane embedding there is exactly one face of infinite area, referred to as the outer face.
For a plane multigraph $G$ we define its dual multigraph $G^*$ with $V(G^*)$ being the set of faces of $G$ and edges given by pairs of distinct faces that are incident to an image of a common edge from $E(G)$.
\mic{Each (a) vertex $v$, (b) edge $e$, and (c) face $f$ of $G$ has a counterpart in $G^*$, respectively: (a) face $v^*$, (b) edge $e^*$, and (c) vertex $f^*$.}

For a plane graph $G$ with a set of faces $F$,
we define its {\em radial graph} $\widehat{G}$ as a bipartite graph with the set of vertices $V(\widehat{G}) = V(G) \cup F$ and edges given by pairs $(v,f)$ where $v\in V(G)$, $f \in F$, and $v$ lies on the face $f$.
For two vertices $u, v \in V(G)$ we define their {\em radial distance} $\rdist_G(u, v)$ 
 to be one less than the minimum length of a sequence of vertices that
starts at $u$, ends in $v$, and in which every two consecutive vertices lie on a common face.
For two sets $X,Y \sub V(G)$ we define $\rdist_G(X,Y) = \min_{x \in X, y \in Y} \rdist_G(x, y)$.
The radial diameter of a plane graph $G$ equals $\max_{u,v\in V(G)} \rdist_G(u,v)$. 

\begin{lemma}[{\cite[Prop. 2.1]{jansen2014near}}]\label{lem:prelim:concentric-radial}
Let $G$ be a plane graph with non-empty disjoint
vertex sets $X$ and $Y$, such that $G[X]$
and $G[Y]$ are connected and $\rdist_G(X,Y) = d \ge 2$.
For any $r$ with $0 < r < d$ there is a cycle $C$ in $G - (X \cup Y)$ such that all
vertices $u \in V(C)$ satisfy $\rdist_G(X, u) = r$, and such
that $V(C)$ is an $(X,Y)$-separator in $G$.
\end{lemma}

A {\em noose} is a subset of $\rr^2$ homeomorphic to the circle $\mathbb{S}^1$.
For a plane graph $G$, a $G$-noose is a noose that intersects $G$ only at vertices; the length of a $G$-noose is defined as the number of vertices it intersects.
For a noose $I$ we define $\disc(I)$ as the closure of the bounded region of $\rr^2 \sm I$.
\mic{For a closed set $D \sub \rr^2$ we define its {\em interior} $\mathsf{int}(D)$ as $D$ minus its boundary $\partial D$.} 
For two nooses $I_{in}, I_{out}$, such that  $I_{in}$ lies in the interior of $\disc(I_{out})$,
we define $\ringfull = \disc(I_{out}) \sm \mathsf{int}(\disc(I_{in}))$.
A~plane graph $G$ is {\em properly embedded} in a set $D \sub \rr^2$ if $G \sub D$ and $G \cap \partial D \sub V(G)$.

\begin{lemma}[{\cite[Prop. 8.2.3]{Mohar2001GraphsOS}}]
\label{lem:prelim:noose-uv}
Let $G$ be a plane graph, $X,Y \sub V(G)$ be non-empty, disjoint, and inducing connected subgraphs of $G$, and $S \sub V(G) \sm (X \cup Y)$ be an inclusion-minimal $(X,Y)$-separator.
Then there exists a $G$-noose $\gamma$ such that $S = \gamma \cap V(G)$
and $\gamma$ separates the plane into two regions, one containing $X$ and the second containing $Y$.
\end{lemma}

\mic{We remark that the original statement in \cite{Mohar2001GraphsOS} involves singleton sets $X,Y$ and triangulated plane graphs, but it can be adapted to our setting by contractions and inserting new vertices inside faces.} 

\begin{lemma}\label{lem:prelim:noose-cycle}
Let $G$ be a plane graph, $C_1,C_2$ be vertex-disjoint cycles in $G$ such that $C_1$ lies in the interior of $C_2$.
Let $S$ be an inclusion-minimal $(C_1,C_2)$-separator, not necessarily disjoint from $C_1,C_2$.
Then there exists a $G$-noose $\gamma$ such that $S = \gamma \cap V(G)$
and $\disc(C_1) \sub \disc(\gamma) \sub \disc(C_2)$.
\end{lemma}
\begin{proof}
    By the minimality of the separator, $S$ does not contain vertices in the  interior of $\disc(C_1)$ and  in the exterior of $\disc(C_2)$.
    Consider a graph $G'$ obtained from $G$ by removing all the vertices in the interior of  $\disc(C_1)$, inserting a vertex $u$ inside $C_1$ adjacent to the entire $V(C_1)$,
    removing all the vertices in the exterior of  $\disc(C_2)$ and inserting a vertex $v$ outside $C_2$ adjacent to the entire $V(C_2)$.
    Then $S$ is an inclusion-minimal $(u,v)$-separator in $G'$, disjoint from $u,v$.
    The claim follows from \cref{lem:prelim:noose-uv}.
\end{proof}

\section{Polynomial kernel for parameter $k + \twsf$}\label{sec:polyKer}

In this section we prove \cref{thm:outline:polyKer}, being a generalization of \cref{thm:polyKer}.
We begin with providing additional preliminaries about linkages.
Next, we present the radial diameter reduction (\cref{sec:tw:radial-reduction}), analyzing the two cases described in the outline, and then combining them into a procedure for the irrelevant edge detection.
Afterwards, we deal with the single-face case (\cref{sec:tw:single-face}) and then apply it to process the low-radial-diameter instances in \cref{sec:tw:cutting}.

\subsection{Preliminaries for processing linkages}
\label{sec:tw:prelim}

We gather several useful facts about linkages that will form our toolbox for proving \cref{thm:outline:polyKer}.
This is mostly a compilation of know facts, adapted to our setting.
We begin with the concept of linkage-equivalency and explain how it helps in compressing subgraphs without terminals.

\begin{definition}
    Two graphs $G_1, G_2$ sharing a set of vertices $X$ are {\em $X$-linkage-equivalent} if for every set of disjoint pairs $\tcal \sub X^2$, $\tcal$ is realizable in $G_1$ if and only if $\tcal$ is realizable in $G_2$.
\end{definition}

\begin{lemma}
\label{lem:prelim:equivalent}
Let $G$ be a graph and $X, Y \sub V(G)$, $U \sub V(G) \sm (X \cup Y)$, $N_G(U) \sub Y$.
Suppose that there is an edge $e \in E(G[U \cup Y])$ such that $G[U \cup Y] \sm e$ is $Y$-linkage-equivalent to $G[U \cup Y]$.
Then $G \sm e$ is $X$-linkage-equivalent to $G$.
\end{lemma}
\begin{proof}
    Let  $\pp$ be an $X$-linkage in $G$.
    For a path $P \in \pp$ let $\Gamma(P)$ be the family of maximal subpaths of $P$ with vertex sets contained in $U \cup Y$.
    Since $U \cap X = \emptyset$ and $N_G(U) \sub Y$ then every path in $\Gamma(P)$ is a $(Y,Y)$-path.
    Therefore, the linkage $\pp_U$ given by the union of $\Gamma(P)$ over $P \in \pp$ is a $Y$-linkage in $G[U \cup Y]$.
    Because $G[U \cup Y] \sm e$ is $Y$-linkage-equivalent to $G[U \cup Y]$, there exists a linkage $\pp'_U$ in $G[U \cup Y] \sm e$ that is aligned with $\pp_U$.
    For $P \in \pp$ let $\widehat P$ be obtained from $P$ by replacing each subpath from $\Gamma(P)$ with its counterpart from $\pp'_U$.
    Then $\{\widehat P \mid P \in \pp\}$ is a linkage in $G \sm e$ that is aligned with $\pp$.
\end{proof}

The following concept has been introduced in the work about treewidth reduction for {\sc Planar Disjoint Paths}~\cite{AdlerKKLST17}.

\begin{definition}[Tight concentric cycles]
Let $G$ be a plane graph, $X_{in},X_{out} \sub V(G)$, and $C_1, \dots , C_m$
be a sequence of cycles in $G$. 
We call $C_1, \dots , C_m$ {\em concentric}, if for all $i \in [m-1]$, the cycle $C_i$ is contained in
the interior of $\disc(C_{i+1})$.
When additionally
$X_{in} \sub \inter(\disc(C_{1}))$ and $X_{out} \cap \disc(C_{m}) = \emptyset$, then we call it an $(X_{in},X_{out})$-sequence of concentric cycles.

A $(X_{in},X_{out})$-sequence of concentric cycles is {\em tight}
if, in addition,
for every $i \in [m-1]$, $\disc(C_{i+1}) \sm \disc(C_{i})$ does not contain a cycle $C$
with $\disc(C_{i}) \sub \disc(C) \subsetneq \disc(C_{i+1})$,
and $\disc(C_{1}) \sm X_{in}$ does not contain a cycle $C$
with $X_{in} \sub \disc(C) \subsetneq \disc(C_{1})$.
\end{definition}

When a~plane graph $G$ properly embedded in \ringfull is clear from the context, we denote $V_{in} = V(G) \cap I_{in}$ and $V_{out} = V(G) \cap I_{out}$.

\begin{lemma}\label{lem:prelim:tight}
    Consider a graph $G$ properly embedded in \ringfull
    with $d = \rdist_G(V_{in}, V_{out})$.
    Let $X_{in} \sub V_{in}$ and $X_{out} \sub V_{out}$.
    Then there exists a tight $(X_{in}, X_{out})$-sequence of concentric cycles $C_1, \dots , C_{d-1}$.
\end{lemma}
\begin{proof}
    Consider a graph $G'$ obtained from $G$ by inserting a vertex $v_{in}$ inside $I_{in}$ adjacent to entire $X_{in}$ and
    a vertex $v_{out}$ outside $I_{out}$ adjacent to entire $X_{out}$.
    The sets $X_{in} \cup \{v_{in}\}$ and $X_{out} \cup \{v_{out}\}$ induce connected subgraphs of $G'$ and their radial distance is at least $d$.
    By \cref{lem:prelim:concentric-radial} there exists cycles in $C_1, \dots , C_{d-1}$ in $G - (X_{in} \cup X_{out})$ that
    are $(X_{in}, X_{out})$-separators in $G$.
    By \cref{lem:prelim:noose-uv} each $C_i$ has $X_{in}$ in its interior and $X_{out}$ in its exterior, so  $C_1, \dots , C_{d-1}$ are concentric.
    As long as this sequence is not tight, we can find a local refinement pushing some cycle closer to $X_{in}$.
    After a finite number of refinements we obtain a tight $(X_{in}, X_{out})$-sequence of concentric cycles.
\end{proof}

\begin{lemma}
\label{lem:prelim:separator-interval}
Consider a graph $G$ properly embedded in \ringfull.
Let $C_1,\dots, C_m$ be a $(V_{in}, V_{out})$-sequence of concentric cycles and
$S$ be an inclusion-minimal $(V_{in}, V_{out})$-separator.
Then there exists an interval $J \sub [m]$ of size at most $|S|$ so that $S \cap V(C_j) \ne \emptyset$ implies $j \in J$.
\end{lemma}
\begin{proof}
    The minimality of $S$ implies that there exists a $G$-noose $\gamma$ with $\gamma \cap V(G) = S$. 
    Therefore each pair of consecutive vertices on $\gamma$ shares a face and the maximal radial distance between vertices of $S$ is at most $|S| - 1$.
    On the other hand, for $u \in C_i$, $v \in C_j$,
    we have $\rdist_G(u,v) \ge |i-j|$.
    The lemma follows.
\end{proof}

When working with linkages traversing concentric cycles, it is convenient to assume that each path in a linkage intersects each cycle exactly once.
We can enforce this property as long as the sequence of concentric cycles is tight.
The following proof is an adaptation of \cite[Lemma 6.15]{lossy-arxiv}.

\begin{lemma}
\label{lem:diam:uncross-cycles}
Consider a graph $G$ properly embedded in \ringfull,
$X_{in} \sub V_{in}$, and $X_{out} \sub V_{out}$.
Let $C_1,\dots, C_m$ be a tight $(X_{in}, X_{out})$-sequence of concentric cycles and $\pp$ be a $(X_{in}, X_{out})$-linkage in $G$ such that each vertex in $X_{in}$ is an endpoint of a path in $\pp$.
Then there exists a linkage $\pp'$ aligned with $\pp$ such that the intersection of each $P \in \pp'$ and each $C_i$ \mic{has exactly one connected component.}
\end{lemma}
\begin{proof}
Let $\pp'$ be a linkage aligned with $\pp$ that minimizes the number of edges in $\bigcup_{P \in \pp'} V(P)$ that do not belong to any cycle $C_i$.
Suppose there is $i \in [d]$ and $P \in \pp'$ 
\mic{such that $P \cap C_i$ has at least two connected components.}
Choose the minimal $i \in [d]$ with this property.
Then $P$ contains a subpath $Q$ with endpoints on $C_i$ and internal vertices disjoint from $V(C_i)$.

First suppose that these vertices lie in the interior of $\disc(C_i)$.
By the choice of $i$ either $i=1$ or the path $Q$ does not intersect $C_{i-1}$.
If $i=1$ then by the assumption 
\mic{that each vertex in $X_{in}$ is an endpoint of some path we infer that} $Q$ cannot intersect $X_{in}$.
In both cases we \mic{could use $Q$ to construct a cycle $C$ enclosing $X_{in}$ (resp. $C_{i-1}$) with $\disc(C) \subsetneq \disc(C_i)$,
contradicting the tightness of the sequence $C_1,\dots, C_m$.}

Now suppose that that the internal vertices of $Q$ are disjoint from $\disc(C_i)$.
\mic{Let $D$ be the bounded region of $\rr^2 \sm (C_i \cup Q)$ incident to $Q$ and
$Q' \sub C_i$ be the path within $C_i$ whose image is $C_i \cap \partial D$; then $Q'$ connects the endpoints of $Q$ (see \Cref{fig:tight}).}
If the internal vertices of $Q'$ were disjoint from all the other paths in $\pp'$ we could replace $Q$ with $Q'$ (and remove some redundant vertices if we get a self-crossing of $P$) but this would contradict the choice of $\pp'$.
Therefore there is some path $P' \in \pp'$, $P' \ne P$, using a vertex $v \in V(C_i)$ being internal in $Q'$.
When $i>1$ then the $(X_{in},v)$-subpath of $P'$ must intersect $C_{i-1}$ so by the choice of $i$ the $(v, X_{out})$-subpath of $P'$ cannot intersect $C_{i-1}$.
As $P'$ is disjoint from $Q$ it must contain a subpath with endpoints on $C_i$ and internal vertices in $\inter(\disc(C_i) \sm \disc(C_{i-1}))$.
When $i=1$ we obtain a subpath of $P'$ with endpoints on $C_1$ and internal vertices in $\inter(\disc(C_1) \sm X_{in}$ because $P'$ does not have internal vertices from $X_{in}$.
In both cases we get a contradiction with the tightness of $C_1,\dots, C_m$. 
\end{proof}

\begin{figure}
    \centering
\includegraphics[scale=0.9]{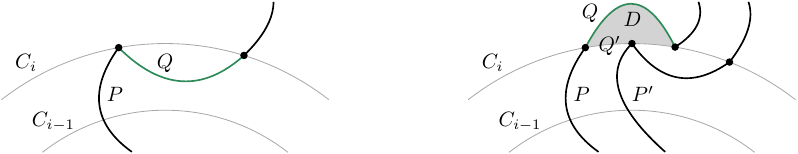}
\caption{An illustration of the two cases in \cref{lem:diam:uncross-cycles}.
Left: The green subpath $Q$ of $P$ is contained in $\disc(C_i) \sm \disc(C_{i-1})$ what contradicts the tightness of the sequence.
Right: $Q$ is internally disjoint from $\disc(C_i)$.
Since the chosen linkage minimizes the number of used edges that do not belong to any cycle,
some other path $P'$ must visit the arc of $C_i$ between the endpoints of $Q$, again leading to a contradiction.
}
\label{fig:tight}
\end{figure}

\mic{We will work with the notion of a cylindrical grid, which can be regarded as an outcome of identifying the opposite sides of a standard grid.} 

\begin{definition}
\label{def:diam:grid}
Let $k \ge 3$, $m \ge 1$.
The {\em $(k,m)$-cylindrical grid} $C_k^m$ is a plane graph constructed as follows.
We draw $m$ concentric cycles referred to as $C^1,\dots, C^m$, counting from the \mic{innermost} one.
Then we draw $k$ pairwise disjoint lines connecting $C^1$ to $C^m$; these lines are called $C_1, \dots, C_k$, counting clockwise.
We turn each intersection of $C_i$ and $C^j$ into a vertex, referred to as $c_i^j$. 
\end{definition}

\mic{We require $k \ge 3$ in order to avoid using parallel edges and restrict our arguments only to simple graphs.
The (4, 8)-cylindrical grid is shown in \Cref{fig:cylindrical}.}
It is known~\cite[Lemma 4.2]{GuT12} that a large linkage between two sets that are far apart in the radial distance entails a minor model of a large cylindrical grid.
We need an intuitive strengthening of this fact relating the endpoints in this linkage to the branch sets in the minor model.

\begin{lemma}
\label{lem:diam:grid-minor}
Consider a graph $G$ properly embedded in \ringfull and $m = \rdist_G(V_{in}, V_{out}) \ge 2$.
Suppose that $\tcal = (s_1,t_1), \dots, (s_k,t_k)$ is realizable in $G$, where \mic{$k \ge 3$, and} 
$s_1, \dots, s_k$ lie in this cyclic order on $I_{in}$ and
$t_1, \dots, t_k$ lie in this cyclic order on $I_{out}$ (both counted clockwise).
Then $G$ contains a minor model of the cylindrical grid $C^{m-1}_k$
such that 
for each $i \in [k]$ the vertex $s_i$ belongs to the branch set of $c^1_{i}$ and 
$t_i$ belongs to the branch set of $c^{m-1}_{i}$.
\end{lemma}
\begin{proof}
Let $X_{in} = V_\tcal \cap V_{in}$ and $X_{out} = V_\tcal \cap V_{out}$; then  $\rdist_G(X_{in}, X_{out}) \ge m$.
By \cref{lem:prelim:tight} there exists a tight $(X_{in}, X_{out})$-sequence of concentric cycles $C_1,\dots, C_{m-1}$ in $G$. 
We can apply \cref{lem:diam:uncross-cycles} to obtain that there exists
a $\tcal$-linkage $\pp$ such that the intersection of each $P \in \pp$ and each $C_i$ is a single segment of $C$.
Let $G'$ be the union of all paths in $\pp$ and the cycles $C_1,\dots, C_{m-1}$.
Let $P_i$ denote the $(s_i,t_i)$-path in $\pp$;
then $P_1, \dots, P_k$ are ordered clockwise.
The set $V(P_i) \cap V(C_j)$ induces a connected subgraph of $G'$; we contract it into a single vertex referred to as $c_i^j$.
Let $G''$ be the graph obtained by these contractions;
then $G''$ is a minor of $G$ and each vertex  $c_i^j$ has exactly four neighbors in $G''$.
Next, we contract every degree-2 vertex with one of its neighbors.
Finally, we contract each degree-1 vertex with its only neighbor so $s_i$ gets contracted with $c_i^1$ and  $t_i$ gets contracted with $c_i^{m-1}$.
The claim follows.
\end{proof}

A linkage in  a cylindrical grid can be transformed into a linkage in graph $G$ using the following observation.

\begin{observation}
\label{lem:diam:minor-realization}
Let $G, H$ be graphs and $(V_h)_{h\in V(H)}$ be a minor model of $H$ in $G$.
Consider a set of pairs  $(s_1,t_1), \dots, (s_k,t_k)$ from $V(G) \times V(G)$ with all the vertices distinct.
Suppose there exists an injection $\pi: \bigcup_{i=1}^k\{s_i\} \cup \{t_i\} \to V(H)$ such that $s_i \in V_{\pi(s_i)}$ and 
 $t_i \in V_{\pi(t_i)}$.
If $(\pi(s_1), \pi(t_1)), \dots, (\pi(s_k), \pi(t_k))$ is realizable in $H$ then  $(s_1,t_1), \dots, (s_k,t_k)$ is realizable in $G$.
\end{observation}

Because $\tw(C_t^t) \ge t$~\cite[\S 7.7.1]{cygan2015parameterized} and treewidth is a monotone measure  with respect to taking minors, we also obtain the following corollary from \cref{lem:diam:grid-minor}.
\mic{Note that $|i-j| \ge t + 1$ implies 
$\rdist(C_i,C_j) \ge t + 1$.}

\begin{corollary}
\label{lem:diam:sep-exists}
Let $G$ be a plane graph and $C_1, \dots, C_m$ be a concentric sequence of cycles.
Suppose there exists $i,j \in [m]$ such that $|i-j| \ge t + 1$ and  $\mu_G(C_i,C_j) \ge t$.
Then $\tw(G) \ge t$.
\end{corollary}

A $(V_{in} \cup V_{out})$-linkage may contain paths (or subpaths) with both endpoints in $V_{in}$ or both in $V_{out}$.
The next lemma assures that we can always assume that such paths  do not go ``too deep'' inside the graph.
The proof uses similar ideas as~\cite[Lemma 3]{AdlerKKLST17}.

\begin{lemma}
\label{lem:diam:visitors}
Let $G$ be a plane graph properly embedded in \ringfull and 
$t = \max(|V_{in}|, |V_{out}|)$. \meir{Use a diff letter? So $k$ is reserved for the number of terminal pairs. \mic{changed $k \to t$ and propagated}}
Let $C_1,\dots, C_m$ be a $(V_{in}, V_{out})$-sequence of concentric cycles.
Then for every $(V_{in} \cup V_{out})$-linkage $\pp$ 
there exists a linkage $\pp'$ aligned with $\pp$ such that
\begin{enumerate}
    \item every inclusion-minimal  $(V_{in} \cup V_{out})$-subpath of a path in $\pp'$, that is a $(V_{in}, V_{in})$-path, intersects at most $t$ first cycles in  $C_1,\dots, C_m$, and
    \item every inclusion-minimal  $(V_{in} \cup V_{out})$-subpath of a path in $\pp'$, that is a $(V_{out}, V_{out})$-path, intersects at most $t$ last cycles in  $C_1,\dots, C_m$.
\end{enumerate}
\end{lemma}
\begin{proof}
Let $\pp'$ be a linkage aligned with $\pp$ that minimizes the number of edges in $\bigcup_{P \in \pp'} V(P)$ that do not belong to any cycle $C_i$.
Let $\widehat \pp$ be the family of the inclusion-minimal
 $(V_{in} \cup V_{out})$-subpaths of paths in $\pp'$.
 Then $\widehat \pp$ is a family of internally disjoint paths and each endpoint can be shared by at most two paths.
 For a  $(V_{in}, V_{in})$-path $P \in \widehat \pp$ we define $R(P)$ as the bounded region of $\rr^2 \sm (P \cup I_{in})$ incident to $P$.
 Let $h(P)$ be the number of paths from $\widehat \pp$, different from $P$, which are contained in $R(P)$.
We have $h(P) \le t - 1$ and for every path $P' \ne P$ contained in $R(P)$ it holds that $h(P') < h(P)$.

 We show inductively that when $h(P) = \ell$ then $P$ intersects at most $\ell + 1$ first cycles in  $C_1,\dots, C_m$.
 First consider $\ell = 0$ and suppose that $P$ intersect $C_2$.
 Let $P'$ be a $(C_1,C_1)$-subpath of $P$ with internal vertices disjoint from $\disc(C_1)$.
 Then there exists a path $P'' \subset C_1 \cap R(P)$ with the same endpoints as $P'$.
 We can replace $P'$ with $P''$ in $P$ and, as a result, obtain a linkage aligned with $\pp'$ which uses less edges not belonging to any cycle $C_i$.
 This gives a contradiction.

 Assume now that the claim holds for $h(P) < \ell$ and consider $P \in \widehat \pp$ with $h(P) = \ell$
 with intersects $C_{\ell + 2}$.
  Let $P'$ be a $(C_{\ell + 1},C_{\ell + 1})$-subpath of $P$ with internal vertices disjoint from $\disc(C_{\ell + 1})$.
Then there exists a path $P'' \subset C_{\ell+1} \cap R(P)$ with the same endpoints as $P'$.
By the assumption, this path is disjoint from all paths in $R(P)$ different than $P$.
Again, by a replacement argument we obtain a linkage aligned with $\pp'$ with a smaller cost. 

This concludes the proof of the first part.
The second part, concerning  $(V_{out}, V_{out})$-paths in $\widehat \pp$, is symmetric.
\end{proof}

\subsection{Radial diameter reduction}
\label{sec:tw:radial-reduction}

As outlined in \cref{sec:outline}, we will reduce the radial diameter of the graph $G$ by repeatedly removing irrelevant edges.
We focus on the scenario where a subgraph of $G$, devoid of terminals, can be properly embedded in \ringfull and $\rdist(V_{in}, V_{out})$ is large.
We inspect two cases, first analyzing {\em non-maximal} linkages, in which the number of $(V_{in}, V_{out})$-paths is less than $\mu(V_{in}, V_{out})$. 
\mic{Later on we will be armed with two strategies for detecting an irrelevant edge, each applicable in a different setting.}

\subsubsection{Rerouting a non-maximal linkage}

We are going to show that when $\tcal \sub V_{in} \times V_{out}$, $|\tcal| < \mu(V_{in}, V_{out})$, and $\rdist(V_{in}, V_{out})$ is large, then the cut-condition $|\tcal| \le \mu(V_\tcal \cap V_{in}, V_\tcal \cap V_{out})$ is sufficient for a $\tcal$-linkage to exist.
We begin with an argument for cylindrical grids.

\begin{lemma}
\label{lem:diam:grid-shift}
Consider the cylindrical grid $C_k^m$ with $m \ge k^2$, $k\ge3$.
Let vertices $s_1, \dots, s_{k-1}$ lie in this cyclic order on $C^1$
and vertices $t_1, \dots, t_{k-1}$ lie in this cyclic order on $C^m$
(both counted clockwise).
Then $\{(s_1,t_1), \dots, (s_{k-1},t_{k-1})\}$ is realizable in $C_k^m$.
\end{lemma}
\begin{proof}
For two vertices $u,v \in V(C_k^m)$ we define $\mathsf{Shift}(u,v)$ as follows.
Let $i,j \in [k]$ be such that $u \in V(C_i)$ and $v \in V(C_j)$.
Then $\mathsf{Shift}(u,v)$ is the minimum non-negative integer $\ell$ satisfying $(i + \ell) \mod k = j \mod k$.
In other words, it is the number of clockwise jumps needed to reach $C_j$ from $C_i$.
We will show the following claim by induction on $j$.

\begin{claim}\label{claim:diam:grid-shift}
Consider the cylindrical grid $C_k^{m}$ where $k\ge3$, $m = (j+1)k$, for some $j \ge 0$.
Let vertices $s_1, \dots, s_{k-1}$ lie in this cyclic order on $C^1$
and vertices $t_1, \dots, t_{k-1}$ lie in this cyclic order on $C^{m}$
(both counted clockwise).
If $\mathsf{Shift}(s_1, t_1) = j$ then
$\{(s_1,t_1), \dots, (s_k,t_k)\}$ is realizable in~$C_k^{m}$.
\end{claim}
\begin{innerproof}
First consider the basic case $j=0$ in which $m=k$.
Since $\mathsf{Shift}(s_1, t_1) = 0$ we can assume w.l.o.g. that $s_1, t_1$ both lie on $C_1$, that is, $s_1 = c^1_1$, $t_1 = c^k_1$.
Let $s^*$ be the unique vertex in $C^1 \sm \{s_1, \dots, s_{k-1}\}$
and $t^*$ be the unique vertex in $C^k \sm \{t_1, \dots, t_{k-1}\}$.
For the clarity of presentation we examine only the extremal case $s^* = c^1_k$, $t^* = c^k_2$ in detail; \mic{the other cases are analogous}.
The pair $(s_1,t_1)$ can be connected directly via the path $P_1 = C_1$.
By the assumption on the cyclic order we have that for each $i \in [2,k-1]$ the vertex $s_i$ lies on $C_i$ and $t_i$ lies on $C_{i+1}$.
For $i \in [2,k-1]$ we define the path $P_i$ as a concatenation of 
\begin{enumerate}
    \item the subpath of $C_i$ from $s_i = c^1_i$ to $c^{k-i+1}_i$,
    \item the edge $c^{k-i+1}_ic^{k-i+1}_{i+1}$,
    \item the subpath of $C_{i+1}$ from $c^{k-i+1}_{i+1}$ to $c^{k}_{i+1} = t_i$. 
\end{enumerate}
See \Cref{fig:cylindrical}.
Note that none of these paths intersect $P_1 = C_1$.
Furthermore, $P_{i+1} \cap C_{i+1}$ is contained in the disc enclosed by $C^{k-i-1}$ (inclusively) while
$P_i \cap C_{i+1}$ is disjoint from the interior of the disc enclosed by $C^{k-i}$.
Hence, the paths $P_1,\dots,P_{k-1}$ are vertex-disjoint.
The construction for different $s^*, t^*$ is analogous.
This concludes the analysis for the case $j=0$.

\begin{figure}
    \centering
\includegraphics[scale=0.7]{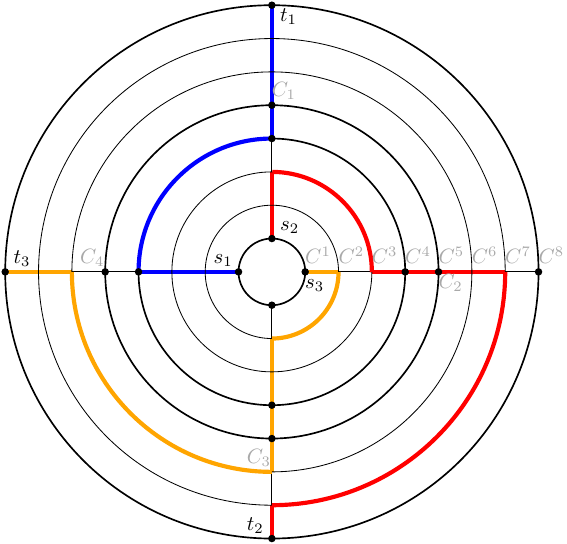}
\caption{A visualization of the cylindrical grid $C_4^8$ and the proof of  \cref{lem:diam:grid-shift}. 
Here $\mathsf{Shift}(s_1, t_1) = j = 1$ so 
we need at least $2\cdot 4 = 8$ concentric cycles in \cref{claim:diam:grid-shift}.
The three $(s_i,t_i)$-paths are drawn in colors.
The four inner cycles illustrate the inductive argument for $j>0$ where each path gets shifted once clockwise. 
The four outer cycles show the argument for $j=0$ where the origin of the path ending at $t_1$ is already correctly positioned and we might only need to shift the remaining ones. 
}
\label{fig:cylindrical}
\end{figure}

Suppose that that $j>0$.
Let $\pi(i) \in [k]$ be that $s_i = c^1_{\pi(i)}$ for $i\in [k]$.
Let $s'_i$ be the vertex $c^k_j$ where $j = \pi(i) + 1$ modulo $k$.
Then $\mathsf{Shift}(s'_1, t_1) = j - 1$.
By the inductive assumption, the $(k, jk)$-cylindrical grid induced by the cycles $C^{k+1}, C^{k+2}, \dots, C^m$
contains a linkage connecting pairs $(s'_1,t_1), \dots, (s'_{k-1}, t_{k-1})$.
Therefore, it suffices to construct a linkage $(P_1,\dots,P_{k-1})$ in $C_k^k$
connecting pairs $(s_1,s'_1), \dots, (s_{k-1}, s'_{k-1})$.
Due to the assumption on the cyclic order
we can assume w.l.o.g. that $\pi(i) = i$ for each $i \in [k-1]$.
Then the path $P_i$ for $i \in [k-1]$ is given by the same concatenation formula as in the case $j=0$.
Again, these paths are vertex-disjoint, which yields the claim.
\end{innerproof}
The lemma follows from the observation that $\mathsf{Shift}(s_1,t_1) < k$ so the claim can be applied.
\end{proof}

Now we generalize the argument from a cylindrical grid to the general case.

\begin{lemma}
\label{lem:diam:non-maximal-criterion}
Let $G$ be a graph properly embedded in \ringfull and $r < p$ be integers.
Suppose that $\mu_G(V_{in}, V_{out}) \ge p$ and $\rdist_G(V_{in}, V_{out}) \ge p^2 + 1$.
Let vertices $s_1, \dots, s_{r}$ lie in this cyclic order on $I_{in}$
and vertices $t_1, \dots, t_{r}$ lie in this cyclic order on $I_{out}$
(both counted clockwise).
Then $\tcal = \{(s_1,t_1), \dots, (s_{r},t_{r})\}$ 
is realizable in $G$ if and only if 
$\mu_G(\{s_1, \dots, s_r\}, \{t_1, \dots, t_r\}) \ge r$.
\end{lemma}
\begin{proof}
\mic{The lemma is trivial for $r=1$ so we will assume $r\ge 2$.}
The condition $\mu_G(V_\tcal \cap V_{in}, V_\tcal \cap V_{out}) \ge r$ is clearly necessary for a $\tcal$-linkage to exist.
Suppose that this condition holds and let $\pp$ be some $(V_\tcal \cap V_{in}, V_\tcal \cap V_{out})$-linkage of size $r$.
Since $r < \mu_G(V_{in}, V_{out})$, by \cref{lem:prelim:augmenting} there exist vertices $s^* \in V_{in} \sm V_\tcal$, $t^* \in V_{out} \sm V_\tcal$, and a
linkage $\pp'$ connecting sets $S = \{s_1, \dots, s_r, s^*\}$ and $T = \{t_1, \dots, t_r, t^*\}$. 
Note that $3 \le r+1 = |S| = |T| \le p$.
By \cref{lem:diam:grid-minor} the graph $G$ contains a minor model of $C_{r+1}^m$, where $m = p^2 \ge (r+1)^2$, and there are bijections $\pi_S \colon S \to C^1$, $\pi_T \colon T \to C^m$ that preserve the cyclic ordering, such that $s \in S$ belongs to the branch set of $\pi_S(s)$ and  $t \in T$ belongs to the branch set of $\pi_T(t)$.
We can thus assume w.l.o.g. that $\pi_S(s_1), \dots, \pi_S(s_r)$ lie in this cyclic order on $C^1$ and $\pi_T(t_1), \dots, \pi_T(t_r)$ lie in this cyclic order on $C^m$, counted clockwise.
By \cref{lem:diam:grid-shift} the set of pairs $\{(\pi_S(s_1), \pi_T(t_1)), \dots, (\pi_S(s_r), \pi_T(t_r)\}$ is realizable in $C_{r+1}^m$.
Then the lemma follows from \cref{lem:diam:minor-realization}.
\end{proof}

\mic{Our goal is to detect an edge that can be safely removed without modifying the family of possible non-maximal linkages.
We can assume that the $(V_{in}, V_{in})$-paths and the $(V_{out}, V_{out})$-paths intersect only few cycles in the concentric family, so the main challenge is to preserve the non-maximal $(V_{in}, V_{out})$-linkages.
As we know that the cut-condition is sufficient for a such a linkage to exist, it remains to find an edge $e$ whose removal does not affect any cut-condition.
We show that this can be guaranteed by two requirements: (a) removing $e$ does not decrease $\mu(V_{in}, V_{out})$ and, (b) $e$ has sufficiently large radial distance from both $V_{in}$ and $V_{out}$.
}

\begin{proposition}
\label{lem:diam:non-maximal-irrelevant}
Let $G$ be a graph properly embedded in \ringfull, 
$t = \max(|V_{in}|, |V_{out}|)$, and  $s = \mu_G(V_{in}, V_{out})$.
Let $C_1,\dots, C_m$ be a $(V_{in}, V_{out})$-sequence of concentric cycles in $G$ and
$m \ge (t+2)^2$. 

Consider $i \in [2t + 1, m-2t]$
and edge $e \in E(C_i)$ such that 
$\mu_{G \sm e}(V_{in} V_{out}) = \mu_G(V_{in}, V_{out})= s$.
Let $\tcal \sub (V_{in} \cup V_{out})^2$. Suppose that $\tcal$  contains less than $s$ pairs
with one element in $V_{in}$ and one in $V_{out}$.
Then $\tcal$ is realizable in $G$ if and only if $\tcal$ is realizable in $G \sm e$.

Furthermore, there exists at least one edge $e$ satisfying the requirements above and it can be found in polynomial time. 
\end{proposition}
\begin{proof}
When $\tcal$ is realizable in $G$ then, by \cref{lem:diam:visitors}, there exists
a $\tcal$-linkage $\pp$ in $G$ such that
every $(V_{in}, V_{in})$-path in $\pp$ intersects at most $t$ of the first cycles in the sequence $C_1, \dots, C_m$, while every $(V_{out}, V_{out})$-path in $\pp$ intersects
at most $t$ of the last cycles in $C_1, \dots, C_m$.
Let $\pp_\textrm{long} \sub \pp$ be the subfamily of paths from $\pp$ with one endpoint in $V_{in}$ and one in $V_{out}$.
By the assumption $|\pp_\textrm{long}| < s$.
If $\pp_\textrm{long} = \emptyset$ then we are done; suppose that this is not the case.

The graph obtained from $G$ by removing the paths from $\pp \sm \pp_\textrm{long}$ has exactly one connected component containing the paths from $\pp_\textrm{long}$; let $G'$ indicate this component.
Let $V'_{in} \sub V(G')$ be the set of vertices lying on the inner face of $G'$ containing $I_{in}$ and $V'_{out} \sub V(G')$ be the set of vertices lying on the outer face of $G'$. \meir{Figure. \mic{done}}

\begin{figure}[t]
    \centering
\includegraphics[scale=0.8]{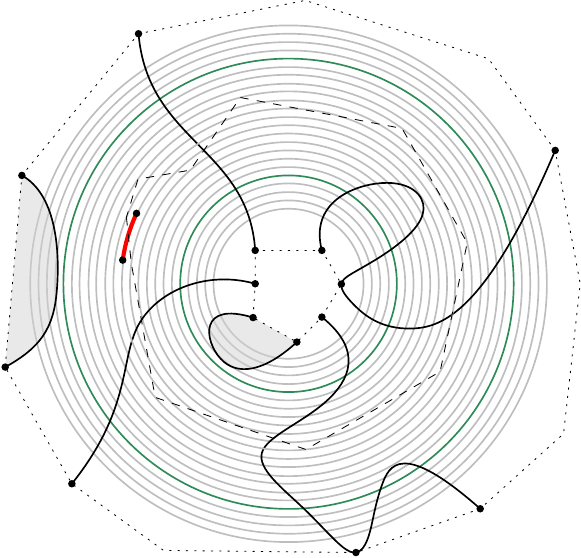}
\caption{An illustration for \cref{lem:diam:non-maximal-irrelevant}. 
The nooses $I_{in}, I_{out}$ are dotted and the $(V_{in}, V_{out})$-sequence of concentric cycles $C_1,\dots, C_m$ is gray.
The cycles $C_{t+1}, C_{m-t}$ are highlighted.
To obtain graph $G'$, we remove from $G$
the paths in $\pp \sm \pp_\textrm{long}$ (which do not intersect $C_{t+1}, C_{m-t}$) together with the shaded area.
The edge $e$ is drawn solid red.
The proposition relies on the observation that any inclusion-minimal separator $S$ in $G' \sm e$ between a vertex inside $C_{t+1}$ and a vertex outside $C_{m-t}$, such that $S$ is not present in $G'$, is represented by a noose (dashed) intersecting $e$.
Since $e \in C_i$, where $i$ is sufficiently far for $t+1$ and $m-t$, this noose cannot intersect $C_{t+1}, C_{m-t}$ so it must be a $(C_{t+1}, C_{m-t})$-separator as well.
}
\label{fig:non-maximal}
\end{figure}

\begin{claim}
\label{claim:non-maximal:max-linkage}
It holds that $\mu_{G'}(V'_{in}, V'_{out}) \ge s$.
\end{claim}
\begin{innerproof}
    Let $\pp_{max}$ be a $(V_{in}, V_{out})$-linkage of size $s$ in $G$.
    Then each path $P$ from $\pp_{max}$ has a non-empty intersection with $V(G')$.
    In particular, $P$ contains a subpath being a $(V'_{in}, V'_{out})$-path in $G'$.
    This gives a $(V'_{in}, V'_{out})$-linkage in $G'$ of size $s$.
\end{innerproof}

All the cycles $C_{t+1}, \dots, C_{m-t}$ are contained in $G'$ and so is $e$.
Hence $\rdist_{G'}(V'_{in}, V'_{out}) \ge 
m -2t \ge t^2 + 2$
and $\rdist_{G' \sm e}(V'_{in}, V'_{out}) \ge t^2 + 1 \ge s^2 + 1$.
Let $e = v_1v_2$.
Then each of $v_1,v_2$ is separated from each of $V'_{in}, V'_{out}$ with at least $t$ cycles from $C_{t+1}, \dots, C_{m-t}$.

\begin{claim}
\label{claim:non-maximal:remove-edge-global}
It holds that $\mu_{G' \sm e}(V'_{in}, V'_{out}) \ge s$.
\end{claim}
\begin{innerproof}
Suppose otherwise.
Then there exists a $(G'\sm e)$-noose $\gamma$ in the plane which separates $V'_{in}$ from $V'_{out}$ and intersects $G' \sm e$ on at most $s-1$ vertices;
let $S = \gamma \cap V(G' \sm e)$.
By \cref{claim:non-maximal:max-linkage} there are no such small $(V'_{in}, V'_{out})$-separators in $G'$. Therefore $\gamma$ must intersect the image of $e$ and so $S$ contains a vertex $u$ that lies on a common face with $e$ in $G'$.
Suppose that $\gamma$ intersects $C_{t+1}$;
then it must also intersect $C_{i-1}$.
\cref{lem:prelim:separator-interval}
implies that $|(i-1) - (t-1)| = |i-t-2| \le s-2 \le t - 2$ but this contradicts the assumption that $i \ge 2t + 1$.
By the symmetric argument, $\gamma$ cannot intersect $C_{m-t}$.
Therefore, the curve $\gamma$ must be contained in the interior of $\ring(C_{t+1}, C_{m-t})$ and so $S$ is a $(C_{t+1}, C_{m-t})$-separator in $G' \sm e$: \mic{see \Cref{fig:non-maximal}}.
As a consequence, $S$ is also a $(V_{in}, V_{out})$-separator in $G \sm e$.
This implies that $\mu_{G \sm e}(V_{in}, V_{out}) \le s - 1$ and contradicts the assumption that 
$\mu_{G \sm e}(V_{in}, V_{out}) = \mu_{G}(V_{in}, V_{out})$.
\end{innerproof}

Note that $\pp_\textrm{long}$ is a $(V'_{in}, V'_{out})$-linkage in $G'$.
Let $T_{in} \sub V'_{in}, T_{out} \sub V'_{out}$ be the sets of endpoints of paths from $\pp_\textrm{long}$.

\begin{claim}
It holds that $\mu_{G' \sm e}(T_{in}, T_{out}) \ge |\pp_\textrm{long}|$.
\end{claim}
\begin{innerproof}
    Clearly, $\mu_{G'}(T_{in}, T_{out}) \ge |\pp_\textrm{long}|$.
    Suppose that such an inequality does not hold in $G' \sm e$.
    Similarly as in \cref{claim:non-maximal:remove-edge-global}, there exists
    a $(G'\sm e)$-noose $\gamma$ in the plane which separates $T_{in}$ from $T_{out}$ and intersects $G' \sm e$ on at most $|\pp_\textrm{long}|-1$ vertices; let $S = \gamma \cap V(G' \sm e)$.
    By the same argument as before we obtain that 
    $\gamma$ is contained in the  interior of $\ring(C_{t+1}, C_{m-t})$ and  $S$ is a $(V_{in}, V_{out})$-separator in $G \sm e$.
    \mic{Recall that $e = v_1v_2$.}
    Observe that any $(V_{in}, V_{out})$-path in $G-S$ must go through $e$ so 
    $S \cup \{v_1\}$ is a $(V_{in}, V_{out})$-separator in $G$. 
    Hence, $\mu_G(V_{in}, V_{out}) \le |S| + 1 \le |\pp_\textrm{long}|$.
    This contradicts the assumption that $|\pp_\textrm{long}| < s = \mu_G(V_{in}, V_{out})$. 
\end{innerproof}

The two claims above allows us to apply the criterion from \cref{lem:diam:non-maximal-criterion} to $G'\sm e, V'_{in}, V'_{out}$, $\pp_\textrm{long}$ with $p = s$, $r = |\pp_\textrm{long}| < p$, and $\rdist_{G' \sm e}(V'_{in}, V'_{out}) \ge p^2 + 1$.
We derive that
there exists a linkage in $G' \sm e$ aligned with $\pp_\textrm{long}$.
By the construction of $G'$ this implies that there exists a linkage in $G \sm e$ aligned with $\pp$.

It remains to justify that $e$ can be efficiently found.
First, $(t+2)^2 - 4t > 0$ so the interval $[2t+1, m-2t]$ is non-empty.
A $(V_{in}, V_{out})$-linkage $\pp_{max}$ of size $s$ can be found in polynomial time.
Then $e$ can be chosen as any edge on $C_i$ that is not used by $\pp_{max}$.
\end{proof}

\subsubsection{Rerouting a maximal linkage}

\mic{We move on to the scenario in which the number of $(V_{in}, V_{out})$-paths in a linkage equals $\mu(V_{in}, V_{out})$.
The crucial special case occurs when $|V_{in}| =  |V_{out}| = \mu(V_{in}, V_{out})$.
This is the same setting that has been studied by
Robertson and Seymour~\cite{GM6} as a subroutine in their FPT algorithm for {\sc Planar Disjoint Paths}.
We shall adopt the same perspective for analyzing this case, based on the following convenient plane embedding.
}

\begin{definition}
    A plane graph $G$ is called {\em $k$-cylindrical} if:
\begin{enumerate}
    \item It is properly  embedded in \ringfull where $I_{in} = \{(x,y) \in \rr^2 \mid x^2 + y^2 = 1\}$ and 
$I_{out} = \{(x,y) \in \rr^2 \mid x^2 + y^2 = 4\}$;
    \item The sets $V_{in} = V(G) \cap  I_{in}$ and $V_{out} =  V(G) \cap I_{out}$ have size $k$ each;
    \item $V_{in} = \{(1,\frac{2j}{k}\pi)\}_{0 \le j < k}$ and
    $V_{out} = \{(2,\frac{2j}{k}\pi)\}_{0 \le j < k}$ in polar coordinates, and
    \item $\mu_G(V_{in}, V_{out}) = k$.
\end{enumerate}   
\end{definition}

We refer to the elements of $V_{in}$ as $s_0, s_1, \dots, s_{k-1}$ so that $s_j$ has polar coordinates $(1,\frac{-2\pi}{k} j)$.
Similarly, $t_0, t_1, \dots, t_{k-1}$ are the elements of $V_{out}$ and $t_j = 
(2,\frac{-2\pi}{k}j)$.

\begin{definition}
For a path $P$ connecting $s \in V_{in}$ and $t \in V_{out}$ we define its {\em winding number} 
$\theta(P) \in \mathbb{Z}$ as $\frac{k}{2\pi}$ times the total angle traversed by the curve corresponding to $P$ (measured clockwise).
\end{definition}

\mic{See \Cref{fig:irrelevant-edge-outline} on page \pageref{fig:irrelevant-edge-outline} for an example.
Intuitively, the winding number measures how many times a path winds around the ring (and in which direction) and what is the difference in the angles of its endpoints.}

\begin{definition}
A {\em cylindrical linkage} in $G$ is a
$(V_{in}, V_{out})$-linkage of size $k$.
When $\pp$ is cylindrical
then every path $P \in \pp$ has the same winding number and we refer to it as $\theta(\pp)$.
We say that $\theta$ is {\em feasible} in $G$ if there is a cylindrical linkage in $G$ with the winding number $\theta$.
\end{definition}

We remark that Robertson and Seymour~\cite{GM6} \mic{defined} the winding number of $P$ as $-\theta(P) / k$ but we choose this convention so we could work with integers and the more intuitive clockwise ordering.

\begin{lemma}[{\cite[Lem. 5.9]{GM6}}]
\label{lem:diam:3thetas}
Let $G$ be a $k$-cylindrical graph.
If $\theta_1 < \theta_2 < \theta_3$ and  $\theta_1, \theta_3$ are feasible in $G$, then so is $\theta_2$.
\end{lemma}

For a $k$-cylindrical graph $G$ let $\Theta^G$ be the set of all feasible values of $\theta$.
By \cref{lem:diam:3thetas} the set $\Theta^G$ forms an interval of integers and it is non-empty because $\mu_G(V_{in}, V_{out}) = k$.
The set $\Theta^G$ is always finite and it can enumerated efficiently.

\begin{lemma}[{\cite[Lem. 5.11]{GM6}}]
\label{lem:diam:enumerate}
There is a polynomial-time algorithm that, given a $k$-cylindrical graph~$G$,
enumerates the set $\Theta^G$.
\end{lemma}

We define $\theta^G_1, \theta^G_2$ as follows.
If $|\Theta^G| < k$ then $\theta^G_1 = \min \Theta^G$ and $\theta^G_2 = \max \Theta^G$.
Otherwise, we set $\theta^G_1 = \min \Theta^G$ and $\theta^G_2 = \theta^G_1 + k - 1$. 

\begin{observation}
\label{obs:diam:theta-interval}
Let $G$ be a $k$-cylindrical graph.
Then $[\theta^G_1, \theta^G_2] \sub \Theta^G$.
Furthermore, if $\theta \in \Theta^G$ then there exists $\theta' \in [\theta^G_1, \theta^G_2]$ such that $\theta' \equiv \theta \mod k$.
\end{observation}

For $j \in [0,k-1]$ let $\tcal_j \sub V_{in} \times V_{out}$ be the set of pairs $(s_i, t_{i + j \mod k})_{i \in [0,k-1]}$.
Clearly, $\tcal_j$ is realizable in $G$ if and only if there exists $\theta \in \Theta^G$ such that $(\theta \mod k) = j$.
This is equivalent to the existence of $\theta \in [\theta^G_1, \theta^G_2]$ with $(\theta \mod k) = j$.
\mic{Combining all 
these observations with the strategy for coping with non-maximal linkages yields a criterion for an edge to be irrelevant in a  $k$-cylindrical~graph.}

\begin{lemma}
\label{lem:diam:criterion}
Let $G$ be a $k$-cylindrical graph
and $C_1,\dots, C_m$ be a $(V_{in}, V_{out})$-sequence of concentric cycles in $G$ with
$m \ge (k+2)^2$.
Consider $i \in [2k + 1, m-2k]$
and edge $e \in E(C_i)$ such that $\theta^G_1, \theta^G_2$ are feasible in $G\sm e$.
Then $G \sm e$ is $(V_{in} \cup V_{out})$-linkage-equivalent to $G$.
\end{lemma}
\begin{proof}
The assumptions imply that $\mu_{G\sm e}(V_{in}, V_{out}) = \mu_{G}(V_{in}, V_{out}) = k$.
Consider some $\tcal \sub (V_{in} \cup V_{out})^2$ that is realizable in $G$.
Let $\ell$ be the number of pairs in $\tcal$ with one endpoint in $V_{in}$ and one in $V_{out}$.
If $\ell < k$ then $\tcal$ is realizable in $G \sm e$ due to
\cref{lem:diam:non-maximal-irrelevant}.
Suppose that $\ell = k$. 
Then $\tcal = \tcal_j$ for some $j \in [0,k-1]$
for which there exists $\theta \in \Theta^G$ such that $(\theta \mod k) = j$.
By \cref{obs:diam:theta-interval} we can assume that $\theta \in [\theta^G_1, \theta^G_2]$.
As both $\theta^G_1, \theta^G_2$ are feasible in $G\sm e$, it follows from \cref{lem:diam:3thetas} that so is $\theta$.
\end{proof}

\paragraph{Disentangling cylindrical linkages.}
\mic{While finding an edge not required by a single linkage is simple, finding a single edge that is not needed by two linkages is more challenging.
A priori, it could be the case then the union of any linkages $\pp_1$ with $\theta(\pp_1) = \theta_1^G$ and $\pp_2$ with $\theta(\pp_2) = \theta_2^G$ is the entire graph.
We show that this is not the case by constructing linkages $\pp_1, \pp_2$ whose intersection pattern is relatively simple.

We need two additional tools to achieve this goal.}
We begin with ordering $(u,v)$-paths in a $k$-cylindrical graph in a clockwise fashion.

\begin{definition}
    Let $G$ be a $k$-cylindrical graph, $u \in V_{in}$, and $v \in V_{out}$. 
    Consider two distinct $(u,v)$ paths $P_1,P_2$ oriented from $u$ to $v$; let $w$ be the last vertex on their longest common prefix.
    
    When $w = u$, let $e_1, \dots, e_d$ be the clockwise ordering of $E_G(u)$ such that $e_1, e_d$ are incident with the face containing $I_{in}$.
    We write $P_1 \sqsubset P_2$ when the first edge of $P_1$ appears earlier in  $e_1, \dots, e_d$ than the first edge of $P_2$.
    
    When $w \ne u$, let $e$ be the edge preceding $w$ in both $P_1,P_2$ and $e_1, \dots, e_d$ be the clockwise ordering of $E(w) \sm e$ such that $e$ lies between $e_d, e_1$.
    We write $P_1 \sqsubset P_2$ when the edge following $w$ in $P_1$ appears earlier in  $e_1, \dots, e_d$ than the edge following $w$ in  $P_2$. \meir{Figure. \mic{done}}
\end{definition}

See \Cref{fig:cylindrical-order} for an example.
The relation $\sqsubset$ is transitive and it \mic{yields} a linear order on the family of $(u,v)$-paths in $G$.

\begin{figure}
    \centering
\includegraphics[scale=0.9]{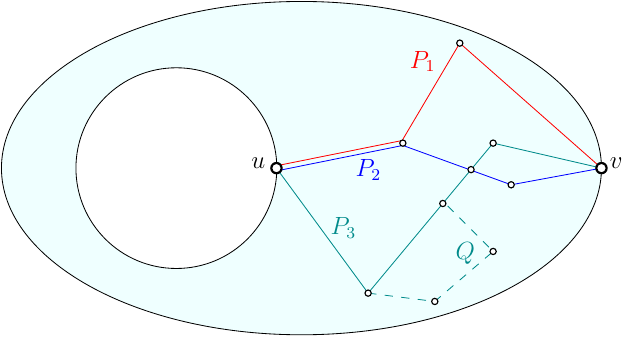}
\caption{Three $(u,v)$-paths satisfying $P_1 \sqsubset P_2 \sqsubset P_3$.
The ring has been deformed for a~better presentation.
The dashed path $Q$ is a clockwise handle of $P_3$.
}
\label{fig:cylindrical-order}
\end{figure}

\begin{definition}[Handle, clockwise-tightness]
    Let $G$ be a $k$-cylindrical graph, $u \in V_{in}$, and $v \in V_{out}$. 
    A path $Q$ is called a {\em handle} of $P$ is the endpoints of $Q$ lie on $P$ and $Q$ is internally disjoint from $P$.
    Let $P^Q$ be the $(u,v)$-path obtained from $P$ by replacing the subpath between the endpoints of $Q$ with the path $Q$.
    We say that $Q$ is a {\em clockwise handle} when $P \sqsubset P^Q$.

    We say that a cylindrical linkage $\pp$ is {\em clockwise-tight}
    if no $P \in \pp$ contains a clockwise handle internally disjoint from all the paths in $\pp$. 
\end{definition}

\mic{Intuitively, being clockwise-tight means that internal points of the paths in the linkage are maximally ``bent'' in the clockwise direction while maintaining disjointedness.
We show that every cylindrical linkage can be modified to be clockwise-tight.}

\begin{lemma}\label{lem:diam:clockwise-tight}
    Let $k \ge 2$ and $G$ be a $k$-cylindrical graph and $\theta \in \Theta^G$.
    There exists a clockwise-tight cylindrical linkage $\pp$ in $G$ with the winding number $\theta$. 
\end{lemma}
\begin{proof}
    Let $\pp$ be any cylindrical linkage in $G$ with the winding number $\theta$.
    We exhaustively apply the following modification to $\pp$:
    while there exists a path $P \in \pp$ with a clockwise handle $Q$  internally disjoint from  $\pp$, replace $P$ with $P^Q$.
    After such a replacement, we obtain a new cylindrical linkage in $G$ with the same winding number $\theta$ (here we use the assumption that $k \ge 2)$.
    We claim that this process must terminate in a finite number of steps.
    If not, an infinite number of replacements happens to a single path $P \in \pp$.
    Hence, there exists an infinite sequence of $(u,v)$-paths $P^1 \sqsubset P^2 \sqsubset \dots$. 
This is impossible because the relation $\sqsubset$ is a linear order and there are only finitely many different $(u,v)$-paths in~$G$. 
\end{proof}

\mic{The second tool is based on the following concept from topology,
used in the analysis of topological spaces with ``holes'', like a torus.
We only provide simple definitions, tailored for our applications. 
}

\begin{definition}[Covering]
    The {\em covering} of \ringfull is a function $\tau  \colon [1,2] \times \rr \to \mathbb{C}$  defined as $\tau((x,y)) = y \cdot \exp\br{\frac{-2i\pi}{k} \cdot x}$.
    We identify the image of $\tau$ with \ringfull. 
\end{definition}

\begin{observation}[Lifting]
    Let $G$ be a $k$-cylindrical graph, $u = (1, \frac{-2\pi}{k}\cdot p)$, $v = (2, \frac{-2\pi}{k}\cdot q)$ (in polar coordinates), and $P$ be a $(u,v)$-path.
    Then for every $\ell \in \mathbb{Z}$ there is a unique curve $P'$ in $[1,2] \times \rr$, called a {\em lifting} of $P$, that starts at $(1,\, \ell\cdot k + p)$, ends at $(2,\, \ell\cdot k + p + \theta(P))$, and $\tau(P') = P$.
    \mic{It holds that $p + \theta(P) \equiv q \mod k$.}
    Moreover, any liftings of two vertex-disjoint paths are disjoint.
\end{observation}

These notions are depicted in \Cref{fig:covering}.
\mic{
Now we can  analyze two linkages with different winding numbers through their liftings in $[1,2] \times \rr$.
Here, we can take advantage of the fact that when
two disjoint curves connect points on a boundary of a topological disc, then these points cannot be intrinsically crossing.}

\begin{figure}
    \centering
\includegraphics[scale=0.9]{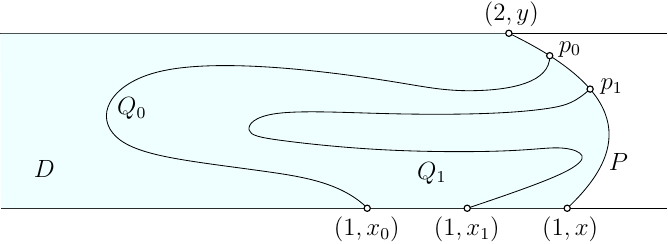}
\caption{An illustration of \cref{obs:diam:lifting-noncross}. The interior of the set $D$ is highlighted. The curves $Q_0, Q_1$ are disjoint so their endpoints on the boundary of $D$ cannot cross. 
}
\label{fig:lifting-noncross}
\end{figure}

\begin{observation}\label{obs:diam:lifting-noncross}
    Consider a curve $P$ in $[1,2] \times \rr$ which starts at $(1,x)$, ends at $(2,y)$, and is internally contained in $(1,2) \times \rr$.
    Let $D$ be the closure of the connected component of $([1,2] \times \rr) \sm P$ containing the point $(1,x-1)$. 
    Next, let $x_0 < x_1 \le x$ and $p_0, p_1 \in P$. 
    Suppose that there exist disjoint curves $Q_0,Q_1$ in $D$ such that
    $Q_0$ connects $(1,x_0)$ to $p_0$ and $Q_1$ connects $(1,x_1)$ to $p_1$.
    Then $p_0$ occurs later than $p_1$ on $P$ when considered oriented from $(1,x)$ to $(2,y)$. 
\end{observation}

\mic{This observation is illustrated in \Cref{fig:lifting-noncross}.
We are ready to prove the main technical lemma about cylindrical graphs, showing that any two cylindrical linkages can be ``disentagled''.
}


\begin{figure}
    \centering
\includegraphics[scale=0.45]{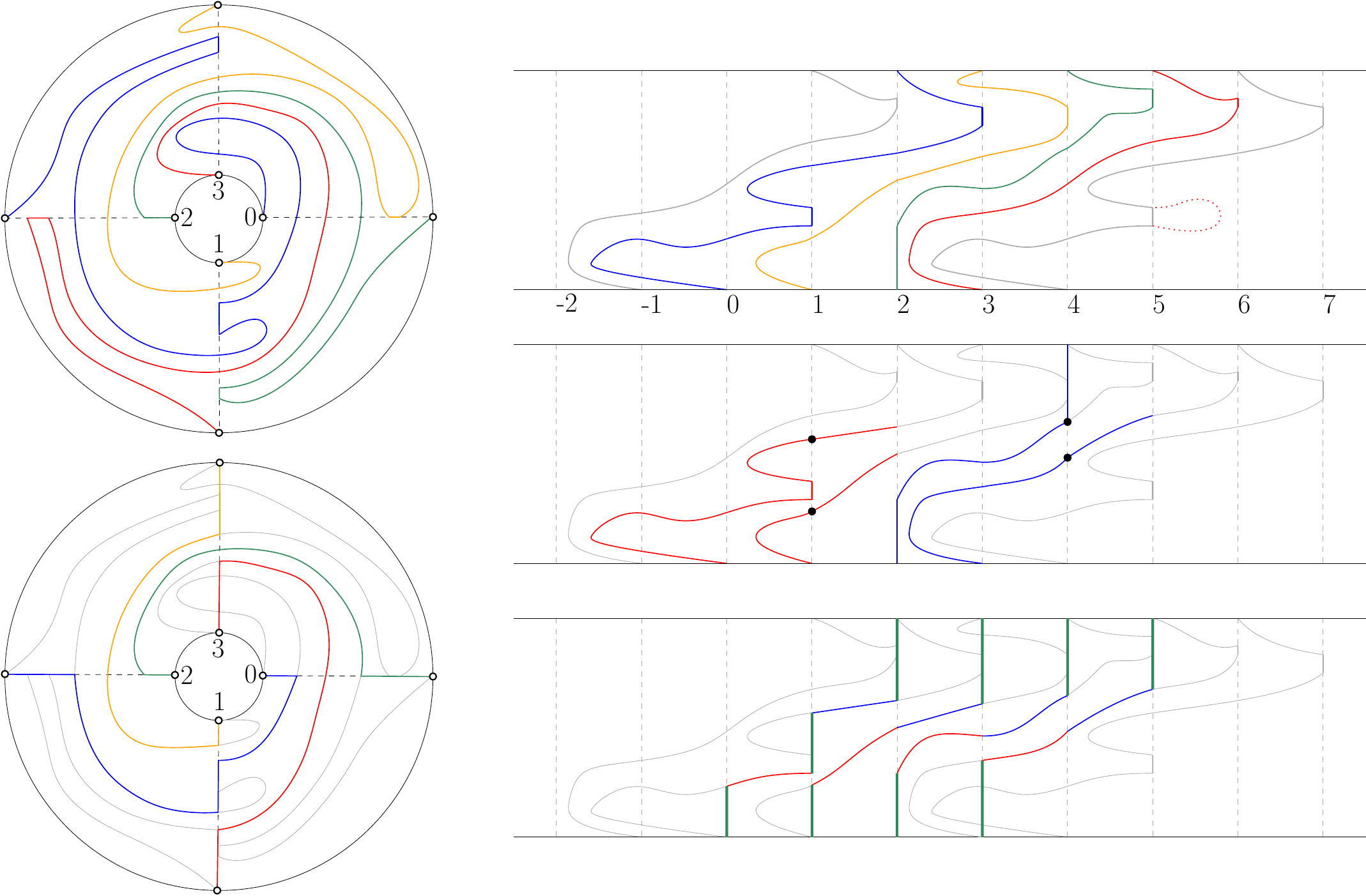}
\caption{An illustration for the proof of \cref{lem:diam:uncross-theta} with $k =4$, $\ell = 2$.
\newline \textcolor{white}{----} Top left: Two cylindrical linkages in a 4-cylindrical graph. For simplicity the linkage $\pp$ is drawn as straight dashed lines and the linkage $\qq$ is drawn in colors.
\newline \textcolor{white}{----} Top right: The covering of \ringfull and the liftings of the paths.
The curves $Q'_0, Q'_1, Q'_2, Q'_3$ are drawn in colors matching their images on the left.
Note that the curve $Q'_{-1}$ coincides with $Q'_3$ modulo a shift.
The same applies to $Q'_0$ and $Q'_4$.
Since $\pp$ is clockwise-tight, none path $P \in \pp$ can have a clockwise handle; hence there cannot be any curve like the red dotted one. 
\newline \textcolor{white}{----} Middle right: The curves $Q_0^2$, $Q_1^1$ (red), and $Q_2^3$, $Q_3^2$ (blue).
The third one is an example of a curve $Q_i^{\ell+1}$ which is not necessarily contained in $Q'_i$ due to the last straight segment.
This forms a special case for property (P5)
but this choice of definition guarantees property (P3).
The relative position of the black disks located on $P'_1$ and $P'_3$ is the subject of
\cref{claim:uncross-theta:later}.
\newline \textcolor{white}{----} Bottom right: The curves $\widehat Q_i^1$ are drawn in red, while the curves $\widehat Q_i^2$ are blue.
Together with the green segments they form the paths
$R'_0, R'_1, R'_2, R'_3$.
\newline \textcolor{white}{----} Bottom left: The images of paths $R'_i$ form the sought family $\mathcal{R}$.
}
\label{fig:covering}
\end{figure}

\begin{lemma}
\label{lem:diam:uncross-theta}
Let $G$ be a $k$-cylindrical graph and
$\theta_1 \le \theta_2 = \theta_1 + \ell$, where $\ell < k$.
If $\theta_1, \theta_2$ are feasible in $G$ then 
there exist cylindrical linkages $\pp, \mathcal{R}$ in $G$ such that 
$\theta(\pp) = \theta_1$,  $\theta(\mathcal{R}) = \theta_2$, and
 and for each $P \in \pp, R \in \mathcal{R}$ the intersection of $P$ and $R$ comprises at most one path.
\end{lemma}
\begin{proof}
    If $k=1$ or $\theta_1 = \theta_2$, then the claim is trivial so we can assume $k \ge 2, \ell \ge 1$.
    Let $\pp, \qq$ be cylindrical linkages with the winding numbers $\theta_1, \theta_2$.
    By \cref{lem:diam:clockwise-tight} we can assume that $\pp$ is clockwise-tight.
    We order the linkages in a clockwise manner: $\pp = P_0, \dots P_{k-1}, \qq = Q_0, \dots Q_{k-1}$, so that $P_i$ and $Q_i$ start at $(1, \frac{-2\pi}{k}i)$.
   Consider the covering $\tau \colon [1,2] \times \rr \to \ringfull$
   and the liftings of $\pp, \qq$.
   More precisely, we consider a unique infinite family $(P'_i)_{i \in \mathbb{Z}}$ of disjoint curves in $[1,2] \times \rr$  
    such that
    $P'_i$ is a path from $(1,i)$ to $(2,i + \theta_1)$ and $\tau(P'_i) = P_{(i \mod k)}$.
    Similarly we define the lifting $(Q'_i)_{i \in \mathbb{Z}}$ of $\qq$.
    Note that each curve $P'_i, Q'_i$ is internally contained in $(1,2) \times \rr$.

    Each curve $Q'_i$ must intersect $P'_i, \dots, P'_{i + \ell}$.
   For $i \in \mathbb{Z}$ and $j \in [1, \ell]$ we define $Q_i^j$ as the minimal prefix of $Q'_i$
    which ends at $P'_{i+j}$.
    Furthermore, let $\widehat Q_i^j$ be the minimal suffix of $Q_i^j$ which starts at $P'_{i+j-1}$. 
    In order to cover the corner cases, we define both $Q^0_i$ and $\widehat Q^0_i$ to be the trivial path from $(1,i)$ to $(1,i)$, $\widehat Q^{\ell+1}_i$ as the trivial path from $(2,i+\theta_2)$ to $(2,i+\theta_2)$ (the last point on $Q'_i$),
    and $Q^{\ell+1}_i$ as the concatenation of $Q^\ell_i$ with the subpath of $P'_{i+\ell}$ from the endpoint of $Q^\ell_i$ to $(2,i+\theta_2)$.
    We make note of the following properties that hold for each $i \in \mathbb{Z}$ and $j \in [0, \ell + 1]$:
    \begin{enumerate}
        \item [(P1)] $\widehat Q_i^j \sub  Q_i^j$, 
        \item [(P2)]$\widehat Q_i^j \sub Q'_i$, 
        \item [(P3)] 
    $Q_i^j$ is internally disjoint from $P'_{i+j}$,
        \item [(P4)] $\tau(Q_i^j)$ is a walk in $G$, 
        \item [(P5)] if $j \le \ell$ then $Q_i^j \sub Q'_i$.
     \end{enumerate}
    
    For $i \in \mathbb{Z}$ we define $R'_i$ as the unique path from $(1,i)$ to $(2,i + \theta_2)$ which is contained  
     in $P'_i \cup \widehat Q^1_i \cup P'_{i+1} \cup \widehat Q^{2}_i \cup \dots \cup \widehat Q^{\ell}_i \cup P'_{i+\ell}$
     \mic{(see \Cref{fig:covering}, bottom right).}
    The intersection of $R'_i$ with $P'_{i+j}$, for $j \in [0,\ell]$,
    is then a subpath of  $P'_{i+j}$ between the endpoints of $\widehat {Q}_{i}^{j}$ and $\widehat {Q}_{i}^{j+1}$.
    \mic{It holds that $\theta(\tau(R'_i)) = \theta_2$.}
     
\begin{claim}\label{claim:uncross-theta:later}
Let $i \in \mathbb{Z}$ and  $j \in [1, \ell]$.
Consider points $p_0, p_{1} \in  [1,2] \times \rr$ such that $p_0\in Q^{j+1}_i \cap P'_{i+j}$ and $p_1 \in Q^j_{i+1} \cap P'_{i+j}$.
Then $p_0$ occurs later than $p_1$ on $P'_{i+j}$, when considered oriented from $(1,i+j)$ to $(2,i+j+\theta_1)$. 
\end{claim}
\begin{innerproof} 
    Here we exploit the fact that $\pp$ is clockwise-tight.
    First consider the case $j < \ell$ as then, by property (P5), the paths $Q^{j+1}_i$ and $Q^{j}_{i+1}$ are disjoint as subpaths of $Q'_i, Q'_{i+1}$.
    By property (P3) both $Q^{j+1}_i$ and $Q^{j}_{i+1}$ are internally disjoint from $P'_{i+j+1}$.
    Let $\tilde Q'_i$ be the prefix of $Q^{j+1}_i$ ending at $p_0$ and
    $\tilde Q'_{i+1}$ be the prefix of $Q^{j}_{i+1}$ ending at $p_1$.   
    Next, let $D$ be the closure of the connected component of $([1,2] \times \rr) \sm P'_{i+j}$ containing the point $(1,i+j-1)$.
    Suppose that $\tilde Q'_i$ or $\tilde Q'_{i+1}$ contains a point $y \not\in D$.
    Then $\tilde Q'_i$ or $\tilde Q'_{i+1}$ 
    has a subpath $Q''$ with both endpoints on $P'_{i+j}$ and internally contained in the region of $[1,2] \times \rr$ between $P'_{i+j}$ and $P'_{i+j+1}$.
    By property (P4) the image $\tau(Q'')$ is a walk  in $G$ and it contains a clockwise handle of $\tau(P'_{i+j})$ which is disjoint from $\pp$; this contradicts $\pp$ being clockwise-tight.
    We obtain that $\tilde Q'_i$, $\tilde Q'_{i+1}$ lie entirely within $D$. 
    The claim follows from \cref{obs:diam:lifting-noncross}.

    Finally, consider the case $j = \ell$ where $Q^{\ell+1}_i$ is not necessarily a subpath of of $Q'_i$.
    However, for $p_0 \in Q^{\ell+1}_i$ chosen as the unique point on $Q^\ell_i \cap P'_{i+\ell}$ the path $\tilde Q'_i$, defined as above, is a subpath of $Q'_i$
    due to property (P5).
    \mic{See \Cref{fig:covering}, middle right.}
    In this case $Q'_i, Q'_{i+1}$ are again disjoint and the same argument applies.
    The general claim follows from the observation that any other point $p \in Q^{\ell+1}_i \cap P'_{i+\ell}$ occurs later than $p_0$ on $P'_{i+j}$.
\end{innerproof}

\begin{claim}\label{claim:uncross-theta:disjoint-covering}
The paths $(R'_i)_{i\in\mathbb{Z}}$ are pairwise disjoint.
\end{claim}
\begin{innerproof}
    It suffices to show that for every $i\in\mathbb{Z}$ the paths $R'_i, R'_{i+1}$ are disjoint.
    By property (P2) the paths of the form $\widehat Q^j_i$, $\widehat Q^{j'}_{i+1}$ belong to the disjoint paths $Q'_i, Q'_{i+1}$ so they cannot intersect each other.
    If  $R'_i, R'_{i+1}$ intersect then there must be  $j \in [1, \ell]$ so that 
    $R'_i \cap P'_{i+j}$ and $R'_{i+1}  \cap P'_{i+j}$ intersect.
    This may happen only if the subpath of $P'_{i+j}$ between the endpoints of $\widehat {Q}_{i}^{j}, \widehat {Q}_{i}^{j+1}$
    and the subpath of $P'_{i+j}$ between the endpoints of $\widehat{Q}_{i+1}^{j-1}, \widehat{Q}_{i+1}^{j}$ have a non-empty intersection.
    This is impossible due to property (P1) and \cref{claim:uncross-theta:later}.
\end{innerproof}

\mic{It follows from the construction that  $\tau(R'_i) = \tau(R'_j)$ whenever $i \equiv j \mod k$.  
We can thus define $R_0, \dots, R_{k-1}$ as the path family in $G$ such that $\tau(R'_j) = R_{(j \mod k)}$ for each $j \in \mathbb{Z}$.
Since $\ell < k$, no path $R'_i$ intersects both $P'_j$ and $P'_{j+k}$ for any $j \in \mathbb{Z}$.
As each intersection $R'_i \cap P'_j$ is either empty or a single path, we infer that also each intersection $R_i \cap P_j$ is either empty or a single path.

Consider  $0 \le i < j < k$: from \cref{claim:uncross-theta:disjoint-covering} we know that $R'_i$, $R'_{j}$ are disjoint.
Moreover, 
the path $R'_j$ is contained in the region $D$ of $[1,2] \times \rr$ between $R'_{i}$ and $R'_{i+k}$, exclusively.
Because $\tau(R'_i) = \tau(R'_{i+k}) = R_i$ and
$\tau(D) \cap R_i = \emptyset$,
we obtain that $R_i$ and $R_j$ form disjoint subsets of \ringfull; hence they are vertex-disoint paths. }
We conclude that $\mathcal{R} = R_0 \dots, R_{k-1}$ is the desired linkage with the winding number $\theta_2$ and single-path intersections with $\pp$.
\end{proof}

\mic{Since the union of two disentangled linkages cannot contain too many concentric cycles, we can now find an edge to which the criterion from \cref{lem:diam:criterion} applies.}

\begin{proposition}
\label{lem:diam:maximal}
Let $G$ be a $k$-cylindrical graph with $\rdist_G(V_{in},V_{out}) \ge (k+2)^2$.
Then there exists an edge $e \in E(G)$ such that
$G \sm e$ is $(V_{in} \cup V_{out})$-linkage-equivalent to $G$.
Furthermore, such an edge can be found in polynomial time.
\end{proposition}
\begin{proof}
Let $C_1,\dots, C_m$ be a $(V_{in}, V_{out})$-sequence of concentric cycles in $G$ with
$m = (k+2)^2 - 1$.
We need to show that there exists an edge satisfying the requirements of \cref{lem:diam:criterion}.
Note that $\theta^G_2 <\theta^G_1 + k$. 
Let $\pp_1, \pp_2$ be the linkages from \cref{lem:diam:uncross-theta} such that $\theta(\pp_1) = \theta^G_1$, $\theta(\pp_2) = \theta^G_2$, and $G'$ be the union of $\pp_1, \pp_2$.
We claim that there exists $i \in [2k+1, m-2k]$ and $e \in E(C_i)$
such that $e \not\in E(G')$.
Suppose otherwise.
Then there is a $(V_{in}, V_{out})$-sequence of $m - 4k$ concentric cycles in $G'$ and so  $\rdist_{G'}(V_{in},V_{out}) \ge m - 4k \ge k + 3$.
Due to \cref{lem:diam:uncross-theta}, for each $P_1 \in \pp_1, P_2 \in \pp_2$ the intersection of $P_1$ and $P_2$ has at most one path.
Consider a region $R$ of \ringfull between two consecutive paths from $\pp_1$. 
Each path from $\pp_2$ has at most one subpath intersecting $R$, which gives at most $k$ subpaths in total.
As a consequence, $\rdist_{G'}(V_{in},V_{out}) \le k+1$ and we get a contradiction with the assumption that there is no edge $e$ obeying the specification above.
Therefore, there exists $i \in [2k+1, m-2k]$ and $e \in E(C_i)$ such that both linkages $\pp_1, \pp_2$ are present in $G \sm e$.
As both $\theta^G_1, \theta^G_2$ are feasible in $G \sm e$, the criterion from \cref{lem:diam:criterion} applies.

In order to detect such an edge, we simply enumerate all edges in $G$ and check for each $e$ whether $e$ satisfies the requirements of \cref{lem:diam:criterion}.
This can be done in polynomial time with \cref{lem:diam:enumerate}.
\end{proof}

\subsubsection{Finding an irrelevant edge}

\mic{We can now combine the two strategies for detecting an irrelevant edge to process any graph properly embedded in \ringfull with sufficiently many concentric cycles.

Note that the notation in the following lemma differs slightly from that in the outline.
The separators $V_{in}, V_{out}$ therein become here $S^1_{in}, S^1_{out}$, the separators $S_{in}, S_{out}$ become  $S^2_{in}, S^2_{out}$, and $S$ corresponds to $S^3_{in}$.}

\begin{figure}[t]
    \centering
\includegraphics[scale=1.3]{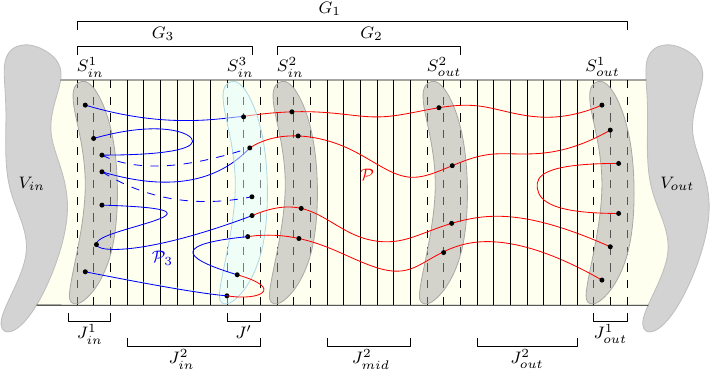} 
\caption{An illustration for the proof of \cref{lem:diam:ring-irrelevant}.
The planar structure is discarded here (cf. \Cref{fig:irrelevant-edge-outline}) and the innermost layers of the graph $G$ are portrayed to the left.
The separators $S^1_{in}, S^2_{in}, S^2_{out}, S^1_{out}$ are sketched gray while the separator $S^3_{in}$ is light blue.
The vertical lines represent the family of cycles $C_1,\dots,C_m$, where the dashed lines are the cycles intersecting one of the separators above.
An $(S^1_{in}, S^1_{out})$-linkage $\pp$ illustrates the argument for the case where $J^2_{in}$ is large.
The blue linkage $\pp_3$ is given by the intersections of $(S^1_{in}, S^1_{out})$-paths from $\pp$ with the subgraph $G_3$.
Since $S^2_{in}$ is the minimal $(S^1_{in}, S^1_{out})$-separator closest to $S^1_{in}$, there exists a $(S^1_{in}, S^3_{in})$-linkage larger than $\pp_3$: this is shown with the dashed paths. 
This observation allows us to use \cref{lem:diam:non-maximal-irrelevant} to find an edge in $G_3$ that is not needed for any choice of $\pp$.
}
\label{fig:irrelevant-edge}
\end{figure}

\begin{lemma}
\label{lem:diam:ring-irrelevant}
Let $G$ be a plane graph properly embedded in \ringfull and 
$C_1,\dots, C_m$ be a $(V_{in}, V_{out})$-sequence of concentric cycles.
Suppose that $m \ge 3(t+4)^2$ where $t = \tw(G) + 1$. 
Then there exists an edge $e \in E(G)$ such that
$G \sm e$ is $(V_{in} \cup V_{out})$-linkage-equivalent to $G$.
Furthermore, such an edge can be found in polynomial time.
\end{lemma}
\begin{proof}
Let $J^1_{in} = [1,t + 2]$ and $J^1_{out} = [m-t-1,m]$.  
By \cref{lem:diam:sep-exists} there exists a $(C_1,C_{t+2})$-separator $S^1_{in}$ of size less than $t$ and a 
$(C_{m-t-1},C_{m})$-separator $S^1_{out}$ of size less than $t$.
By \cref{lem:prelim:noose-cycle} there exist $G$-nooses $N^1_{in}, N^1_{out}$ contained respectively in $\ring(C_{1},C_{t+2})$, $\ring(C_{m-t-1},C_{m})$, so that $N^1_{in} \cap V(G) = S^1_{in}$ and likewise for $N^1_{out}$.
Let $G_1$ be the subgraph of $G$ induced by the vertices located in $\ring(N^1_{in}, N^1_{out})$.
By \cref{lem:prelim:equivalent} it suffices to find an edge $e \in  E(G_1)$ so that $G_1 \sm e$ is $(S^1_{in} \cup S^1_{out})$-linkage-equivalent to $G_1$.
This simplifies our task because $|S^1_{in}|, |S^1_{out}| < t$.

There are at least $m - 2(t+1) \ge 3(t+3)^2 + 2t$ cycles from $C_1,\dots, C_m$ lying entirely between $N^1_{in}$ and  $N^1_{out}$.
Let $S^2_{in}, S^2_{out}$ be the minimum-size $(S^1_{in}, S^1_{out})$-separators in $G_1$ that are closest to $S^1_{in}, S^1_{out}$, respectively.
They can be found in polynomial time (\cref{thm:prelim:closets-cut}).
Let $p$ denote $|S^2_{in}| = |S^2_{out}|$.
By \cref{lem:diam:sep-exists} we have $p < t$.
By \cref{lem:prelim:separator-interval}, each of $S^2_{in}, S^2_{out}$ intersects at most $p$ consecutive cycles from $C_1, \dots, C_m$.
Let $J^2_{in}, J^2_{mid}, J^2_{out}$ denote the intervals representing indices of cycles in $C_1, \dots, C_m$ which lie respectively: between $S^1_{in}$ and  $S^2_{in}$, between $S^2_{in}$ and $S^2_{out}$, between $S^2_{out}$ and  $S^1_{out}$. 
\mic{Note that the separators $S^2_{in}, S^2_{out}$ may intersect; in this case $J^2_{mid} = \emptyset$.}
We have $|J^2_{in}| +  |J^2_{mid}| + |J^2_{out}| \ge 3(t+3)^2$.
One of these intervals must contain at least $(t+3)^2$ elements.
We distinguish two scenarios.

\paragraph{Deep cylindrical subgraph in the middle.} First suppose that $|J^2_{mid}| \ge (t+3)^2$.
Let $N^2_{in}, N^2_{out}$ be the $G_1$-nooses corresponding to the separators $S^2_{in}, S^2_{out}$ and
$G_2$ be the subgraph of $G_1$ induced by the vertices lying in $\ring(N^2_{in}, N^2_{out})$. 
Note that $\mu_{G_2}(S^2_{in}, S^2_{out}) = p$ due to the choice of $S^2_{in}, S^2_{out}$.
We can thus transform $G_2$ in a homotopic way to a $p$-cylindrical graph $H$.
There may be many ways to obtain $H$ which differ by a relative cyclic shift between $S^2_{in}$, $S^2_{out}$ and give different sets $\Theta^H$ but we can choose an arbitrary one.
\cref{lem:diam:maximal} works regardless of the chosen embedding of $H$ and it gives a polynomial-time algorithm to find an edge $e \in E(H) = E(G_2)$ such that $G_2 \sm e$ is $(S^2_{in} \cup S^2_{out})$-linkage-equivalent to $G_2$.
\mic{Since $S^2_{in} \cup S^2_{out}$ separates the endpoints of $e$ from $S^1_{in} \cup S^1_{out}$, we obtain from \cref{lem:prelim:equivalent} that 
$G_1 \sm e$ is $(S^1_{in} \cup S^1_{out})$-linkage-equivalent to $G_1$.}
\meir{Add another sentence why it follows. \mic{done}}

\paragraph{Deep subgraph with a non-maximal linkage.} 
Suppose now that $|J^2_{in}| \ge (t+3)^2$ or $|J^2_{out}| \ge (t+3)^2$.
These cases are symmetric so we only examine the first one.
Let $J'$ be the subinterval of $J^2_{in}$ comprising its last $t+2$ elements
(representing the cycles closest to $S^2_{in}$).
By the same argument as before, there exists an $(S^1_{in}, S^2_{in})$-separator $S^3_{in}$ of size at most $t$ given by a $G_1$-noose $N^3_{in}$ which may intersect only these cycles from  $C_1, \dots, C_m$ with indices within $J'$. 
Let $G_3$ be the subgraph of $G_1$ induced by the vertices lying in $\ring(N^1_{in}, N^3_{in})$. 
There are at least $(t+3)^2 - (t+2) \ge (t+2)^2$ cycles from $C_1, \dots, C_m$ lying in  the interior of  $\ring(N^1_{in}, N^3_{in})$.
Since $S^2_{in}$ was chosen as the minimum $(S^1_{in}, S^1_{out})$-separator in $G_1$ closest to $S^1_{in}$, 
there are no 
 such separators within $\ring(N^1_{in}, N^3_{in})$ of size $p$ or smaller. 
This implies that $\mu_{G_3}(S^1_{in}, S^3_{in}) \ge p + 1$. 
As a result, $G_3$ satisfies the preconditions of \cref{lem:diam:non-maximal-irrelevant} with $\max(|S^1_{in}|, |S^3_{in}|) \le t$ and $s \ge p+1$; let $e \in E(G_3)$ be an edge provided by that proposition. 
We will show that $G_1 \sm e$ is $(S^1_{in} \cup S^1_{out})$-linkage-equivalent to~$G_1$.

Let $\pp$ be an $(S^1_{in} \cup S^1_{out})$-linkage in $G_1$.
Recall that $|S^1_{in}|, |S^1_{out}| \le t$.
Because the interval $J^2_{in} \sm J'$ is sufficiently long and due to \cref{lem:diam:visitors}, there exists a linkage $\pp'$ in $G_1$ that is aligned with $\pp$ and every inclusion-minimal $(S^1_{in} \cup S^1_{out})$-subpath of $P \in \pp'$ which is an $(S^1_{in}, S^1_{in})$-path does not intersect $S^3_{in}$. 
Let $\pp'_\textrm{long}$ be the family of inclusion-minimal $(S^1_{in} \cup S^1_{out})$-subpaths of paths in $\pp'$ which are $(S^1_{in}, S^1_{out})$-paths.
There are at most $p$ paths in $\pp'_\textrm{long}$ because every such path must intersect the separator $S^2_{in}$ of size $p$.
\mic{Observe that
every minimal $(S^1_{in} \cup S^1_{out})$-subpath in $\pp'$ that visits both $S^1_{in}, S^3_{in}$ belongs to $\pp'_\textrm{long}$.}
Let $\pp_3$ be a linkage in $G_3$ given by the maximal intersections of paths from $\pp'$ with $V(G_3)$; then $\pp_3$ is a $(S^1_{in} \cup S^3_{in})$-linkage (see \Cref{fig:irrelevant-edge}).
By the observation above, each $(S^1_{in}, S^3_{in})$-path $Q \in \pp_3$
\mic{intersects} some path $Q' \in \pp'_\textrm{long}$ and the mapping $Q \to Q'$ is injective.
We infer that $\pp_3$ contains at most $p$ many  $(S^1_{in}, S^3_{in})$-paths.
We apply \cref{lem:diam:non-maximal-irrelevant} 
with $s \ge p+1$ to derive that there exists a linkage in $G_3 \sm e$ aligned with $\pp_3$. 
As a result, $G_1 \sm e$ contains a linkage aligned with $\pp$.
This concludes the proof.
\end{proof}

\mic{Finally, we show that when a plane graph $G$ has sufficiently large radial diameter, then we can find a subgraph of $G$ to which \cref{lem:diam:ring-irrelevant} applies.
This allows us to detect an irrelevant edge in $G$.}

\begin{proposition}
\label{prop:diam:final}
Let $G$ be a plane graph, $X \sub V(G)$ be of size $k$, and $t = \tw(G) + 1$.
Suppose that the radial diameter of $G$ is at least $7(k+1)(t+5)^2$.
Then we can find, in polynomial time, an edge $e \in E(G)$ such that $G\sm e$ is $X$-linkage-equivalent to $G$.
\end{proposition}
\begin{proof}
Let $v$ be a vertex on the outer face of $G$.
By the triangle inequality we obtain that there must exist a vertex $u$ with $\rdist_G(u,v) > 3(k+1)(t+5)^2$.
By \cref{lem:prelim:concentric-radial} there exists a $(\{u\}, \{v\})$-sequence of concentric cycles  $C_1, \dots, C_m$, where $m = 3(k+1)(t+5)^2$.
For $i\in [m-1]$ let $V_i = V(G) \cap( \disc(C_{i+1}) \sm \disc(C_{i}))$.
These sets are disjoint and there may be at most $k$ of them containing a vertex from $X$.
Hence, there is an interval $J \sub [m-1]$ of length $(m - 1 - k) / (k+1) \ge 3(t+4)^2 + 1$ where $V_j \cap X = \emptyset$ for $j \in J$.
Let $i = \min(J)$, $j = \max(J)$, and $U$ be the set of vertices lying between $C_i$ and $C_{j+1}$.
Then $U \cap X = \emptyset$.
Let $G' = G[U \cup V(C_i) \cup V(C_{j+1})]$;
then $\tw(G') \le \tw(G) = t-1$ and $G'$ is properly embedded in \ringfull for some curves $I_{in}, I_{out}$
such that $I_{in} \cap G' = C_i$ and $I_{out} \cap G' = C_{j+1}$.
Furthermore, $C_{i+1}, \dots, C_{j}$ forms a $(V(C_i), V(C_{j+1}))$-sequence of concentric cycles in $G'$ of length $3(t+4)^2$.
We apply \cref{lem:diam:ring-irrelevant} to find an edge $e \in E(G')$ such that $G' \sm e$ is $(V(C_{i}) \cup V(C_{j+1}))$-linkage equivalent to $G'$.
It follows from \cref{lem:prelim:equivalent} that $G \sm e$ is $X$-linkage-equivalent to $G$.
\end{proof}

\subsection{Single-face case}
\label{sec:tw:single-face}

Let $G$ be a plane graph properly embedded in $\disc(I)$ and $X = V(G) \cap I$.
We say the a set of disjoint pairs $\tcal \sub X^2$ is {\em cross-free} if it does not contain pairs $(a,c), (b,d)$ so that $a,b,c,d$ lie on $I$ in this order.
A division $X = X_1 \cup X_2$ is called {\em canonical} if both $X_1,X_2$ are non-empty and there are points $y_1, y_2 \in I$
so that $X_1, X_2$ belong to different connected components of $I \sm \{y_1,y_2\}$.
We define $\mu_\tcal(X_1,X_2)$ to be the number of pairs in $\tcal$ with one element in $X_1$ and the other one in $X_2$.

\mic{It turns out that when a graph is properly embedded in a disc and all the terminals occur at the boundary of the disc, then again the cut-condition is sufficient for a linkage to exist.}

\begin{lemma}[{\cite[Lem. 3.6]{GM6}}]
\label{lem:single:criterion}
Let $G$ be properly embedded in $\disc(I)$, $V(G) \cap I = X$, and  $\tcal \sub X^2$.
Then $\tcal$ is realizable in $G$ if and only if $\tcal$ is cross-free and for every canonical division $(X_1,X_2)$ of $X$ it holds that $\mu_G(X_1,X_2) \ge \mu_\tcal(X_1,X_2)$.
\end{lemma}

\mic{Our goal now is to compress a given graph with terminal set $X$ located on the boundary of the disc, to an $X$-linkage-equivalent graph of size $|X|^{\Oh(1)}$.
By the lemma above, it is sufficient to preserve the sizes of minimum $(X_1,X_2)$-separators for all canonical divisions  $(X_1,X_2)$.
Our strategy is to mark $|X|^{\Oh(1)}$ vertices covering all the relevant separators and then replace the remaining parts of the graph with gadgets that are ``at least as good''.}

\begin{lemma}
\label{lem:single:gadget}
Let $I$ be a noose and $X \sub I$ be a finite set of size $k$.
There exists a plane graph $H$ 
properly embedded in $\disc(I)$
on at most $k^2$ vertices such that $H \cap I = X$ and every cross-free $\tcal \sub X^2$ is realizable in $H$.
This graph can be constructed in time polynomial in~$k$.
\end{lemma}
\begin{proof}
\mic{
    The lemma is trivial for $k \le 2$, so we will assume $k \ge 3$.
    We use the $(k,k)$-cylindrical grid $C_{k}^k$ (\cref{def:diam:grid}), which has $k^2$ vertices, and identify the vertices on the outer cycle $C^k$ with the points from $X$.
    It is easy to see that $C_{k}^k$ can be properly embedded in $\disc(I)$ in such a way that $C^k_k \cap I = C^k$.
    We argue that $C^k_k$ satisfies the lemma using the criterion from \cref{lem:single:criterion}.
    We claim that for any canonical division $(X_1,X_2)$ of $X = C^k$ it holds that  $\mu_{C^k_k}(X_1,X_2) = \min(|X_1|, |X_2|)$.
    Let $p =  \min(|X_1|, |X_2|)$ and
    assume w.l.o.g. that $X_1 = \{c^k_i \mid i \in [p]\}$.
    For $i \in [p]$ let $P_i$ be the unique $(c^k_i,c^k_{k+1-i})$-path contained in the paths $C_i$, $C_{k+1-i}$ and the subpath of $C^{k+1-i}$ between $C_i$, $C_{k+1-i}$ that contains the vertex $c^{k+1-i}_1$.
    Then $P_1, \dots, P_\ell$ form an $(X_1,X_2)$-linkage implying that $\mu_{C^k_k}(X_1,X_2) \ge p$.
    In fact, we get equality because each of the sets $X_1,X_2$ is an $(X_1,X_2)$-separator.
    For any $\tcal \sub X^2$ it holds that $\mu_\tcal(X_1,X_2) \le  \min(|X_1|, |X_2|)$.
    Hence any cross-free $\tcal$ is realizable in $C^k_k$ due to \cref{lem:single:criterion}.
    }
\end{proof}

\mic{The following fact will come in useful for estimating the number of necessary gadgets.}

\begin{lemma}[{\cite[Lem. 13.3]{fomin2019kernelization}}]
\label{lem:single:bipartite}
Let $G$ be a planar graph, $X \subseteq V(G)$, and let $N_3$ be a set of vertices from $V(G) \setminus X$ such that every vertex from $N_3$ has at least three neighbors in {$X$}. Then $|N_3| \le 2\cdot |X|$.
\end{lemma}

\begin{proposition}
\label{prop:single:final}
Let $G$ be properly embedded in $\disc(I)$, $V(G) \cap I = X$,
and $k = |X|$.
One can construct, in polynomial time, a graph $\widehat G$ on $\Oh(k^6)$ vertices, properly embedded in $\disc(I)$, that is $X$-linkage-equivalent to $G$ and with $V(\widehat G) \cap I = X$.
\end{proposition}
\begin{proof}
For each canonical division $(X_1,X_2)$ of $X$ we compute a minimum-size $(X_1,X_2)$-separator $S(X_1,X_2)$.
Clearly, $|S(X_1,X_2)| \le k$.
When one of the sets $X_1,X_2$ is a singleton $\{v\}$ then we can assume that $S(X_1,X_2) = \{v\}$.
Otherwise the separator $S(X_1,X_2)$ can be represented by a simple curve $N(X_1,X_2) \sub \disc(I)$ connecting two points on $I$ so that 
$N(X_1,X_2) \cap G = S(X_1,X_2)$
and every curve within  $\disc(I)$ connecting points from $X_1$ and $X_2$ must intersect $N(X_1,X_2)$.
Let $N$ be the union of all the curves $N(X_1,X_2)$ and $S$ be the union of all the sets $S(X_1,X_2)$.
By the argument above, we have $X \sub S$.
There are at most $k^2$ canonical divisions
so $|S| \le k^3$.

We say that $F$ is a face of $(I,N)$ if it is \mic{a closure of} an inclusion-maximal subset of $\disc(I) \sm N$.
\mic{Every face $F$ is of the from $F = \disc(\partial F)$.}
For a face $F$ let $V_F = V(G) \cap \inter(F)$, and $X_F = V(G) \cap \partial F$; clearly $X_F \sub S$. 
\mic{Note that $G\cap F = G[V_F \cup X_F]$ is properly embedded in $F$.}
If $V_F \ne \emptyset$, we replace the subgraph 
$G \cap F$ with
a graph $H_F$ given by \cref{lem:single:gadget} applied to $\partial F$ and $X_F$.
Then  $H_F \cap \partial F = X_F$ and $|V(H_F)| \le |X_F|^2$.
Observe that this modification does not affect $G \cap N$.
Let $G'$ be obtained from $G$ by applying this modification to every face $F$ of $(I,N)$.

\begin{claim}
For every canonical division $(X_1,X_2)$ of $X$ it holds that $\mu_{G'}(X_1,X_2) \le \mu_{G}(X_1,X_2)$.
\end{claim}
\begin{innerproof}
By construction $G' \cap N(X_1,X_2) = G \cap N(X_1,X_2) = S(X_1,X_2)$. Therefore $\mu_{G'}(X_1,X_2) \le |S(X_1,X_2)| = \mu_{G}(X_1,X_2)$.
\end{innerproof}

\begin{claim}
For every canonical division $(X_1,X_2)$ of $X$ it holds that $\mu_{G'}(X_1,X_2) \ge \mu_{G}(X_1,X_2)$.
\end{claim}
\begin{innerproof}
Let $\pp$ be an $(X_1,X_2)$-linkage in $G$ of size  $\mu_{G}(X_1,X_2)$.
We claim that there exists an $(X_1,X_2)$-linkage $\pp'$ in $G'$ of the same size.
Let $\pp = \{P_1, \dots, P_s\}$.
Let $F$ be a face of $(I,N)$ and $\pp_F$ be the family of maximal subpaths of $P_1, \dots, P_s$ that traverse $F$.
This is an $X_F$-linkage in $G \cap F$.
Let $\tcal_F$ be the set of pairs representing the endpoints of $\pp_F$; then $\tcal_F$ is cross-free with respect to $\partial F$.
As a consequence of \cref{lem:single:gadget}, $\tcal_F$ is realizable in $H_F$.
By applying this argument to every $F$, we turn $\pp$ into a linkage $\pp'$ in $G'$ connecting the same pairs of vertices in $X$.
Hence, $\mu_{G'}(X_1,X_2) \ge |\pp'| = |\pp| = \mu_{G}(X_1,X_2)$.
\end{innerproof}

With the two claims above we apply \cref{lem:single:criterion} to infer that $G'$ is $X$-linkage-equivalent to $G$.
Finally, we apply two reduction rules to bound the size of $G'$.
When there exists a vertex set $C \sub V(G) \sm S$ such that 
$N_{G'}(C) \sub  S$ and (a) $|N_{G'}(C)| = 1$, then remove $C$; or
(b) $N_{G'}(C) = \{u,v\}$, then replace $C$ with the edge $uv$ (when such an edge already exists, do nothing).
Let $\widehat G$ be the result of applying these rules to $G'$.
These modifications preserve $X$-linkages and $\widehat G$ remains properly embedded in $\disc(I)$.
\mic{Moreover, for every face $F$ of $(I,N)$ with a non-empty set $V(\widehat G) \cap \inter(F)$ it holds that $|X_F| \ge 3$.}

\begin{claim}
The graph $\widehat G$ has $\Oh(k^6)$ vertices.
\end{claim}
\begin{innerproof}
Consider a graph $\widehat G^c$ obtained from $\widehat G$ by contracting each connected component of $\widehat G - S$ into a single vertex. 
 Let $B$ be the set of the vertices created due to contractions.
Each vertex from $B$ corresponds to some face $F$ of $(I,N)$
with $V(\widehat G) \cap \inter(F) \ne \emptyset$ and $|X_F| \ge 3$,
so the minimum degree in $B$ is at least 3.
By \cref{lem:single:bipartite} the size of $B$ is at most $2|S| \le 2k^3$.
\mic{Due to planarity,} the number of edges in $\widehat G^c$ is at most $3\cdot (|B| + |S|) = \Oh(k^3)$.
Let $\ff$ be the set of faces $F$ of $(I,N)$ with $V(\widehat G) \cap \inter(F) \ne \emptyset$. 
We have $\sum_{F \in \ff} |X_F| \le |E(\widehat G^c)|$.
In turn, $\sum_{F \in \ff} |V(H_F)| \le  \sum_{F \in \ff} |X_F|^2 \le  (\sum_{F \in \ff} |X_F|)^2 = \Oh(k^6)$.
This entails the claimed bound on the size of $\widehat G$.
\end{innerproof}

The construction of $\widehat G$ can be easily performed in polynomial time.
The proposition follows.
\end{proof}

\subsection{Cutting the graph open}
\label{sec:tw:cutting}

In this section, we finalize the construction of a polynomial kernel.
After reducing the radial diameter, we can find a tree of moderate size spanning the set of terminals $X$ in the radial graph.
We shall {\em cut the graph open} alongside this tree to reduce the problem to the case where all the terminals lie on a single face.
The following transformation has been used in the algorithms for \textsc{Steiner Tree} \cite{BorradaileKM09, PilipczukPSvL18}
and \textsc{Vertex Multiway Cut} \cite{JansenPvL19} on planar graphs.

\begin{definition}[Cut alongside a tree]
\label{def:open:cut}
    Let $G$ be a plane graph 
    and $T$ be a tree in the radial graph of $G$. 
    The plane graph $G^T$ is obtained from $G$ as follows.
    Consider an Euler tour of $T$ that traverses each edge twice in different directions, and respects the plane embedding of $T$.
    We replace each vertex $v \in V(T) \cap V(G)$ with $\deg_T(v)$ many copies, reflecting its occurrences on the Euler tour, and distribute the copies in the plane, creating a new face 
    incident to all the created copy-vertices (we refer to the set of these vertices as $V_T \sub V(G^T))$. 
    For $v \in V(T) \cap V(G)$ let $\Gamma_T(v) \sub V_T$ be the set of copies of $v$ created during this process.
\end{definition}

This construction is depicted in \Cref{fig:split-outline} on page \pageref{fig:split-outline}.
Since the sum of vertex degrees in a tree is at most twice the number of its vertices, we obtain the following observation.

\begin{observation}\label{lem:open:size}
    For a plane graph $G$ and a tree $T$ in the radial graph of $G$, we have $|V_T| \le 2 \cdot |V(T)|$.
\end{observation}

We show that for the sake of obtaining an equivalent instance of {\sc Planar Disjoint Paths}, we can focus on the new instance obtained via the cutting operation.

\begin{lemma}\label{lem:open:equivalent}
    Let $G_1, G_2$ be plane graphs sharing a vertex set $Y$.
    Next, let $T_1, T_2$ be trees in the radial graphs of $G_1, G_2$, respectively, such that $Y = V(T_1) \cap V(G_1) = V(T_2) \cap V(G_2)$, $V_{T_1} = V_{T_2}$ (we refer to this set as $Y'$) and for each $v \in Y$ it holds that $\Gamma_{T_1}(v) = \Gamma_{T_2}(v)$.
    If $G_1^{T_1}$ and $G_2^{T_2}$ are $Y'$-linkage-equivalent, 
    then $G_1$ and $G_2$ are $Y$-linkage-equivalent.
\end{lemma}
\begin{proof}
    By symmetry, it suffices to prove that when $\tcal \sub Y \times Y$ is realizable in $G_1$, then it is also realizable in $G_2$.
    Let $\pp_1$ be a $\tcal$-linkage in $G_1$.
    We say that a path $Q$ in $G_1$ does not cross $T_1$ if for each pair of consecutive edges $e_1, e_2 \in E_{G_1}(V)$ on $Q$ there is a single $v$-copy $v' \in \Gamma_{T_1}(v)$ such that $e_1, e_2 \in E_{G_1^{T_1}}(v')$.
    Note that when a subpath $Q$ of a $(Y,Y)$-path does not cross $T_1$ then $Q$ corresponds to a unique path in $G_1^{T_1}$.
    We partition each path $P \in \pp_1$ into maximal subpaths that do not cross $T_1$; let $\Gamma_1(P)$ denote the family of corresponding paths in $G_1^{T_1}$.
    We define $\tcal'  \sub Y' \times Y'$ as follows. 
    First, we insert to $\tcal'$ the endpoints of each path from  $\Gamma_1(P)$ for every  $P \in \pp_1$.
    Next, when $v' \in V_{T_1}$ does not belong to any of the paths above, we insert the pair $(v',v')$ to $\tcal'$; we then say that $v'$ is {\em blocked}.
    Clearly, $\tcal'$ is realizable in $G_1^{T_1}$ and, by the assumption, it is realizable in $G_2^{T_2}$ as well.

    Let $\pp'_2$ be a $\tcal'$-linkage in $G_2^{T_2}$.
    We need to show that it can be merged back to a $\tcal$-linkage in $G_2$.
    For $P \in \pp_1$ let $\Gamma_2(P) \sub \pp'_2$ be the linkage aligned with $\Gamma_1(P)$; these families are disjoint for distinct $P \in \pp_1$.
    The paths from  $\Gamma_2(P)$ can be merged into a path $\widehat P$ in $G_2$ with the same endpoints as $P$; let $\pp_2$ be the union of such paths.
    We argue that different paths from $\pp_2$ are vertex-disjoint.
    There is a 1-1 mapping between the vertices from $V(G_2^{T_2}) \sm Y'$ and $V(G_2) \sm Y$ so we only need to check that no vertices from $Y$ collide.
    Consider $v \in Y$; 
    if $v$ is not being visited by any path from $\pp_1$, then for each  $v' \in \Gamma_{T_2}(v)$ the pair $(v',v')$ belongs to $\tcal'$ (i.e., $v'$ is blocked) and so $v$ cannot belong to any path from~$\pp_2$.
    Suppose that $v \in V(P)$ for some $P \in \pp_1$; we consider two cases.
    If $v$ is an endpoint of a maximal subpath of $P$ that does not cross $T_1$, then either $v$ is an endpoint of $P$ or there are two vertices $v', v'' \in \Gamma_{T_2}(v)$ which are endpoints of paths in $\Gamma_2(P)$ whereas all the remaining vertices from $\Gamma_{T_2}(v)$ are blocked.
    Therefore $v$ is being visited only by $\widehat P$.
    In the last case, $v$ is being visited by $P$ but no vertex from $\Gamma_{T_2}(v)$ is an endpoint of a path in  $\Gamma_2(P)$.
    Then there is exactly one vertex $v' \in \Gamma_{T_2}(v)$ and one path $P' \in \Gamma_1(P)$ that visits $v'$ while  the remaining vertices from $\Gamma_{T_2}(v)$ are blocked.
    So also in this case $v$ is being visited only by $\widehat P$.
    We infer that $\pp_2$ is a $\tcal$-linkage in $G_2$ aligned with $\pp_1$; this concludes the proof. 
\end{proof}

We are ready to prove \cref{thm:outline:polyKer} which, in turn, implies \cref{thm:polyKer}.

\polyKer*
\begin{proof}
    Consider an arbitrary plane embedding of $G$.
    While the radial diameter of $G$ is larger than $7(k+1)(\twsf+6)^2$, we apply \cref{prop:diam:final}
    to find an irrelevant edge and reduce the size of $G$ while maintaining $X$-linkage-equivalency.
    By applying this reduction exhaustively, we can
    assume that $G$ has radial diameter $d = \Oh(k\cdot \twsf^2)$.
    We greedily construct a Steiner tree $T$ of $X$ in the radial graph of $G$: we order the vertices $x_1, x_2, \dots, x_k$ of $X$ arbitrarily and for $i = 1, 2,\dots, k$ we construct a tree $T_i$ by finding a shortest radial path between $x_i$ and a vertex of $T_{i-1}$.
    In each step we augment the tree with a path of length at most  $2d = \Oh(k\cdot \twsf^2)$, so, eventually, $|V(T)| = \Oh(k^2\cdot \twsf^2)$.
    
    We cut $G$ open alongside $T$ obtaining a graph $G^T$.
    It can be properly embedded in $\disc(I)$ for some noose $I$ in such a way that $G^T \cap I = V_T$ (this requires flipping the embedding to turn the newly created face into the outer face). 
     \cref{lem:open:size} implies that $|V_T| = \Oh(k^2\cdot \twsf^2)$.
    We apply  \cref{prop:single:final} to replace $G^T$ with a graph $H$ on $\Oh(|V_T|^6) = \Oh(k^{12}\cdot \twsf^{12})$ vertices,
     properly embedded in $\disc(I)$, such that  $V(H) \cap I = V_T$ and $H$ is $V_T$-linkage-equivalent to $G^T$.
    Then we merge the split vertices back together: for each $v \in V(T) \cap V(G)$
    we identify the vertices from $\Gamma_T(v)$ in $H$.
    The graph $G'$ obtained this way remains planar.
    Since $X \sub V(T) \cap V(G)$, \cref{lem:open:equivalent} implies that $G'$ is $X$-linkage-equivalent to $G$. 
     The theorem follows.
\end{proof}

\section{Kernelization hardness for parameter $k$}\label{sec:hardness}

In this section, we prove Theorems~\ref{thm:noPolyKer}, \ref{thm:MK2}, and \ref{thm:noPolyKerEdge}.
First, we introduce the intermediate problem \nullnoncross and  present a reduction to it from {\sc Set Cover}.
In \cref{sec:hard:vector} we construct the most internal gadget, allowing us to encode a large family of sets using homotopy classes. 
We use it as a building block to design a subset gadget in \cref{sec:hard:subset}.
In order to introduce the ideas gradually, we first give a simplified construction working under an overly-optimistic assumption (\cref{sec:subset-simple}), followed by a proper one  (\cref{sec:subset-full}), reflecting the exposition in \cref{sec:outline}.
The reduction to \nullnoncross is finalized in \cref{sec:hard:setcover}.
Afterwards, we explain how to get rid of the weighted requests (\cref{sec:hard:weights}) and to translate the hardness to {\sc Planar (Edge-)Disjoint Paths} (\cref{sec:hard:final}).
Note that several definitions in this section differ from their simplified versions in  \cref{sec:outline}.

\subsection{Non-crossing multicommodity flow}

\micr{reformulated stuff in this subsection}

We introduce the concept of a non-crossing flow and establish some notation to work with it.

\begin{definition}
Let $G$ be a plane  multigraph.
Consider a vertex $v \in V(G)$ and
two pairs of edges $(e_1,f_1)$ and $(e_2,f_2)$, such that $\{e_1,f_1,e_2,f_2\} \sub E_G(v)$.
We say that these pairs
{\em cross} if $e_1,e_2,f_1,f_2$ appear in this, or opposite, order in the cyclic ordering $\pi_G(v)$ of $E_G(v)$.
Two edge-disjoint walks $W_1, W_2$ in $G$ are {\em non-crossing} if there are no
pairs of consecutive edges $(e_1,f_1)$ in $W_1$ and $(e_2,f_2)$ in $W_2$ that cross.
\end{definition}

\begin{figure}
    \centering
\includegraphics[scale=.7]{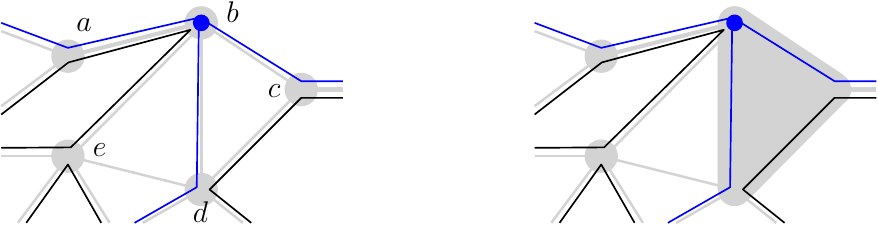} 
\caption{Left: An example of a non-crossing flow.
The vertex $b$ is a terminal for the three blue walks.
Because of the condition (b) in \cref{def:reduction:flow}, these walks can be connected within the image of $b$ without crossing the other walks.
The edges $ab$ and the one incident to $c$ (solid gray lines) have multiplicities 2 so the walks are pairwise edge-disjoint.
Right: After contracting $b,c,d$ into a single vertex, we still obtain a non-crossing flow.
Note that the edges $eb$ and $ed$ now become parallel.
}
\label{fig:non-crossing-contract}
\end{figure}


For  two walks $W_1 = (v_1, e_1, \dots, e_p, v_{p+1})$, $W_2 = (u_1, f_1, \dots, f_q, u_{q+1})$ with $ v_{p+1} = u_1$, we define their concatenation $W_1+W_2$ as a walk  given by $(v_1, e_1, \dots, e_p, u_1, f_1, \dots, f_q, u_{q+1})$.

\begin{definition}\label{def:reduction:flow}
Let $G$ be a plane multigraph 
and $\tcal$ be a multiset of triples from $V(G) \times V(G) \times \nn$.
A family $\pp$ of edge-disjoint {walks} in $G$ is a {\em $\tcal$-flow} if for every triple $(s_i, t_i, d_i) \in \tcal$ there exists a subfamily $\pp_i \sub \pp$ of $d_i$ many $(s_i,t_i)$-walks, and these subfamilies are disjoint for distinct triples from $\tcal$. 

A $\tcal$-flow $\pp$ is called {\em non-crossing} if (a) each pair of walks in $\pp$ is non-crossing, and (b) for any $(s_i, t_i, d_i) \in \tcal$, any two walks $W_1,W_2 \in \pp_i$ and a walk $W' \in \pp \sm \pp_i$, the walks $W_1 + W_2$ and $W'$ are non-crossing.
\end{definition}

The last condition enforces that all the walks from $\pp_i$ can ``touch'' each other at vertex $s_i$ (or $t_i$) without the need to cross the other walks.
One can imagine 
each vertex to have a positive area so a non-crossing flow can be depicted as a family of disjoint curves on the plane.
By connecting the images of the endpoints of paths in $\pp_i$ we can draw
 $E(\pp_1), \dots, E(\pp_k)$
as connected pairwise-disjoint subsets of the plane (see \Cref{fig:non-crossing-contract}).
This interpretation 
leads to the following observation.

\begin{observation}\label{obs:reduction:contract}
    Let $G$ be a plane multigraph,  $\tcal$ be a multiset of triples from $V(G) \times V(G) \times \nn$, 
   and \mic{$D \sub \rr^2$ be a topological disc such that $G[V(G) \cap D]$ is connected.}
   Next, let $G'$ be obtained from $G$ by contracting the set $V(G) \cap D$ to a single vertex $s$
   and $\tcal'$ be obtained from $\tcal$ by replacing each occurrence of a vertex from $V(G) \cap D$ with $s$.
   Suppose that there exists a non-crossing $\tcal$-flow  in $G$.
   Then there exists a non-crossing $\tcal'$-flow  in $G'$.
\end{observation}

Observe that this property would not hold if we replaced ``walks'' in the definition of a non-crossing flow with ``paths'' because a path might enter and exit $V(G) \cap D$ multiple times and all these visits may be necessary even after contraction to avoid crossings.
This is the main reason why we prefer to work with walks.
Also note that the opposite implication in \cref{obs:reduction:contract} does not necessarily hold even if $D$ contains no vertices from $\tcal$.

\defparproblem{\nullnoncross}
{Plane multigraph $G$, 
set $\mathcal{T}$ of $k$ vertex-disjoint requests $(s_1,t_1,d_i), \dots ,(s_k,t_k,d_k) \in V(G) \times V(G) \times \nn$. }
{$k$}
{Determine whether there exists a non-crossing $\tcal$-flow in $G$.}

We will refer to each triple $(s_i, t_i, d_i) \in \tcal$ as a {\em request} and to the integer $d_i$ as the {\em demand} of this request.
All vertices occurring in $\tcal$ are referred to as {\em terminals}.

An instance of \nullnoncross is called 
{\em unitary} if each demand $d_i$
equals 1 and every terminal 
has degree 1.

\subsection{Vector-containment gadget}
\label{sec:hard:vector}

{We begin describing the reduction from the innermost gadget.}
\mic{In order to provide an interface that will be consistent with the other gadgets, we need to impose a technical condition on the desired flow.
Whereas in the problem definition we require that each vertex may occur in at most one request, here we will consider families $\tcal$ that violate this condition.
Therefore, we need to specify in which order the walks enter a terminal.}

\begin{definition}\label{def:homo:seeing-order}
    Let $\pp$ be a \mic{non-crossing} flow in a plane multigraph $G$ whose outer face is confined by a simple cycle, $U \sub V(G)$, and $v \in V(G) \sm U$ lie on the outer face of $G$.
    Next, let $\pp_{v,U} \sub \pp$ be the family of walks in $\pp$ with one endpoint at $v$ and the other one at a vertex in $U$.
    We define $\pi$ as the clockwise order on $\pp_{v,U}$ given by the ordering of edges incident to $v$, starting from the one being next to the outer face.
    
    
    We say that $v$ {\em sees} vertices from $U$ in the order $u_1, u_2, \dots, u_k$ (with respect to flow $\pp$) if (a) for each $i \in [k]$ the occurrences of $(v,u_i)$-walks in $\pp_{v,U}$ form a continuous interval with respect to $\pi$, and (b) the order of these intervals matches the order $u_1, u_2, \dots, u_k$.
\end{definition}

{A visualization of this property is given in \Cref{fig:element-gadget}.}

\begin{definition}\label{def:homo:elemenet-gadget}
Consider $k \in \nn$, $\func \colon \{0,1\}^k \to \nn$, and $Z \sub  \{0,1\}^k$.
A plane multigraph $G$ is a $(k,\func,Z)${-$\mathsf{Vector\, Containment\, Gadget}$} if the following conditions hold.
\begin{enumerate}
    \item $G$ has distinguished vertices $z_1,z_2,\dots,z_k$ and $w_0,w_1$, where the last two lie on the outer face.
    \item Let $\bb \in \{0,1\}^k$, $d \in \nn$, and $\tcal_{\bb,d}$ be the family of following requests:
     \begin{enumerate}
        \item $(w_0, z_i, 2^k)$ for each $i \in [k]$ with $\bb_i = 0$,
        \item $(w_1, z_i, 2^k)$ for each $i \in [k]$ with $\bb_i = 1$,
        \item $d$ copies of the request $(w_0, w_1, 1)$. 
    \end{enumerate}
    Then the following conditions are equivalent:
    \begin{enumerate}
        \item $d \le \func(\bb) + 1_{[\bb \in Z]}$,
        \item there exists a  $\mathcal{T}_{\bb,d}$-flow in $G$,
        \item there exists a non-crossing $\mathcal{T}_{\bb,d}$-flow in $G$,
        \mic{in which $w_0$ sees $\{z_i \mid b_i = 0\}$ in the order of decreasing $i$ and $w_1$ sees $\{z_i \mid b_i = 1\}$ in the order of increasing $i$.}
    \end{enumerate}
\end{enumerate}
\end{definition}

\begin{figure}
    \centering
\includegraphics[scale=1]{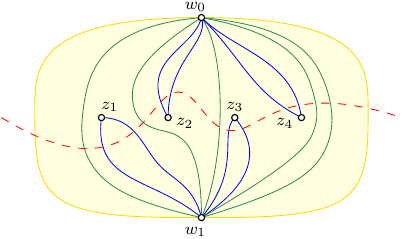} 
\caption{A conceptual sketch of a $(4,\func,Z)${-$\mathsf{Vector\, Containment\, Gadget}$}.
The flow on the picture corresponds to $\bb = (1010)$: the vertex $z_i$ sends the blue flow to $w_0$ when $\bb_i = 0$ or to $w_1$ when $\bb_i = 1$.
\mic{The vertex $w_0$ sees vertices $z_4, z_2$ in this order while $w_1$ sees vertices $z_1, z_3$ in this order.}
The amount of the available $(w_0,w_1)$-flow (green) depends on $\func(\bb)$ and on whether $\bb \in Z$.
Observe that each path in the flow must cross the red dashed curve whose homotopy class (with respect to $z_1,z_2,z_3,z_4$) agrees with the vector $\bb$.
We will rely on this observation when constructing the gadget.
}
\label{fig:element-gadget}
\end{figure}

For the existence of a  $\mathcal{T}_{\bb,d}$-flow, using $d$ copies of the request $(w_0, w_1, 1)$ is equivalent to using a single request  $(w_0, w_1, d)$.
This however does matter for the existence of a non-crossing $\mathcal{T}_{\bb,d}$-flow.
The reason for considering $d$ copies of $(w_0, w_1, 1)$ instead of just $(w_0, w_1, d)$ comes from condition (b) in \cref{def:reduction:flow}: we want to allow the $(w_0,w_1)$-walks 
to be arbitrarily intertwined with the other walks at $w_0$ or $w_1$ (see Figure~\ref{fig:element-gadget}).
\mic{On the other hand, we require those other walks to be well-structured.}

\mic{Basically, the vector $\bb$ specifies for each terminal $z_i$ whether it should send the flow to the left $(w_0)$ or to the right $(w_1)$.
In turn, the condition $\bb \in Z$ governs how many $(w_0,w_1)$-walks can be allocated on top of the walks above.
}

The rest of \cref{sec:hard:vector} is devoted to a construction of 
a $(k,\func,Z)${-$\mathsf{Vector\, Containment\, Gadget}$} of size $2^{\Oh(k)}$ for a certain function $\func$. 

\subsubsection{Homotopy classes and shortest paths}

\mic{Instead of constructing a vector-containment gadget directly, we begin from a prototype of its dual.
Its most important property is the uniqueness of a shortest $(s,t)$-path in each homotopy class (to be defined later).}

\paragraph*{Operations on bit vectors.}
We number the coordinates in a size-$k$ vector  from 1 to $k$.
Consider a~binary vector $\bb = (b_1,b_2, \dots, b_k)$.
When referring to indices or performing arithmetic, we implicitly use the big-endian binary decoding $\{0,1\}^k \to [0, 2^k)$ given as $\sum_{i=1}^{k} b_{i}\cdot 2^{k-i}$.
For $i \in [k]$ let $\bb^{[i]}$ denote a binary string obtained from $\bb$ by reversing its prefix of length $i$.
Note that for every $i \in [k]$ the mapping $\bb \to \bb^{[i]}$ is a bijection; for $i=1$ it is identity.

\paragraph*{Construction of the graph $H_k$.}
We define a plane graph $H_k$ as follows.
We draw $2k$ vertical lines $Q_1, Q'_1, Q_2, Q'_2, \dots, Q_k, Q'_k$, in this order from left to right,
and mark $2^k$ vertices on each of them.
The vertices on $Q_i$ are referred to as $v_{i,j}$, where $j \in [0, 2^k)$, counting from the top to the bottom.
Similarly, vertices on $Q'_i$ are referred to as $v'_{i,j}$.
We add two additional vertices: $s$ to the left of $Q_1$, and $t$ to the right of $Q'_k$.

For each $\bb \in  \{0,1\}^k$ we draw a curve $P_\bb$ which starts at $s$, crosses all the lines $Q_1, Q'_1, Q_2, Q'_2, \dots,$ $Q_k, Q'_k$, in this order, and ends at $t$.
The curve $P_\bb$ crosses the line $Q_i$ (resp. $Q'_i$) at the vertex $v_{i, \bb^{[i]}}$ (resp. $v'_{i, \bb^{[i]}}$).
Each segment of $P_\bb$ between crossing consecutive vertical lines is a straight line.
The plane graph $H_k$ is obtained from this drawing by turning every crossing of curves $P_{\bb}$, $P_{\bb'}$ into a vertex. 
See Figure \ref{fig:homotopy2} for an illustration.
We retain the names $P_{\bb}, Q_i, Q'_i$ to denote the corresponding paths in~$H_k$.

\begin{figure}
    \centering
\includegraphics[scale=1.3]{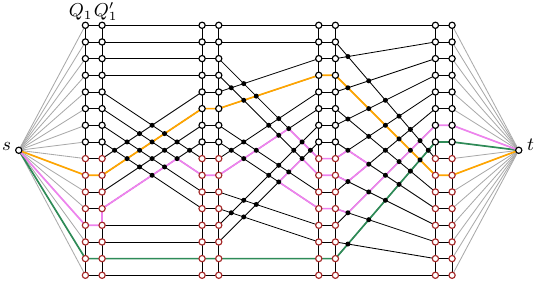}
\caption{The graph $H_4$. The vertical lines are  $Q_1, Q'_1, \dots, Q_4, Q'_4$, counting from left to right.
For each $i \in [4]$ the vertices $v_{i,j}$ and  $v'_{i,j}$ are drawn with black boundary when 
$j < 2^3$ and with brown boundary when $j \ge 2^3$.
For $\bb = (1001)$ the path $P_\bb$ is highlighted in orange and
for $\vv = (1110)$ the path $P_\vv$ is highlighted in green.
Observe that the positions of the first vertices on $P_\bb, P_\vv$ correspond to the numbers encoded by $\bb, \vv$ in binary. 
\mic{The violet path is an example of a different $\vv$-homotopic path.}
} 
\label{fig:homotopy2}
\end{figure}


\begin{observation}\label{obs:homo:signs}
Let $\bb, \vv \in  \{0,1\}^k$ be distinct and $j \in [k-1]$. The paths $P_{\bb}$ and $P_{\vv}$ in $H_k$ intersect between $Q'_j$ and $Q_{j+1}$ if and only if 
$\bb^{[j]} - \vv^{[j]}$ and $\bb^{[j+1]} - \vv^{[j+1]}$ have different signs.
\end{observation}

\mic{We will study the geometric relations between $P_{\bb}$ and $P_{\vv}$ through bit substrings in $\bb$ and $\vv$.}

\begin{lemma}\label{obs:homo:bits}
Let $\bb, \vv \in  \{0,1\}^k$ be distinct and $j \in [k-1]$.
If $\bb^{[j]} < \vv^{[j]}$ and $\bb^{[j+1]} > \vv^{[j+1]}$ then $\bb_{j+1} = 1$ and $\vv_{j+1} = 0$.
Furthermore, there exists $i \in [j]$ for which $\bb_{i} = 0$, $\vv_{i} = 1$, and
$\bb_{h} = \vv_{h}$ for $i < h < j + 1$.
\end{lemma}
\begin{proof}
First, targeting a contradiction, suppose that $\bb_{j+1} = \vv_{j+1}$.
Observe that $\bb^{[j]}$ can be obtained from $\bb^{[j+1]}$ by moving the first bit to the position $j+1$, and likewise for $\vv^{[j]}, \vv^{[j+1]}$.
When $\bb_{j+1} = \vv_{j+1}$ then this operation does not affect the ``$<$'' relation between the encoded integers, so this implies $\bb^{[j]} > \vv^{[j]}$, contrary to the assumption.
Hence $\bb_{j+1} \ne \vv_{j+1}$ and it must be $\bb_{j+1} > \vv_{j+1}$.

Now, suppose that $\bb, \vv$ coincide on first $j$ coordinates.
Since $\bb_{j+1} > \vv_{j+1}$, this implies that $\bb^{[j]} > \vv^{[j]}$, a contradiction.
As a consequence, there is an index in $[j]$ at which $\bb, \vv$ differ; let $i$ denote the last such index. Then, the choice of $i$ implies $\bb_{h} = \vv_{h}$ for $i < h < j + 1$.
Consider the first coordinate at which $\bb^{[j]}, \vv^{[j]}$ differ: for the first vector this bit equals $\bb_i$ and for the second one it is $\vv_i$.
Since $\vv^{[j]}>\bb^{[j]}$, this implies $\vv_i > \bb_i$.
\end{proof}

We will refer to the structure observed in the last lemma as a {\em crossing pair}.

\begin{definition}
Consider two vectors $\bb, \vv \in \{0,1\}^k$.
We say that $(i,j) \in [k]^2$ is a {\em crossing pair} for $(\bb, \vv)$ if (a) $i < j$, (b) the pair $((\bb_i, \vv_i), (\bb_j, \vv_j))$ equals either $((0,1),(1,0))$ or $((1,0),(0,1))$
 and (c) $\bb_h = \vv_h$ for each $i < h < j$.
We define $C(\bb, \vv)$ to be the set of crossing pairs for $(\bb, \vv)$.
\end{definition}

\mic{As an example, consider $\bb = (1001101)$ and $\vv = (0101000)$.
They differ at positions $1, 2, 5, 7$ and $C(\bb, \vv) = \{(1,2), (2,5)\}$.}

\begin{observation}\label{obs:homo:pairs-disjoint}
    When $(i_1, j_1)$ and $(i_2,j_2)$ are different crossing pairs for some $(\bb, \vv)$ then $j_1 \le i_2$ or $j_2 \le i_1$.
\end{observation}

We would like to employ some notion of a homotopy class for the $(s,t)$-paths.
However, instead of working with the topological notion of homotopy, we introduce a simpler definition, tailored just for our analysis of the graph $H_k$.
Let $E^i \subseteq E(H_k)$ denote the set of edges between $V(Q_i)$ and $V(Q'_i)$.
Each edge from $E^i$ is of the form $v_{i,j}v'_{i,j}$ for $j \in [0,2^k)$.
For $b \in \{0,1\}$ we define $\mathsf{Half}(k, b) \sub [0, 2^k)$ to be $[0, 2^{k-1})$ when $b=0$ and $[2^{k-1}, 2^{k})$ when $b=1$.
Next, we define $E^i_b \subseteq E^i$ as $\{v_{i,j}v'_{i,j} \mid j \in \mathsf{Half}(k, b)\}$. 
The edges from $E^i_0$ have black circles as endpoints on Figure~\ref{fig:homotopy2} whereas the ones from $E^i_1$ have brown circles as endpoints.

\begin{definition}\label{def:homo:homotopic}
For $\bb \in \{0,1\}^k$
we say that an $(s,t)$-path $P$ in $H_k$ is $\bb$-homotopic if $E(P) \cap E^i \sub E^i_{\bb_i}$ for each $i \in [k]$.
\end{definition}

A canonical example of a $\bb$-homotopic path is $P_\bb$.
Note that it might be the case that some $(s,t)$-path is not $\bb$-homotopic for any $\bb \in \{0,1\}^k$ according to our definition.
\mic{We can now express some geometric properties of $(s,t)$-paths in terms of crossing pairs.}

\begin{lemma}\label{lem:homo:two-paths-cross}
Consider two vectors $\bb, \vv \in \{0,1\}^k$.
Let $R$ be a $\vv$-homotopic $(s,t)$-path in $H_k$.
The graph given by the intersection  $R \cap P_\bb$ contains at least $|C(\bb, \vv)|$ connected components disjoint from $s$ and $t$.
\end{lemma}
\begin{proof}
Let $(i,j) \in C(\bb, \vv)$.
The path $R$ must contain a subpath $R^{i,j}$ that starts at $Q'_i$, ends at $Q_j$, and is internally contained between  $Q'_i$ and $Q_j$.
Similarly, let $P^{i,j}_{\bb}$ be the subpath of $P_\bb$ between $Q'_i$ and $Q_j$.
By the definition of a crossing pair, the endpoints of $R^{i,j}$, $P^{i,j}_{\bb}$ are all distinct and they lie in different orders on $Q'_i$ and $Q_j$; hence these paths must intersect between $Q'_i$ and $Q_j$, exclusively.
From \cref{obs:homo:pairs-disjoint} we obtain that the paths $R^{i,j}$ constructed for distinct crossing pairs are disjoint.
Therefore, each $(i,j) \in C(\bb, \vv)$ contributes at least one connected component of  $R \cap P_\bb$ that is disjoint from $s,t$.
\end{proof}

\mic{For the special case $R = P_\vv$ we can make a stronger observation.}

\begin{lemma}\label{lem:homo:two-paths-cross-shortest}
Consider distinct vectors $\bb, \vv \in \{0,1\}^k$.
The number of \mic{internal} vertices shared by $P_{\bb}$ and $P_{\vv}$ equals $|C(\bb, \vv)|$. \meir{Why not split this into 2 lemmas? \mic{done}}
\end{lemma}
\begin{proof}
Clearly,  $P_{\bb}$ and $P_{\vv}$ cannot intersect between $Q_i$ and $Q'_i$ for any $i \in [k]$.
Let $J \sub [k-1]$ be the set of indices $j$ for which $\bb^{[j]} - \vv^{[j]}$ and $\bb^{[j+1]} - \vv^{[j+1]}$ have different signs.
By \cref{obs:homo:signs}, the number of common internal vertices in $P_{\bb}$ and $P_{\vv}$ equals $|J|$. 
Let $j \in J$ and assume w.l.o.g. that $\bb^{[j+1]} > \vv^{[j+1]}$ and $\bb^{[j]} < \vv^{[j]}$.
By \cref{obs:homo:bits} we have $\bb_{j+1} = 1$ and $\vv_{j+1} = 0$.
Furthermore, $\bb_i = 0$ and  $\vv_i = 1$, where $i$ is the last index in $[j]$ at which $\bb, \vv$ differ.
Hence $(i,j+1)$ forms a crossing pair for $(\bb,\vv)$.
The crossing pairs obtained for different $j_1,j_2 \in J$ must be different, what implies $|J| \le |C(\bb,\vv)|$.
\mic{The equality follows from the \cref{lem:homo:two-paths-cross}, as the number of shared internal vertices is no less than the number of components in $P_\bb \cap P_\vv$ disjoint from $s$ and $t$.}
\end{proof}

\mic{As a next step, we compute the length of the path $P_\bb$.
It can be expressed with a very convenient formula which will come in useful later.
}

\begin{definition}
We define function $\func_k \colon \{0,1\}^k \to \nn$ as follows.
\[
\func_k(b_1b_2\dots b_k) = \sum_{1 \le j < i \le k} 1_{[b_i \ne b_j]} \cdot 2^{k-i+j-1}.
\]
\end{definition}

When the parameter $k$ is clear from the context, we abbreviate $\func = \func_k$.

\begin{lemma}\label{lem:homo:length-Pb}
For each $\bb \in \{0,1\}^k$ the length of the path $P_{\bb}$ in $H_k$ equals $2k+1 + \func(\bb)$.
\end{lemma}
\begin{proof}
The length of $P_b$ equals the total number of its crossings with $P_\vv$ for $\vv \ne \bb$ plus the number of crossings with $Q_i, Q'_i$ ($2k$ in total) plus one.
By \cref{lem:homo:two-paths-cross-shortest} it suffices to show that the sum of $|C(\bb,\vv)|$ over $\vv \ne \bb$ equals $\func(\bb)$.
To this end, we change the order of summation and, for each pair $1 \le i < j \le k$, we count the number of vectors $\vv \ne \bb$ for which $(i,j) \in C(\bb,\vv)$. So, consider some pair $1 \le i < j \le k$.
First, note that if $(i,j)\in C(\bb,\vv)$ is non-empty for any $\vv$, then $\bb_i \ne \bb_j$.
In this case, $(i,j) \in C(\bb,\vv)$ if and only if $\vv_i \ne \bb_i$, $\vv_j \ne \bb_j$, and $\bb, \vv$ coincide between $i$ and $j$.
These conditions fix exactly $(j-i+1)$ coordinates of $\vv$ and on the remaining coordinates $\vv$ may be arbitrary.
Therefore there are exactly $2^{k-i+j-1}$ vectors $\vv$ for which $(i,j) \in C(\bb,\vv)$.
This agrees with the definition of function $\func$.
\end{proof}

\mic{Finally, we prove the most crucial property of the graph $H_k$: the uniqueness of a shortest $(s,t)$-path in each homotopy class.}

\begin{lemma}\label{lem:homo:unique}
For each $\bb \in \{0,1\}^k$ the path $P_\bb$ is the unique shortest $\bb$-homotopic $(s,t)$-path in~$H_k$.
\end{lemma}
\begin{proof}
Let $R \ne P_\bb$ be a $\bb$-homotopic $(s,t)$-path in $H_k$.
We are going to show that $|R| > |P_\bb|$.
First observe that $R$ cannot be the path $P_\vv$ for any $\vv \ne \bb$ as then it would not be $\bb$-homotopic.

Let $\pp$ be the family of all paths of the form $P_\vv$, $Q_i$, $Q'_i$.
Every vertex $v \in V(H_k) \sm \{s,t\}$ is an intersection of some two paths
$P_1,P_2 \in \pp$. 
For $P \in \pp$ let $\Gamma_R(P)$ be the set of internal vertices $v$ in $R$ such that $v \in V(P)$ but the predecessor of $v$ on $R$ does not belong to $V(P)$.
\mic{
Observe that when $v$ is an internal vertex of $R$ and $v$ is an intersection of paths $P_1,P_2 \in \pp$ then the predecessor of $v$ is either $s$ (which belongs to exactly one of $V(P_1)$, $V(P_2)$) or it is an intersection of paths $P'_1,P'_2 \in \pp$ with $|\{P_1,P_2\} \cap \{P'_1,P'_2\}| = 1$.
Consequently, the sets $\Gamma_R(P)$ are pairwise disjoint and
 every internal vertex of $R$ belongs to some set $\Gamma_R(P)$.} \meir{Elaborate a bit on the last sentence. \mic{done}}
We infer that the length of $R$ equals $\sum_{P \in \pp} |\Gamma_R(P)|$ + 1.
\mic{When $P = P_\vv$ for some $v \in \{0,1\}^k$ then $|\Gamma_R(P_\vv)|$ is lower bounded by the number of connected components of $R \cap P_\vv$ disjoint from $s,t$, which in turn is at least $|C(\bb,\vv)|$ due to 
\cref{lem:homo:two-paths-cross}.
So $|\Gamma_R(P_\vv)| \ge |C(\bb,\vv)|$.} \meir{Elaborate a bit on the last sentence.\mic{done}}
When $P = Q_i$ or $P = Q'_i$ then $|\Gamma_R(P)| \ge 1$.
If in both cases we always had equalities then the length of $R$ would be the same as $P_\bb$ (by \cref{lem:homo:two-paths-cross-shortest}).
Therefore, it is sufficient to show that for some $\vv \in \{0,1\}^k$ we have strict inequality $|\Gamma_R(P_\vv)| > |C(\bb,\vv)|$.

Let $\vv$ be the vector for which the last edge on $R$ (the one incident to $t$) belongs to $P_\vv$ (possibly $\vv = \bb$).
\mic{Since $R \ne P_\vv$, we obtain that 
$R \cap P_\vv$ has a connected component of size at least $2$ containing $t$ but not $s$.
The first vertex (with respect to $R$) in this component belongs to $\Gamma_R(P_\vv)$ but this component 
has not been taken into account in the bound $|\Gamma_R(P_\vv)| \ge |C(\bb,\vv)|$ above (because it contains $t$).
Therefore, $|\Gamma_R(P_\vv)| > |C(\bb,\vv)|$, what implies $|R| > |P_\bb|$.}
This concludes the proof.\meir{Why do you need 2 cases? Where do you use $C(\bb,\vv) \neq \emptyset$? \mic{reformulated}}

\end{proof}

\subsubsection{Dual flows}

Before we 
\mic{are ready to finish the construction of a} vector-containment gadget, we need to establish a method for constructing non-crossing flows with certain properties in a dual graph. 

\begin{lemma}[{\cite[Prop.~2.6.4]{Mohar2001GraphsOS}}]
\label{lem:homo:dual-cut-cycle}
Let $G$ be a 2-connected plane multigraph and $G^*$ be its dual.
Suppose that $S \sub E(G)$ is an inclusion-minimal $(u,v)$-edge-separator for some $u,v \in V(G)$.
Then $S^*$ is an edge set of a cycle in $G^*$ 
such that 
the vertices $u,v$ belong to different connected components of $\rr^2 \sm S^*$.
\end{lemma}

\begin{figure}
    \centering
\includegraphics[width=.5\linewidth]{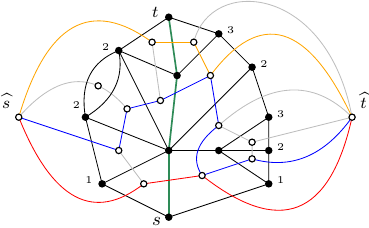} 
\caption{An $(s,t)$-dual $(G^\circ, {{\widehat s}}, {{\widehat t}})$ of a multigraph $G$. The vertices of $G$ are black whereas the vertices of $G^\circ$ are hollow. A shortest $(s,t)$-path $P$ in $G$ is drawn with solid green lines.
A non-crossing family of three edge-disjoint $({{\widehat s}}, {{\widehat t}})$-paths in $G^\circ$ is highlighted with colors.
Each of these paths must cross some edge of $P$.
They illustrate the construction from \cref{lem:homo:dual-flow}.
We have indicated the distances from $s$ in $G$ for the vertices on the paths $P^s$ and $P^t$. 
}
\label{fig:dual-flow}
\end{figure}

\begin{definition}\label{def:homo:dual}
Let $G$ be a connected plane multigraph and $s,t \in V(G)$ lie on the outer face.
Let $G_{st}$ be obtained from $G$ by inserting the edge $st$ within the outer face and $G_{st}^*$ be the dual of $G_{st}$.
Let ${{\widehat s}}, {{\widehat t}}$ denote the endpoints of the edge $(st)^*$ in $G_{st}^*$ so that ${{\widehat s}}$ corresponds to the face incident to the edge $st$ on the right when considering the orientation of $st$ from $s$ to $t$.
The {\em $(s,t)$-dual of $G$} is the triple $(G^*_{st} \sm (st)^*, {{\widehat s}}, {{\widehat t}})$.
\end{definition}

See Figure \ref{fig:dual-flow} for an example of an $(s,t)$-dual. 
We use notation $G^\circ$ to refer to the graph $G^*_{st} \sm (st)^*$ (when $s,t$ are clear from context). 
For an edge $e \in E(G)$ we refer to its counterpart in $G^\circ$ as $e^\circ$.
Similarly, for an internal face $f$ in $G$ we refer to the corresponding vertex in $G^\circ$ as $f^\circ$.

\mic{We will utilize the correspondence between the length of the shortest $(s,t)$-path in $G$ and the maximal size of an $({{\widehat s}}, {{\widehat t}})$-flow in the $(s,t)$-dual of $G$.
The following lemma also reveals which of the edges incident to ${{\widehat s}}, {{\widehat t}}$ are used in this flow.}

\begin{lemma}\label{lem:homo:dual-flow}
Let $G$ be a 2-connected plane multigraph whose outer face is confined by a simple cycle $C$, $s,t \in V(C)$, and $d = \dist_G(s,t)$.
Let $(G^\circ, {{\widehat s}}, {{\widehat t}})$ be the $(s,t)$-dual of $G$.
Next, let $P^s, P^t$ be the $(s,t)$-paths in $C$ such that the edges of $P^s$ (resp. $P^t$) are incident to ${{\widehat s}}$ (resp. ${{\widehat t}}$). 

For $i \in [d]$ let $v^s_i$  be the last vertex on $P^s$ with $\dist_G(s,v^s_i) < i$ and $u^s_i$ be its successor on~$P^s$.
Analogously we define vertices $v^t_i, u^t_i \in V(P^t)$.
Then there exists a non-crossing family of $d$ edge-disjoint $({{\widehat s}},{{\widehat t}})$-paths $P^\circ_1, P^\circ_2, \dots, P^\circ_d$ in $G^\circ$, such that $(v^s_iu^s_i)^\circ\in P^\circ_i$ and $(v^t_iu^t_i)^\circ \in P^\circ_i$ 
for each $i \in [d]$.
\end{lemma}
\begin{proof}
 For $i \in [d]$ let $\widehat {V}_i = \{v \in V(G) \mid \dist_G(s,v) \ge i\}$ and let $V_i \sub \widehat V_i$ induce the connected component of $G[\widehat V_i]$ that contains $t$.
 We set $S_i = E(V_i, V(G) \sm V_i)$.
 The vertices $u^s_i, u^t_i$ belong to $V_i$ because they can be connected to $t$ with subpaths of $P^s, P^t$ contained in  $G[\widehat V_i]$.
Therefore, $v^s_iu^s_i \in S_i$ and $v^t_iu^t_i \in S_i$.

We claim that $S_i$ is an inclusion-minimal $(s,t)$-separator.
Let $vu \in S_i$ and $u \in V_i$, $v \not\in V_i$.
\mic{Observe that $v \not\in \widehat V_i$ because otherwise it would belong to the connected component of $G[\widehat V_i]$ containing $t$, which would imply $v \in V_i$.}
\meir{Not immediate from construction. Add a short explanation (or show in a fig). \mic{added}} Hence $\dist_G(s,v) = i - 1$ and there is an $(s,v)$-path in $G \sm S_i$. 
As $G[V_i]$ is connected, this implies that $S_i \sm vu$ is not an $(s,t)$-separator, hence $S_i$ is minimal.

Recall that $G_{st}$ is obtained from $G$ by inserting the edge $st$ and $G_{st}^*$ is the dual of $G_{st}$.
 Then $ {S}_i \cup \{st\}$ is an inclusion-minimal $(s,t)$-separator in $G_{st}$.
 By \cref{lem:homo:dual-cut-cycle} the set ${S}^*_i \cup \{(st)^*\} \sub E(G_{st}^*)$ forms a cycle in $G^*_{st}$ separating the vertices $s$ and $t$ on the plane.
 This cycle goes through vertices  ${{\widehat s}}$,  ${{\widehat t}}$, and the edge ${{\widehat s}}{{\widehat t}} = (st)^*$.
 Therefore ${S}^*_i$ forms an edge set of an $({{\widehat s}},  {{\widehat t}})$-path in $G^\circ = G^*_{st} \sm (st)^*$; this shall be the path  $P^\circ_i$.
 We have  $(v^s_iu^s_i)^\circ\in  P^\circ_i$, $(v^t_iu^t_i)^\circ \in P^\circ_i$, and
these paths are edge-disjoint because the sets $S_1,\dots,S_d$ are disjoint.

Finally we argue that  $P^\circ_1, P^\circ_2, \dots, P^\circ_d$ are non-crossing.
Let $D_i \subset \rr^2$ be the connected component of $\rr^2 \sm ({S}^*_i \cup \{(st)^*\})$ containing $t$ (it may be unbounded).
We have $D_i \cap V(G_{st}) = V_i$ for each $i \in [d]$.
Therefore, $D_i$ is a union of the faces in the dual $G_{st}^*$ corresponding to vertices from $V_i$. 
Observe that for $i < d$ the set $V_{i+1}$ is contained in $V_i$.
Consequently, we have $D_1 \supset D_2 \supset \dots \supset D_d$.
Since the path $P_i$ is an arc of $\partial D_i$, these paths cannot cross.
 See Figure~\ref{fig:dual-flow} for an illustration.

\end{proof}

\subsubsection{Construction of a non-crossing flow}

A direct approach to construct a vector-containment gadget would be to consider the $(s,t)$-dual $(H^\circ, {{\widehat s}},  {{\widehat t}})$ of the graph $H_k$ and set $z_i$ to be the vertex corresponding to the face between the last edge from $E^i_0$ and the first edge from $E^i_1$.
Consider some $\bb \in \{0,1\}^k$ and a flow $\pp$ in $H^\circ$ that consists of (a) $({{\widehat s}}, {{\widehat t}})$-paths, (b) $({{\widehat s}},  z_i)$-paths for $\bb_i = 0$, and (c) $({{\widehat t}},  z_i)$-paths for $\bb_i = 1$.
Then every path in $\pp$ must cross any $\bb$-homotopic $(s,t)$-path in $H_k$ (see Figure~\ref{fig:element-gadget}) and so the length of the shortest such path upper bounds the size of $\pp$.

Since $P_\bb$ is the unique shortest $\bb$-homotopic $(s,t)$-path in $H_k$, subdividing its first edge (the one incident to $s$) increases the length of the shortest $\bb$-homotopic $(s,t)$-path (see \Cref{fig:homotopy3}).
This allows $\pp$ to have one more element. 
We could thus encode the set $Z$ by subdivisions of the edges incident to $s$, increasing the upper bound for the size of $\pp$ exactly when $\bb \in Z$.
It is more complicated though to obtain the implication in the other direction: when $\bb \in Z$ we want to construct a non-crossing flow $\pp$ satisfying certain requests of the three mentioned types.
Performing such a construction ``by hand'' would be very tedious and instead we will take advantage of \cref{lem:homo:dual-flow}.
To this end, we need to first subdivide more edges in $H_k$ to make it amenable to this lemma. 

\begin{definition}\label{def:homo:subdivision}
Let $Z \sub \{0,1\}^k$.
The graph $H'_{k,Z}$ is obtained from $H_k$ as follows.
\begin{enumerate}
    \item For each $i \in [k]$ and $j \in [0,2^k)$, each of the two edges incident to $v'_{i,j}$ but not contained in $Q'_i$ gets subdivided $2^k-1$ times.
    \item For each $j \in Z$ the edge $sv_{1,j}$ gets subdivided once.
\end{enumerate}
\end{definition}

\mic{An example is given in \Cref{fig:homotopy3}.}
There is a 1-1 correspondence between $(s,t)$-paths in $H_k$ and $H'_{k,Z}$ therefore by a slight abuse of notation we can consider $\bb$-homotopic $(s,t)$-paths in  $H'_{k,Z}$.

\begin{figure}
    \centering
\includegraphics[scale=1.3]{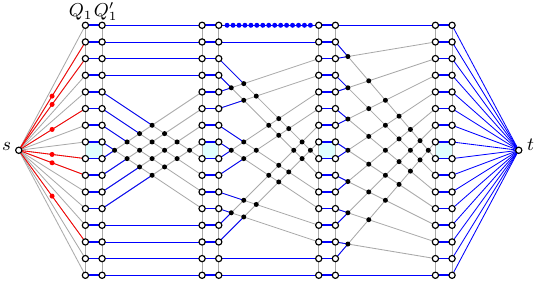}
\caption{Graph $H'_{4,Z}$ for $Z = \{1,2, 5, 8, 9, 13\}$.
The red edges are subdivided once and the blue edges are subdivides $2^4 - 1 = 15$ times. 
For legibility only one blue edge on the top is drawn subdivided. The gray edges are not being subdivided.
The faces $f_1,f_2,f_3,f_4$ are highlighted in light blue.
} 
\label{fig:homotopy3}
\end{figure}

\begin{lemma}\label{def:homo:subdivision:length}
Let $Z \subseteq \{0,1\}^k$ and $\bb \in \{0,1\}^k$.
The length of the shortest $\bb$-homotopic $(s,t)$-path in $H'_{k,Z}$ equals $k\cdot 2^{k+1} + \func(\bb)+2$ if $\bb \in Z$ and $k\cdot 2^{k+1} + \func(\bb) + 1$ otherwise.
\end{lemma}
\begin{proof}
    First consider an intermediate graph $H'_k$ obtained from $H_k$ by the first modification from \cref{def:homo:subdivision}.
    Every $(s,t)$-path $P$ in $H_k$ must cross $Q'_i$ for each $i \in [k]$ and so it must contain two edges incident to $Q'_i$.
    Due to the subdivisions, the length of $P$ in $H'$ increases by at least $2k \cdot (2^k - 1)$. 
    The length of $P_\bb$ increases by exactly $2k \cdot (2^k - 1)$ and so its length in $H'_k$ becomes $k\cdot 2^{k+1} + \func(\bb)+1$.
    Since $P_\bb$ is the unique shortest $\bb$-homotopic $(s,t)$-path in $H_k$, it is also the unique shortest $\bb$-homotopic $(s,t)$-path in $H'_k$.

    Now consider the second modification from \cref{def:homo:subdivision}.
    Clearly, it cannot decrease the length of any path.
    If $\bb \not\in Z$ then this modification does not affect any edge on $P_\bb$ so its length in $H'_{k,Z}$ is again $k\cdot 2^{k+1} + \func(\bb)+1$.
    However, if $\bb \in Z$ then we have subdivided an edge on the unique shortest $\bb$-homotopic $(s,t)$-path so now the length of the shortest $\bb$-homotopic $(s,t)$-path becomes $k\cdot 2^{k+1} + \func(\bb)+2$.
\end{proof}

\paragraph*{Carving off the cavities.}
We introduce some additional notation for the two following lemmas.
For $i \in [k]$ we distinguish face $f_i$ as the face between $Q_i$ and $Q'_i$ incident to the vertices $v_{i,2^{k-1}-1}$ and $v_{i,2^{k-1}}$
(the four highlighted faces in \Cref{fig:homotopy3}).
Recall that $\mathsf{Half}(k, b)$ 
stands for $[0, 2^{k-1})$ when $b=0$ and for $[2^{k-1}, 2^{k})$ when $b=1$.  
For $\bb \in \{0,1\}^k$ we define $H^\bb_{k,Z}$ as the plane graph obtained from $H'_{k,Z}$
by removing all the internal vertices in the subdivided 
$(v_{i,j}, v'_{i,j})$-edge
for each $i \in [k]$ and $j \in \mathsf{Half}(k, 1 - {\bb}_i)$.
See Figure \ref{fig:homotopy4} for a visualization.

An edge $e \in E(H^\bb_{k,Z})$ is called {\em exposed} \meir{Better name? \mic{special $\to$ exposed} that's great} if 
$e$ belongs to the subdivided $(v_{i,j},v'_{i,j})$-edge
where $j = 2^{k-1}$ if $b_i = 0$ and $j = 2^{k-1} + 1$ if $b_i = 1$.
Note that every exposed edge is incident to the outer face of  $H^\bb_{k,Z}$.
An edge $e^\circ$ in an $(s,t)$-dual of  $H^\bb_{k,Z}$ is called exposed if $e$ is exposed in  $H^\bb_{k,Z}$.

\vspace{3mm}
\mic{In order to construct a non-crossing $\tcal_{\bb,d}$-flow in the $(s,t)$-dual of $H'_{k,Z}$, we first construct a flow in the $(s,t)$-dual of $H^\bb_{k,Z}$ and then translate it to the dual above.
We will work with flows consisting of paths instead of walks, what obviously meets our definition of a non-crossing~flow.
}

\begin{figure}[t]%
    \centering
    \subfloat{
    {\includegraphics[width=11.2cm,valign=c]{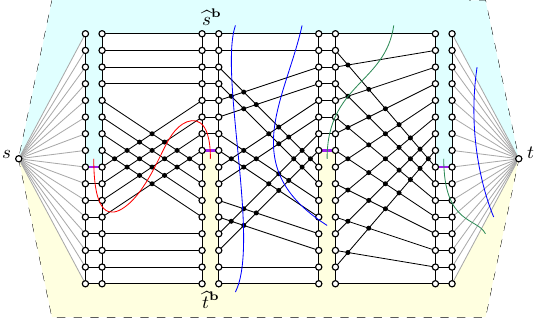}}
    }%
    \qquad\qquad
    \subfloat 
    {{\includegraphics[width=3.2cm,valign=c]{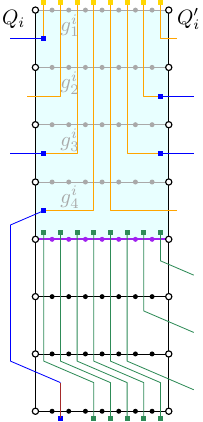}}}%
\caption{Left: The graph $H^\bb_{4,Z}$ for $\bb = (1001)$ and arbitrary $Z$ (the edge subdivisions are omitted here).
The edges on the purple paths are {\em exposed}.
The areas highlighted in color correspond to vertices
${{\widehat s}}^\bb,{{\widehat t}}^\bb$ in the $(s,t)$-dual of $H^\bb_{4,Z}$.
The blue and green curvy lines are examples of the paths from the family constructed in \cref{lem:homo:special-edges}.
The crux of the lemma is that this family does not contain paths like the red one.
The green paths belong to subfamilies $\pp_3, \pp_4$ from \cref{lem:homo:flow-exists} while the blue ones belong to $\pp_{long}$.
\newline \textcolor{white}{----} Right: An illustration for \cref{lem:homo:flow-exists} showing a fragment of the graph $H^\bb_{3,Z}$ with $\bb_i = 1$.
The gray paths are present in  $H'_{3,Z}$ but not in  $H^\bb_{3,Z}$ while the highlighted faces of $H'_{3,Z}$ become merged with $\widehat s^\bb$ in $H^\bb_{3,Z}$.
The purple path consists of the exposed edges.
The face $g^i_4$ is the same as $f_i$ (cf. Figure \ref{fig:homotopy3}).
The paths in the $(s,t)$-dual of $H^\bb_{3,Z}$ are drawn as paths going through the faces of $H^\bb_{3,Z}$.
\newline \textcolor{white}{----} The condition of edge-disjointedness means that each drawn edge may be crossed by only a single path.
The paths from $\pp_i \sub \pp$ are sketched in green; their common upper endpoint is $\widehat s^\bb$ but when we consider them in the $(s,t)$-dual of $H'_{3,Z}$ this endpoint becomes $(f_i)^\circ = (g^i_4)^\circ$.
The paths from $\pp_{long} \sub \pp$ which enter $\widehat s^\bb$ through $Q_i$ or $Q'_i$ are sketched in blue; they can be extended to reach $\widehat s$ in $H'_{3,Z}$ using the orange paths. 
}
\label{fig:homotopy4}
\end{figure}

\begin{lemma}\label{lem:homo:special-edges}
Let $Z \subseteq \{0,1\}^k$, $\bb \in \{0,1\}^k$, and $d = k\cdot 2^{k+1} + \func(\bb) + 1 + 1_{[\bb \in Z]}$.
Furthermore, let $(H^{\circ\bb}, {{\widehat s}^\bb}, {{\widehat t}}^\bb)$ be the $(s,t)$-dual of $H^\bb_{k,Z}$.

Then there exists a non-crossing family of $d$ edge-disjoint $({{\widehat s}}^\bb,{{\widehat t}}^\bb)$-paths $P^\circ_1, P^\circ_2, \dots, P^\circ_d$ in $H^{\circ\bb}$ such that (1)
every exposed edge belongs to some path $P^\circ_i$, and (2)
every path $P^\circ_i$ contains at most one exposed edge.
\end{lemma}
\begin{proof}
    The distance between $s$ and $t$ in $H^\bb_{k,Z}$ equals the length of the shortest $\bb$-homotopic $(s,t)$-path in $H'_{k,Z}$, which is $d = k\cdot 2^{k+1} + \func(\bb) + 1 + 1_{[\bb \in Z]}$ due to \cref{def:homo:subdivision:length}.
    Let $P^s, P^t$ be the $(s,t)$-paths within the outer cycle of $H^\bb_{k,Z}$, defined as in \cref{lem:homo:dual-flow}.
    We also reuse the definitions of vertices $v^s_i, u^s_i, v^t_i, u^t_i$ for $i\in[d]$.
    In order to derive the claim from \cref{lem:homo:dual-flow} we need to prove the following.
    \begin{enumerate}
        \item[(P1)] Every exposed edge is of the form  $v^s_iu^s_i$ or $v^t_iu^t_i$ for some $i\in[d]$.
        \item[(P2)] For each $i\in[d]$ only one of the edges $v^s_iu^s_i$, $v^t_iu^t_i$ can be exposed.
    \end{enumerate}
    
    All distances considered in this proof are measured with respect to the graph $H^\bb_{k,Z}$.

\begin{claim}\label{claim:homo:special-edges:monotone}
For each $i \in [k]$ and any vertices $x \in V(Q_i)$, $y \in V(Q'_i)$, it holds that $\dist(s,x) < \dist(s,y)$.
Furthermore, for $i \in [k-1]$ and any vertices $x \in V(Q'_i)$, $y \in V(Q_{i+1})$ it holds that $\dist(s,x) < \dist(s,y)$.
\end{claim}
\begin{innerproof}
We will use the following three observations.
First, for each $i \in [k]$ any path from $s$ to $ V(Q'_i)$ must intersect $V(Q_i)$ and, when $i > 1$, any path from $s$ to $ V(Q_i)$ must intersect $V(Q'_{i-1})$.
Next, the minimal distance between $V(Q_i)$ and $V(Q'_i)$ or $V(Q'_{i-1})$ is $2^k$.
Finally, for each two $u,v \in V(Q_i)$ (resp. $u,v \in V(Q'_i)$) we have $\dist(u,v) < 2^k$. 

We prove only the first claim in detail, as the second one has an analogous proof.
Let $y$ be the vertex from  $V(Q'_i)$ that minimizes distance from $s$ and $P$ be the shortest $(s,y)$-path in $H^\bb_{k,Z}$.
Then $V(P) \cap V(Q'_i) = \{y\}$. 
Let $x'$ be a vertex from $V(Q_i) \cap V(P)$.
Since the minimal distance between $V(Q_i)$ and $V(Q'_i)$ is $2^k$ we have $\dist(s,x') \le \dist(s,y) - 2^k$.
Now, for any other $x'
\in V(Q_i)$ we have $\dist(x,x') < 2^k$, what implies $\dist(s,x) \le \dist(s,x') + \dist(x',x) < \dist(s,y)$ and proves the 
claim due to the choice of $y$. 
\end{innerproof}

\begin{claim}\label{claim:homo:special-edges:eq}
For each $i \in [k]$ and $j \in \mathsf{Half}(k, {b}_i)$, it holds that $\dist(s,v'_{i,j}) = \dist(s,v_{i,j}) + 2^k$.
\end{claim}
\begin{innerproof}
 By the triangle inequality we have  $\dist(s,v'_{i,j}) \le \dist(s,v_{i,j}) + 2^k$.
 Let $P$ be a shortest $(s,v'_{i,j})$-path in $H^\bb_{k,Z}$.
 Let $j' \in \mathsf{Half}(k, {b}_i)$ be such that $v'_{i,j'}$ is the first vertex from $V(Q'_i)$ on $P$.
 Then $P$ visits $v_{i,j'}$ as well and so $\dist(s,v_{i,j'}) = \dist(s,v'_{i,j}) - \dist(v_{i,j'}, v'_{i,j}) =  \dist(s,v'_{i,j}) - 2^k - |j-j'|$.
 Next, $\dist(s,v_{i,j}) \le \dist(s,v_{i,j'}) + |j-j'| = \dist(s,v'_{i,j}) - 2^k$, what gives the inequality in the second direction.
\end{innerproof}

It follows that every exposed edge is of the form $vu$ where $\dist(s,u) = \dist(s,v) + 1$.
Moreover, when $v$ is the $\ell$-th vertex on the subdivided $(v_{i,j}, v'_{i,j})$-edge (counting from $v_{i,j}$) then $\dist(s,v) = \dist(s,v_{i,j}) + \ell$.

\begin{claim}\label{claim:homo:special-edges:long-enough}
For each $i \in [k]$ and any vertex $x \in V(Q_i)$, it holds that $\dist(s,x) + 2^k < \dist(s,t)$.
\end{claim}
\begin{innerproof}
Due to \cref{claim:homo:special-edges:monotone} it suffices to consider $i = k$. Let $x \in V(Q_k)$.
Let $P$ be a shortest $(s,t)$-path in $H^\bb_{k,Z}$.
This path must intersect  $V(Q_k)$; let $x'$ be a vertex from $V(Q_k) \cap V(P)$.
The minimal distance from $V(Q_k)$ to $t$ is $2^{k+1}$, hence $\dist(s,x') \le \dist(s,t) - 2^{k+1}$.
Since $\dist(x',x) < 2^k$, the claim follows from the triangle inequality.
\end{innerproof}

Let $vu$ be an exposed edge and $u$ be the vertex with 
$\ell = \dist(s,u) = \dist(s,v) + 1$.
Suppose w.l.o.g. that $vu \in E(P^s)$.
From Claims \ref{claim:homo:special-edges:monotone} and \ref{claim:homo:special-edges:eq} we obtain that 
for every vertex $v'$ that lies further than $v$ on $P^s$ it holds $\dist(s,v') \ge \ell$.
\cref{claim:homo:special-edges:long-enough} implies that $\ell \le \dist(s,t) = d$.
We therefore obtain property (P1): $vu = v^s_\ell u^s_\ell$ for some $\ell \in [d]$.
Similarly, when $vu \in E(P^t)$ then $vu = v^t_\ell u^t_\ell$ for some $\ell \in [d]$.

Finally, consider $vu \in E(P^s), v'u' \in V(P^t)$, 
such that $vu, v'u'$
are exposed.
Then there is $i \in [k]$ such that $V(Q'_i)$ separates $v$ from $v'$ and one of $v,v'$ belongs to the same connected component of $H^\bb_{k,Z} - V(Q'_i)$ as $s$.
By Claims \ref{claim:homo:special-edges:monotone} and \ref{claim:homo:special-edges:eq} 
the distances $\dist(s,v)$ and $\dist(s,v')$ are different, what implies property (P2).
We apply \cref{lem:homo:dual-flow} to obtain a family of $({{\widehat s}}^\bb,{{\widehat t}}^\bb)$-paths that satisfies the conditions of the lemma.
\end{proof}

As the last step, we want to employ \cref{lem:homo:special-edges} to construct a certain flow in the $(s,t)$-dual of $H'_{k,Z}$.
\mic{The only modification needed involves extending the paths from the $(s,t)$-dual of $H^\bb_{k,Z}$ that end at ${{\widehat s}}^\bb$ or ${{\widehat t}}^\bb$ but,
when considered in the $(s,t)$-dual of $H'_{k,Z}$, this endpoint corresponds to an internal face of $H'_{k,Z}$.}
Note that when an $({{\widehat s}}^\bb, {{\widehat t}}^\bb)$-path in the $(s,t)$-dual of $H^\bb_{k,Z}$ reaches its endpoint 
through an exposed edge then this endpoint corresponds to $(f_i)^\circ$ for some $i \in [k]$ in the $(s,t)$-dual of $H'_{k,Z}$.
\mic{For the remaining cases, we will extend the path}
to reach ${{\widehat s}}$ or ${{\widehat t}}$ using the subdivided edges between $Q_i$ and $Q'_i$.

\begin{lemma}\label{lem:homo:flow-exists}
Let $k \in \nn$, $Z \subseteq \{0,1\}^k$, $\bb \in \{0,1\}^k$, and $d = k\cdot 2^{k} + \func(\bb) + 1 + 1_{[\bb \in Z]}$.
Let $(H^\circ, {{\widehat s}}, {{\widehat t}})$ be the $(s,t)$-dual of $H'_{k,Z}$ and
$\tcal_{\bb,d}$ be the family of following requests:
     \begin{enumerate}
        \item $({{\widehat s}}, f^\circ_i, 2^k)$ for each $i$ with $b_i = 0$,
        \item $({{\widehat t}}, f^\circ_i, 2^k)$ for each $i$ with $b_i = 1$,
        \item $d$ copies of the request $({{\widehat s}}, {{\widehat t}}, 1)$. 
    \end{enumerate}
    Then there exists a non-crossing $\tcal_{\bb,d}$-flow in $H^\circ$,
    \mic{in which ${{\widehat s}}$ sees $\{f^\circ_i \mid b_i = 0\}$ in the order of decreasing $i$ and ${{\widehat t}}$ sees $\{f^\circ_i \mid b_i = 1\}$ in the order of increasing $i$ (recall \cref{def:homo:seeing-order}). }
\end{lemma}

\meir{Rotate fig:flow-extention? Will take more space, but will be clearer. \mic{Done. It looks better when merged with the previous figure.}}

\begin{proof}
Let $(H^{\circ\bb}, {{\widehat s}^\bb}, {{\widehat t}^\bb})$ be the $(s,t)$-dual of $H^\bb_{k,Z}$. 
Let $\ell = d + k\cdot 2^k$ and, for $i \in [k]$,
$E^{\circ\bb}_i \sub E(H^{\circ\bb})$ be the set of $2^k$ exposed edges 
located between $Q_i$ and $Q'_i$.
We apply \cref{lem:homo:special-edges} to obtain a non-crossing $({{\widehat s}^\bb},{{\widehat t}^\bb})$-flow $\pp = \{P^\circ_1, \dots, P^\circ_\ell\}$  in $H^{\circ\bb}$.
For each $i \in [k]$ there is a subfamily $\pp_i \sub \pp$ of $2^k$ paths containing an edge from $E^{\circ\bb}_i$.
Moreover, the subfamilies $\pp_1, \dots, \pp_k$ are disjoint.
Let $\pp_{long} = \pp \sm (\pp_1 \cup \dots \cup \pp_k)$.

Every internal face of $H^\bb_{k,Z}$ is also an internal face of $H'_{k,Z}$.
Therefore, every path in $H^{\circ\bb}$ that is internally disjoint from ${{\widehat s}^\bb}, {{\widehat t}^\bb}$
is also a path in $H^\circ$.
We can thus consider the flow $\pp$ in $H^\circ$.
When $P \in \pp_i$ for some $i \in [k]$ then one of its endpoints (incident to an exposed edge) becomes $f^\circ_i$.
The other endpoint (incident to a non-exposed edge) will be either ${{\widehat s}}$, ${{\widehat t}}$ or $g^\circ$, where $g$ is an internal face of $H'_{k,Z}$ that is not present in  $H^\bb_{k,Z}$.
When $P \in \pp_{long}$ then the latter scenario applies to both endpoints of $P$.

Let $g^i_1, \dots, g^i_{2^{k-1}}$ be the internal faces of  $H'_{k,Z}$ between $Q_i$ and $Q'_i$ that are not present in  $H^\bb_{k,Z}$, ordered in such a way that $g^i_1$ is incident to the outer face,  $g^i_{2^{k-1}} = f_i$,
and $g^i_{j+1}$ shares an edge with $g^i_j$.
When $e \in E(H^{\circ\bb})$ is incident to ${{\widehat s}}^\bb$ or ${{\widehat t}}^\bb$ in $H^{\circ\bb}$ then either $e$ is incident to 
${{\widehat s}}$, ${{\widehat t}}$, or some $(g^i_{j})^\circ$  in $H^{\circ}$.
In the last case, when $e$ is non-exposed then $e$ crosses $Q_i$ or $Q'_i$.

For each $i \in [k]$, $j \in [2^{k-1}]$, there is exactly one edge in $H^{\circ}$ incident to $(g^i_j)^\circ$ that crosses $Q_i$ and exactly one that crosses $Q'_i$; in total $2^k$ for fixed $i$.
Observe that for each $j \in [1, 2^{k-1}-1]$ there are $2^k$ parallel edges between the vertices $(g^i_j)^\circ$ and $(g^i_{j+1})^\circ$.
Therefore, paths from $\pp$ that reach some $(g^i_j)^\circ$ via a non-exposed edge 
can be extended to reach ${{\widehat s}}$ (resp.  ${{\widehat t}}$) in a non-crossing manner
(see Figure \ref{fig:homotopy4}, right, and the caption below).
As a result, the paths from families $\pp_i$ are being extended to satisfy requests of types (1, 2), while the paths from $\pp_{long}$ are being extended to satisfy requests of type (3).

\mic{Finally, we argue that ${{\widehat s}}$ (resp.  ${{\widehat t}}$) sees the vertices $f^\circ_1, f^\circ_2, \dots, f^\circ_k$ in the right order.
First, consider the flow $\pp$ in the graph $H^{\circ \bb}$, in which
every path is an $({{\widehat s}^\bb}, {{\widehat t}^\bb})$-path.
Observe that each subfamily $\pp_i \sub \pp$ forms a continuous interval in $\pp$ ordered with respect to the ordering of edges incident to ${{\widehat s}^\bb}$ (or ${{\widehat t}^\bb}$).
Consider now $i < j$ with $\bb_i = \bb_j = 0$, and let $P_i \in \pp_i$, $P_j \in \pp_j$.
Let $e_i, e_j \in E(H^{\circ \bb})$ be the edges on respectively $P_i,P_j$ that are incident to ${{\widehat s}^\bb}$. 
Then $e_i$ occurs later than $e_j$ in the clockwise ordering of edges  incident to ${{\widehat s}^\bb}$ (starting from the edge next to the outer face of $H^{\circ \bb}$), reflecting the relative order of the edges at the other ends of $P_i,P_j$.
The transformation between the flow in $H^{\circ \bb}$ and the flow in $H^\circ$ preserves the relative order of paths $P_i, P_j$, yielding the same relation with respect to  ${{\widehat s}}$.
The analysis for the paths ending at ${{\widehat t}}$ is symmetric: when $i < j$ and $\bb_i = \bb_j = 1$ then any path from $\pp_i$ occurs earlier than any path from $\pp_j$ with respect to the clockwise ordering of edges  incident to ${{\widehat t}^\bb}$.
This concludes the proof.
}
\end{proof}

\mic{The flow considered above is exactly the one that we require in the vector-containment gadget.
We can therefore summarize the construction.}

\begin{proposition}\label{prop:homo:final}
Let $\widehat{\func}_k \colon  \{0,1\}^k \to \nn$ be defined as $\widehat{\func}_k(\bb) = k \cdot 2^{k} + \func_k(\bb) + 1$.
For each $k$ and $Z \sub  \{0,1\}^k$ there exists a $(k,\widehat{\func}_k,Z)${-$\mathsf{Vector\, Containment\, Gadget}$} \mic{of size} $2^{\Oh(k)}$ and it can be constructed in time~$2^{\Oh(k)}$.
\end{proposition}
\begin{proof}
    Let $(H^\circ, {{\widehat s}}, {{\widehat t}})$ be the $(s,t)$-dual of the plane graph $H'_{k,Z}$.
    Both graphs can be constructed in time polynomial in their size, which is $2^{\Oh(k)}$. 
    We construct a $(k,\widehat{\func}_k,Z)${-$\mathsf{Vector\, Containment\, Gadget}$} using the graph $H^\circ$.
    We set $w_0 = {{\widehat s}}$, $w_1 = {{\widehat t}}$, and $z_i = (f_i)^\circ$ for $i \in [k]$.
    
    Let $\tcal_{\bb,d}$ be the family of requests defined as in 
    \cref{lem:homo:flow-exists}.
    We need to show the following conditions to be equivalent:
    \begin{enumerate}[label=(\alph*)]
        \item $d \le \widehat{\func}_k(\bb) + 1_{[\bb \in Z]}$,
        \item there exists a  $\mathcal{T}_{\bb,d}$-flow in $H^\circ$,
        \item there exists a non-crossing $\mathcal{T}_{\bb,d}$-flow in $H^\circ$,
         \mic{in which ${{\widehat s}}$ sees $\{f^\circ_i \mid b_i = 0\}$ in the order of decreasing $i$ and ${{\widehat t}}$ sees $\{f^\circ_i \mid b_i = 1\}$ in the order of increasing $i$.}
    \end{enumerate}

    The implication (c) $\Rightarrow$ (b) is trivial. To see (b) $\Rightarrow$ (a), consider some shortest $\bb$-homotopic $(s,t)$-path $P$ in $H'_{k,Z}$.
    By \cref{def:homo:subdivision:length} the length of $P$ equals $k\cdot 2^{k+1} + \func(\bb) + 1 + 1_{[\bb \in Z]}$.
    Each \mic{walk} $Q$ in a $\tcal_{\bb,d}$-flow in $H^\circ$ must cross the path $P$, i.e., there exists $e \in E(P)$ such that $e^\circ \in E(Q)$.
    Therefore the number of paths in a $\tcal_{\bb,d}$-flow, that is $k\cdot 2^k + d$, cannot be greater than the length of~$P$.
    This implies $d \le \widehat{\func}_k(\bb) + 1_{[\bb \in Z]}$.
    The last implication (a) $\Rightarrow$ (c) is proven in 
    \cref{lem:homo:flow-exists}.
\end{proof}

\subsection{Subset gadget}
\label{sec:hard:subset}

The aim of the following gadget is to determine whether a given set $F \sub [k]$ is a subset of one of $2^r$ sets from a family $\scal$.
To make the notation consistent with the previous gadget, we encode the set family as a function from the bit vectors to the subsets of $[k]$.
We use the notation $2^{[k]}$ to distinguish the family of subsets of $[k]$ from the family of vectors of length $k$.

\begin{definition}\label{def:subset:gadget}
Let $r,k$ be integers and $\mathcal{S} \colon \{0,1\}^r \to 2^{[k]}$ be a function.
We say that a pair $(G,\tcal)$ is an $(r,k,\mathcal{S})$-$\mathsf{Subset\, Gadget}$ if the following conditions hold.
\begin{enumerate}
    \item $G$ is a plane multigraph with $2k$ distinguished vertices $s_1,t_1, s_2, t_2, \dots, s_k,t_k$ lying on the outer face in this clockwise order.
    \item $\tcal$ is a set of triples from $V(G) \times V(G) \times \nn$. 
    \item For $F \sub [k]$ let $\tcal_F = \{(s_i,t_i, 1) \mid i \in F\}$.
    Then there exists a non-crossing $(\tcal \cup \tcal_F)$-flow in $G$ if and only if 
    there exists $\bb \in \{0,1\}^r$ for which $F \subseteq \scal(\bb)$.
\end{enumerate}  
\end{definition}

\noindent Our goal is to construct an $(r,k,\mathcal{S})${-$\mathsf{Subset\, Gadget}$} of size $2^{\Oh(r)}\cdot k^{\Oh(1)}$ 
with $|\tcal| = (r+k)^{\Oh(1)}$.

\subsubsection{The first attempt}
\label{sec:subset-simple}
Let $\gamma^0 \colon \{0,1\}^r \to \nn$ be the zero-function, i.e., $\gamma^0(\bb) = 0$ for all $\bb \in \{0,1\}^r$.
In this section we present a simplified construction working under the assumption that for any set $Z \sub \{0,1\}^r$ there exists an $(r,{\func^0},Z)${-$\mathsf{Vector\, Containment\, Gadget}$} of size $2^{\Oh(r)}$. 
Of course, this assumption is overly optimistic but it allows us to first present the pattern propagation mechanism alone. 
We strongly encourage the Reader to first familiarize with this simplified construction before reading the proper proof in \cref{sec:subset-full}.
The proper proof builds atop this construction in an incremental way. 

We use following conventions to describe the constructed graphs.
When $H$ is a graph with a distinguished vertex named $v$
and a graph $G$ is constructed using explicit vertex-disjoint copies of the graph $H$, referred to as $H_1, H_2, \dots, H_\ell$, 
we refer to the copy of $v$ within the subgraph $H_i$ as~$H_i[v] \in V(G)$.
When vertices $u, v$ are connected by multiple parallel edges, we refer to the number of such edges as the capacity of $uv$.

\begin{figure}
    \centering
\centerline{\includegraphics[scale=0.63]{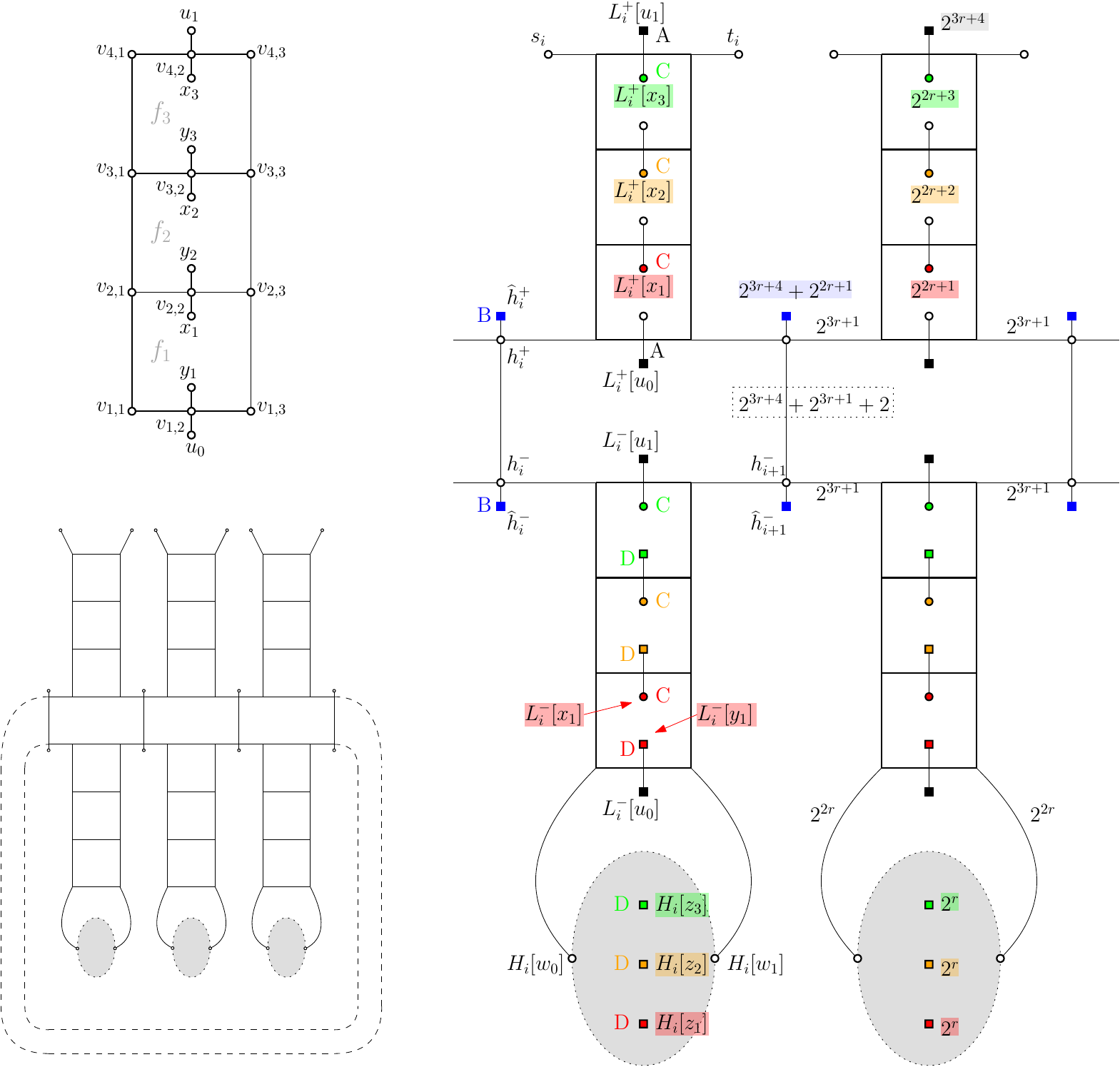}}
\caption{
Top left: a $3$-ladder labeled with vertices' and faces' names.
Right: a fragment of the graph $\mathsf{Ring}(r,k,\mathcal{H})$ for $r=3$.
The subgraph $R_i$ is labeled with the vertices' names while the subgraph $R_{i+1}$ is labeled with the edge capacities and the numbers of walks requested between each terminal pair (on a colorful background). 
The ladder edges with capacities $2^{3r+5}$ are drawn with thicker lines. 
The vertices that need to be connected in the flow $\tcal_{r,k}$ share common colors and shapes.
They are also labeled with letters indicating the types of requests.
The gray ovals in the bottom represent the subgraphs $H_i$, $H_{i+1}$ (the vector-containment gadgets).
Bottom left: a sketch of the ring structure obtained from combining the subgraphs $R_1, R_2, \dots, R_k$.
}
\label{fig:ladder1}
\end{figure}

\paragraph*{The ladder.}
An {\em $r$-ladder} 
is a plane multigraph defined as follows (see \Cref{fig:ladder1}).
We begin construction from $r+1$ disjoint paths $(v_{1,1}, v_{1,2}, v_{1,3}), \dots, (v_{r+1,1}, v_{r+1,2}, v_{r+1,3})$, followed by adding additional edges forming paths $(v_{1,1}, v_{2,1}, \dots, v_{r+1,1})$ and $(v_{1,3}, v_{2,3}, \dots, v_{r+1,3})$.
Next, we duplicate each edge $2^{3r+5}$ times (i.e., we place this many parallel edges).

We create vertices $u_0, u_1$ on the outer face adjacent 
respectively to $v_{1,2}$ and $v_{r+1,2}$.
Let $f_1, \dots, f_r$ be the internal faces in the already constructed graph, numbered in such a way that $v_{1,1}$ is incident to $f_1$ and for each $i \in [r-1]$ the faces $f_i, f_{i+1}$ share an edge.
For each $j \in [r]$ we create vertices $x_j, y_j$ within the face $f_j$, and then
insert edges $x_jv_{j+1,2}$ and $y_jv_{j,2}$.

\paragraph*{The ring.}
Let $\mathcal{H} = H_1, \dots, H_k$ be a sequence of plane multigraphs that
satisfy condition (1) of \cref{def:homo:elemenet-gadget}: each $H_i$ has $r+2$ distinguished vertices $z_1, \dots, z_r, w_0, w_1$ so that the last two lie on the outer face. 
We build the graph $\mathsf{Ring}(r,k,\mathcal{H})$ from $k$ blocks arranged in a ring-like structure, so that $H_i$ will be installed inside the $i$-th block.

For $i \in [k]$ we start construction of the plane multigraph $R_i$ from two copies of an $r$-ladder, $L_i^+, L_i^-$.
For $\odot \in \{+,-\}$ we duplicate the edges incident to $L_i^\odot[u_0], L_i^\odot[u_1]$ times $2^{3r+4}$.
For $j \in [r]$ we duplicate the edge incident to $L^\odot_i[x_j]$ times $2^{2r+j}$ and the edge incident to $L^-_i[y_j]$ times $2^{2r}$.

Next, we create six vertices: $s_i, t_i, h_i^+, \widehat h_i^+, h_i^-, \widehat h_i^-$.
For $\odot \in \{+,-\}$ we insert $2^{3r+4} + 2^{2r}$ parallel edges
$h_i^\odot\widehat h_i^\odot$ 
and $2^{3r+4} + 2^{3r+1} + 2$ parallel edges $h_i^+ h_i^-$.
Next, we put $2^{3r+1}$ parallel edges between $h_i^+$ and $L_i^+[v_{1,1}]$ (the bottom-left corner vertex of the upper ladder),
and $2^{3r+1}$ parallel edges between $h_i^-$ and $L_i^-[v_{r+1,1}]$ (the upper-left corner vertex of the lower ladder).
We insert a single edge between $s_i$ (resp. $t_i$) and $L_i^+[v_{r+1,1}]$ (resp. $L_i^+[v_{r+1,3}]$). 

 Recall that $H_i$ has vertices $H_i[w_0], H_i[w_1]$ on its outer face. 
 We connect $H_i[w_0]$ (resp. $H_i[w_1]$) to $L_i^-[v_{1,1}]$ (resp. $L_i^-[v_{1,3}]$) (the bottom corners of the lower ladder) via $2^{2r}$ parallel edges.
The arrangement of the vertices on the plane is depicted in \Cref{fig:ladder1}.

Finally, we arrange the multigraphs $R_1, R_2, \dots, R_k$ into a ring.
For $i \in [k-1]$ 
we insert $2^{3r+1}$ parallel edges between $L_i^+[v_{1,3}]$ and  $h_{i+1}^+$, as well as between $L_i^-[v_{r+1,3}]$ and  $h_{i+1}^-$.
The graphs $R_k, R_1$ get connected in the same way.
The constructed ring encloses a bounded region incident to the minus-sides of the multigraphs $R_1, R_2, \dots, R_k$.

\vspace{3mm}
\mic{Next, we define the requests of the gadget, divided into four groups.
Although it is possible to achieve the same properties with slightly smaller demands, we choose to use the same numbers that appear in \cref{sec:subset-full} in order to reduce the edit distance between the proofs.
}

\paragraph{The requests.}
We define a family $\tcal_{r,k}$ of requests over  $\mathsf{Ring}(r,k,\mathcal{H})$. 
\begin{enumerate}[label=\Alph*.]
    \item In each ladder $L$ of the form $L^+_{i}, L^-_{i}$ we create a request $(L[u_0], L[u_1], 2^{3r+4})$, $2k$ in total.
    \item For each $i \in [k]$ we create a request $(\widehat h_i^+, \widehat h_i^-, 2^{3r+4} + 2^{2r})$. 
    \item For each $i \in [k]$ and $j \in [r]$ we create a request $(L^+_{i}[x_j], L^-_{i}[x_j], 2^{2r + j})$. 
    \item For each $i \in [k]$ and $j \in [r]$ we create a request $(L^-_{i}[y_j], H_i[z_j], 2^r)$. 
\end{enumerate}

For a  $\tcal_{r,k}$-flow $\pp$ we use variables $\pp^{A+}_i$, $\pp^{A-}_i$, $\pp^B_i$, $\pp^C_{i,j}$, $\pp^D_{i,j}$ to refer to subfamilies of $\pp$ satisfying the respective types of requests.

\begin{observation}\label{obs:numEdgesIncident}
    For every vertex $v$ of the form $\widehat h_i^\odot$, $L^\odot_{i}[u_j]$, $L^\odot_{i}[x_j]$, $L^-_{i}[y_j]$, 
    the number of edges incident to $v$ equals the number of walks in a  $\tcal_{r,k}$-flow that have an endpoint at $v$.
\end{observation}

This observation allows us to exclude cases where one walk would start at a vertex $v$ and another walk would only pass through $v$.

\paragraph{Intuition.}
The aim of the (A) requests is to draw a pattern through each ladder which splits it into the left and the right side (see \Cref{fig:ring-solution}).
Then the (C) requests must be satisfied by walks between the upper ladder and the lower ladder that traverse the middle belt either through the left or the right side, according to the pattern above.
The (B) requests work as guards and ensure that no other walks of type (C) are possible (ruling out the possibility that some walk winds around the entire ring).
Since the (C) requests encode powers of two, a partition of them encodes an integer from 0 to $2^r - 1$.
Because the blocks are arranged in a ring structure and the number of walks that can cross $h^+_ih^-_i$ is limited, these integers must coincide.
As a consequence, the pattern drawn in each ladder must be the same; in this way, the flow ``chooses'' a vector $\bb \in \{0,1\}^r$. 

The (D) requests originate from the same faces in the lower ladders as their (C) counterparts, so they also need to stay either on the left or right side of the ladder.
This determines which entrance to the vector-containment gadget $H_i$ they can use (left or right).
Then the vertex $H_i[z_j]$ must be connected to $H_i[w_0]$ when $\bb_j = 0$ or to $H_i[w_1]$ when $\bb_j = 1$.
The vector-containment gadget governed by the zero-function $\gamma^0$ alters its behavior depending on whether $\bb$ belongs to a certain set: it
allows one additional walk between $H_i[w_0]$ and $H_i[w_1]$ if and only if the containment occurs.
Subsequently, this leaves space for an $(s_i,t_i)$-walk, which must go through $H_i$ because the walks of type (A) and (B) block the other passages.
By supplying appropriate sets to the vector-containment gadgets, we enforce that an $(s_i,t_i)$-walk can be accommodated exactly when $i \in \scal(\bb)$. 

\paragraph{}
For $\mathcal{S} \colon \{0,1\}^r \to 2^{[k]}$ and $i \in [k]$ we define $Z^\scal_i \sub \{0,1\}^r$ as the set of vectors $\bb$ for which $i \in \scal(\bb)$.
Assume that for each $i \in [k]$ there exists a plane multigraph $H_i$ which is an $(r,\func^0,Z^\scal_i)$-$\mathsf{Vector\, Containment\, Gadget}$.
We define $\mathcal{H}^{vcg}_\mathcal{S}$ as $H_1, \dots, H_k$.
We are going to show that then
$(\mathsf{Ring}(r,k,\mathcal{H}^{vcg}_\mathcal{S}), \tcal_{r,k})$ forms an $(r,k,\mathcal{S})$-$\mathsf{Subset\, Gadget}$.

We need to prove two implications to establish condition (3) of \cref{def:subset:gadget}.
We begin from the easier one: when $F \sub \scal(\bb)$ for some vector $\bb$, then the desired non-crossing flow exists.
For a flow $\pp$ and a walk $W$, we say that $W$ is non-crossing with $\pp$ if $W$ does not cross or share an edge with any walk in $\pp$.

\begin{figure}
    \centering
\centerline{\includegraphics[scale=0.55]{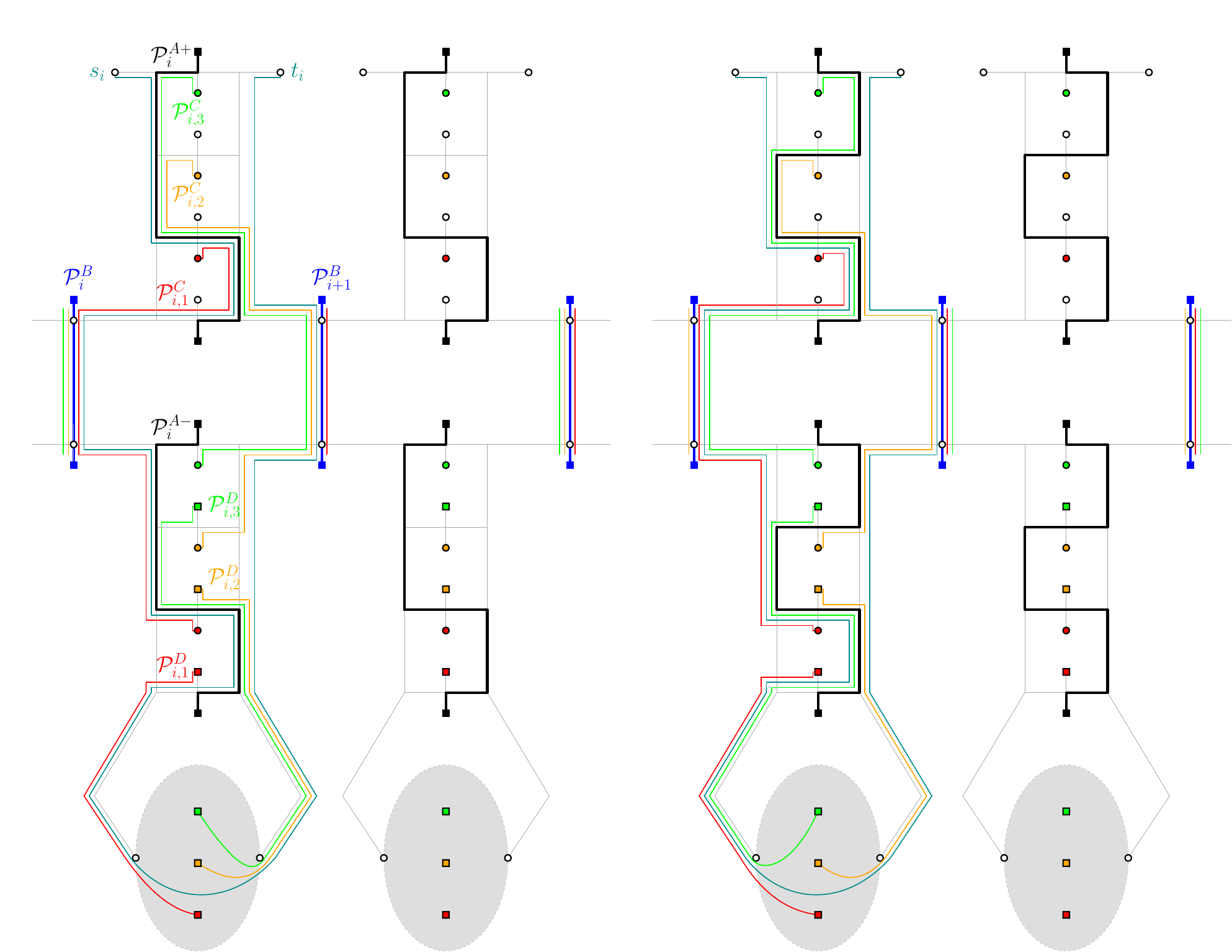}}
\caption{Two examples of solutions constructed in \cref{lem:subset:simple:forward}. 
The walks from families $\pp^{A\odot}_i$ and $\pp^B_i$ are drawn in black and blue, respectively.
The colors red, orange, and green represent 
walks from families $\pp^{C}_{i,j}$ and $\pp^{D}_{i,j}$ (single color for each $j \in \{1,2,3\}$).
In each solution, the choice of which colors go through the left or right passage is fixed for all $i\in [k]$ because otherwise some edge passing the middle belt (along the blue walks from $\pp^B_i$) would be overloaded.
The $(s_i,t_i)$-walk is drawn in~cyan.
}
\label{fig:ring-solution}
\end{figure}

\meir{Mark the families on the figure (\ref{fig:ring-solution})? \mic{done}}

\begin{lemma}\label{lem:subset:simple:forward}
Consider $r,k\in \nn$  and $\mathcal{S} \colon \{0,1\}^r \to 2^{[k]}$. 
Assume that the sequence of multigraphs $\mathcal{H}^{vcg}_\mathcal{S}$ exists.
For $F \sub [k]$ let $\tcal_F = \{(s_j, t_j, 1) \mid j \in F\}$.
If $F \sub \scal(\bb)$ for some $\bb \in \{0,1\}^r$ there exists a non-crossing $( \tcal_{r,k} \cup \tcal_F)$-flow in $\mathsf{Ring}(r,k,\mathcal{H}^{vcg}_\mathcal{S})$.
\end{lemma}
\begin{proof}
The construction is depicted in \Cref{fig:ring-solution}.
We begin by describing which vertices are being visited by each walk and later we check that the edge capacities are sufficient to accommodate all the walks.

First, consider $i \in [k]$ and $\odot \in \{+,-\}$.
Let $P$ be a path that traverses the ladder $L^\odot_i$ in such a way that the face $f_j$ is to the right of $P$ exactly when $\bb_j = 1$. 
Each  walk from $\pp^{A\odot}_i$ traverses $L^\odot_i$ through the same vertices as $P$.
Next, every walk in the family $\pp^{B}_i$ is of the form
$(\widehat h_i^+, h_i^+, h_i^-, \widehat h_i^-)$.

Now we describe the families $\pp^{C}_{i,j}$, $\pp^{D}_{i,j}$.
First, suppose that $i \not\in F$.
Whenever $\bb_j = 0$ all the walks from $\pp^{C}_{i,j}$ go through an edge $h^+_ih^-_{i}$ and otherwise they go through $h^+_{i+1}h^-_{i+1}$ (counting modulo $k$).
We arrange them from left to right in such a way 
that the families
 $\pp^{C}_{i,j}$ appear in the order of increasing $j$ (see \Cref{fig:ring-solution}).
For each $j \in [r]$ there is a common internal face of $L^-_i$ incident to $L^-_i[x_j]$ (the endpoint for $\pp^{C}_{i,j}$) and 
$L^-_i[y_j]$ (the endpoint for $\pp^{D}_{i,j}$).
When constructing the walks in a top-bottom fashion, one can imagine that the walks from $\pp^{D}_{i,j}$ replace the ones from $\pp^{C}_{i,j}$ in the same position among the other walks.
The walks from the families $\pp^{D}_{i,j}$ with $\bb_j=0$ reach the vertex $H_i[w_0]$ in a monotone order (with respect to $j$).
Similarly, the walks from $\pp^{D}_{i,j}$ with $\bb_j=1$ reach the vertex $H_i[w_1]$,
however this time we arrange them from left to right in the order of decreasing $j$.
It remains to connect $H_i[z_j]$ to $H_i[w_0]$ (when $\bb_j = 0$) or to $H_i[w_1]$ (when $\bb_j = 1$) in a non-crossing way.
Since we do not need to accommodate any additional $H_i[w_0]H_i[w_1]$-walks, the value of $d$ in \cref{def:homo:elemenet-gadget} is 0 and the desired non-crossing flow exists.
\mic{The condition (2c) in the definition ensures that the order of walks entering $w_0$ (resp. $w_1$) in $H_i$ matches the order of walks outside $H_i$.}

When $i \in F$ we begin from constructing a non-crossing $\tcal_{\bb,1}$-flow in $H_i$: we request $2^r$ walks from  $H_i[z_j]$ to $H_i[w_0]$ (when $\bb_j = 0$) or to $H_i[w_1]$ (when $\bb_j = 1$) and a single $H_i[w_0]H_i[w_1]$-walk.
Because $i \in F \sub \scal(b)$ we have $\bb \in Z^\scal_i$.
By the definition of a $(r,\gamma^0, Z^\scal_i)$-$\mathsf{Vector\, Containment\, Gadget}$, there exists a non-crossing $\tcal_{\bb,1}$-flow in $H_i$.
We construct the remainders of walks from $\pp^{D}_{i,j}$, as well as walks from $\pp^{C}_{i,j}$, similarly as before.
The only difference is that we place an $H_i[w_0]s_i$-walk $W^0_i$ in between the left-side walks and an $H_i[w_1]t_i$-walk $W^1_i$ in between the right-side walks, creating an $s_it_i$-walk as a result.
We keep the same relative position of $W^0_i$ (resp. $W^1_i$) among the families $\pp^{C}_{i,j}$, $\pp^{D}_{i,j}$ as the position on which it leaves the subgraph $H_i$ (see \Cref{fig:ring-solution}).

Finally, we check that we have a sufficient number of parallel edges.
First consider the edges $h^+_ih^-_i$: there are $2^{3r+4} + 2^{3r+1} + 2$ copies of each.
There are $2^{3r+4} + 2^{2r+1}$ walks in $\pp^B_i$.
Next, for each $j \in [r]$, exactly one of the families  $\pp^C_{i-1,j}$, $\pp^C_{i,j}$ go through $h^+_ih^-_i$.
These sum up to
\[
\sum_{j=1}^r 2^{2r+j} = 2^{2r+1}\cdot (2^r - 1) = 2^{3r+1} - 2^{2r+1}.
\]
The only additional walks that might go through $h^+_ih^-_i$ are $W^1_{i-1}$ and $W^0_i$.
In total we obtain
\[
(2^{3r+4} + 2^{2r+1}) +  (2^{3r+1} - 2^{2r+1}) + 2 = 2^{3r+4} + 2^{3r+1} + 2
\]
walks, as intended.

The number of walks passing between a vertex $h^\odot_i$ and any ladder is at most $\sum_{j=1}^r 2^{2r+j} + 1$ (the additive 1 depends on whether $i \in F$) which is bounded by $2^{3r+1}$, the capacity of this passage.
The number of walks going from $H_i[w_0]$ to $L^-[v_{1,1}]$ (resp. from $H_i[w_1]$ to $L^-[v_{1,3}]$) is at most $r2^r + 1 \le 2^{2r}$.
Within each ladder, every pair of adjacent vertices of the form $v_{x,y}$ is connected via $2^{3r+5}$ parallel edges, which upper bounds the total number of walks in $\pp^{A\odot}_i$, $\pp^C_{i,j}$, $\pp^D_{i,j}$, together with $W^0_i, W^1_i$.
This concludes the proof.
\end{proof}

We have thus established the first implication in the proof of correctness. 
Next, we show that any non-crossing $\tcal_{r,k}$-flow must obey certain properties.
They will allow us to prove the second implication: when adding requests encoding a set $F \sub [k]$ results in a satisfiable instance, then $F$ must be a subset of some $\scal(\bb)$.

\begin{lemma}\label{lem:subset:simple:reverse}
Consider $r,k\in \nn$  and $\mathcal{S} \colon \{0,1\}^r \to 2^{[k]}$.
Let $\mathcal{H} = H_1, \dots, H_k$ be arbitrary and
$\pp$ be a non-crossing $\tcal_{r,k}$-flow in $\mathsf{Ring}(r,k,\mathcal{H})$. 
Then the following hold.
\begin{enumerate}
    \item Let  $i \in [k]$, $\odot \in \{+,-\}$, and $W$ be an $(h_i^\odot, h_{i+1}^\odot)$-walk internally contained in $L^\odot_i$.
    Then $W$ cannot be non-crossing with $\pp^{A\odot}_i$.\label{claim:subset:simple:reverse:AvsW}
    
    \item For each $i \in [k]$ the family $\pp^B_i$ contains a walk on vertices $\{\widehat h_i^+, h_i^+, h_i^-, \widehat h_i^-\}$. Moreover, every walk from $\pp^B_i$ goes through an edge $h_i^+h_i^-$.\label{claim:subset:simple:reverse:B} 

    \item There exists a vector $\bb \in \{0,1\}^r$ such that
    for each $i\in [k]$, 
    $j \in [r]$, every walk $P \in \pp^D_{i,j}$ contains an $(H_i[w_0],H_i[z_j])$-walk in $H_i$ when $\bb_j = 0$, or an $(H_i[w_1],H_i[z_j])$-walk in $H_i$ when $\bb_j = 1$. \label{claim:subset:simple:reverse:D-subwalk}
    
\end{enumerate}
\end{lemma}
\begin{proof}
    Within this proof, whenever we perform addition +1 or subtraction -1 from $i \in [k]$, we do it modulo $k$, that is, we \mic{adopt the convention that} $k + 1 = 1$ and $1-1 = k$.

     \begin{innerproof}[Proof of (1)] 
     First we argue that there exists a walk $Q \in \pp^{A\odot}_i$ contained entirely within $L^\odot_i$.
    To see this, we count the  total number of edges leaving   $L^\odot_i$; there are at most $2\cdot 2^{3r+1} + 2\cdot 2^{2r}$ of them, which is less than $|\pp^{A\odot}_i| = 2^{3r+4}$. 
    Therefore, at least one walk from $\pp^{A\odot}_i$ never leaves the subgraph~$L^\odot_i$.
     
    Now suppose that $W,Q$ are edge-disjoint and non-crossing; then $W$ must go through either $L^+_i[u_0]$ or $L^+_i[u_1]$.
    But the number of edges incident to each of these vertices equals the number of walks in $\pp^{A+}_i$.
    Therefore $W$ cannot be edge-disjoint with every walk in  $\pp^{A+}_i$; a contradiction.
        
     \end{innerproof}

    \begin{innerproof}[Proof of (2)] 
    As before, we count the total number of edges leaving the subgraph induced by  $\{\widehat h_i^+, h_i^+, h_i^-, \widehat h_i^-\}$ to be $4 \cdot 2^{3r+1} = 2^{3r+3}$.
    Since this number is less than the size of the family $\pp^B_i$, at least one walk from $\pp^B_i$ does not leave the vertex set $\{\widehat h_i^+, h_i^+, h_i^-, \widehat h_i^-\}$. 
    
    Suppose now that $P \in \pp^B_i$ does not go through any edge $h_i^+h_i^-$.
    Then $P$ goes through the vertex $h_{i+1}^+$ or $h_{i-1}^+$.
    Assume w.l.o.g. the first scenario, so $P$ needs to traverse $L^+_i$.
    Then $P$ contains a subwalk that meets the specification of Part (\ref{claim:subset:simple:reverse:AvsW}) of the lemma.
    This contradicts the assumption that $\pp$ is a non-crossing flow.
    \end{innerproof}
    
    We need two intermediate observations to reach the last claim of the lemma.
    \begin{claim}\label{claim:subset:simple:reverse:C-decision}
        For each $i \in [k]$ there exists a
        vector $\bb^i \in \{0,1\}^r$
        such that when $\bb^i_j = 0$ then all the walks from  $\pp^{C}_{i,j}$ go through an edge $h_i^+h_i^-$
        and when $\bb^i_j = 1$ then all the walks from  $\pp^{C}_{i,j}$ go through an edge $h_{i+1}^+h_{i+1}^-$. 
    \end{claim}
    \begin{innerproof}
        First we argue that for every $j \in [r]$ each walk from $P \in \pp^{C}_{i,j}$ goes either through $h_i^+h_i^-$ or $h_{i+1}^+h_{i+1}^-$.
        The edges  $\widehat h_i^+h_i^+$ are saturated by the walks from $\pp^B_i$, so $P$ cannot visit the vertex $\widehat h_i^+$.
        By Part (\ref{claim:subset:simple:reverse:B}), there is a walk $W \in \pp^B_i$ on vertex set $\{\widehat h_i^+, h_i^+, h_i^-, \widehat h_i^-\}$.
        Because $P$ and $W$ do not cross, $P$ is blocked from the left by the path $(\widehat h_i^+, h_i^+, h_i^-, \widehat h_i^-)$.
        Similarly, $P$ is blocked from the right by the path $(\widehat h_{i+1}^+, h_{i+1}^+, h_{i+1}^-, \widehat h_{i+1}^-)$.
        Therefore $P$ must proceed alongside one of these paths.
    
        Suppose now that for some $j \in [r]$ the family $\pp^{C}_{i,j}$
        contains a walk $W_0$ that does not go through any edge $h_{i+1}^+h_{i+1}^-$ and
        a walk $W_1$ that does not go through any edge $h_i^+h_i^-$.
        By \cref{def:reduction:flow} of a non-crossing flow, the concatenation $W_0 + W_1$ does not cross any walk from $\pp^{A+}_i$.
        But then $W_0 + W_1$ contains a subwalk that meets the specification of Part (\ref{claim:subset:simple:reverse:AvsW}).
        This  contradicts the assumption that $\pp$ is a non-crossing flow.
        Therefore for each $j \in [r]$ the choice whether to go via the left passage or the right one is fixed.

    \end{innerproof}

    Let us keep the variable $\bb^i$ to indicate the vector defined in  \cref{claim:subset:simple:reverse:C-decision}.
    
    \begin{claim}\label{claim:subset:simple:reverse:C-equal}
        There exists a single vector $\bb \in \{0,1\}^r$ so that $\bb^i = \bb$ for all $i \in [k]$. 
    \end{claim}
    \begin{innerproof}
        We define $\tau(b_1b_2\dots b_r) = \sum_{h=1}^r b_h\cdot 2^{h-1}$.
        Suppose that the claim does not hold.
        Because we work on a ring structure, there exists $i \in [k]$ for which $\tau(\bb^i) < \tau(\bb^{i+1})$.
        By \cref{claim:subset:simple:reverse:C-decision} the number of walks from $ \pp^C_{i,1} \cup  \pp^C_{i,2} \cup \dots  \pp^C_{i,r}$ that go through an edge $h_{i+1}^+h_{i+1}^-$ equals $2^{2r+1} \cdot \tau(\bb^i)$.
        On the other hand, the number of walks from $ \pp^C_{i+1,1} \cup  \pp^C_{i+1,2} \cup \dots  \pp^C_{i+1,r}$ that go through an edge $h_{i+1}^+h_{i+1}^-$  equals $2^{2r+1} \cdot (2^{r} - 1 - \tau(\bb^{i+1}))$.
        Since $\tau(\bb^i) < \tau(\bb^{i+1})$, this quantity is at least $2^{2r+1} \cdot (2^{r} - \tau(\bb^{i}))$.
        In total, we obtain at least $2^{3r+1}$ walks that go through $h_{i+1}^+h_{i+1}^-$.
        Due to Part (\ref{claim:subset:simple:reverse:B}) of the lemma, all $2^{3r+4} + 2^{2r+1}$ walks from $\pp^B_{i+1}$ also go through $h_{i+1}^+h_{i+1}^-$.
        But there are only $2^{3r+4} + 2^{3r+1} + 2$ parallel edges $h_{i+1}^+h_{i+1}^-$, which are too few to accommodate all $2^{3r+4} + 2^{3r+1} + 2^{2r+1}$ walks above, and       
        so we arrive at a~contradiction.
    \end{innerproof}

    \begin{innerproof}[Proof of (3)]  Let $\bb \in \{0,1\}^r$ be the vector from \cref{claim:subset:simple:reverse:C-equal}.
    Consider some $i \in [k]$ and $j \in [r]$.

    When $\bb_j = 0$ then any walk from $\pp^C_{i,j}$ must enter $L^-_i$ via $L^-_i[v_{r+1,1}]$, that is, the upper left corner, due to \cref{claim:subset:simple:reverse:C-decision}.
    Since $|\pp^C_{i,j}| > 2^{2r}$, there is at least one walk in $\pp^C_{i,j}$ does not use any edge $L^-_i[v_{1,1}]H_i[w_0]$  (there are only $2^{2r}$ such parallel edges).
    This walk includes an $(h^-_i,L^-_i[x_j])$-walk  
    internally contained in $L^-_i$.
    Note that both the vertices $L^-_i[x_j]$, $L^-_i[y_j]$ lie on the face $f_j$ of $L^-_i$ and the number of edges incident to each of them equals the number of walks in $\pp^C_{i,j}$, $\pp^D_{i,j}$, respectively.
    Therefore no walk from $\pp^{A-}_{i}$ can visit  $L^-_i[x_j]$ nor  $L^-_i[y_j]$.

    Consequently, when $\bb_j = 0$ then each walk $W \in \pp^D_{i,j}$ must also stay ``on the left'' of any walk from $\pp^{A-}_{i}$.
    The walk $W$ cannot 
    pass through $L^+_i$ or cross the path $(\widehat h_i^+, h_i^+, h_i^-, \widehat h_i^-)$ by the same argument as in \cref{claim:subset:simple:reverse:C-decision}.
    Therefore, the only possibility for $W$ to reach $H_i[z_j]$ is to enter $H_i$ through $H_i[w_0]$ and utilize some $(H_i[w_0],H_i[z_j])$-walk in $H_i$.
    
    The argument for the case $\bb_j = 1$ is analogous.
    \end{innerproof}
    
 This concludes the proof of \cref{lem:subset:simple:reverse}.
\end{proof}

\mic{Having imposed a structure of a non-crossing $\tcal_{r,k}$-flow, we can finish the correctness proof.} 

\begin{lemma}\label{lem:subset:simple:reverse-final}
Consider $r,k\in \nn$  and $\mathcal{S} \colon \{0,1\}^r \to 2^{[k]}$.
Assume that the sequence of multigraphs $\mathcal{H}^{vcg}_\mathcal{S}$ exists.
For $F \sub [k]$ let $\tcal_F = \{(s_j, t_j, 1) \mid j \in F\}$.
Suppose that there exists a non-crossing $( \tcal_{r,k} \cup \tcal_F)$-flow in $\mathsf{Ring}(r,k,\mathcal{H}^{vcg}_\mathcal{S})$.
Then there exists $\bb \in \{0,1\}^r$ for which $F \sub \scal(\bb)$.
\end{lemma}
\begin{proof}
    Let $\pp$ be a $\tcal_{r,k}$-flow and $\pp_F$ be  a $\tcal_{F}$-flow
    so that $\pp \cup \pp_F$ is non-crossing in $\mathsf{Ring}(r,k,\mathcal{H}^{vcg}_\mathcal{S})$.
    We apply \cref{lem:subset:simple:reverse} to $\pp$; let $\bb \in \{0,1\}^r$ be the vector given by Part (\ref{claim:subset:simple:reverse:D-subwalk}) of the lemma.
    Fix $i \in F$.
    We obtain that when $\bb_j = 0$ then
    each walk $P \in \pp^D_{i,j} \sub \pp$
    contains an $(H_i[w_0],H_i[z_j])$-walk within $H_i$, and 
    when $\bb_j = 1$ each walk  $P \in \pp^D_{i,j}$ contains an $(H_i[w_1],H_i[z_j])$-walk within $H_i$.
    Therefore, the subwalks of $\pp^D_{i,j}$ within $H_i$ satisfy request $(H_i[w_0], H_i[z_j], 2^r)$ when  $\bb_j = 0$ or request $(H_i[w_1], H_i[z_j], 2^r)$ when  $\bb_j = 1$.

    Now consider the $(s_i, t_i)$-walk $P_i \in \pp_F$.
    By \cref{lem:subset:simple:reverse}(\ref{claim:subset:simple:reverse:B}), the walk $P_i$ can cross neither the path
    $(\widehat h_i^+, h_i^+, h_i^-, \widehat h_i^-)$ nor the path $(\widehat h_{i+1}^+, h_{i+1}^+, h_{i+1}^-, \widehat h_{i+1}^-)$.
    Next, due to \cref{lem:subset:simple:reverse}(\ref{claim:subset:simple:reverse:AvsW}),
    the walk $P_i$ cannot contain any subwalk that traverses $L^+_i$ nor $L^-_i$ from left to right.
    Hence $P_i$ must go through the following vertices: 
    \[L^+_i[v_{r+1,1}], L^+_i[v_{1,1}], L^-_i[v_{r+1,1}], L^-_i[v_{1,1}], H_i[w_0], H_i[w_1], L^-_i[v_{1, 3}], L^-_i[v_{r+1,3}], L^+_i[v_{1,3}], L^+_i[v_{r+1,3}].\]
    Consequently, $P_i$ contains an $(H_i[w_0], H_i[w_1])$-walk contained in $H_i$.
    Because $H_i$ is an $(r,\func^0,Z^\scal_i)$-$\mathsf{Vector\, Containment\, Gadget}$ and $\func^0(\bb) = 0$, this implies $\bb \in Z^\scal_i$ (\cref{def:homo:elemenet-gadget}, (2b) $\Rightarrow$ (2a)).

  The argument above works for every $i \in [k]$, and so the definition of $Z^\scal_i$
  implies that $i \in \scal(\bb)$ whenever $i \in F$.
\end{proof}

Lemmas \ref{lem:subset:simple:forward} and \ref{lem:subset:simple:reverse-final} imply that if the $(r,\gamma^0, Z^\scal_i)$-$\mathsf{Vector\, Containment\, Gadget}$s existed, then indeed $(\mathsf{Ring}(r,k,\mathcal{H}^{vcg}_\mathcal{S}), \tcal_{r,k})$ would form an $(r,k,\mathcal{S})$-$\mathsf{Subset\, Gadget}$.

\subsubsection{Dynamic flow generators}
\label{sec:subset-full}

We will now get rid of the unrealistic assumption that an $(r,\gamma^0, Z^\scal_i)$-$\mathsf{Vector\, Containment\, Gadget}$ exists.
The construction from the previous section could be easily extended to a setting where the function $\gamma$ is constant for fixed $r$, i.e., $\gamma_r(\bb) = f(r)$ for some function $f$.
One could then simply insert additional requests of the form $(s_i, t_i, f(r))$ to generate this many additional {units of flow to be pushed through} each vector-containment gadget.
The real issue is that \cref{prop:homo:final} provides us with a gadget governed by the following function 

\[\widehat{\func}_r(b_1b_2\dots b_r) = r \cdot 2^{r} + 1 + \sum_{1 \le p < q \le r} 1_{[b_p \ne b_q]} \cdot 2^{r-q+p-1}.\]

This means that the amount of additional flow passing through the vector-containment gadget must depend on the pattern encoded by the bit vector $\bb =b_1b_2\dots b_r$.
We will take advantage of the special form of the function $\widehat{\func}_r$ to extend the previous construction with ``dynamic flow generators'': new requests that could be satisfied either locally, within their ladder, or via walks passing through $H_i$.
We are going to insert ${r \choose 2}$ new blocks between each pair of blocks in the ring structure.
Using the pattern propagation mechanism, we will guarantee that the new block inserted after the $i$-th one, labeled with a triple $(i,p,q)$, generates $ 2^{r-q+p-1}$ additional units of flow exactly when the $p$-th bit and the $q$-th bit in the pattern differ, matching the formula for $\widehat{\func}_r$.

\paragraph{The extended ring.}
Similarly as before,
for $\mathcal{S} \colon \{0,1\}^r \to 2^{[k]}$ and $i \in [k]$ we define $Z^\scal_i \sub \{0,1\}^r$ as the set of vectors $\bb$ for which $i \in \scal(\bb)$.
For $i \in [k]$ let $H_i$ be the $(r,\widehat\func_r,Z^\scal_i)$-$\mathsf{Vector\, Containment\, Gadget}$ provided by \cref{prop:homo:final}. 

We construct the graph $\mathsf{ExRing}(r,k,\scal)$ by extending the building blocks of $\mathsf{Ring}(r,k,\mathcal{H})$ from the previous construction (see \Cref{fig:ladder-full}).
Since the family $H_1, \dots, H_k$ is now fixed for given $\scal$, we directly pass $\scal$ as a parameter of the construction.
We reuse the notion of $r$-ladder from~\cref{sec:subset-simple}.

Let $\Gamma_r$ be the set of pairs $(a,b) \in [r]^2$ with $1\le a < b \le r$, plus one special element $\bot$.
We have $|\Gamma_r| = {r \choose 2} + 1$.
We define an ordering on $\Gamma_r$ so that $\bot$ is the smallest element and the pairs are ordered lexicographically.
For $i \in [k]$ and $q \in \Gamma_r$ we start constructing the plane multigraph $R_{i,q}$ from two copies of an $r$-ladder, $L_{i,q}^+, L_{i,q}^-$.
For $\odot \in \{+,-\}$ we duplicate the edges incident to $L_{i,q}^\odot[u_0], L_{i,q}^\odot[u_1]$ times $2^{3r+4}$.
For $j \in [r]$ we duplicate the edge incident to $L^\odot_{i,q}[x_j]$ times $2^{2r+j}$.
We duplicate
the edge incident to $L^-_{i,\bot}[y_j]$ (only for $q=\bot$) times $2^{r}$, for all $j \in [r]$.
When $q = (a,b)$ for some $1\le a < b \le r$, we duplicate the edges incident to  $L^-_{i,q}[y_a]$,  $L^-_{i,q}[y_b]$ times $2^{r-b+a-1}$.

\begin{figure}
    \centering
\centerline{\includegraphics[scale=0.53]{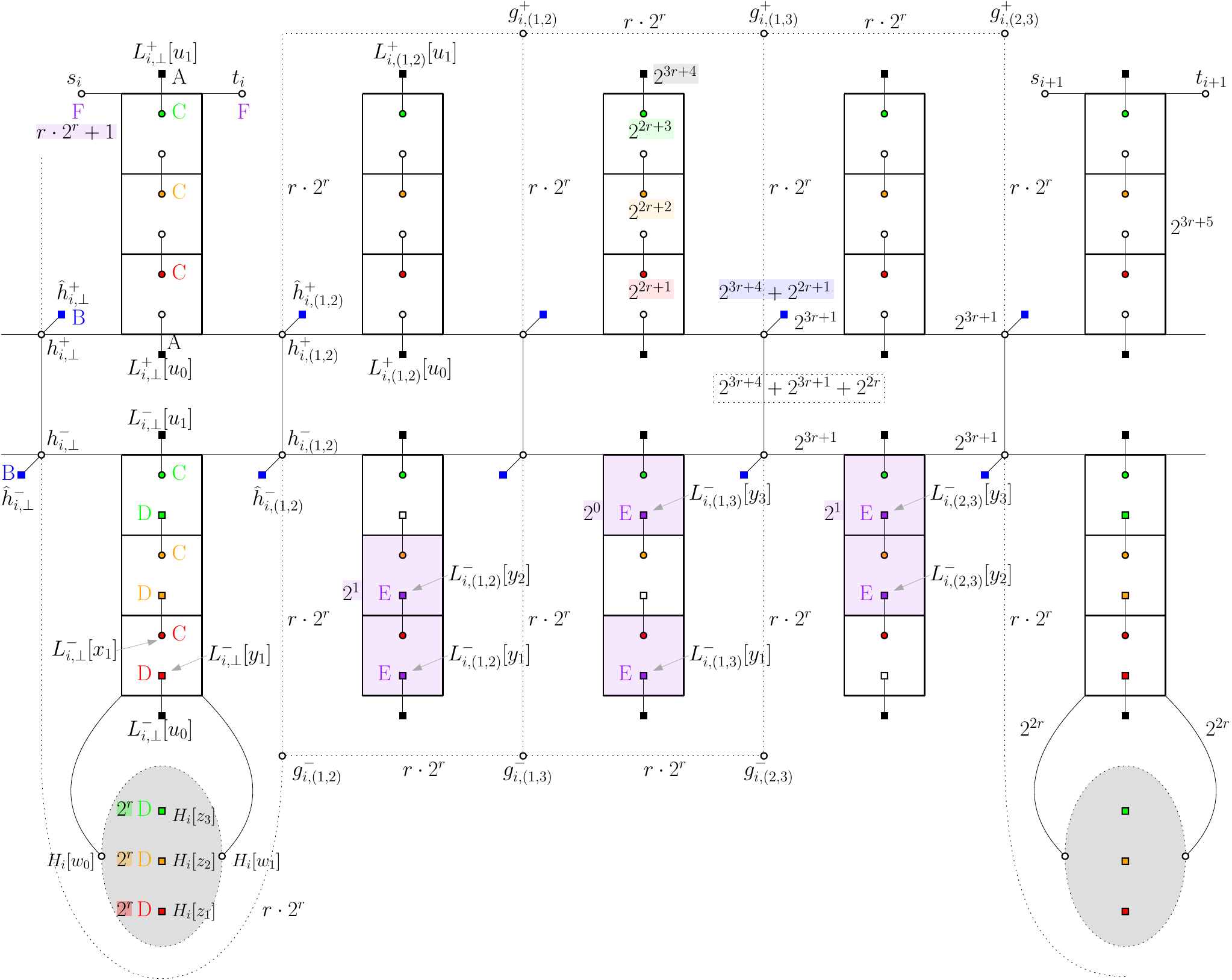}}
\caption{A fragment of the graph $\mathsf{ExRing}(r,k,\scal)$ for $r=3$, comprising subgraphs $R_{i,\bot}$, $R_{i,(1,2)}$, $R_{i,(1,3)}$, $R_{i,(2,3)}$, $R_{i+1,\bot}$.
The edges incident to vertices of the form $g^\odot_{i,q}$. which are not present in the previous construction, are dotted.
The gray ovals in the bottom represent the subgraphs $H_i$, $H_{i+1}$ (the vector-containment gadgets).
The vertices' names and edges' capacities (i.e., numbers of parallel edges) are provided.
The pairs of vertices that need to be connected in a $\widehat\tcal_{r,k}$-flow share common colors and shapes.
The colorful letters indicate the requests' types.
The number of walks requested between each terminal pair is given on a colorful background.
In each ladder of the form $L^-_{i,(a,b)}$ the faces $f_a,f_b$ are highlighted.
Note that the lower part of the figure becomes the interior of the ring structure, so the vertices $s_i,t_i$ end up on the outer face.
}
\label{fig:ladder-full}
\end{figure}

Next, for each $i \in [k]$ and $q \in \Gamma_r$ we create four vertices: $h_{i,q}^+, \widehat h_{i,q}^+, h_{i,q}^-, \widehat h_{i,q}^-$.
For $\odot \in \{+,-\}$ we insert $2^{3r+4} + 2^{2r}$ parallel edges
$h_{i,q}^\odot\widehat h_{i,q}^\odot$ 
and $2^{3r+4} + 2^{3r+1} + 2^{2r}$ parallel edges $h_i^+ h_i^-$.
Note that this is different from the previous construction where the third summand was just 2.
Next, we put $2^{3r+1}$ parallel edges between $h_{i,q}^+$ and $L_{i,q}^+[v_{1,1}]$ (the bottom-left corner vertex of the upper ladder),
and $2^{3r+1}$ parallel edges between $h_{i,q}^-$ and $L_{i,q}^-[v_{r+1,1}]$ (the upper-left corner vertex of the lower ladder).
The arrangement of the vertices on the plane is presented in \Cref{fig:ladder-full}.

 We connect $H_i[w_0]$ (resp. $H_i[w_1]$) to $L_{i,\bot}^-[v_{1,1}]$ (resp. $L_{i,\bot}^-[v_{1,3}]$) (the bottom corners of the lower ladder) via $2^{2r}$ parallel edges.
 We create vertices $s_i, t_i$ and connect each of them with $r\cdot2^r+2$ parallel edges to $L_{i,\bot}^+[v_{r+1,1}]$ or $L_{i,\bot}^+[v_{r+1,3}]$, respectively.
 These steps are omitted for $q \ne \bot$.

Analogously as before, we arrange the multigraphs $R_{i,q}$ into a ring.
We consider the lexicographic order on the set $[k] \times \Gamma_r$.
For each  $i \in [k]$ and $q \in \Gamma_r$
let $(i^\rightarrow, q^\rightarrow)$ denote the successor of $(i,q)$ in this order.
When $(i,q)$ is the last element in $[k] \times \Gamma_r$, then $(i^\rightarrow, q^\rightarrow$) becomes the first element, i.e., $(1,\bot)$.
We insert $2^{3r+1}$ parallel edges between $L_{i,q}^+[v_{1,3}]$ and  $h_{i^\rightarrow, q^\rightarrow}^+$, as well as between $L_{i,q}^-[v_{r+1,3}]$ and  $h_{i^\rightarrow, q^\rightarrow}^-$.
The constructed ring encloses a bounded region incident to the minus-sides of the multigraphs $R_{i,q}$.

The last step is novel compared to the previous construction.
For each $i \in [k]$ and $q \in \Gamma_r$, $q\ne \bot$, we create vertices $g^+_{i,q}, g^-_{i,q}$, connected via $r \cdot 2^{r}$ parallel edges to $h^+_{i^\rightarrow, q^\rightarrow}$ or $h^-_{i,q}$, respectively.
The new edges incident to $h^+_{i^\rightarrow, q^\rightarrow}$ (resp. $h^-_{i,q}$) are located between the edges to 
$L^+_{i,q}[v_{1,3}]$ and $\widehat h^+_{i^\rightarrow, q^\rightarrow}$
(resp. between $\widehat h^-_{i,q}$ and $L^-_{i,q}[v_{1,1}]$).

For $\odot \in \{+,-\}$ we connect $g^\odot_{i,q}$ via $r \cdot 2^{r}$ parallel edges to its predecessor in the ordering given by $\Gamma_r$, unless $q=(1,2)$ (i.e., $q$ is first in the ordering).
The vertex $g^+_{i,(1,2)}$ gets connected to $h^+_{i,(1,2)}$ via $r \cdot 2^{r}$ parallel edges
while $g^-_{i,(1,2)}$ gets connected to $h^-_{i,\bot}$ via $r \cdot 2^{r}$ parallel edges, in an analogous way as before.

\paragraph{The requests.}
We define a family $\widehat \tcal_{r,k}$ of requests over  $\mathsf{ExRing}(r,k,\mathcal{S})$.
The first four groups are analogous to those from \cref{sec:subset-simple}.
\begin{enumerate}[label=\Alph*.]
    \item $(L[u_0],\, L[u_1],\, 2^{3r+4})$ for each ladder $L$ of the form $L^+_{i,q}, L^-_{i,q}$. 
    \item $(\widehat h_{i,q}^+,\, \widehat h_{i,q}^-,\, 2^{3r+4} + 2^{2r})$ for each $i \in [k]$, $q \in \Gamma_r$.
    \item  $(L^+_{i,q}[x_j], \,L^-_{i,q}[x_j], \,2^{2r + j})$ for each $i \in [k]$, $q \in \Gamma_r$, $j \in [r]$. 
    \item $(L^-_{i,\bot}[y_j],\, H_i[z_j],\, 2^r)$ for each $i \in [k]$, $j \in [r]$. 
    \item $(L^-_{i,(a,b)}[y_a],\, L^-_{i,(a,b)}[y_b],\, 2^{r-b+a-1})$ for each $i \in [k]$, $1 \le a < b \le r$.
    \item $(s_i,\, t_i,\, r\cdot 2^r + 1)$  for each  $i \in [k]$. 
\end{enumerate}

For a  $\widehat\tcal_{r,k}$-flow $\pp$ we use variables $\pp^{A+}_{i,q}$, $\pp^{A-}_{i,q}$, $\pp^B_{i,q}$, $\pp^C_{i,q,j}$, $\pp^D_{i,j}$, $\pp^E_{i,a,b}$, $\pp^F_{i}$ to refer to subfamilies of $\pp$ satisfying the respective types of requests.

We make note of an observation analogous to the one (\ref{obs:numEdgesIncident}) from the previous construction.

\begin{observation}
    For every vertex $v$ of the form $\widehat h_{i,q}^\odot$,  $L^\odot_{i,q}[u_j]$, $L^\odot_{i,q}[x_j]$, $L^-_{i,\bot}[y_j]$, $L^-_{i,(a,b)}[y_a]$, $L^-_{i,(a,b)}[y_b]$, 
    the number of edges incident to $v$ equals the number of walks in a  $\widehat \tcal_{r,k}$-flow that have an endpoint at $v$.
\end{observation}

In order to keep the calculations as clean as possible, we will neglect small values of $r$ (for which we will be able to solve the instance we reduce from in polynomial time) and work in the setting where the following convenient inequalities hold.
They compare the maximal amount of flow that needs to go through particular edges with these edges' capacities. 

\begin{lemma}\label{lem:subset:full:capacities}
     For $r \ge 6$, $k \ge 1$, and a $\widehat\tcal_{r,k}$-flow $\pp$, the following bounds hold for each fixed $i \in [k]$, $q \in \Gamma_r$, and $\odot \in \{+,-\}$,
     with summation over all $q' \in \Gamma_r$ and $j \in [r]$.
\begin{enumerate}
    \item $\sum  |\pp^{E}_{i,q'}| \le r\cdot 2^{r}$
    \item $|\pp^{B}_{i,q}| + \sum  |\pp^{C}_{i,q,j}|
    + 4\cdot (\sum  |\pp^{E}_{i,q'}| + |\pp^{F}_{i}| + 1) \le 2^{3r+4} + 2^{3r+1} + 2^{2r}$
    \item $\sum |\pp^{D}_{i,j}| +
     \sum  |\pp^{E}_{i,q'}| + |\pp^{F}_{i}| + 1 \le 2^{2r}$
    \item $|\pp^{A\odot}_{i,q}| + \sum |\pp^{C}_{i,q,j}| + \sum |\pp^{D}_{i,j}| +
     4\cdot (\sum  |\pp^{E}_{i,q'}| + |\pp^{F}_{i}| + 1) \le 2^{3r+5}$
     \item $\sum |\pp^{C}_{i,q,j}| + 
     2\cdot (\sum  |\pp^{E}_{i,q'}| + |\pp^{F}_{i}| + 1) \le 2^{3r+1}$
\end{enumerate}
\end{lemma}
\begin{proof}
    (Part 1.) We estimate 
\[\sum_{1 \le a < b \le r}  |\pp^{E}_{i,(a,b)}| =  \sum_{1 \le a < b \le r} 2^{r-b+a-1} = \sum_{b=1}^r \br{2^{r-b} \cdot \sum_{a=1}^{b-1} 2^{a-1}} \le 
\sum_{b=1}^r 2^{r-b} \cdot 2^b = r\cdot 2^{r}.
\]

    (Part 2.) Using the bound above, we obtain
\[
\sum_{{1 \le a < b \le r}}  |\pp^{E}_{i,(a,b)}| + |\pp^{F}_{i}| + 1 \le r\cdot 2^{r+1} + 2.
\]
Starting from $r = 6$ the right-hand side becomes bounded by $2^{2r-2}$.
By multiplying this by 4 we obtain $2^{2r}$.
It remains to inspect the remaining summands.
\[
|\pp^{B}_{i,q}| + \sum_{j=1}^r |\pp^{C}_{i,q,j}|  = (2^{3r+4} + 2^{2r+1}) + \sum_{j=1}^r 2^{2r+j}  = (2^{3r+4} + 2^{2r+1}) + (2^{3r+1} - 2^{2r+1}) = 2^{3r+4} + 2^{3r+1}.
\]

    (Part 3.) We have already established that
$\sum_{{1 \le a < b \le r}}  |\pp^{E}_{i,(a,b)}| + |\pp^{F}_{i}| + 1 \le r\cdot 2^{r+1} + 2 \le 2^{2r-2}$.
From the second inequality we can derive
$\sum_{j=1}^r  |\pp^{D}_{i,j}| =  r\cdot 2^{r} \le 2^{2r-2}$.
Summing these two terms leads to the claimed bound.

    (Part 4.)
    We have $|\pp^{A\odot}_{i,q}| = 2^{3r+4}$
    and $\sum_{j=1}^r |\pp^{C}_{i,q,j}| \le  2^{3r+1}$.
    From the previous calculations we get $\sum_{j=1}^r  |\pp^{D}_{i,j}| \le 2^{2r-2}$
    and the last term is bounded by $2^{2r}$.
    In total, these numbers clearly cannot exceed $2^{3r+5}$.
    
    (Part 5.)
    We have $\sum_{j=1}^r |\pp^{C}_{i,q,j}| = 2^{3r+1} - 2^{2r+1}$ while the second term can be at most $2^{2r}$. The bound follows.
\end{proof}

\mic{Recall that $(i^\rightarrow, q^\rightarrow)$ stands for the successor of $(i,q)$ in the cyclic ordering of $[k]\times \Gamma_r$; \meir{Already defined \mic{rephrased}}we will also denote by $(i^\leftarrow, q^\leftarrow)$ the predecessor of $(i,q)$. }
We are going to show that $(\mathsf{ExRing}(r,k,\scal), \widehat\tcal_{r,k})$ forms a truly working $(r,k,\scal)$-$\mathsf{Subset\, Gadget}$.
We move on to the first implication in the proof of correctness.
\mic{The sketch of the construction is provided in Figures \ref{fig:ring-solution-full}, \ref{fig:ring-flow-gadget}, and \ref{fig:ladder-full-outline} (on page \pageref{fig:ladder-full-outline}) but
we need to check (in a rather tedious way) that the edge capacities suffice to accommodate the~flow. 
}

\begin{figure}
    \centering
\includegraphics[scale=0.45]{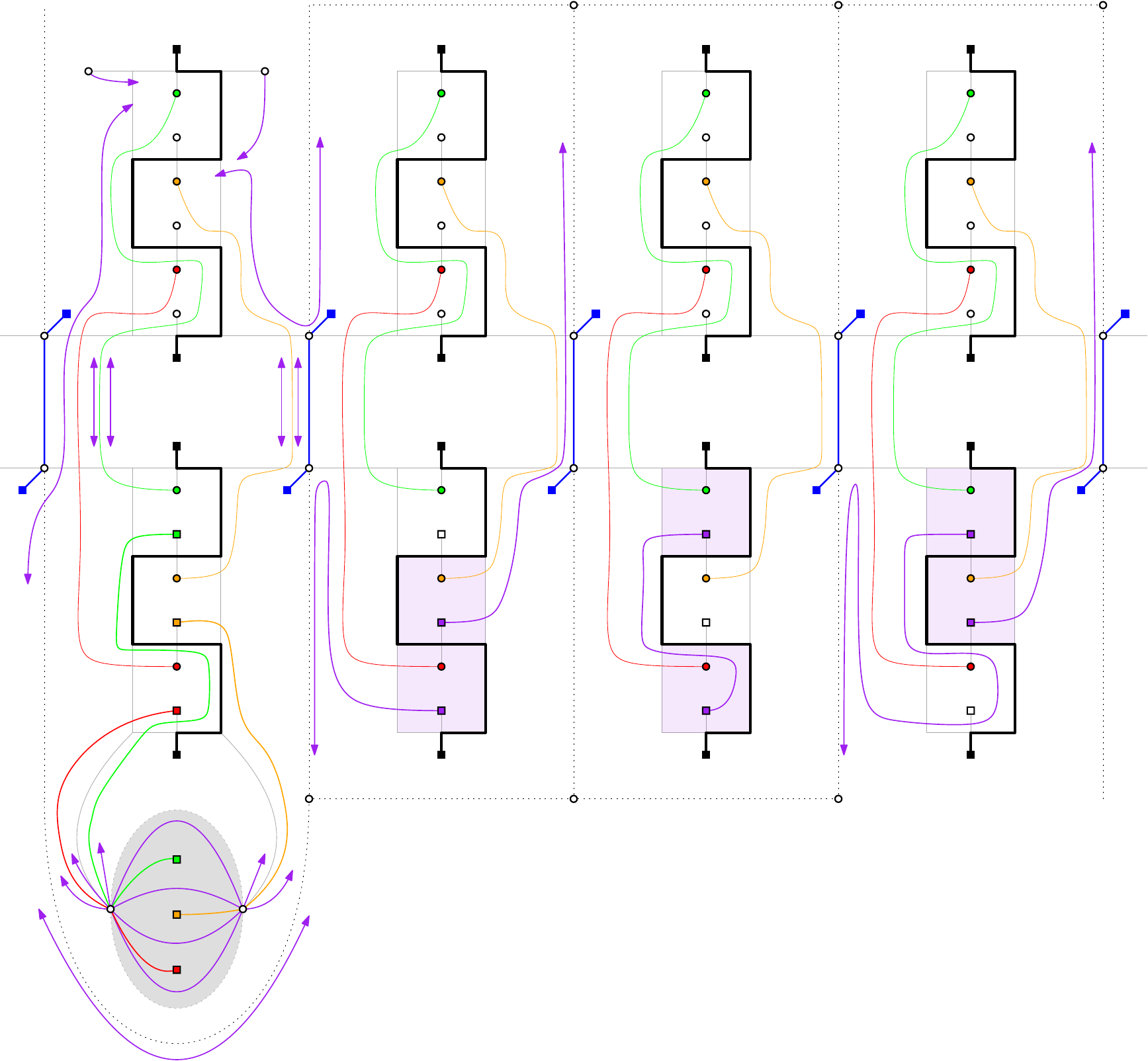}
\caption{
An illustration for the proof of \cref{lem:subset:full:forward}.
For a less detailed version with only walks of types (E, F) see \Cref{fig:ladder-full-outline}.
Due to the abundance of different walks in the flow, the walks are only roughly sketched with the colorful curves.
Their colors represent the types of requests: black (A), blue (B), green (C, D), orange (C, D), red (C, D), and purple (E,~F).
For $a = 1, b= 2$ the request of type (E) cannot be satisfied within the ladder $L^-_{1,2}$ because the vertices 
$L^-_{1,2}[y_1]$, $L^-_{1,2}[y_2]$ lie on different sides of the black curve (the flow for a request of type (A)).
The same applies to  $a = 2, b= 3$.
The respective flow must use the lower dotted edges to reach the subgraph $R_{i,\bot}$, traverse the subgraph $H_i$, and proceed through the upper dotted edges. 
The general strategy of bundling the walks on parallel edges stays the same as in \Cref{fig:ring-solution}, while the detailed view on how the purple walks are routed is provided in \Cref{fig:ring-flow-gadget}.
}
\label{fig:ring-solution-full}
\end{figure}

\begin{lemma}\label{lem:subset:full:forward}
Consider $r \ge 6$, $k\ge 1$,  and $\mathcal{S} \colon \{0,1\}^r \to 2^{[k]}$. 
For $F \sub [k]$ let $\tcal_F = \{(s_j, t_j, 1) \mid j \in F\}$.
If $F \sub \scal(\bb)$ for some $\bb \in \{0,1\}^r$, then there exists a non-crossing $( \widehat\tcal_{r,k} \cup \tcal_F)$-flow in $\mathsf{ExRing}(r,k,\mathcal{S})$.
\end{lemma}
\begin{proof}
    We deal with the requests of types A,B,C,D similarly as in \cref{lem:subset:simple:forward}.
    Again, we begin with describing which vertices are being visited by each walk and later we check that the edge capacities are large enough to receive all the walks.

    For each $i\in[k]$, $q\in\Gamma_r$, and $\odot \in \{+,-\}$, we consider a path $P$ that traverses the ladder $L^\odot_{i,q}$ in such a way that the face $f_j$ is to the right of $P$ exactly when $\bb_j = 1$.
    Each walk from $\pp^{A\odot}_{i,q}$ traverses $L^\odot_{i,q}$ through the same vertices as $P$.
    Every walk in the family $\pp^B_{i,q}$ is of the form 
    $(\widehat h_{i,q}^+, h_{i,q}^+, h_{i,q}^-, \widehat h_{i,q}^-)$.

    Now consider $i\in[k]$, $q\in\Gamma_r$, and $j \in [r]$.
    When $\bb_j = 0$ then all the walks from the family $\pp^C_{i,q,j}$ go through an edge $h_{i,q}^+h_{i,q}^-$ and otherwise they go through $h_{i^\rightarrow, q^\rightarrow}^+h_{i^\rightarrow, q^\rightarrow}^-$. 
    When  $\bb_j = 0$ then all the walks from the family $\pp^D_{i,j}$ go through an edge $L^-_{i,\bot}[v_{1,1}]H_{i}[w_0]$ and otherwise through  $L^-_{i,\bot}[v_{1,3}]H_{i}[w_1]$.

    Let $d_i = \gamma_r(\bb) + 1_{[i \in \scal(\bb)]}$.
    Note that the condition $i \in \scal(\bb)$ is equivalent to $\bb \in Z^\scal_i$.
    We define family $\tcal_i$ following \cref{def:homo:elemenet-gadget}:
    for each $j \in [r]$ with $\bb_j = 0$ we add request $(H_i[w_0], H_i[z_j], 2^r)$, for each $j \in [r]$ with $\bb_j = 1$ we add request $(H_i[w_1], H_i[z_j], 2^r)$, and finally we add request $(H_i[w_0], H_i[w_1], d_i)$.
    From the definition of an $(r,\gamma_r,Z^\scal_i)$-$\mathsf{Vector\, Containment\, Gadget}$, we obtain that there exists a non-crossing $\tcal_i$-flow $\pp^H_i$
    in $H_i$.
    Moreover, in this flow the vertices $H_i[w_0]$, $H_i[w_1]$ see the vertices $H_i[z_j]$ in the order consistent with the ordering of families $\pp^D_{i,j}$ on the edges  leaving $H_i$ (recall \cref{def:homo:seeing-order}).
    However, we have no control how the $(H_i[w_0], H_i[w_1])$-walks in $\pp^H_i$ are intertwined with the other walks at $H_i[w_0]$ and at $H_i[w_1]$.
    Therefore, we will adjust the routing of the remaining walks to fit between the families  $\pp^D_{i,j}$ in the same fashion.

    We move on to the requests of type (E).
    Consider $i \in [k]$ and $1 \le a < b \le r$.
    If $\bb_a = \bb_b$ then the faces $f_a, f_b$ of the ladder $L^-_{i,(a,b)}$ are on the same side of the walks from family $\pp^{A-}_{i,(a,b)}$.
    The walks in family $\pp^{E}_{i,(a,b)}$ must connect $L^-_{i,(a,b)}[y_a]$ to $L^-_{i,(a,b)}[y_b]$.
    The only issue is to avoid crossing the walks from families $\pp^{C}_{i,(a,b),j}$, $j\in[r]$, which also need to reach vertices within $L^-_{i,(a,b)}$.
    Observe that both vertices $L^-_{i,(a,b)}[y_a]$, $L^-_{i,(a,b)}[y_b]$ can be reached from 
    $L^-_{i,(a,b)}[v_{1,1}]$ (when $\bb_a = \bb_b = 0$) or from $L^-_{i,(a,b)}[v_{1,3}]$ (when $\bb_a = \bb_b = 1$) without crossing the other walks.
    Hence every walk in $\pp^{E}_{i,(a,b)}$ can be obtained via a concatenation of such two walks.

    Suppose now that $\bb_a \ne \bb_b$ and assume w.l.o.g. that $\bb_a = 0$ and $\bb_b = 1$.
    First, one can reach $L^-_{i,(a,b)}[v_{r+1,1}]$ from $L^-_{i,(a,b)}[y_a]$ without crossing other walks, and then one can move to   $h_{i,(a,b)}^-$.
    Due to ordering of edges incident to $ h_{i,(a,b)}^-$, the walks can proceed ``down'' to $g_{i,(a,b)}^-$ and then follow the ``lower'' path from $g_{i,(a,b)}^-$ to $h^+_{i,\bot}$ (see \Cref{fig:ring-solution-full}).
    Looking from the other end, starting at $L^-_{i,(a,b)}[y_b]$, one first reaches $L^-_{i,(a,b)}[v_{r+1,3}]$, then $h_{i^\rightarrow, q^\rightarrow}^-$, $h_{i^\rightarrow, q^\rightarrow}^+$, 
     $g_{i,(a,b)}^+$, and follows the ``upper'' path towards $h^+_{i,(1,2)}$ (the vertex just to the right of the subgraph $R_{i,\bot}$).

    So far we have explained how to reach the left side  the subgraph $R_{i,\bot}$ from each $L^-_{i,(a,b)}[y_a]$ with $\bb_a = 0$ and how to reach the right side $R_{i,\bot}$ from each $L^-_{i,(a,b)}[y_b]$ with $\bb_b=1$.
    The total amount of flow from families $\pp^E_{i,(a,b)}$ that need to traverse $R_{i,\bot}$ equals
    \[
    \sum_{1\le a < b \le r} 1_{[\bb_a \ne \bb_b]} \cdot 2^{r-b+a-1}.
    \]
    
    We also need to take care of the request $(s_i, t_i, r\cdot 2^r + 1)$ and, when $i \in \scal(\bb)$, the request $(s_i, t_i, 1)$ from $\tcal_F$.
    In total, we need to push exactly $d_i$ units of flow through $R_{i,\bot}$, from left to right.
    The flow $\pp^H_i$ in $H_i$ already includes this many $(H_i[w_0], H_i[w_1])$-walks.
    It remains to group the walks coming from the left side (that is, from $s_i$ and  $h^+_{i,\bot}$) into bundles, reflecting the bundles of $(H_i[w_0], H_i[w_1])$-walks in $\pp^H_i$ between the $(H_i[w_0], H_i[z_j])$-walks at vertex $H_i[w_0]$, and accommodate these bundles respectively between the families $\pp^C_{i,\bot,j}$.
    Similarly, we group the walks coming from the right side (that is, from $t_i$ and  $h^+_{i,(1,2)}$) accordingly to the bundles of  $(H_i[w_0], H_i[w_1])$-walks in $\pp^H_i$ between the $(H_i[w_1], H_i[z_j])$-walks at vertex $H_i[w_1]$.
    The order of walks in these two families is symmetric, so one can match the requested endpoints (see \Cref{fig:ring-flow-gadget}).

    Finally, we check that the edge capacities are sufficient to accommodate all the walks.
    The edges incident to vertices of the form $g^\odot_{i,q}$ have capacity $r \cdot 2^{r}$. 
In an extreme case, such an edge might be utilized by all families $\pp^E_{i,(a,b)}$ for $1 \le a < b \le r$ and fixed $i$.
The total number of walks in these families is bounded by  $r \cdot 2^{r}$ due to \cref{lem:subset:full:capacities}(1).
    
    Now consider the edges $h^+_{i,q}h^-_{i,q}$:
    there are $2^{3r+4} + 2^{3r+1} + 2^{2r}$ copies of each.
    In our construction each walk from family $\pp^B_{i,q}$ uses one of these edges.
    For each $j \in [r]$, exactly one of the families  $\pp^C_{i^\leftarrow, q^\leftarrow,j}$, $\pp^C_{i,q,j}$ goes through $h^+_ih^-_i$.
    This gives a number of walks equal to $\sum_{j=1}^r |\pp^C_{i,q,j}|$.
    We might also need to accommodate the walks of types (E), (F), and the $(s_i,t_i)$-walk requested by $\tcal_F$.
    Each of these walks might go at most twice through $h^+_{i,q}h^-_{i,q}$.
    For $q = \bot$ there might be also walks indexed with $(i^\leftarrow, q^\leftarrow)$ that traverse $h^+_{i,q}h^-_{i,q}$ to the left of family $\pp^B_{i,q}$.
    We upper bound the number of all these walks by multiplying $\sum_{1\le a < b \le r} |\pp^E_{i,(a,b)}| + |\pp^F_i| +1$ by 4.
    In total we get a quantity bounded $2^{3r+4} + 2^{3r+1} + 2^{2r}$, due to \cref{lem:subset:full:capacities}(2).

The edges of the form $H_i[w_0]L^-_{i,\bot}[v_{1,1}]$
or $H_i[w_1]L^-_{i,\bot}[v_{1,3}]$ have capacity $2^{2r}$.
Such an edge might be used by all families $\pp^D_{i,j}$, $\pp^E_{i,(a,b)}$, $\pp^F_{i}$, for fixed $i$, and the $(s_i,t_i)$-walk requested in $\tcal_F$.
By \cref{lem:subset:full:capacities}(3), the total number of these walks is at most $2^{2r}$.

For each $i \in [k]$, $q \in \Gamma_r$, and $\odot \in \{+,-\}$, the edges within the ladder $L^\odot_{i,q}$, which are not adjacent to any terminal, might need to accommodate the all respective families of types (A), (C), (D), (E), (F), and the $(s_i,t_i)$-walk requested by $\tcal_F$.
The walks of the last three types might traverse up to four copies of a single edge: by going ``up'' and ``down'' on each of two sides of the ladder.
We use \cref{lem:subset:full:capacities}(4) to bound the total number of necessary parallel edges by $2^{3r+5}$, that is, the number of copies for each edge.

The last non-trivial case is the passage between a vertex of the form $h^\odot_{i,q}$ and a ladder.
The edges therein are utilized by the walks of type (C) for a fixed pair $(i,q)$ and possibly the walks of types (E), (F) plus the $(s_i,t_i)$-walk requested by $\tcal_F$. 
We multiply by 2 the total amount of flow from the last three types to cover potential detours of walks.
By \cref{lem:subset:full:capacities}(5), we need no more than $2^{3r+1}$ parallel edges, which is exactly the capacity.

This concludes the construction of a non-crossing $(\widehat\tcal_{r,k} \cup \tcal_F)$-flow.
\end{proof}

\begin{figure}
    \centering
\centerline{\includegraphics[scale=0.65]{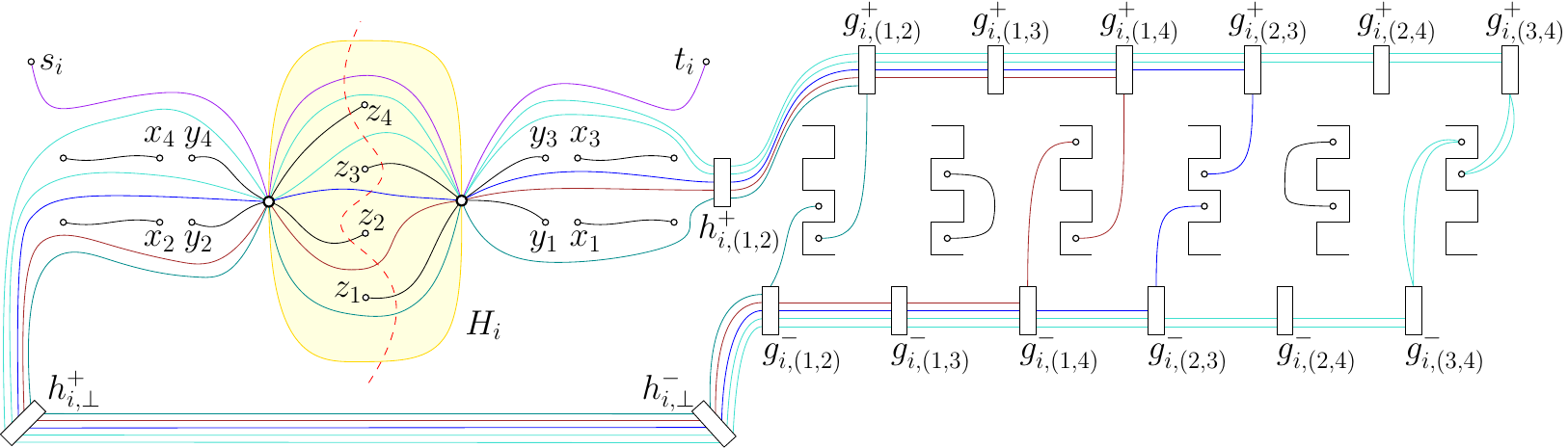}}
\caption{
A topological view on the construction from \cref{lem:subset:full:forward} for $r=4$ and $\bb = (1010)$, illustrating walks traversing the vector-containment gadget $H_i$.
The length of the red dashed curve, measuring how many  $(H_i[w_0], H_i[w_1])$-walks can pass through the gadget, is governed by the choice of the vector $\bb$.
The walks of types (C) and (D) are drawn in black, with a single curve representing each family $\pp^C_{i,\bot,j}$ or $\pp^D_{i,j}$.
The labels $x_j, y_j$ mark vertices $L^-_{i,\bot}[x_j]$, $L^-_{i,\bot}[y_j]$.
For $(a,b) \in \{(1,3), (2,4)\}$ the vertices $L^-_{i,(a,b)}[y_a]$, $L^-_{i,(a,b)}[y_b]$ lie on the same side of the pattern drawn by walks of type (A), so they can be connected within the ladder $L^-_{i,(a,b)}$.
For the remaining pairs $(a,b)$, the flow $\pp^E_{i,(a,b)}$ must go through the subgraph $H_i$.
Each such family, as well as the family of $(s_i,t_i)$-walks, is represented by a single curve, except for $\pp^E_{i,(3,4)}$.
In that last case, two walks are drawn to demonstrate that even walks satisfying a single request may exhibit different behavior when traversing $H_i$ (they pass vertex $H_i[z_4]$ from different sides).
The coloring of the curves illustrates that all the requests can be satisfied in a non-crossing way. 
}
\label{fig:ring-flow-gadget}
\end{figure}

\mic{We move on to the second implication in the correctness proof.}
It will be convenient to formally define a {\em pattern} of a walk $Q \in \pp^{A\odot}_{i,q}$.
We do not use the term ``homotopy'' in order to avoid confusion with \cref{def:homo:homotopic}.

\begin{definition}\label{def:subset:full:homotopic}
Consider $i \in [k]$, $q \in \Gamma_r$, and $\odot \in \{+,-\}$.
Let $Q$ be an $(L^\odot_{i,q}[u_0],L^\odot_{i,q}[u_1])$-walk in $\mathsf{ExRing}(r,k,\scal)$ and $\bb \in  \{0,1\}^r$.
We say that $Q$ has {\em pattern $(i,q,\odot,\bb)$} if $Q$ is contained in the subgraph $L^\odot_{i,q}$ and for each $j \in [r]$ there exists a walk $W_j$ such that:
\begin{enumerate}
    \item $W_j$ starts at $L^\odot_{i,q}[x_j]$ and ends at $h^\odot_{i,q}$ (when $\bb_j = 0$) or $h^\odot_{i^\rightarrow, q^\rightarrow}$ (when $\bb_j = 1$);
    \item $W_j$ is internally contained in $L^\odot_{i,q}$;
    \item $W_j$ and $Q$ are non-crossing.
\end{enumerate}
\end{definition}

\mic{Intuitively, the face $f_j$ is to the left of $Q$ when $\bb_j = 0$ and to the right when $\bb_j = 1$.}
The next lemma plays the same role as \cref{lem:subset:simple:reverse} in \cref{sec:subset-simple}.
We prove that any non-crossing $\widehat\tcal_{r,k}$-flow must enjoy essentially the same structure as the solution constructed in \cref{lem:subset:full:forward}.
\mic{The difference with \cref{lem:subset:simple:reverse} is that now some walks of types (A), (B), (C) may potentially use the new connections through the $g$-vertices.
We argue that their numbers are large enough, compared to those edges' capacities, to ensure that the majority of walks exhibits the same behaviour as before.}

\begin{lemma}\label{lem:subset:full:reverse}
Consider $r \ge 6$, $k\ge 1$,  and $\mathcal{S} \colon \{0,1\}^r \to 2^{[k]}$.
Let $\pp$ be a non-crossing $\widehat\tcal_{r,k}$-flow in $\mathsf{ExRing}(r,k,\mathcal{S})$. 
Then the following hold.
\begin{enumerate}
    \item Let  $i \in [k]$, $q\in \Gamma_r$, $\odot \in \{+,-\}$, and $W$ be an $(h_{i,q}^\odot, h_{i^\rightarrow, q^\rightarrow}^\odot)$-walk internally contained in $L^\odot_{i,q}$.
    Then $W$ crosses with $\pp^{A\odot}_{i,q}$. \label{claim:subset:full:reverse:AvsW}  
    \item For each $i \in [k]$, $q \in \Gamma_r$, the family $\pp^B_{i,q}$ contains a walk on vertices $\{\widehat h_{i,q}^+, h_{i,q}^+, h_{i,q}^-, \widehat h_{i,q}^-\}$. Moreover, there are at least $|\pp^B_{i,q}| - r\cdot 2^r$ walks in $\pp^B_{i,q}$ that go through an edge $h_{i,q}^+h_{i,q}^-$. \label{claim:subset:full:reverse:B}
    \item There exists a vector $\bb \in \{0,1\}^r$ such that
    for each $i\in [k]$,  $q\in \Gamma_r$, and $\odot \in \{+,-\}$, there is a walk $Q^\odot_{i,q} \in \pp^{A\odot}_{i,q}$ with pattern $(i,q,\odot,\bb)$. \label{claim:subset:full:reverse:D-subwalk}
\end{enumerate}
\end{lemma}
\begin{proof}
    The proof of Part (\ref{claim:subset:full:reverse:AvsW}) is analogous to the one of \cref{lem:subset:simple:reverse}(\ref{claim:subset:simple:reverse:AvsW}).

    \begin{innerproof}[Proof of (\ref{claim:subset:full:reverse:B}).]
    We count the total number of edges leaving the subgraph induced by $\{\widehat h_{i,q}^+, h_{i,q}^+, h_{i,q}^-, \widehat h_{i,q}^-\}$ to be $4 \cdot 2^{3r+1} + 2r \cdot 2^{r}$. For $r \ge 6$ this is less than $2^{3r+4} < |\pp^B_{i,q}|$. Hence there is at least one walk from $\pp^B_{i,q}$ that does not leave the vertex set $\{\widehat h_{i,q}^+, h_{i,q}^+, h_{i,q}^-, \widehat h_{i,q}^-\}$.

    Suppose now that $P \in \pp^B_{i,q}$ does no go through any edge $h_{i,q}^+h_{i,q}^-$ nor $h_{i,q}^+g_{i,q}^+$. Then $P$ needs to traverse $L^+_{i,q}$ or $L^+_{i^\leftarrow, q^\leftarrow}$.
    This means that $P$ contains a subwalk that meets the specification of Part (\ref{claim:subset:full:reverse:AvsW}) of the lemma.
    This contradicts the assumption that $\pp$ is a non-crossing flow.
    Since the number of edges $h_{i,q}^+g_{i,q}^+$ is $r\cdot 2^r$, we obtain that at least $|\pp^B_{i,q}| - r\cdot 2^r$ walks in $\pp^B_{i,q}$ go through $h_{i,q}^+h_{i,q}^-$.
    \end{innerproof}

    Similarly to the proof of  \cref{lem:subset:simple:reverse}, we first establish two intermediate claims.
    We say that a walk $W$ leaves a subgraph $H$ through  vertex $v$ if exactly one of the endpoints of $W$ belongs to $V(H)$ and $v$ is the first vertex on $W$ (counting from this endpoint) that does not belong to $V(H)$.    
    Unlike \cref{claim:subset:simple:reverse:C-decision} in the previous section, we first only specify the vertex through which the walks from $\pp^{C}_{i,q,j}$ leave 
     $L^+_{i,q}$, and then inspect the next edge on these walks in \cref{claim:subset:full:reverse:C-equal}.

    \begin{claim}\label{claim:subset:full:reverse:C-decision}
        For each $i \in [k]$ and $q \in \Gamma_r$ there exists a
        vector $\bb^{i,q} \in \{0,1\}^r$
        such that when $\bb^{i,q}_j = 0$ then all the walks from  $\pp^{C}_{i,q,j}$ leave $L^+_{i,q}$ through vertex $h_{i,q}^+$,
        and when $\bb^{i,q}_j = 1$ then all the walks from  $\pp^{C}_{i,q,j}$ leave $L^+_{i,q}$ through vertex $h_{i^\rightarrow, q^\rightarrow}^+$.
    \end{claim}
    \begin{innerproof}
        Suppose that there are walks $W_0, W_1 \in \pp^{C}_{i,q,j}$ so that $W_0$ leaves   $L^+_{i,q}$ through  $h_{i,q}^+$ and $W_1$ leaves $L^+_{i,q}$ through $h_{i^\rightarrow, q^\rightarrow}^+$.
        By \cref{def:reduction:flow} of a non-crossing flow, the concatenation $W_0 + W_1$ does not cross any walk from $\pp^{A+}_i$.
        But then $W_0 + W_1$ contains a subwalk that meets the specification of Part (\ref{claim:subset:full:reverse:AvsW}).
        This  contradicts the assumption that $\pp$ is a non-crossing flow.
        Therefore for each $j \in [r]$ the choice whether to leave $L^+_{i,q}$ through  $h_{i,q}^+$ or $h_{i^\rightarrow, q^\rightarrow}^+$ is fixed.
    \end{innerproof}

    \begin{claim}\label{claim:subset:full:reverse:C-equal}
        There exists a single vector $\bb \in \{0,1\}^r$ so that $\bb^{i,q} = \bb$ for all $i \in [k]$ and $q \in \Gamma_r$.
    \end{claim}
    \begin{innerproof}
        We define $\tau(b_1b_2\dots b_r) = \sum_{h=1}^r b_h\cdot 2^{h-1}$.
        Suppose that the claim does not hold.
        Because we work on a ring structure, there exist $i \in [k]$, $q \in \Gamma_r$, for which $\tau(\bb^{i,q}) < \tau(\bb^{i^\rightarrow, q^\rightarrow})$.
        By \cref{claim:subset:full:reverse:C-decision} the number of walks from $ \pp^C_{i,q,1} \cup  \pp^C_{i,q,2} \cup \dots  \pp^C_{i,q,r}$ that leave $L^+_{i,q}$ through $h_{i^\rightarrow, q^\rightarrow}^+$ equals $2^{2r+1} \cdot \tau(\bb^i)$.
        On the other hand, the number of walks from $ \pp^C_{i^\rightarrow, q^\rightarrow,1} \cup  \pp^C_{i^\rightarrow, q^\rightarrow,2} \cup \dots  \pp^C_{i^\rightarrow, q^\rightarrow,r}$ that leave $L^+_{i^\rightarrow, q^\rightarrow}$ through $h_{i^\rightarrow, q^\rightarrow}^+$  equals $2^{2r+1} \cdot (2^{r} - 1 - \tau(\bb^{i^\rightarrow, q^\rightarrow}))$.     
        Since $\tau(\bb^{i,q}) < \tau(\bb^{i^\rightarrow, q^\rightarrow})$, this quantity is at least $2^{2r+1} \cdot (2^{r} - \tau(\bb^{i,q}))$.
        In total, we obtain at least $2^{3r+1}$ walks from the two ladders that meet at  $h_{i^\rightarrow, q^\rightarrow}^+$.
        
        By Part (\ref{claim:subset:full:reverse:AvsW}), when a walk of type (C) leaves the ladder $L^+_{i,q}$ (resp. leaves $L^+_{i^\rightarrow, q^\rightarrow}$) from the right side (resp. the left side), it 
        must at some point reach a neighbor of $h_{i^\rightarrow, q^\rightarrow}^+$ 
        \mic{that does not belong to $V(L^+_{i,q})$ (resp. $V(L^+_{i^\rightarrow, q^\rightarrow})$).}
        \meir{Be more explicit about ``the corner of the ladder of origin''? \mic{changed}}
        There are at most $r\cdot 2^r$ walks that can use an edge $h_{i^\rightarrow, q^\rightarrow}^+q_{i,q}^+$
        \mic{and no (C)-type walk can use an edge $h_{i^\rightarrow, q^\rightarrow}^+\widehat h_{i^\rightarrow, q^\rightarrow}^+$ because they are all used by walks from $\pp^B_{i^\rightarrow, q^\rightarrow}$.
        Also, due to Part (\ref{claim:subset:full:reverse:B}), at least $2^{3r+4} + 2^{2r+1} - r\cdot 2^r$ walks from $\pp^B_{i^\rightarrow, q^\rightarrow}$ go through $h_{i^\rightarrow, q^\rightarrow}^+h_{i^\rightarrow, q^\rightarrow}^-$.
        Since the remaining walks of type (C) need to go through $h_{i^\rightarrow, q^\rightarrow}^+h_{i^\rightarrow, q^\rightarrow}^-$ as well,}
        in total we have at least $2^{3r+4} + 2^{3r+1} + 2^{2r+1} - r\cdot 2^{r+1}$ walks that need to go through this passage.
        On the other hand, there are only $2^{3r+4} + 2^{3r+1} + 2^{2r}$ parallel edges $h_{i^\rightarrow, q^\rightarrow}^+h_{i^\rightarrow, q^\rightarrow}^-$.
        Since for $r \ge 6$ we have $r\cdot 2^{r+1} < 2^{2r}$, there are too few edges to accommodate all the walks above, and      
        so we arrive at a~contradiction.
    \end{innerproof}

    \begin{innerproof}[Proof of (\ref{claim:subset:full:reverse:D-subwalk}).]
    Let $\bb$ be the vector from \cref{claim:subset:full:reverse:C-equal}.
    Fix $i \in [k]$, $q\in \Gamma_r$, and $\odot\in\{+,-\}$.
    First we argue that there exists a walk $Q^\odot_{i,q} \in \pp^\odot_{i,q}$ entirely contained in $L^\odot_{i,q}$.
    This follows from counting the edges leaving $L^\odot_{i,q}$: there are at most $2\cdot 2^{3r+1} + 2\cdot 2^{2r}$ of them, which is less than $|\pp^\odot_{i,q}|$.
    To see that $Q^+_{i,q}$ has pattern $(i,q,+,\bb)$,
    fix $j \in [r]$.
    The existence of the walk $W_j$ in \cref{def:subset:full:homotopic} follows from
     \cref{claim:subset:full:reverse:C-decision}: it can be chosen as a subwalk of any walk from $\pp^C_{i,q,j}$.
     The argument 
     that $Q^-_{i,q}$ has pattern $(i,q,-,\bb)$
     is the~same.
    \end{innerproof}

    This concludes the proof of \cref{lem:subset:full:reverse}.
\end{proof}

We can now take advantage of the structure imposed on a non-crossing $\widehat\tcal_{r,k}$-flow to analyze which walks need to go through the subgraphs $H_i$.
The following lemma is based on the analogous observations as \cref{lem:subset:simple:reverse-final}.

\begin{lemma}\label{lem:subset:full:reverse-final}
Consider $r \ge 6$, $k\ge 1$,  and $\mathcal{S} \colon \{0,1\}^r \to 2^{[k]}$. 
For $F \sub [k]$ let $\tcal_F = \{(s_j, t_j, 1) \mid j \in F\}$.
If there exists a non-crossing $( \widehat\tcal_{r,k} \cup \tcal_F)$-flow in $\mathsf{ExRing}(r,k,\mathcal{S})$,
then $F \sub \scal(\bb)$ for some $\bb \in \{0,1\}^r$.
\end{lemma}
\begin{proof}
    Let $\pp$ be a $\widehat\tcal_{r,k}$-flow and $\pp_F$ be  a $\tcal_{F}$-flow
    so that $\pp \cup \pp_F$ is non-crossing in $\mathsf{ExRing}(r,k,\scal)$.
    We apply \cref{lem:subset:full:reverse} to $\pp$; let $\bb \in \{0,1\}^r$ be the vector given by Part (\ref{claim:subset:full:reverse:D-subwalk}) of the lemma. 
    Fix $i \in [k]$ for the rest of the proof.
    For each $q \in \Gamma_R$ and $\odot \in \{+,-\}$ we apply \cref{lem:subset:full:reverse}(\ref{claim:subset:full:reverse:AvsW})
    to obtain that no walk $W$ that traverses the ladder $L^\odot_{i,q}$ from left to right can be non-crossing with $\pp^{A\odot}_{i,q}$.
    Let $j \in [r]$.
    By the same concatenation argument as in \cref{claim:subset:full:reverse:C-decision} we arrive at the following.
    
    \begin{observation}\label{obs:subset:full:reverse-final:block}
    If there is walk $W'_j$, non-crossing with $\pp^{A\odot}_{i,q}$, that starts at $L^\odot_{i,q}[x_j]$ and leaves the ladder through vertex $h^\odot_{i,q}$, then $\bb_j = 0$.
    Symmetrically, if $W'_j$ leaves the ladder through vertex $h^\odot_{i^\rightarrow, q^\rightarrow}$, then $\bb_j = 1$.
    \end{observation}

    When $L^-_{i,q}[y_j]$ occurs as a terminal in $\widehat\tcal_{r,k}$ (that is, when (I) $q=\bot$ or (II) $q=(a,b)$ and $j \in \{a,b\}$)
    then the number of edges incident to  $L^-_{i,q}[x_j]$ or $L^-_{i,q}[y_j]$ equals the number of walks in $\pp$ ending at this vertex.
    Therefore, no walk from $\pp^{A-}_{i,q}$ can visit $L^-_{i,q}[x_j]$ nor  $L^-_{i,q}[y_j]$.
    Since these two vertices share a face,
    they are both to the left or both to the right of the walk $Q^{A-}_{i,q} \in \pp^{A-}_{i,q}$ with pattern $(i,q,-,\bb)$.
    This means that in these cases \cref{obs:subset:full:reverse-final:block} remains true if we can replace $L^-_{i,q}[x_j]$ with  $L^-_{i,q}[y_j]$.   

    By \cref{lem:subset:full:reverse}(\ref{claim:subset:full:reverse:B}),
    each family $\pp^B_{i,q}$ contains a walk on vertices $\{\widehat h_{i,q}^+, h_{i,q}^+, h_{i,q}^-, \widehat h_{i,q}^-\}$.
    Because the number of edges incident to $\widehat h_{i,q}^+$ or $\widehat h_{i,q}^-$ equals $|\pp^B_{i,q}|$, no walk from $\pp$ can cross the path $(\widehat h_{i,q}^+, h_{i,q}^+, h_{i,q}^-, \widehat h_{i,q}^-)$. Together with \cref{obs:subset:full:reverse-final:block}, this rules out all the connections between the two sides of a ladder other than the way through $H_i$.

    \begin{observation}
        For each $q \in \Gamma_r$ and $a,b \in [r]$ with $\bb_a \ne \bb_b$, any  $(L^-_{i,q}[y_a],L^-_{i,q}[y_b])$-walk $W$, that is non-crossing with the walks in $\pp$ of types (A) and (B), must contain an $(H_i[w_0], H_i[w_1])$-subwalk.
    \end{observation}

    The sum of $|P^E_{i,(a,b)}|$ over $(a,b)$ satisfying  $\bb_a \ne \bb_b$ equals  
    \[
    \sum_{1\le a < b \le r} 1_{[\bb_a \ne \bb_b]} \cdot 2^{r-b+a-1}.
    \]
    
    Next, each $(s_i, t_i)$-walk in $\pp \cup \pp_F$ must contain an $(H_i[w_0], H_i[w_1])$-subwalk as well.
    In total, the number $d_i$ of walks that traverse $H_i$ from left to right equals $\widehat\gamma_r(\bb) + 1_{[i \in F]}$.
    Furthermore, the only way for the walks of type (D) to reach the vertices $H_i[z_j]$ is to enter $H_i$ from the same side of the walk $Q^{A-}_{i,\bot}$ as $L^-_{i,\bot}[y_j]$ is located.  

    \begin{observation}\label{obs:subset:full:reverse-final:D}
        For each $j \in [r]$, every walk $P \in \pp^D_{i,j}$ contains an $(H_{i,}[w_0],H_i[z_j])$-walk in $H_i$ when $\bb_j = 0$, or an $(H_i[w_1],H_i[z_j])$-walk in $H_i$ when $\bb_j = 1$.
    \end{observation}
    Therefore, the subwalks of $\pp^D_{i,j}$ within $H_i$ satisfy request $(H_i[w_0], H_i[z_j], 2^r)$ when  $\bb_j = 0$ or request $(H_i[w_1], H_i[z_j], 2^r)$ when  $\bb_j = 1$.
    Together with $d_i$ units of the $(H_i[w_0], H_i[w_1])$-flow, these match the flow requested in \cref{def:homo:elemenet-gadget}.
    Since $H_i$ is an $(r,\widehat\gamma_r, Z^\scal_i)${-$\mathsf{Vector\, Containment\, Gadget}$} this implies that $\bb \in Z^\scal_i$ whenever $i \in F$.
    By the definition of $ Z^\scal_i$, we obtain $i \in F \Rightarrow i \in \scal(\bb)$.
    This concludes the proof of the containment $F \sub \scal(\bb)$.
\end{proof}

We can finally summarize the entire construction of the subset gadget.

\begin{proposition}\label{prop:subset:full:final}
    There is a polynomial-time algorithm that, given $r\ge 6$, $k \ge 1$, and a function $\mathcal{S} \colon \{0,1\}^r \to 2^{[k]}$, outputs an $(r,k,\scal)$-$\mathsf{Subset\, Gadget}$ $(G,\tcal)$ with $|V(G)| + |E(G)| = k \cdot 2^{\Oh(r)}$ and $|\tcal| = \Oh(k \cdot r^3)$.
    Moreover, for each request $(u_i,v_i, d_i) \in \tcal$ it holds that $d_i \le \Oh(2^{3r})$. 
\end{proposition}
\begin{proof}
    For each $i \in [k]$ we use \cref{prop:homo:final} to construct an $(r,\gamma_r, Z^\scal_i)$-$\mathsf{Vector\, Containment\, Gadget}$ of size $2^{\Oh(r)}$ in time polynomial in $2^r$, which is the input size.
    This allows us to construct the graph $\mathsf{ExRing}(r,k,\scal)$.
    It is divided into $k \cdot \br{{r \choose 2} + 1}$ blocks, each equipped with $\Oh(r)$ requests from $\widehat\tcal_{r,k}$.
    Due to Lemmas \ref{lem:subset:full:forward} and \ref{lem:subset:full:reverse-final}, the pair $(\mathsf{ExRing}(r,k,\scal), \widehat\tcal_{r,k})$ forms an $(r,k,\scal)$-$\mathsf{Subset\, Gadget}$.
\end{proof}

\subsection{From a set cover to a non-crossing flow}
\label{sec:hard:setcover}

In this section, we finish the reduction from \textsc{Set Cover} to \nullnoncross.
We need two more simple gadgets, which are adaptations of the gadgets used in the NP-hardness proof of \textsc{Planar Disjoint Paths}~\cite{kramer1984complexity}.
The first gadget encodes which of $\ell$ sets should cover an element $i \in [k]$.

\begin{definition}\label{def:setcover:existential}
    For $\ell \in \nn$, a pair $(G,\tcal)$ is an {\em $\ell$-$\mathsf{Existential\, Gadget}$} if the following conditions hold.
\begin{enumerate}
    \item $G$ is a plane graph with $2\ell$ distinguished vertices $s_1,t_1, s_2, t_2, \dots, s_\ell,t_\ell$ lying on the outer face in this clockwise order.
    \item $\tcal$ is a set of triples from $V(G) \times V(G) \times \{1\}$. 
    \item For $F \sub [\ell]$, let $\tcal_F = \{(s_i,t_i, 1) \mid i \in F\}$.
    Then, there exists a non-crossing $(\tcal \cup \tcal_F)$-flow in $G$ if and only if 
    $|F| < \ell$.
\end{enumerate}
\end{definition}

When seeking a set cover of size $\ell$,
we make a single copy of an $\ell$-$\mathsf{Existential\, Gadget}$ for each $i \in [k]$.
We will allow an index $j \in [\ell]$ to belong to $F$ when the element $i$ is not covered by the set $S_j$ in a solution $S_1,S_2,\dots, S_\ell$ to \textsc{Set Cover}.
Condition (3) ensures that one of the indices will be missing in $F$, implying that $i$ gets covered.

\begin{figure}[t]
    \centering
\includegraphics[scale=0.8]{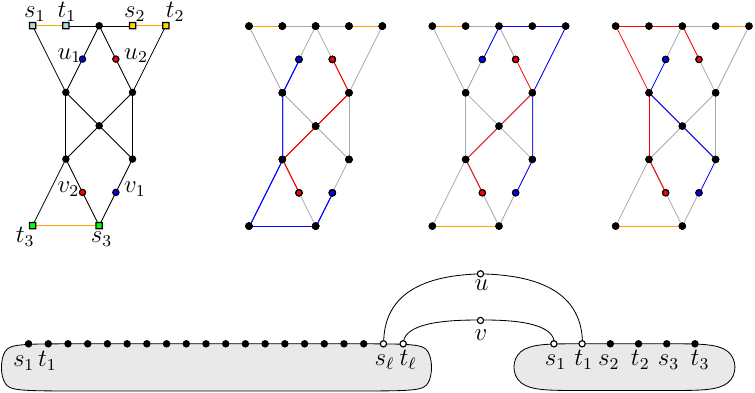} 
\caption{
An illustration for \cref{lem:setcover:existential-exists}.
Top left: A 3-$\mathsf{Existential\, Gadget}$ $(G_3,\tcal_3)$ with $\tcal_3$ given as $\{(u_1,v_1,1), (u_2,v_2,1) \}$.
Top right: Constructing a non-crossing $(\tcal_3 \cup \tcal_F)$-flow in $G_3$ is possible whenever  $F$ misses some element from $\{1,2,3\}$.
Bottom: A construction of an $(\ell+1)$-$\mathsf{Existential\, Gadget}$ using an $\ell$-$\mathsf{Existential\, Gadget}$ and a~3-$\mathsf{Existential\, Gadget}$.
The black vertices become the new terminals.
}
\label{fig:existential}
\end{figure}

\begin{lemma}\label{lem:setcover:existential-exists}
    For each $\ell \ge 3$ there exists an $\ell$-$\mathsf{Existential\, Gadget}$ $(G_\ell, \tcal_\ell)$ with $|V(G_\ell)| + |E(G_\ell)| + |\tcal_\ell| = \Oh(\ell)$.
    Furthermore, $(G_\ell, \tcal_\ell)$ can be constructed in time $\ell^{\Oh(1)}$.
\end{lemma}
\begin{proof}
    A construction of a 3-$\mathsf{Existential\, Gadget}$ $(G_3,\tcal_3)$ is given in \Cref{fig:existential}.
    We refer to this figure in the arguments below.
    We define $\tcal_3$  as $\{(u_1,v_1,1), (u_2,v_2,1) \}$.
We argue that $(G_3,\tcal_3)$ satisfies condition (3) of \cref{def:setcover:existential}.
If $F = [3]$, then each orange edge must be used by an $(s_i,t_i)$-walk in any $(\tcal_3 \cup \tcal_F)$-flow, as otherwise the walks in the flow would not be edge-disjoint. 
Let $G'_3$ be obtained from $G_3$ by removing the orange edges. 
Then $u_1,u_2,v_1,v_2$ lie on the outer face of $G'_3$ in this order.
Since the pairs $(u_1,v_1)$ and $(u_2,v_2)$ cross and each of these vertices has degree 2 in $G'_3$, no non-crossing $\tcal_3$-flow exists in $G'_3$.
Consequently, a non-crossing $(\tcal_3 \cup \tcal_F)$-flow exists in $G_3$.
On the other hand, whenever $|F| < 3$, then a non-crossing $(\tcal_3 \cup \tcal_F)$-flow exists in $G_3$: see the top right of the figure.

    Suppose now that an  $\ell$-$\mathsf{Existential\, Gadget}$ $(G_\ell, \tcal_\ell)$ with the claimed size exists.
    We show inductively how to construct an  $(\ell+1)$-$\mathsf{Existential\, Gadget}$ $(G_{\ell+1}, \tcal_{\ell+1})$.
    We build $G_{\ell+1}$ from a disjoint union of $G_\ell$ and $G_3$, insert new vertices $u,v$ on the outer face, and add edges $uG_\ell[s_\ell]$, $uG_3[t_1]$, $vG_\ell[t_\ell]$, $vG_3[s_1]$ (see \Cref{fig:existential}, bottom).
    We define $\tcal_{\ell+1}$ as a union of the requests $\tcal_{\ell}$ in $G_\ell$ and $\tcal_3$ in $G_3$, together with request $(u,v,1)$.
    The distinguished vertices of $G_{\ell+1}$ are: $G_\ell[s_1], G_\ell[t_1], \dots, G_\ell[s_{\ell-1}], G_\ell[t_{\ell-1}], G_3[s_2], G_3[t_2], G_3[s_3], G_3[t_3]$.
    There are $(\ell+1)$ pairs of them; let $\widehat\tcal$ denote a family of $\ell+1$ unitary requests, one for each pair.

    Clearly, the terminals of $G_{\ell+1}$ can be arranged on the outer face in the presented order.
    To establish condition (3) of \cref{def:setcover:existential}, we need to show that there is no non-crossing $(\tcal_{\ell+1} \cup \widehat\tcal)$-flow in $G_{\ell+1}$
    but removing any request from $\widehat\tcal$ suffices to construct the flow.
    Suppose that there exists a non-crossing $(\tcal_{\ell+1} \cup \widehat\tcal)$-flow $\pp$ in $G_{\ell+1}$.
    It contains a $(u,v)$-walk $W_{uv}$ that needs to go either through $G_\ell$ or $G_3$.
    Suppose w.l.o.g. the first scenario.
    Then $W_{uv}$ contains a $(G_\ell[s_\ell], G_\ell[t_\ell])$-subwalk $W$ in $G_\ell$.
    Let $\widehat\tcal_\ell$ be the family of $\ell-1$ requests from $\widehat\tcal$ concerning the terminals of $G_\ell$.
    Next, let $\pp_\ell$ be the non-crossing $(\tcal_\ell \cup \widehat\tcal_\ell)$-flow contained in $\pp$.
    Since no walk from $\pp_\ell$ can use edges from $E(W_{uv})$, this flow must be entirely contained in the graph $G_\ell$.
    Therefore, the non-crossing flow $\pp_\ell \cup \{W\}$ satisfies all the requests from $\tcal_\ell$ and all of the form $(G_\ell[s_i], G_\ell[t_i])$ for $i \in [\ell]$. This contradicts the assumption that $(G_\ell, \tcal_\ell)$ is an $\ell$-$\mathsf{Existential\, Gadget}$.

    Next, consider a family $\widehat\tcal'$ obtained from $\widehat\tcal$ by removal of any single request.
    We argue that there exists a non-crossing $(\tcal_{\ell+1} \cup \widehat\tcal')$-flow in $G_{\ell+1}$.
    Suppose w.l.o.g. that $\widehat\tcal'$ is missing a request concerning a pair of terminals from $G_\ell$.
    Let $\widehat\tcal'_\ell = \br{\widehat\tcal' \sm \{(G_3[s_2], G_3[t_2], 1), (G_3[s_3], G_3[t_3], 1)\}} \cup \{(G_\ell[s_\ell], G_\ell[t_\ell], 1)\}$.
    Note that $\widehat\tcal'_\ell$ has less than $\ell$ elements. 
    By the definition of an $\ell$-$\mathsf{Existential\, Gadget}$, there exists a non-crossing  $(\tcal_\ell \cup \widehat\tcal'_\ell)$-flow $\pp_\ell$ in $G_\ell$.
    Next, let $\widehat\tcal'_3 = \{(G_3[s_2], G_3[t_2], 1), (G_3[s_3], G_3[t_3], 1)\}$.
    Again by the definition, there exists a non-crossing  $(\tcal_3 \cup \widehat\tcal'_3)$-flow $\pp_3$ in $G_3$.
    We take the union of $\pp_\ell$ and $\pp_3$ and extend the $(G_\ell[s_\ell], G_\ell[t_\ell])$-walk from $\pp_\ell$ with edges   $uG_\ell[s_\ell]$,  $vG_\ell[t_\ell]$, so it becomes a $(u,v)$-walk.
    This forms a non-crossing $(\tcal_{\ell+1} \cup \widehat\tcal')$-flow in $G_{\ell+1}$.

    We have thus established that $(G_{\ell+1}, \tcal_{\ell+1})$ is indeed an $(\ell+1)$-$\mathsf{Existential\, Gadget}$.
    In the inductive step we increase the size of the graph and the number of requests by $\Oh(1)$, so the claimed bound holds.
    The construction of $(G_\ell, \tcal_\ell)$ can be easily performed in time polynomial in size of~$G_\ell$.
\end{proof}

Suppose we want to encode an instance $(k,\scal,\ell)$ of \textsc{Set Cover} with $|\scal| = 2^r$.
Here is the first (naive) attempt. We make $k$ copies of an $\ell$-$\mathsf{Existential\, Gadget}$, $\ell$ copies of an $(r,k,\scal)$-$\mathsf{Subset\, Gadget}$ and, for each $i \in [k]$, $j \in [\ell]$, we add terminals $u_{i,j}$, $v_{i,j}$, connected to the $j$-th pair of terminals in the $i$-th existential gadget  and the $i$-th pair of terminals in the $j$-th subset gadget.
For each created pair $u_{i,j}$, $v_{i,j}$, we demand a single unit of flow between $u_{i,j}$ and $v_{i,j}$.
By condition (3) of \cref{def:setcover:existential}, for each $i \in [k]$ there needs to be at least one $j \in [\ell]$ for which the $(u_{i,j}, v_{i,j})$-walk goes through the $j$-th subset gadget.
Next, condition (3) of \cref{def:subset:gadget} implies that for each $j \in [\ell]$ the set of such indices $i$ 
forms a subset of some set from $\scal$. 
The problem is that already for $\ell = k = 3$ such a graph contains $K_{3,3}$ as a minor, so it cannot be planar.
To circumvent this issue, we need yet another gadget to allow the links between each $i$-th existential gadget and each $j$-th subset gadget to cross.
Since the number of such links is only $k\cdot \ell$, we can afford adding $\Oh(1)$ new requests to implement such a crossing in a planar fashion.

\begin{figure}
    \centering
\includegraphics[scale=0.8]{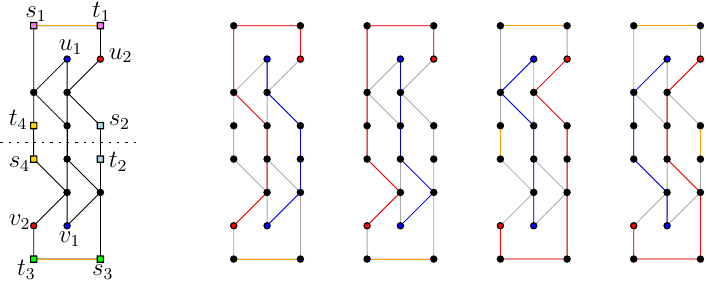} 
\caption{An illustration for \cref{lem:setcover:junction-exists}: a $\mathsf{Junction\, Gadget}$. The graph $G$ is on the left and $\tcal$ is given as $\{(u_1,v_1,1), (u_2,v_2,1) \}$. The four flows on the right demonstrate that whenever $F$ excludes one of $1,3$ and one of $2,4$, then we can construct a non-crossing $(\tcal \cup \tcal_F)$-flow.
}
\label{fig:junction}
\end{figure}

\begin{lemma}\label{lem:setcover:junction-exists}
    There exists a pair $(G,\tcal)$ (called a $\mathsf{Junction\, Gadget}$) with the following properties.
\begin{enumerate}
    \item $G$ is a plane graph with $8$ distinguished vertices $s_1,t_1, s_2, t_2, s_3, t_3, s_4,t_4$ lying on the outer face in this clockwise order.
    \item $\tcal$ is a set of triples from $V(G) \times V(G) \times \{1\}$. 
    \item For $F \sub [4]$, let $\tcal_F = \{(s_i,t_i, 1) \mid i \in F\}$.
    Then, there exists a non-crossing $(\tcal \cup \tcal_F)$-flow in $G$ if and only if 
    $\{1,3\} \not\sub F$ and $\{2,4\} \not\sub F$.
\end{enumerate}
\end{lemma}
\begin{proof}
    The graph $G$ is depicted in \Cref{fig:junction}.
    We refer to this figure in the arguments below.
    The family $\tcal$ is given as $\{(u_1,v_1,1), (u_2,v_2,1) \}$.
Let $F \sub [4]$.
The three edges crossing the dotted line separate $\{u_1,u_2,s_2,t_4\}$ from $\{v_1,v_2,t_2,s_4\}$; so, when $\{2,4\} \sub F$, then there can be no $(\tcal \cup \tcal_F)$-flow in $G$.
Suppose now that $\{1,3\} \sub F$ and there exists a non-crossing $(\tcal \cup \tcal_F)$-flow in $G$. 
The orange edges must be utilized by the $(s_1,t_1)$-walk and the $(s_3,t_3)$-walk, as otherwise the walks would not be edge-disjoint with the $\tcal$-flow.
Let $G'$ be obtained from $G$ by removing the orange edges. 
Then $u_1,u_2,v_1,v_2$ lie on the outer face of $G'$ in this order.
Since the pairs $(u_1,v_1)$ and $(u_2,v_2)$ cross and each of these vertices has degree 2 in $G'$, no non-crossing $\tcal$-flow exists in $G'$.
Hence we arrive at a contradiction.
The four flows on the right of the figure demonstrate that whenever $\{1,3\} \not\sub F$ and $\{2,4\} \not\sub F$, then a non-crossing $(\tcal \cup \tcal_F)$-flow exists in~$G$. 
\end{proof}

We are ready to present the proper reduction.

\begin{theorem}
\label{prop:setcover:final}
    There is a polynomial-time algorithm that, given an instance $(k,\scal, \ell)$ of \textsc{Set Cover},
    outputs an equivalent instance $(G,\tcal)$ of $\nullnoncross$ with $|\tcal| = \Oh(k^5)$.
    \mic{The demands $d_i$ for $(s_i,t_i,d_i) \in \tcal$ are bounded by $2^{\Oh(k)}$.}
\end{theorem}
\begin{proof}
    By removing duplicates in the family $\scal$, we can assume that it contains at most $2^k$ sets.
    Next, by padding the family $\scal$ with empty sets and increasing its size \micr{removed ``by''} 
    at most twice, we can assume that the size $\scal$ is a power of $2$.
    Let $r \in \nn$ be such that $|\scal| = 2^r$.
    We can solve $(k,\scal, \ell)$ in polynomial time when $r < 6$ or $\ell < 3$.
    Therefore, we can assume that $6 \le r \le k$ and $\ell \ge 3$, so we will meet preconditions of the used lemmas.
    Note that the running time of the form $2^{\Oh(r)}$ is polynomial in the input size.

    By a slight abuse of notation, we represent the set family $\scal$ as a function $\{0,1\}^r \to 2^{[k]}$.
    The choice of this function is irrelevant.
    We begin the construction of the graph $G$ by creating $k$ copies $A_1,\dots,A_k$ of an $\ell$-$\mathsf{Existential\, Gadget}$ (\cref{lem:setcover:existential-exists}) and $\ell$ copies $B_1,\dots,B_\ell$ of an $(r,k,\scal)$-$\mathsf{Subset\, Gadget}$ (\cref{prop:subset:full:final}).
    We arrange them on the plane as illustrated in \Cref{fig:setcover}.
    For each $i \in [k]$, $j \in [\ell]$, we draw two L-shaped paths, first connecting vertices $A_i[s_j]$ and $B_j[t_i]$, the second one connecting  vertices $A_i[t_j]$ and $B_j[s_i]$.
    We refer to the union of these two paths as the $(i,j)$-road.
    When the $(i_1,j_1)$-road and the $(i_2,j_2)$-road cross, we place there a copy of the $\mathsf{Junction\, Gadget}$ $J_{i_1, j_1, i_2, j_2}$ (\cref{lem:setcover:junction-exists}).
    Let $\mathcal{H}$ be the family of all constructed gadgets of three types.
    Now each road becomes divided into subroads, each connecting a pair of terminals $s_x,t_x$ from some gadget $H_1 \in \mathcal{H}$ to a pair of terminals $s_y,t_y$ from some gadget $H_2 \in \mathcal{H}$.
    We place two new vertices $u,v$ in the middle of this subroad
    and insert edges $uH_1[s_x]$, $uH_2[t_y]$, $vH_1[s_y]$, $vH_2[t_x]$.
    The edges are chosen in such a way as to avoid any edge crossings; note that the ordering of terminals $s_1,t_1, s_2, t_2,\dots$ for each gadget is clockwise around its outer face.
    For each such pair $(u,v)$, we insert a request $(u,v,1)$ to a family $\tcal_{road}$.
    For a gadget $H \in \mathcal{H}$, let $\tcal[H]$ denote the family of its internal requests.
    We define $\tcal$ as $\tcal_{road} \cup \bigcup_{H \in \mathcal{H}} \tcal[H]$.
    This finishes the construction of the instance $(G,\tcal)$.

\begin{figure}[t]
    \centering
\includegraphics[scale=0.7]{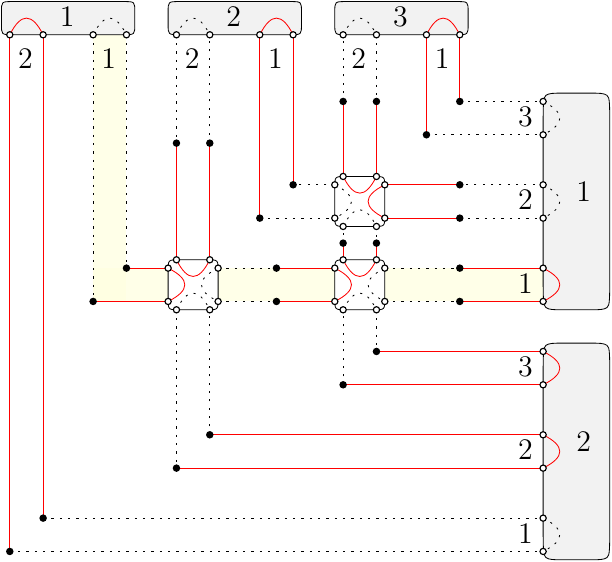}
\caption{(\Cref{fig:non-crossing-setcover} restated) A visualization of the reduction in \cref{prop:setcover:final} with $k = 3$, $\ell = 2$.
The existential gadgets are on the top and the subset gadgets are on the right.
The terminal pairs in each gadget are numbered in a clockwise manner. 
The three squares in the middle are the junction gadgets.
The (1,1)-road is highlighted.
The red $\tcal_{road}$-flow encodes a solution $S_1 = \{1\}$, $S_2 = \{2,3\}$.
}
\label{fig:setcover}
\end{figure}

    \begin{claim}
        If there exist $\ell$ vectors $\bb_1, \dots, \bb_\ell \in \{0,1\}^r$ for which $\bigcup_{i=1}^\ell \scal(\bb_i) = [k]$, then there exists a non-crossing $\tcal$-flow in $G$. 
    \end{claim}
    \begin{innerproof}
        Consider the $(i,j)$-road for $i \in [k]$, $j \in [\ell]$.
        If $i \in \scal(\bb_j)$, then for each pair $(u,v)$ of terminals on the $(i,j)$-road, we route the $(u,v)$-walk through the gadget $H \in \mathcal{H}$ that is closer to the subset gadget $B_j$.
        Otherwise, we route the  $(u,v)$-walk through the gadget closer to the existential gadget $A_i$.
        First, observe that in every junction gadget 
        $J$ we either use a connection between $J[s_1]$ and $J[t_1]$ or between $J[s_3]$ and $J[t_3]$.
        We also either use a connection between $J[s_2]$ and $J[t_2]$ or between $J[s_4]$ and $J[t_4]$.
        By condition (3) of \cref{lem:setcover:junction-exists}, such two connections can be realized via a non-crossing flow within $J$ together with $\tcal[J]$.

        Consider the existential gadget $A_i$ for $i \in [k]$.
        Let $F_i \sub [\ell]$ indicate the indices $j$ for which a $(u,v)$-walk from the $(i,j)$-road goes through $A_i$.
        Since there is at least one $j \in [\ell]$ with $i \in \scal(\bb_j)$, we obtain $|F_i| < \ell$.
        By condition (3) of \cref{def:setcover:existential}, there exists a non-crossing $(\tcal[A_i] \cup \tcal_{F_i})$-flow in $A_i$.
        Finally, consider the subset gadget $B_j$.
        Let $F'_j \sub [k]$ indicate the indices $i$ for which a $(u,v)$-walk from the $(i,j)$-road goes through $B_j$.
        By construction, we have $F'_j = \scal(\bb_j)$.
        Hence by condition (3) of \cref{def:subset:gadget}, a non-crossing $(\tcal[B_j] \cup \tcal_{F'_j})$-flow exists in $B_j$.
        The claim follows.
    \end{innerproof}

      \begin{claim}
       If there exists a non-crossing $\tcal$-flow $\pp$ in $G$, then
        there exist vectors $\bb_1, \dots, \bb_\ell \in \{0,1\}^r$ for which $\bigcup_{i=1}^\ell \scal(\bb_i) = [k]$.
    \end{claim}
    \begin{innerproof}
        Since the walks from $\pp$ are edge-disjoint, the requests from family $\tcal_{road}$ forbid any walk from family $\tcal[H]$, for $H \in \mathcal{H}$, to use edges outside the subgraph $H$.
        Therefore, for every $H \in \mathcal{H}$, the $\tcal[H]$-flow included in $\pp$ is entirely contained in $H$.

        Let $F_i \sub [\ell]$ indicate the indices $j$ for which a $(u,v)$-walk from the $(i,j)$-road goes through $A_i$.
        Similarly, let $F'_j \sub [k]$ indicate the indices $i$ for which a $(u,v)$-walk from the $(i,j)$-road goes through $B_j$.
        Consider the $(i,j)$-road for $i \in [k]$, $j \in [\ell]$.
        By the properties of a junction gadget, no two pairs $(u_1,v_1)$, $(u_2,v_2)$ of terminals located on this road can make use of a single junction gadget.
        This implies that either $j \in F_i$ or $i \in F'_j$ or possibly both conditions hold.
        By the properties of an existential gadget, for each $i \in [k]$ there exists $\tau(i) \in [\ell]$ such that $\tau(i) \not\in F_i$.
        Consequently, $i \in F'_{\tau(i)}$.
        By the properties of a subset gadget, for each $j \in [\ell]$ there exists a vector $\bb_j$ such that $F'_j \sub \scal(\bb_j)$.
        We infer that an element $i \in [k]$ is contained in the set $\scal(\bb_{\tau(i)})$.
        Therefore, a set cover of size $\ell$ exists.
    \end{innerproof}

    We have thus established that the instances $(k,\scal,\ell)$ and $(G,\tcal)$ are equivalent.
    \mic{All the requests in $\tcal$ are unitary apart from those in the subset gadget.
    \cref{prop:subset:full:final} guarantees that the demands in this gadget are bounded by $\Oh(2^{3r}) = 2^{\Oh(k)}$.}
    It remains to count the number of requests in $\tcal$.
    Each of $k$ existential gadgets requires $\Oh(\ell)$ requests.
    Each of $\ell$ subsets gadgets requires $\Oh(k r^3)$ requests.
    The number of junction gadgets equals the number of crossings between the roads,
    which is at most the number of roads squared, that is, $\Oh(k^2\ell^2)$.
    Each such gadget contains only $\Oh(1)$ requests. 
    The size of the family $\tcal_{road}$ is proportional to the number of junction gadgets, so also $|\tcal_{road}| = \Oh(k^2\ell^2)$.
    Since both $r,\ell$ are bounded by $k$, we obtain $|\tcal| = \Oh(k^5)$.
    Finally, each of the three types of gadgets can be constructed in time polynomial in the input size.
    This concludes the proof.
\end{proof}

\subsection{Implementing weights}
\label{sec:hard:weights}

As the last step, we need to get rid of large demands in the request family $\tcal$.
Our construction of the subset gadget requires demands as large as $2^{\Oh(k)}$ and we cannot afford requesting this many vertex-disjoint paths in a meaningful reduction.   
Instead, we shall implement such a request with $\Oh(r^2)$ unitary requests by utilizing the construction by Adler and Krause~\cite{AdlerLB}.

We begin by simplifying the requests so that the demands are of the form $2^i-1$ and the number of edges incident to each terminal $v$ equals the number of paths starting at $v$.
We also place additional ``guarding'' requests of demand 1 that will come in handy in further topological arguments.
This operation is depicted in \Cref{fig:binary}.

\begin{definition}[Binary simplification]
\label{def:weight:binary}
Let $(G, \tcal)$ be an instance of \textsc{Non-crossing Multicommodity Flow} and $(s,t,d) \in \tcal$.
We obtain an instance $(G', \tcal')$ from  $(G, \tcal)$ as follows.
First, we replace  the vertex $s$ (resp. $t$) in $G$ with a cycle $C_s$ (resp. $C_t)$ having a single vertex for each edge incident to $s$ (resp. $t$).
We multiply each of the edges on the cycle $|E(G)|$ times.
We pick an arbitrary vertex $v$ on $C_s$ (resp. $C_t$) and create a new vertex $s'$ (resp. $t'$) in the interior of the cycle, connected to $v$ via $d$ parallel edges.
Next, let $D \sub \nn$ be the set of 1's in the binary representation of $d$, i.e., $d = \sum_{i \in D}2^i$.
For each $i \in D$ we create vertices $v_s^i$, $u_s^i$, adjacent to $s'$, and  vertices $v_t^i$, $u_t^i$, adjacent to $t'$.
We place them in a clockwise manner around $s'$ and in an counter-clockwise manner around $t'$.
For every vertex of the form $u_s^i, u_t^i$, we multiply the only edge incident to it times $2^i-1$.
We remove $(s,t,d)$ from $\tcal$ and replace it with $\bigcup_{i \in D} \{(v_s^i, v_t^i, 1), (u_s^i, u_t^i, 2^i-1)\}$.
We say that $(G', \tcal')$ is obtained from $(G, \tcal)$ via {\em binary simplification} of $(s,t,d)$.
\end{definition}

\begin{figure}
    \centering
\includegraphics[scale=0.8]{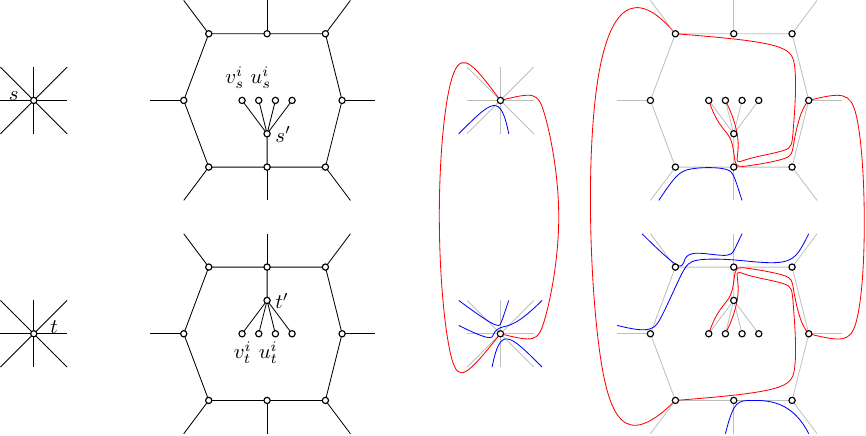} 
\caption{Left: An example of binary simplification of a request $(s,t,d)$, where $d$ has two 1's in the binary representation.
Right: A correspondence between non-crossing flows in the original \mic{multigraph} and the \mic{multigraph} after modification.
Note that the edges on the cycles are duplicated.
Two exemplary $(s,t)$-paths are drawn in red.
}
\label{fig:binary}
\end{figure}

\begin{lemma}\label{lem:weight:binary}
    Let $(G', \tcal')$ be obtained from $(G, \tcal)$ via {binary simplification} of $(s,t,d) \in \tcal$.
    Then these two instances of $\nullnoncross$ are equivalent.
\end{lemma}
\begin{proof}
    Note that by the definition of the problem, no other terminals from $\tcal$ coincide with $s, t$.
    
    Consider a non-crossing $\tcal$-flow $\pp$ in $G$ and let $\pp_{st} \sub \pp$ be the family of $d$ walks connecting $s$ with $t$.
    Since the vertices of the cycle $C_s$ (resp. $C_t$) are connected via $|E(G)|$ parallel edges and each walk from $\pp$ visiting $s$ or $t$ must use some edge incident to this vertex, there is enough space to route all the walks  from $\pp \sm \pp_{st}$ 
    alongside the cycle.
    The walks from $\pp_{st}$ can be translated into $(V(C_s), V(C_t))$-walks in $G'$.
    Let us order them as $P_1, \dots, P_d$ in a clockwise manner around $C_s$, starting from an arbitrary one.
    Then  $P_1, \dots, P_d$ arrive at $C_t$ in an counter-clockwise manner.
    By \cref{def:reduction:flow} of a non-crossing flow, no two walks $P_i, P_j$ cross with any other walk from $\pp$ at $C_s$ (resp. $C_t$).
    Therefore we can route them using the innermost edges on the cycles $C_s, C_t$ to reach $s'$ and $t'$ in an order reflecting the terminal pairs $(v_s^i, v_t^i)$ and $(u_s^i, u_t^i)$ (see \Cref{fig:binary}).

    Consider now a  non-crossing $\tcal'$-flow $\pp'$ in $G'$.
    Let $\pp'_{st} \sub \pp'$ be the family of walks realizing the requests created due to  {binary simplification} of $(s,t,d)$.
    We contract the vertex set $V(C_s)$ together with all the vertices lying inside $C_s$ into a single vertex $s$, and similarly for $C_t$, thus obtaining the graph $G$ again.
    By \cref{obs:reduction:contract}, this operation transforms $\pp'$ into a non-crossing flow in $G$.
    \mic{Because there are exactly $d$ parallel edges between $s'$ and the fixed vertex on $C_s$ (resp. $t'$ and the fixed vertex on $C_t$),
    no paths from $\pp' \sm \pp'_{st}$ can visit $s'$ or $t'$.}
     The only terminals contracted into $s$ (resp. $t$) are the ones created due to binary simplification of $(s,t,d)$, so we obtain a non-crossing $\tcal$-flow.
\end{proof}

The binary simplification imposes a convenient structure on the multigraph, which we will utilize next to replace each request of demand $2^i-1$ with just $i$ unitary requests.
We will analyze the reduction using the following artificial problem.

\begin{definition}[{\cite[Def. 3]{AdlerLB}}]
\label{def:weight:dpp}
    Given a subset $X$ of the plane and a set of $k$ pairs of terminals
$\tcal \sub X^2$, the {\sc Topological Disjoint Paths} problem is to determine whether there exist $k$ pairwise disjoint curves in $X$, such that each
curve $P_i$ is homeomorphic to $[0, 1]$ and its ends are $s_i$ and $t_i$ where $(s_i, t_i) \in \tcal$. 
\end{definition}

We will use this problem only for the sake of analysis, so we do not have to specify how the set $X$ is encoded.
For $k \in \nn$ we define an instance $(X_k, \tcal_k)$ of {\sc Topological Disjoint Paths}.
This is a concise version of~\cite[Definition 4]{AdlerLB}, where the set $X_k$ is called a {\em disc-with-edges}.
See \Cref{fig:topological-dpp} for an illustration.

\begin{figure}
    \centering
\centerline{\includegraphics[scale=0.43]{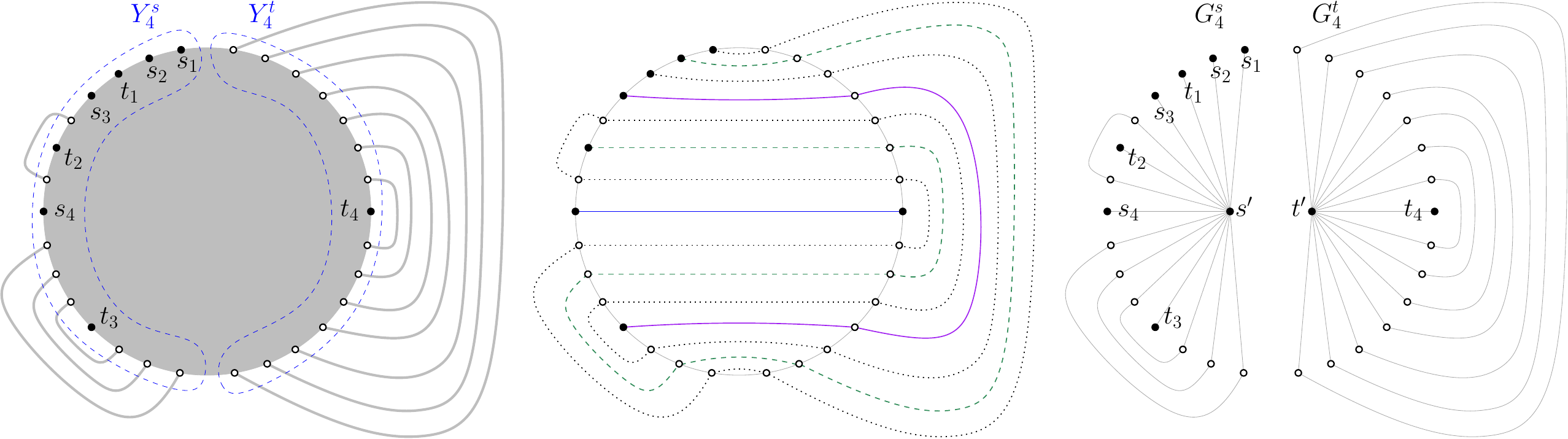}} 
\caption{Left: The instance $(X_4, \tcal_4)$ of {\sc Topological Disjoint Paths}.
The set $X_4$ is gray and the points from $\tcal_4$ are black discs.
The black and hollow discs form the set $Y_4$, which is divided into $Y_4^s$ and $Y_4^t$.
Middle: The unique solution to $(X_4, \tcal_4)$ traverses the disc $2^4-1$ times from left to right.
Right: The gadgets $G_4^s$ and $G_4^t$.
}
\label{fig:topological-dpp}
\end{figure}

\begin{definition}\label{def:weight:disc}
    Let $L_k$ be an ordered list of elements from $\{s_1,t_1,s_2,t_2,\dots,s_k,t_k\}$ defined inductively.
    For $k=1$ we set $L_1 = (s_1,t_1)$.
    Assume that $L_k$ is already constructed with $t_k$ being its last element.
    We obtain $L_{k+1}$ from $L_k$ by inserting $s_{k+1}$ before $t_k$ and inserting $t_{k+1}$ after $t_k$.
    \mic{We define $\tcal_k = \{(s_i,t_i) \mid i \in [k]\}$.}

    Now we construct the set $X_k \sub \rr^2$. 
    Let $D \sub \rr^2$ be a closed disc.
    We place the elements of $L_k$ on the boundary of $D$ in the counter-clockwise order. 
    For $i \in [k]$ let $S_i,T_i$ be the connected components of $\partial D \sm L_k$ neighboring $t_i$.
    For each  $i \in [k]$ let $E_i$ be a family of $2^{i-1}-1$ curves connecting $S_i$ and $T_i$ outside $D$, in such a way that the are no crossings in $\bigcup_{i=1}^k E_i$. 
    We set $X_k = D \cup \bigcup_{i=1}^k E_i$ and
    refer to the curves from $E_1,\dots,E_k$
    as the edges of $X_k$.

    Next,
    let $Y_k \sub X_k \cap \partial D$ be the union of the set $L_k$
    and all the endpoints of edges in $X_k$.
    We set $Y_k^t = Y_k \cap (E_k \cup \{t_k\})$ and
    $Y_k^s = Y_k \sm Y_k^t$.
\end{definition}

Note that $2^{1-1}-1 = 0$ so $E_1 = \emptyset$.
The sets $Y_k^s, Y_k^t$ divide the distinguished points in $X_k$ into the left side and the right side.
Observe that $|Y_k^s| = |Y_k^t| = 2^k - 1$.
The instance $(X_k, \tcal_k)$ is designed to have a unique solution, depicted in \Cref{fig:topological-dpp}.
In this solution, the $(s_i,t_i)$-path must ``go around'' the $(s_{i-1},t_{i-1})$-path, thus traversing the disc twice as many times.

\begin{lemma}[{\cite[Lem. 1, Thm. 4]{AdlerLB}}]
\label{lem:weight:unique}
    For each $k \in \nn$ the instance $(X_k, \tcal_k)$ of {\sc Topological Disjoint Paths} has a unique solution (up to homeomorphism). 
    The curves in this solution contain $2^k - 1$ subcurves connecting a point in $Y_k^s$ to a point in $Y_k^t$.
\end{lemma}

Adler and Krause~\cite{AdlerLB} used this lemma to construct an instance of {\sc Planar $k$-Disjoint Paths} with a $(2^k-1)\times(2^k-1)$-grid, in which every vertex is used by the unique solution.
This constitutes an example of a large-treewidth instance in which no vertex is irrelevant.

We want to turn the instance $(X_k, \tcal_k)$ into two gadgets that can be plugged into a plane multigraph, enforcing a flow of size $2^k-1$ between the gadgets by using only $k$ terminal pairs.
To this end, we need to translate the topological structure of the disc-with-edges into a graph structure.

\begin{definition}
    Consider  the instance $(X_k, \tcal_k)$ of {\sc Topological Disjoint Paths}.
    We define a gadget $G_k^s$ by placing a vertex $s'$ in the interior of $D$, inserting edges between $s'$ and each element in $Y_k^s$, and adding the edges of $X_k$ with both endpoints in $Y_k^s$.
    Analogously we obtain a gadget $G_k^t$ by placing a vertex $t'$ in the interior of $D$, inserting edges between $t'$ and each element in $Y_k^t$, and adding the edges of $X_k$ with both endpoints in $Y_k^t$.
    We refer to the vertices $s', t'$ as the roots of the gadgets. 
\end{definition}

The gadgets are depicted in \Cref{fig:topological-dpp}.
By a slight abuse of notation, we will treat
\mic{$s_i,t_i$ as vertices from $G_k^s \cup G_k^t$ and $\tcal_k$ as a set of pairs of vertices.} 
\meir{Maybe say this in the def itself? \mic{I prefer to keep it as it is. The def of $\tcal_k$ does not change but the points $s_i,t_i$ now refer to vertices}}

\begin{definition}[Weight implementation]
\label{def:weight:implement}
    Let $(G, \tcal)$ be an instance of \textsc{Non-crossing Multicommodity Flow} and $(s,t,d) \in \tcal$ be such that $d = 2^i-1$, $i > 0$, and each of $s, t$ has exactly $d$ incident edges, which are parallel.
    We obtain an instance $(\widehat G, \widehat \tcal)$ from  $(G, \tcal)$ as follows.
    Let $s',t'$ be the only neighbors of $s,t$, respectively, in $G$.
    We replace $s$ with the gadget $G_i^s$ rooted at $s'$ and replace $t$ with the gadget $G_i^t$ rooted at $t'$.
    We replace  $(s,t,d)$ in $\tcal$ with the set $\{(s_j,t_j,1) \mid (s_j,t_j) \in \tcal_i\}$. \meir{$s,t$ overloaded. In ${\cal T}_i$, they have a diff name. \mic{added a subscript}}
\end{definition}

\begin{figure}
    \centering
\centerline{\includegraphics[scale=0.5]{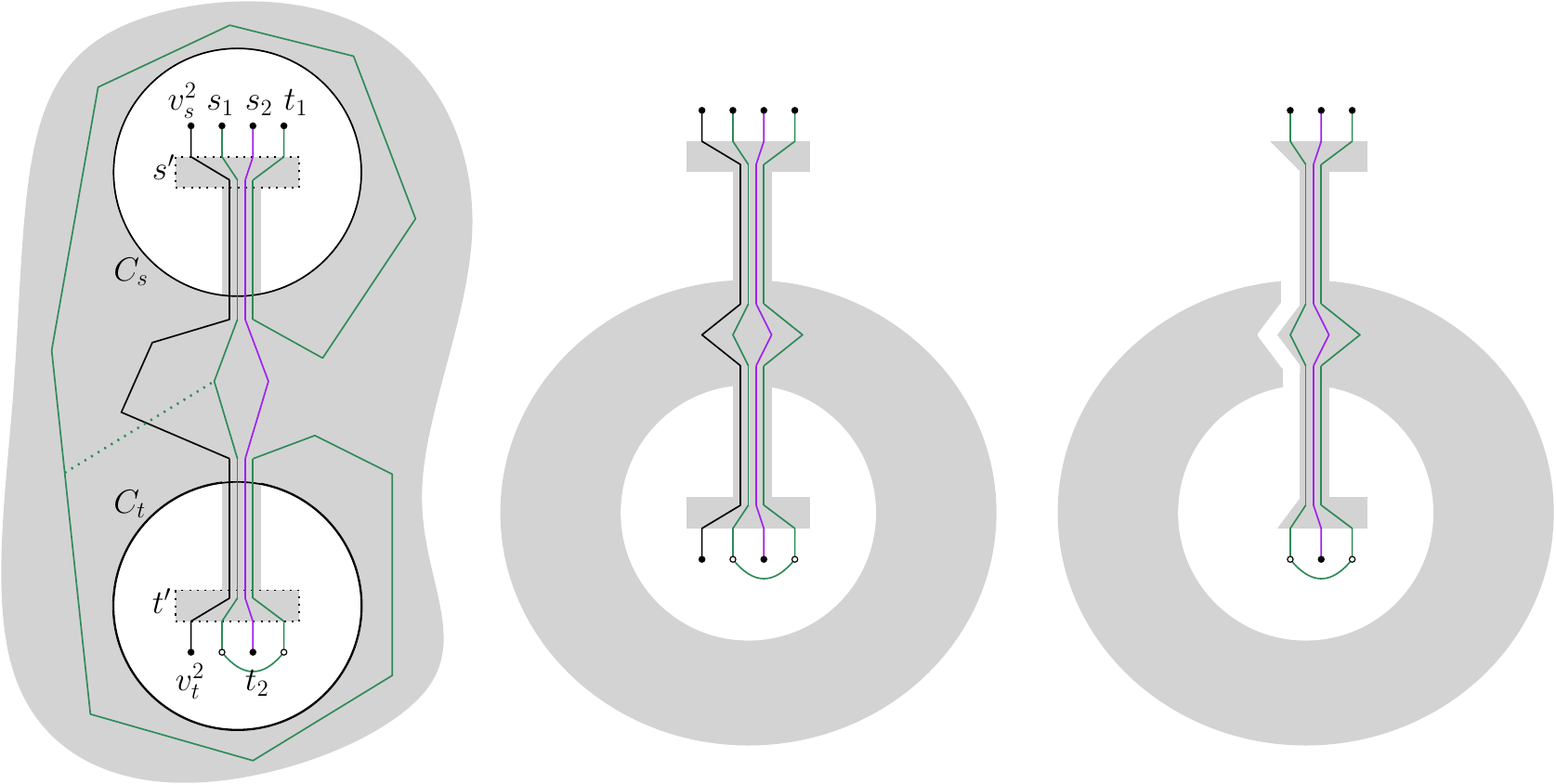}} 
\caption{An illustration to the proof of \cref{lem:weight:implement} for $d=4, i=2$. 
Left: A multigraph after binary simplification and weight implementation. 
The vertices $s,t$ got replaced by cycles $C_s, C_t$.
The vertices $s', t'$ and the incident parallel edges are part of the gray area.
Next, the vertices $u_s^2, u_t^2$ and request $(u_s^2, u_t^2, 3)$ got replaced with the gadgets $G_2^s, G_2^t$ and requests   $(s_1, t_1, 1)$, $(s_2, t_2, 1)$.
The crux of the lemma is to show that the green walk cannot use the dotted shortcut and it needs to enter the cycle $C_t$.
Middle: After flipping the embedding we can assume that $s'$ lies on the outer face.
Right: After removing the $(v_s^2, v_t^2)$-walk, the remaining gray area becomes a topological disc.
This reduces the analysis to the instance $(X_2, \tcal_2)$ where the unique solution traverses $2^i-1 = 3$ times between $C_s$ and $C_t$.
}
\label{fig:weights}
\end{figure}

\mic{In order to prove correctness of this transformation, we need to show that $\tcal_i$-walks in a non-crossing $\widehat\tcal$-flow in $\widehat G$ traverse $2^i-1$ many times between the gadgets $G_i^s, G_i^t$.
We will take advantage of the ``guarding'' request $(v_s^i, v_t^i,1)$
to reduce the analysis to the case where the $\tcal_i$-flow traverses a topological disc.
Then we could apply \cref{lem:weight:unique} to reveal the structure of the flow.}

\begin{lemma}\label{lem:weight:implement}
    Let $(G', \tcal')$ be obtained from $(G, \tcal)$ via {binary simplification} of $(s,t,d) \in \tcal$ and $i \in \nn$ belong to the binary representation of $d$.
    Next, let $(\widehat G, \widehat \tcal)$ be obtained from $(G', \tcal')$ by applying the transformation from \cref{def:weight:implement} to the request $(u_s^i, u_t^i, 2^i - 1) \in \tcal'$.
    Then the instances $(G', \tcal')$ and $(\widehat G, \widehat \tcal)$ are equivalent.
\end{lemma}
\begin{proof}
    First, consider a non-crossing $\tcal'$-flow $\pp'$ in $G'$ and the family $\pp'_u \sub \pp'$ realizing the $(u_s^i, u_t^i)$-paths.
    They all must visit $s'$ and $t'$.
    Because $(v_s^i, v_t^i, 1) \in \pp'$, there is a $(v_s^i, v_t^i)$-path $P_v$ in $\pp'$, also visiting $s'$ and $t'$.
    Since $P_v$ is non-crossing with $\pp'_u$, the order in which the paths from $\pp'_u$ enter $u_s^i$ is symmetric to the order in which they enter $u_t^i$.
    Therefore, $\pp'_u$ can be translated into a non-crossing $\tcal_i$-flow.
    
    Next,  consider a non-crossing $\widehat\tcal$-flow $\widehat\pp$ in $\widehat G$ and the family $\widehat\pp_u \sub \widehat\pp$ being a $\tcal_i$-flow.
    We need to show that $\widehat\pp_u$ contains $2^i - 1$ subwalks connecting $s'$ and $t'$.
    We can flip the embedding of $\widehat G$ to make the face incident to $s'$ the outer face, without changing the rotation system; this transformation preserves the property of being a non-crossing flow (see \Cref{fig:weights}).
    Let $D_G \sub \rr^2$ be the subset of the plane without the outer face of  $\widehat G$ and the face incident to $t'$.
    Due to the existence of a $(v_s^i, v_t^i)$-path $P_v$ in $\widehat\pp$, which is non-crossing with $\widehat\pp_u$ and has endpoints of degree 1, the flow $\widehat\pp_u$ can be drawn as a family of non-crossing curves in $D_G \sm P_v$.
    This set is a topological disc, i.e.,
    there is a homotopy that translates $(D_G \sm P_v, \tcal_i)$ into $(X_i, \tcal_i)$.
    Hence \cref{lem:weight:unique} implies that $\widehat\pp_u$ is equivalent to the unique solution to the {\sc Topological Disjoint Paths} instance  $(X_i, \tcal_i)$, and so it contains $2^i - 1$ walks connecting $s'$ and $t'$.
    This concludes the proof.
\end{proof}

We are ready to summarize our reduction.
Recall that an instance  $(G, \tcal)$ of \nullnoncross is called unitary if every demand in $\tcal$ is 1 and every terminal occurring in $\tcal$ has degree 1 in $G$.
Observe that the newly created requests are unitary and the new terminals are of degree 1.
Hence applying both transformations to all original requests yields a unitary instance.

\begin{proposition}\label{prop:weight:final}
    Let $(G, \tcal)$ be an instance of $\nullnoncross$ and $\ell\in\nn$ be such that $d_i \le 2^\ell$ for each $(s_i, t_i, d_i) \in \tcal$.
    Then in polynomial time we can transform $(G, \tcal)$ into an equivalent unitary instance $(\widehat G, \widehat\tcal)$ of $\nullnoncross$ satisfying $|\widehat\tcal| = \Oh(\ell^2) \cdot |\tcal|$.
\end{proposition}
\begin{proof}
    We first apply binary simplification to all requests in $\tcal$, thus increasing the number of requests by the factor of $\Oh(\ell)$.
    By \cref{lem:weight:binary}, the obtained instance $(G',\tcal')$ is equivalent to $(G,\tcal)$.
    Next, we apply weight implementation to each non-unitary request in $\tcal'$, obtaining a unitary instance $(\widehat G, \widehat\tcal)$.
    Again, the number of requests is being multiplied by $\Oh(\ell)$.
    The instance $(\widehat G, \widehat\tcal)$ is equivalent to $(G',\tcal')$ due to \cref{lem:weight:implement}. 
\end{proof}

\mic{We remark that almost all the requests used in our reduction from {\sc Set Cover} have demands being powers of two.
The only exceptions are the requests of type (F) in \cref{sec:subset-full}.
While it is possible to replace each of them with $r$ requests 
of demand $2^r$ (without increasing the total number of requests significantly) and shave off a single $\Oh(\ell)$-factor in weight implementation, we have chosen not to further complicate the description of the already complex gadget.
Besides, we think that 
\cref{prop:weight:final} in its general form might find applications also outside our context.}

\subsection{From a non-crossing flow to disjoint paths}
\label{sec:hard:final}

Recall that for a vertex set $X \sub V(G)$, a set of pairs $\tcal \sub X^2$ is called realizable if there exists a $\tcal$-linkage (that is, a $\tcal$-family of vertex-disjoint paths) in $G$.
Also recall the notion of a proper embedding from \cref{sec:prelims} and the property of being cross-free from \cref{sec:tw:single-face}. 

\begin{lemma}
\label{lem:weight:gadget}
Let $I$ be a noose and $X \sub I$ be a finite set of size $k$.
There exists a subcubic plane graph $H$ 
properly embedded in $\disc(I)$
such that $H \cap I = X$, $\deg_H(x) \le 2$ for every $x \in X$, and every cross-free $\tcal \sub X^2$ is realizable in $H$.
This graph can be constructed in time polynomial in~$k$.
\end{lemma}
\begin{proof}
    The claim is trivial for $k \le 2$.
    For $k \ge 3$ we use a construction similar to 
    that from \cref{lem:single:gadget} but instead of a $(k,k)$-cylindrical grid (that have vertices of degree 4) we use a $k$-cylindrical wall (with $k$ cycles and $k$ edges between each pair of consecutive cycles; see \Cref{fig:non-crossing-setcover} on page \pageref{fig:non-crossing-setcover} with 3 cycles) and identify the degree-2 vertices on the outer face with $X$.
    A $k$-cylindrical wall is a subcubic graph. 
    The argument that every cross-free $\tcal \sub X^2$ is realizable is the same as for the cylindrical grid and follows from the criterion from
    \cref{lem:single:criterion}.  
\end{proof}

Armed with this gadget, we can assume that a given graph is simple and subcubic. 

\begin{lemma}\label{lem:weight:subcubic}
    There is a polynomial-time algorithm that, given a unitary instance $(G, \tcal)$ of $\nullnoncross$, transforms it into an equivalent unitary instance $(G', \tcal')$ such that $G'$ is simple, subcubic, and $|\tcal'| = |\tcal|$.
\end{lemma}
\begin{proof}
    By the definition of a unitary instance, when $v \in V(G)$ occurs as a terminal in $\tcal$ then $\deg_G(v) = 1$.
    Therefore it suffices to reduce the degrees of non-terminal vertices.
    Let $v \in V(G)$ be a vertex of degree $k = \deg_G(v) \ge 2$; then $v$ does not appear in $\tcal$.
    We draw a noose $I$ around $v$, intersecting each edge from $E_G(v)$ once, and no more edges or vertices from $G$.
    Let $X$ be the set of intersections of $I$ with  $E_G(v)$.
    We replace $v$ with a gadget from \cref{lem:weight:gadget}, creating a subcubic subgraph $H$ properly embedded in $\disc(I)$.
    Since $\deg_H(x) \le 2$ for for every $x \in X$, the degree of $x$ in $G'$ becomes 3.
    
    We argue that this transformation yields an equivalent instance.
    Consider a non-crossing $\tcal$-flow $\pp$ in $G$.
    Let $\tcal_v \sub X^2$ be the set of pairs representing
    pairs of consecutive edges from $E_G(v)$ traversed by 
 walks from $\pp$ (recall that no path from $\pp$ has $v$ as an endpoint).
    Then $\tcal_v$ is cross-free with respect to $H$.
    By \cref{lem:weight:gadget}, there exists a $\tcal_v$-linkage in $H$.
    Since vertex-disjoint paths are clearly non-crossing, this allows us to transform $\pp$ into  a non-crossing $\tcal$-flow in $G'$.
    On the other hand, when a non-crossing $\tcal$-flow exists in $G'$, it can be turned into a non-crossing $\tcal$-flow in $G$ by simply contracting $H$ back to a single vertex (see \cref{obs:reduction:contract}).

    We apply this modification to every vertex in $G$ with degree at least 2, thus creating a subcubic graph.
    Note that every pair of parallel edges in $G$ connects two vertices of degree at least 2, and during the transformation their endpoints become distinct.
    Hence the outcome is also a simple graph.
\end{proof}

In a subcubic graph, the notions of (non-crossing) edge-disjointness and vertex-disjointness coincide, so we can treat an instance of \nullnoncross as an instance of {\sc Planar (Edge-)Disjoint Paths}.
Note that in the first problem we assume that a plane embedding is provided in the input, while in the last two we do not.
However, once the reduction is done, we can discard the fixed embedding.

\begin{theorem}\label{thm:weight:setcover}
    There is a polynomial-time algorithm that, given an instance $(k,\scal, \ell)$ of \textsc{Set Cover},
    outputs an equivalent instance $(G,\tcal)$ of {\sc Planar (Edge-)Disjoint Paths} with $|\tcal| = \Oh(k^7)$.
\end{theorem}
\begin{proof}
    \cref{prop:setcover:final} allows us to 
    compute an instance $(G,\tcal)$ of \nullnoncross with $|\tcal| = \Oh(k^5)$ that is equivalent to $(k,\scal, \ell)$.
    The demands in $(G,\tcal)$ are bounded by $2^{\Oh(k)}$.
    Next, we use \cref{prop:weight:final} to transform $(G,\tcal)$ into an equivalent unitary instance $(\widehat G, \widehat \tcal)$ with $|\widehat\tcal| = \Oh(k^7)$.
    Subsequently, $(\widehat G, \widehat \tcal)$ can be transformed into an equivalent unitary instance $(G',\tcal')$ such that $G'$ is simple, subcubic, and $|\tcal'| = |\widehat\tcal|$ (\cref{lem:weight:subcubic}).
    Let $\tcal_{DP} = \{(s,t) \mid (s,t,1) \in \tcal'\}$.
    In a subcubic graph with all terminals of degree 1, every walk in a solution is a path and two paths are edge-disjoint if and only if they are vertex-disjoint. 
    In turn, vertex-disjointness implies being non-crossing.
    Hence the following statements are equivalent.
    \begin{enumerate}[nolistsep]
        \item There is a non-crossing $\tcal'$-flow in $G'$.
        \item There is a $\tcal_{DP}$-family of vertex-disjoint paths in $G'$.
        \item There is a $\tcal_{DP}$-family of edge-disjoint paths in $G'$.
    \end{enumerate}
    As a consequence, $(G',\tcal_{DP})$ is a yes-instance of {\sc Planar Disjoint Paths} (resp. {\sc Planar Edge-Disjoint Paths}) if and only if $(G',T')$ is a yes-instance of \nullnoncross.
    This concludes the proof.
\end{proof}

The {\sc Set Cover} problem parameterized by the universe size is known not to admit a polynomial kernel unless  coNP $\subseteq$ NP/poly~\cite{DomLS09}.
\cref{thm:weight:setcover} implies the same for {\sc Planar (Edge-)Disjoint Paths} parameterized by the number of the requests.
Since  {\sc Set Cover} under this parameterization is also
WK[1]-complete~\cite{hermelin2015completeness},
we establish WK[1]-hardness of {\sc Planar (Edge-)Disjoint Paths} as well.
As a consequence, we obtain
Theorems~\ref{thm:noPolyKer}, \ref{thm:MK2}, and \ref{thm:noPolyKerEdge}.

\section{Conclusion}\label{sec:conclusion}

We conclude the paper with several open questions. Possibly, the ideas used in this paper will be useful in solving some of these questions.
\mic{In particular, we believe that our construction of a $2^{\Oh(k^2)}\cdot n^{\Oh(1)}$-time algorithm for {\sc Planar Disjoint Paths}, based on the irrelevant edge rule, could be easier to generalize to bounded-genus graphs than the approach based on enumerating homotopy classes of~\cite{cho2023parameterized, LokshtanovMPSZ20}.
It is unclear, however, how to extend the treewidth reduction procedure~\cite{AdlerKKLST17} and the $n^{\Oh(k)}$-time algorithm~\cite{schrijver1994finding}.
This leads to the first open question.}

\begin{itemize}
    \item Can {\sc Disjoint Paths} on graph classes substantially larger than the class of planar graphs (such as \mic{bounded-genus or} minor-free graphs) be solved in time $2^{k^{\Oh(1)}}\cdot n^{\Oh(1)}$, and does it admit a polynomial kernel when parameterized by $k+\mathsf{tw}$? Currently, the best known running time for \mic{proper minor-closed graph classes} is galactic as in the general case. 

    \item Can {\sc Planar Disjoint Paths} be solved in time \mic{$2^{o(k^2)}\cdot n^{\Oh(1)}$} or even $2^{\Oh(k)}\cdot n^{\Oh(1)}$? 
    We remark that the existing NP-hardness proof for {\sc Planar Disjoint Paths}~\cite{kramer1984complexity} only implies that the problem is not solvable in time $2^{o(\sqrt{k})}\cdot n^{\Oh(1)}$ unless the ETH is false.
    \mic{Note that even though  $2^{o(\sqrt{k})}$ may seem a natural parameter dependency for a ``genuinely planar'' problem, for the related {\sc Planar Steiner Tree} problem a $2^{o({k})}\cdot n^{\Oh(1)}$ lower bound is known~\cite{MarxPP18}.}
   
    \item Can the extension of {\sc Planar Disjoint Paths} to directed graphs be solved in time $2^{k^{\Oh(1)}}\cdot n^{\Oh(1)}$, and does it admit a polynomial kernel when parameterized by $k+\mathsf{tw}$? Currently, the best known running time is $2^{2^{k^{\Oh(1)}}}\cdot n^{\Oh(1)}$~\cite{DBLP:conf/focs/CyganMPP13}.   
    
    \item While {\sc Disjoint Paths} is long known to be in FPT, the existing algorithms are very complicated. Can one design a ``simple'' $f(k)\cdot n^{\Oh(1)}$-time (or even $n^{f(k)}$-time) algorithm for this problem?   
    
    \item Being a close relative of {\sc Disjoint Paths}, can {\sc Topological Minor Testing} on the class of planar graphs be solved in time $2^{k^{\Oh(1)}}\cdot n^{\Oh(1)}$? Currently, the best known running time is $2^{2^{2^{k^{\Oh(1)}}}}\cdot n^{\Oh(1)}$~\cite{DBLP:conf/stoc/FominLP0Z20}. Note that the question is not asked for {\sc Minor Testing} since it becomes ``easy'' on planar graphs (due to reasons that do not apply to {\sc Disjoint Paths} and {\sc Topological Minor Testing} on planar graphs, and not to {\sc Minor Testing} in general)~\cite{adler2012fast}.
    
    \item The {\sc Min-Sum Disjoint Paths} problem is the optimization version of {\sc Disjoint Paths} where we do not only need to determine whether a solution exists, but, if the answer is positive, find one where the sum of the lengths of the solution paths is minimized. The {\sc Shortest Disjoint Paths} is a restricted case of this problem where the task is to determine whether there exists a solution where every path is a shortest one between its endpoints. Although introduced more than 20 years ago~\cite{DBLP:journals/dam/Eilam-Tzoreff98}, to date, all we know on {\sc Shortest Disjoint Paths} is that it is W[1]-hard~\cite{DBLP:conf/soda/Lochet21} (and hence also {\sc Min-Sum Disjoint Paths} is W[1]-hard), in XP~\cite{DBLP:conf/soda/Lochet21}, and, on digraphs, in P when $k=2$~\cite{BercziK17} (in contrast to {\sc Disjoint Paths}). The status of {\sc Min-Sum Disjoint Paths} is grimmer: all we know is that it is in P when $k=2$~\cite{DBLP:journals/siamcomp/BjorklundH19}, and, on digraphs, it is NP-hard even when $k=2$ (because it generalizes {\sc Disjoint Paths}). Specifically, we ask: Is {\sc Shortest Disjoint Paths} (or even {\sc Min-Sum Disjoint Paths}) on planar graphs in FPT? Is {\sc Min-Sum Disjoint Paths} (on undirected graphs) in XP? Is {\sc Shortest Disjoint Paths} on digraphs in XP? 
    
    \item Does {\sc Disjoint Paths} admit a polynomial kernel when restricted to chordal graphs? Currently, a positive answer is known for split graphs~\cite{heggernes2015finding,yang2018kernelization}, and, more generally, well-partitioned chordal graphs~\cite{ahn2020well}.

    \item For which other problems that admit a treewidth reduction can we show the impossibility of a polynomial treewidth reduction? We refer to \cite{DBLP:conf/stoc/FominLP0Z20,golovach2013obtaining,DBLP:journals/jcss/Grohe04,DBLP:conf/stoc/GroheKMW11,DBLP:conf/soda/LokshtanovP0SZ18,marx2013finding} for \mic{some} examples of problems other than {\sc Disjoint Paths} and {\sc Minor Testing} that are known to admit \mic{super-polynomial} treewidth reductions.
\end{itemize}
 Lastly, we remark that it might be interesting to study lossy kernels and FPT approximation algorithms~\cite{DBLP:journals/algorithms/FeldmannSLM20} for optimization versions of the above-mentioned problems in future works.  

\bibliographystyle{plainurl}
\bibliography{main}

\end{document}